%% file: main.tex
\documentclass[12pt]{report}

\usepackage{geometry}
\usepackage{utepcsthesis}
\usepackage{titling} 
\usepackage[x11names]{xcolor} 

\usepackage{amssymb,amsmath}
\usepackage{amsthm} 
\usepackage{multirow} 
\usepackage{relsize} 
\usepackage[ruled]{algorithm2e}
\usepackage{multibib}
\usepackage{txfonts}

\usepackage{url}
\makeatletter
\g@addto@macro{\UrlBreaks}{\UrlOrds}
\makeatother

\usepackage{listings}

\usepackage{subcaption}
\usepackage{graphicx}
\usepackage{wrapfig}
\usepackage{adjustbox}

\newtheorem{thm}{Theorem}

\newtheorem{cor}{Corollary}

\usepackage{titlesec}
\usepackage{indentfirst}

\titleformat{\section}{\normalfont\bfseries}{\thesection}{1em}{}
\titleformat{\subsection}{\normalfont\bfseries}{\thesubsection}{1em}{}
\titleformat{\subsubsection}{\normalfont\bfseries}{\thesubsubsection}{1em}{}
\titlespacing*{\section}{0pt}{0pt}{0pt}
\titlespacing*{\subsection}{0pt}{0pt}{0pt}
\titlespacing*{\subsubsection}{0pt}{0pt}{0pt}

\usepackage{setspace}
\usepackage{tocloft}

\cftsetindents{chapter}{0em}{5em}

\renewcommand{\headrulewidth}{0pt}
\begin{document}

\title{
  REAL-TIME CONTROL OVER WIRELESS NETWORKS
}
\author{VENKATA PRASHANT MODEKURTHY}
\date{2020}


\CommitteeChair{Advisor \hspace{2in} Date}
\CommitteeMembers{}{}{}{}

\maketitlepage

%
%
%

\newpage

\pagestyle{fancy} \chead{} \rhead{} \lhead{}
\pagenumbering{roman} \lfoot{}\cfoot{\thepage}\rfoot{}
\setcounter{page}{2}

\addcontentsline{toc}{chapter}{Dedication}
\chapter*{Dedication}
\vspace{30pt}

\begin{center}
To Venkata Sita Rama Murthy Modekurthy, Vijaya Kumari Modekurthy, Venkata Pavan Kumar Modekurthy, and Kalianne Kinsey Modekurthy.
\end{center}

\newpage

\addcontentsline{toc}{chapter}{Acknowledgements}
\chapter*{ACKNOWLEDGEMENTS}

\indent I would like to sincerely thank my advisor, Dr. Abusayeed Saifullah, for his guidance, cooperation, ideas and most-importantly time. I would not be graduating if not for his advice over the last 5 years. I extend my thanks to my co-advisor Dr. Sanjay Madria for his guidance and motivation. 

I would like to thank Dr. Nathan Fisher, Dr. Zhishan Guo, and Dr. Marco Brocanelli with whom I had wonderful opportunities for collaboration. These collaborations have made my Ph.D. experience productive and inspiring. I would also like to thank my lab mates from the CRI lab at Wayne State University and then Advanced Networking Lab at Missouri S\&T. My hearty gratitude for Dr. Corey Tessler and Dr. Amartya Sen for both the personal and professional time which lead to stimulating discussions and collaborations. I am greatly indebted to the the faculty members and staff at both Wayne State for their helping hands in every step of my Ph.D. and notably for making the transfer from Missouri S\&T effortless. Most importantly, I would like to thank my parents and wife for their perpetual support and love. 

I acknowledge the financial support from Dr. Abusayeed Saifullah and Dr. Sanjay Madria in the form of research assistantship through NSF grants, and from Department of Computer Science at Wayne State in the form of research assistantship from internal grants and teaching assistantship. Last but not least, I thank all committee members for their time, suggestions, and service. 
\newpage

\begin{singlespace}
\tableofcontents{\normalsize}

\newpage
\addcontentsline{toc}{chapter}{List of Tables}
\listoftables

\newpage
 \addcontentsline{toc}{chapter}{List of Figures}
\listoffigures

\end{singlespace}

\include{null}          
\pagestyle{fancy} 
\pagenumbering{arabic}  
\chead{\thepage} \rhead{} \lhead{}
\lfoot{}\cfoot{}\rfoot{}
\setcounter{page}{1}    

\renewcommand{\headrulewidth}{0pt}



\fancypagestyle{plain}{
  \fancyhf{}
  \fancyhead[C]{\thepage}
}

\include{intro/intro}
\include{ICDCN_routing/ICDCN-distributedGraph}

\include{DistributedHart_TMC/TMC19}

\include{snow_latency/rtss_integration}
\include{utilization/ICII}
\include{online_period_selection/ICII19_DynamicPeriod}
\include{conclusion}
\include{references/references}

\include{abstract}

\chapter*{Autobiographical Statement}
\addcontentsline{toc}{chapter}{Autobiographical Statement}

Venkata Prashant Modekurthy is a Ph.D. candidate in the Department of Computer Science at Wayne State University. He received his master's degree from Missouri University of Science and Technology in December 2015 and his undergraduate degree from Raghu Engineering College, Visakhapatnam, India, in May 2011. His research interests include Real-Time Wireless Network, Cyber-Physical Systems (CPS), Real-Time and Embedded Systems, Parallel and Distributed Computing, and Low-Power Wide Area Networks. He authored a paper that was awarded the best paper award at IEEE ICII 2018.
\end{document}

%% file: intro/intro.tex
\chapter{Introduction}\label{ch:intro}
Industrial Internet-of-Things (IoT) evolved from industrial wireless standards like Wireless-HART \cite{WirelessHART2007_standard} and ISA100 \cite{isa100_standard} that facilitate low-power, flexible, and cost-efficient communication for a broad range of applications like process control \cite{songRTAS11}, smart manufacturing \cite{RTSS10paper}, smart grid \cite{gungor2011smart}, and data center power management \cite{capnet, capnetRTSS}. Industrial IoT enables closed-loop communication over wireless networks, where a sensor measures the state of a plant and delivers it to a controller. The controller generates control commands based on the measured state and then sends the control commands to an actuator through a wireless network. Industrial IoT requires a low end-to-end latency and reliable communication between devices to avoid catastrophic outcomes such as cause plant shutdown, accidents causing deaths, and economic/environmental impacts. To realize a predictable and reliable communication in a highly unreliable wireless environment, wireless standards use a centralized wireless stack design. In a centralized wireless stack design, a central manager generates graph routes and a communication schedule for a multi-channel time division multiple access communication (TDMA) based medium access control (MAC).

Although current wireless standards meet the reliability requirements, they are less suitable for large-scale deployments. Today, industrial IoT and cyber-physical systems are emerging in large-scale applications. Specifically, agricultural fields \cite{agri1, agri2, oregonFarming}   and oil/gas fields \cite{oilgas}  may extend over hundreds of square kms.  For example,  the East Texas Oil-field  extends over an area of $74\times 8$km$^2$   requiring tens of thousands of sensors for management \cite{texasof}. Emerson  is targeting to deploy 10,000 nodes for managing an oil-field in Texas \cite{WH10000, WirelessHART_Emerson}. To cover such large areas, we need highly scalable and energy-efficient protocols. 

Route and schedule dissemination by a central manager for a large network can be highly energy consuming, less scalable, and does not support frequent changes to networks or workloads. Distributing the scheduling and routing decision to the nodes or a set of nodes is known to be highly energy-efficient and scalable. However, there is less work on a distributed or decentralized wireless network stack design. Furthermore, designing a decentralized or a distributed wireless stack that ensures scalable, real-time (low latency), and reliable communication is a significant challenge. Moreover, a distributed wireless network induces an intricate problem involving the control application requiring a co-design of wireless network and control performance. A co-design of wireless network stack and control application has seen little progress and is highly challenging. 

In this dissertation, the above-mentioned challenges are addressed through the following contributions. 
\begin{itemize}
	\item A scalable and distributed routing algorithm for industrial IoT which generates graph routes with a high degree of redundancy while consuming less energy than the existing approach.
	\item A local and online scheduling algorithm that is scalable, energy-efficient, and supports network/workload dynamics while ensuring reliability and real-time performance.
	\item  An approach to minimize latency for in-band integration of multiple low-power wide area networks.
	\item A fast and efficient test of schedulability that determines if an application meets the real-time performance requirement for given network topology.
	\item A distributed scheduling and control co-design that balances the control performance requirement and real-time performance for industrial IoT.
\end{itemize}

This dissertation is organized as follows. Chapter \ref{chapter:distributedgraphrouting} describes the distributed routing for industrial IoT. Chapter \ref{ch:distributedhart} describes the local and online scheduling for industrial IoT. Chapter \ref{chapter:low_latency_integration} presents a low latency and low-power in-band integration of multiple networks. Chapter \ref{ch:utilization} describes the test for schedulability. Chapter \ref{ch:period} describes the distributed scheduling control co-design for industrial IoT. Chapter \ref{ch:conclusion} concludes this dissertation.

%% file: ICDCN_routing/ICDCN-distributedGraph.tex
\chapter{Distributed Graph Routing for WirelessHART Networks}
\label{chapter:distributedgraphrouting}
Wireless Sensor and Actuator Network (WSAN) heralds an efficient  communication infrastructure for industrial process monitoring and control. Stability of process control requires a high degree of reliability on communication between sensors and actuators which is quite challenging in industrial environments. To make reliable and real-time communication in highly unreliable environments, industrial wireless standards such as WirelessHART adopt graph routing. In graph routing, each packet is scheduled on multiple time slots using multiple channels, on multiple links along multiple paths on a routing graph between a source and a destination. While high redundancy is crucial to reliable communication, determining and maintaining graph routing is challenging in terms of execution time and energy consumption for resource constrained WSAN. Existing graph routing algorithms use centralized approach, do not scale well in terms of these metrics, and are less suitable under network dynamics. To address these limitations, we propose the first distributed graph routing protocol for WirelessHART networks. Our distributed protocol is based on the Bellman-Ford Algorithm, and generates all routing graphs together using a single algorithm.
We prove that our proposed graph routing can include a path between a source and a destination with cost (in terms of hop-count) at most 3 times the optimal cost. We implemented our proposed routing algorithm on TinyOS and evaluated through experiments on TelosB motes and simulations using TOSSIM. The results show that it is scalable and consumes at least $82.4\%$ less energy and needs at least $66.1\%$ less time at the cost of 1kB of extra memory compared to the  state-of-the-art centralized approach for generating routing graphs.

\input{ICDCN_routing/introduction}
\input{ICDCN_routing/related_work}
\input{ICDCN_routing/network_model}
\input{ICDCN_routing/distributed_routing}
\input{ICDCN_routing/Theoretical_evaluation}
\input{ICDCN_routing/evaluation} 

\section{Summary}\label{ICDCN_routing_sec:Conclusion} 
This chapter presents the first distributed graph routing protocol for WirelessHART networks. The proposed approach is an adaptation of the Bellman-Ford algorithm, and it scales well in terms of execution time and energy consumption. We have implemented our protocol in TinyOS and evaluated its effectiveness through both experiments and TOSSIM simulations. Our simulation results using TOSSIM show that it consumes about 86.4\% less energy with 66.1\% reduced convergence time at the cost of 1kB of additional memory compared to the state-of-the-art centralized approaches, thereby, to our knowledge, demonstrating it as the first practical distributed reliable graph routing for WirelessHART networks.

%% file: ICDCN_routing/introduction.tex
\section{Introduction}

Wireless Sensor and Actuator Network (WSAN) provides an efficient communication infrastructure for a broad range of industrial control applications \cite{chen2011survey}. Reliability of communication in a WSAN has a high impact on the stability of critical control applications like process control \cite{songRTAS11}, smart manufacturing \cite{RTSS10paper}, smart grid \cite{gungor2011smart}, and data center power control \cite{RTSS15paper}. Packet losses in such applications may lead to highly unstable systems. Feedback control loops in a WSAN, therefore, impose stringent dependability requirements on communication between sensors and actuators. However, industrial environments make it difficult to meet these requirements because of frequent transmission failures due to channel noise, limited bandwidth, physical obstacles, multi-path fading, and interference from coexisting wireless devices \cite{RTSS10paper}.

Of late, WSAN has received a new impulse with the advent of industrial wireless standards such as WirelessHART \cite{WirelessHART2007_standard}. WirelessHART was specifically designed to operate in highly unreliable environments. With approximately 30 million HART devices installed across the world, it is predominantly being used worldwide for  process management \cite{reliability_wh}. To make reliable and real-time communication in highly unreliable environments, a key technique adopted by WirelessHART is reliable {\slshape graph routing} \cite{WirelessHART2007_standard}. In graph routing, a packet is scheduled on redundant time slots, on redundant links on multiple paths leading to its destination, and on multiple channels based on TDMA (Time Division Multiple Access) for enhanced end-to-end reliability. Packets from all sensor nodes are routed to the controller through an {\slshape uplink graph}. For each actuator, there is a {\slshape downlink graph} from the controller through which control messages are delivered to it. While graph routing with such a high degree of redundancy is crucial to reliable and real-time communication, determining and maintaining such routes is challenging in terms of run time and energy consumption at resource constrained WSAN devices. In this chapter, we address these challenges and study highly efficient and scalable graph routing for WirelessHART networks.  

While there are numerous routing algorithms for wireless networks \cite{abbasi2007survey}, graph routing for WSAN has been studied recently \cite{songRTAS11, wu2016maximizing}. These existing algorithms use centralized heuristics, which do not scale well with the number of nodes. Emerson \cite{WirelessHART_Emerson} and MOL \cite{WirelessHART_MOL} are targeting to deploy WSAN networks that span across thousands of nodes making scalability in WirelessHART networks a notable problem. Creating all routing graphs centrally and disseminating to all nodes is less suitable for large industrial WSANs, particularly in the presence of network dynamics.  Centralized heuristics, in principle, rely on a central manager to compute and disseminate the routes to all nodes in the network. As link quality changes in the wireless network due to network dynamics, the central manager has to re-compute and re-disseminate the routes to all the nodes. Frequent dissemination of routes is highly energy and bandwidth-consuming and can hinder critical control operations. 

To address the problems mentioned above, we propose the first fully distributed protocol for graph routing for WirelessHART networks in this chapter. A distributed algorithm obviates the need for a central manager to create and disseminate routes. While a distributed protocol is highly scalable and responsive to network dynamics, developing such a protocol is challenging as it must guarantee convergence and achieve near-optimal (in terms of routing cost) performance based on local computations only. 

We design the proposed distributed graph routing protocol by adapting the Bellman-Ford algorithm \cite{topdown} for WSANs. In the Bellman-Ford algorithm, each node maintains a routing table that contains the cost of forwarding a packet to a destination via all neighbors. That is, if a link to one of its neighbors fails, it can use alternate paths to make a transmission, which is the essence of graph routing. We leverage this feature of the Bellman-Ford algorithm to develop our distributed graph routing protocol.  However, as the routing table at a node contains costs to all other nodes, message size and energy required to converge pose a key challenge for low power WSAN. To handle this, we use node clustering where the routing table at a node needs to contain link cost information of only the nodes in its cluster and cluster heads of the other clusters and generate routes only to the required destinations. Another key feature of our protocol is that it can generate all routing graphs concurrently using a single algorithm, unlike existing approaches that use a separate algorithm for each.  When the nodes in a cluster are single-hop away from its cluster head, the least cost path on each routing graph yielded by our approach is guaranteed to have a cost (in terms of hop-count) at most 3 times the optimal cost. 

We implemented our proposed routing algorithm on TinyOS and evaluated through experiments on TelosB motes and TOSSIM simulations. We determine the effectiveness of our distributed routing in terms of scalability, energy, and under network dynamics. The simulation results demonstrate that our protocol consumes at least 82.4\% less energy and needs at least 66.1\% less time compared to the state-of-the-art centralized approach \cite{songRTAS11, wu2016maximizing} for generating routing graphs. 

In this chapter, Section \ref{ICDCN_routing_sec:related_work} presents related work. Section \ref{ICDCN_routing_sec:network_model} presents our network model. Section \ref{ICDCN_routing_sec:graphRoutingOverview} gives an overview of graph routing. Section \ref{ICDCN_routing_sec:routing_dist} describes our distributed routing protocol. Section \ref{ICDCN_routin_sec:bounds} analyzes approximation ratio of routing cost.  Section \ref{ICDCN_routing_sec:eval} presents our implementation and evaluation. Section \ref{ICDCN_routing_sec:Conclusion} concludes the chapter.

%% file: ICDCN_routing/related_work.tex
\section{Related Work}
\label{ICDCN_routing_sec:related_work}
Routing in wireless ad hoc and sensor networks has been studied extensively on issues like energy efficiency \cite{pantazis2013energy}, time-dependency \cite{hao2012routing}, hierarchy \cite{karaboga2012cluster, tyagi2013systematic} and reliability through multi-path routing \cite{marina2002ad, cai2015bee}. The multi-path routing protocols like AOMDV \cite{marina2002ad} typically provide an alternate path from the source to the destination, and do not provide an alternate path from every node in the route. More recently, Bee-Sensor-C \cite{cai2015bee} integrates clustering and agents to obtain multi-path routes. GDMR \cite{zhang2015load} presents a heuristic multi-path routing for optimizing network performance by load balancing traffic in IP networks. Two distributed multi-path routing algorithms were developed in \cite{geng2015algebra} to let nodes efficiently find next-hops for each destination and guarantee a few metrics of QoS routing. 

The above multi-path routing protocols provide a few either node-disjoint or link-disjoint paths between a source and a destination, thereby providing a certain degree of end-to-end reliability. In contrast, graph routing is particularly designed to achieve a high degree of reliability for real-time process monitoring and control applications in highly unreliable industrial environments. The WirelessHART standard mandates each node in a routing graph to have at least two neighbors that can relay a packet towards the destination.  Thus, a routing graph for a control loop can consist of an exponential (in the number of nodes) number of routes between its source and destination.  For example, assuming this minimum requirement of two neighbors for every node in a routing graph with $l$ nodes, there can be $2^{l}$ paths between a source and a destination. Graph routing provides a high degree of reliability, which cannot be generated by existing multi-path routing algorithms.  Note that the paths of the graph route may overlap each other. The path taken by a packet depends on the current link conditions.   

Real-time scheduling for process control based on graph routing has received considerable interest over the last five years \cite{RTSS10paper, zhang2013highly, RTSS15paper, ECRTSpaper, RoamingHART, songRTAS11}. Specifically, real-time scheduling \cite{RTSS10paper, zhang2013highly,duquennoy2015orchestra}, schedulability and end-to-end communication delay analysis \cite{ RTSS15paper}, scheduling-control co-design \cite{RTAS2012}, priority assignment \cite{ECRTSpaper}, and localization \cite{RoamingHART}, has been studied recently for WirelessHART networks based on graph routing.  However, all these works assume that routing in the network layer is given. While there are enormous routing algorithms for wireless networks \cite{adibi2006multipath}, graph routing for WirelessHART has been studied only recently \cite{songRTAS11, wu2016maximizing} and these algorithms use {\slshape centralized} heuristics to generate the uplink graph and each downlink graph separately at the expense of scalability. Furthermore, frequent dissemination of the routes to all nodes is highly energy consuming, particularly in the presence of network dynamics. In contrast, we propose a fully distributed protocol for graph routing for WirelessHART networks, thus obviating the need for centrally creating and disseminating the routes. 

%% file: ICDCN_routing/network_model.tex
\section{Network Model}
\label{ICDCN_routing_sec:network_model}
WirelessHART, designed on the top of IEEE 802.15.4, operates in the 2.4GHz band. It forms a multi-hop mesh network consisting of a Gateway, field devices, and access points. The {\slshape network manager} and controller are at the {\slshape Gateway}. {\slshape Field devices} are wirelessly networked sensors and actuators which are equipped with a {\slshape half-duplex} omnidirectional radio transceiver that cannot both transmit and receive at the same time and can receive from at most one sender at a time. {\slshape Access points} provide redundant paths between the wireless network and the Gateway. The network involves feedback control loops between sensors and actuators through the controllers.

\begin{figure}[t]
    \centering
    \begin{subfigure}[b]{0.3\textwidth}
        \includegraphics[width=\textwidth]{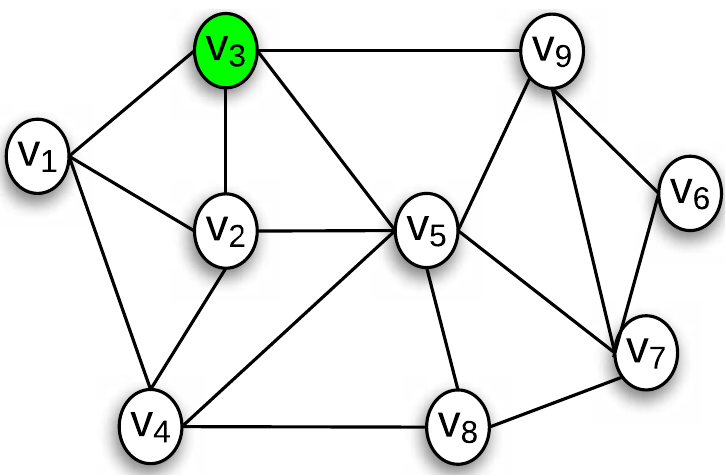}
        \caption{Network Topology}
        \label{ICDCN_routing_fig:originalGraph}
    \end{subfigure}
    \quad
    \begin{subfigure}[b]{0.3\textwidth}
        \includegraphics[width=\textwidth]{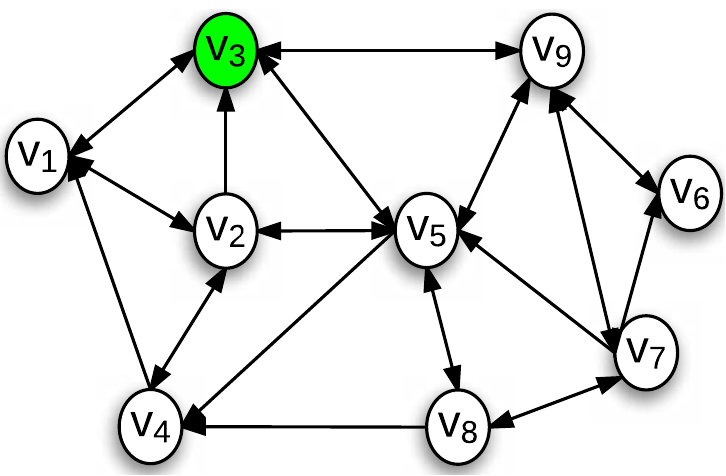}
        \caption{Uplink Graph}
        \label{ICDCN_routing_fig:uplink}
    \end{subfigure}
     \quad
    \begin{subfigure}[b]{0.26\textwidth}
        \includegraphics[width=\textwidth]{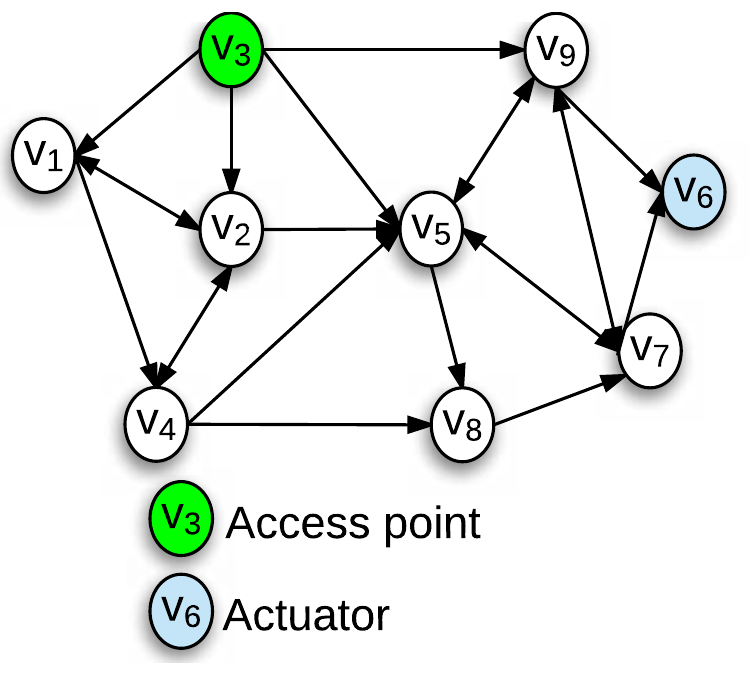}
        \caption{Downlink graph for $v_6$}
        \label{ICDCN_routing_fig:downlink}
    \end{subfigure}
   \caption{Example of Graph Routing}
   \label{ICDCN_routing_fig:example}
\end{figure}

Transmissions are scheduled based on a multi-channel TDMA (Time Division Multiple Access) protocol. The network uses the channels defined in IEEE $802.15.4$. Time is globally synchronized. Each time slot is of fixed length. Packet transmission and its acknowledgment transpire in one time slot. A transmission between a sender and its receiver can take place on a dedicated or a shared time slot. In a {\slshape dedicated slot}, only one sender is allowed to transmit to a receiver. In a { \slshape shared slot}, multiple senders can attempt to send to a common receiver. For enhanced reliability, the network adopts {\slshape graph routing} \cite{WirelessHART2007_standard}. Graph routing enables scheduling a packet on multiple links using multiple channels on multiple time slots through multiple paths to deliver a packet to a destination, thereby ensuring high reliability in unreliable environments.  

The network is represented as a graph $G = (V,E)$, where $V$ is the set of nodes and $E$ is the set of edges. We consider a set of destination nodes $D \subset V$ and a node $v \in D$ represents a process controller or an actuator node. We use $|V|, |E|$, and $|D|$ to denote the number of nodes, the number of edges, and the number of destination nodes in the network, respectively. We represent a control loop $C_a$ as $C_i = \{T_a, D_a, P_a\}$ where $T_a$, $D_a$ and $P_a$ represent the time period, deadline and priority of control loop $a$. In addition we represent $\tau$  as the length of superframe which can be obtained as $\tau = LCM(T_1, T_2, \cdots, T_n)$.

\section{An Overview of Graph Routing}
\label{ICDCN_routing_sec:graphRoutingOverview}

\begin{figure}[th]
\centering
\includegraphics[width=0.3\textwidth]{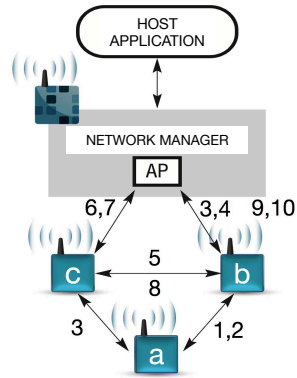}
\caption{Example of Scheduling in WirelessHART}
\label{ICDCN_routing_fig:scheduling}%
\end{figure} 

A {\em routing graph} in a WirelessHART network is a directed list of loop-free paths between a source and a destination. It is a directed acyclic graph from a source to a destination, and each node in the graph route except the destination has a minimum of 2 neighboring nodes. Redundancy in paths at each node in the routing graph circumvents temporary link or node failures allowing retransmission through redundant links/paths.  

Graph routing defines two types of routing graphs; uplink graphs and downlink graphs. Packets from all sensor nodes are routed to the Gateway through the {\em uplink graph}. For every actuator, there is a {\em downlink graph} from the Gateway through which the Gateway delivers control messages. In each routing graph, a node can have multiple neighbors to which it can transmit and retransmit a packet multiple times to be delivered to the corresponding destination. For a network shown in Fig. \ref{ICDCN_routing_fig:originalGraph}, an example of uplink graph routing and downlink graph for actuator $v_6$ are presented in Fig. \ref{ICDCN_routing_fig:uplink} and Fig. \ref{ICDCN_routing_fig:downlink} respectively.

In a routing graph, we define a primary path as the least cost from source to a destination. We consider all other paths in the routing graph as back-up paths. Typically, a sender transmits a packet to the destination through the primary path. Upon failure of a transmission on a primary path, it is delivered through a back-up path of the routing graph. In the uplink graph shown in Fig. \ref{ICDCN_routing_fig:uplink}, for sensor node $v_8$ if the primary path on uplink graph is computed to be $v_8 \rightarrow v_5 \rightarrow v_3$, then all other paths, like $v_8 \rightarrow v_5 \rightarrow v_9 \rightarrow v_3$, are back-up paths.

For a control loop scheduling in the uplink graph, the scheduler allocates two dedicated slots for each link on the primary path. A node on the primary path uses the second dedicated time slot to handle transmission failures on the first slot. Then, to handle transmission failures of both slots, the scheduler allocates a third shared slot on a separate path to handle another retry. It then schedules the links on the downlink graph similarly. Fig. \ref{ICDCN_routing_fig:scheduling} shows an example of transmission scheduling from node $a$ to an access point (AP) for one control loop in a network of 4 nodes. The labels on the link represent the time slot assigned by the scheduler. On the primary path $a\rightarrow b \rightarrow AP$, the scheduler allocates 2 dedicated time slots for each link, i.e.,  time slots 1 and 2 for link $a\rightarrow b$  and time slots 3 and 4 for link $b \rightarrow AP$. The scheduler allocates a shared slot 3 for the link $a\rightarrow c$ (of a back-up path) to handle transmission failures on the link $a\rightarrow b$. 

%% file: ICDCN_routing/distributed_routing.tex
\section{Distributed Graph Routing}
\label{ICDCN_routing_sec:routing_dist}

In this section, we present our distributed graph routing protocol for WirelessHART networks. Our protocol is highly scalable, and adaptable to network dynamics when compared to existing solutions that rely on a central manager to create and disseminate routes to all nodes in the network. Furthermore, our approach is energy efficient for resource-constrained WSAN nodes. Our algorithm can converge quickly since our approach requires a lesser number of message communications than existing approaches.

Our proposed distributed graph routing is an adaptation of {\bf\slshape Distributed Bellman-Ford} algorithm \cite{topdown,perkins1994highly} for WirelessHART networks. Distributed Bellman-Ford algorithm \cite{topdown,perkins1994highly} has been used widely for distributed routing in the Internet. In the Bellman-Ford algorithm, a node maintains a routing table that contains the cost of routing a packet to all nodes through its neighbors.  A node broadcasts its routing table whenever there is a change, and neighboring nodes update their routing tables accordingly. As a result, all nodes maintain a least-cost path to all nodes, in the face of constant changes to link qualities.

The Bellman-Ford algorithm has two advantages for use in WSAN. Firstly, it can efficiently maintain and update routing tables. This feature is helpful in the face of long term link/node failures. Secondly, a node knows all possible neighbors that can deliver a packet to a destination.  That is, it keeps track of multiple paths to a destination, which is the essence of graph routing.  This feature is helpful in the face of short term link/node failures, where a node can use the available alternate paths to transmit a packet. Due to these advantages, our proposed approach depends on the Bellman-Ford algorithm to generate graph routes.  In the proposed approach, we consider the least-cost path obtained from the Bellman-Ford algorithm as the primary path of the graph route and all other paths as back-up paths. 

Since the routing table at a node contains costs to all other nodes, the energy required to converge to a solution and message size for communicating the routing table pose a {\bf key challenge} for low power and low bandwidth WSAN nodes. Furthermore, the Bellman-Ford algorithm does not scale well with the number of nodes. It requires messages in the order of $O(|V| \times |E|)$ to converge to an optimal solution. To address these challenges, we propose a novel distributed graph routing protocol based on the Bellman-Ford algorithm for WirelessHART networks, as described in the following sub-section. ributed graph routing protocol based on the Bellman-Ford algorithm for WirelessHART networks as described in the following sub-section.

\subsection{Protocol Description} \label{ICDCN_routing_sec:proposedGraphRouting}
To address the challenges on energy overhead for WirelessHART networks, we propose to adopt the distributed Bellman-Ford algorithm through node clustering. That is, we propose to divide the nodes into clusters $C_1, C_2, \cdots, C_q$. $CH_i$ denotes the cluster head of cluster $C_i$. There exist many algorithms for distributed clustering in wireless sensor networks such as \cite{abbasi2007survey, leach2002}, and our proposed approach works with any clustering algorithm. For the sake of simplicity, we assume that the network is clustered during network initialization.
There exist many algorithms for distributed clustering in wireless sensor networks such as \cite{abbasi2007survey, leach2002} and our proposed approach works with any clustering algorithm. For the sake of simplicity, we assume that the network is clustered during network initialization. 

After clustering the network, we use the Bellman-Ford algorithm to generate routes to nodes within the same cluster and all cluster heads at each node. That is, each node in the network will have the least-cost path to all nodes within its cluster, $v \in C_i$, and to all other nodes through their respective cluster heads, $CH_j$ $where$ $j \neq i$.  Thus, creating $k$ clusters reduces the energy, memory, and convergence time requirement by a factor of $O(k)$ when compared to the Bellman-Ford algorithm. After the convergence of the Bellman-Ford algorithm, cluster-heads broadcast the information of nodes within its cluster to all nodes.  For a packet with its destination in the same cluster, a node forwards the packet on the routing graph to the destination (since it is aware of routes to all nodes within the cluster). For a packet with a destination ($v_j$) in a different cluster ($C_i$), a node forwards the packet on the routing graph to $CH_i$ (cluster head of the destination node's cluster).  Since nodes in a cluster have routes to all other nodes within the same cluster, upon entering the cluster $C_i$, the packet will be forwarded to the destination node $v_j$, which may or may not pass through the cluster head $CH_i$. 

To further reduce the energy usage, nodes use the Bellman-Ford algorithm to maintain routes and cluster head information of the destination nodes (i.e., controller and actuator nodes). Under the proposed approach, the number of message communications in the network and the energy consumption of the nodes to maintain the routes is in the order of $O(|D| \times |E|)$, where $|D| < |V|$. Routes and cluster head information to nodes other than the destination nodes are not generated as intermediate nodes can not process the information in the packet.  In section  \ref{ICDCN_routin_sec:convergence}, we analyze the convergence and optimality property while considering just destination nodes. In the proposed method, we adopt DSDV \cite{perkins1994highly} to avoid generating loops in the system. Our approach can support any other technique like Poison Reverse \cite{topdown}, but we choose DSDV for the sake of simplicity. 

\begin{algorithm}[!th]
\footnotesize
	\caption{Distributed Graph Routing Algorithm} \label{ICDCN_routin_alg:Routing}
    \SetAlgoLined 
    \SetKwInOut{Input}{Input}
    \SetKwInOut{Output}{Output}
    \DontPrintSemicolon
    \Input{destinationNode, nodeId}
    \Output{routeTable, clusterHeadInfo}
    
    Initially,
    \textit{routeTable} $\leftarrow$ NULL \; 
    \textit{clusterHeadInfo} $\leftarrow$ NULL\;
    \;
    
   clusterHead = generateClusters()\;
   
    \If{destinationNode}{
    	send routeTable\;
    }    
    \If{nodeId = clusterHead} {
        send routeTable\;
            sleep until algorithm converges\;
            send clusterHeadInfo\;
    }
    
    \If{received message} {
    \For{row in routeTable} {
    	\If{row is for destination node within same cluster \textbf{or} row is for cluster head}
        {
        	update routeTable \;
        	send  routeTable\;
        }
   }
    }

    \If{received message \textbf{and} contains clusterHeadInfo}{
    	update clusterHeadInfo\;
       send clusterHeadInfo\;
   } 
  
\end{algorithm}

The pseudo-code of our proposed distributed graph routing algorithm is presented in Algorithm \ref{ICDCN_routin_alg:Routing}. Each node executes this algorithm to construct routing graphs and cluster head information. The algorithm takes a boolean value if the node is destination node or not ($destinationNode$) and ID of a node ($nodeId$) as input and generates routing table and cluster head information as output. Given the inputs, nodes first run a clustering algorithm (Leach in our case) to generate and identify cluster heads. After clustering, destination nodes and cluster heads initiate the Bellman-Ford algorithm by sending routing tables. If the received message at a node contains a route to a cluster head, then the node updates the cost in its routing table with the received new cost value. Upon updating the routing table with new cost value, each node re-broadcasts the routing table to all other nodes. If a node receives a message from a destination node, it updates the cost only for destination nodes within its cluster. If the destination is from a different cluster, then the routing table update message is ignored and is not re-broadcasted. Once the Bellman-Ford algorithm converges, cluster heads broadcast IDs of destination nodes within its cluster to all nodes in the network, thereby allowing all nodes to generate both graph routing and cluster head information.

{\bf  Route Maintenance.}  Nodes in the network can use management superframes and periodic maintenance communication (that WirelessHART devices already do) to periodically update routes with recent link quality information. Since the graph route provides redundant paths for packet delivery, routes may not need frequent reconstructions, and hence reconstruction during management superframes can be sufficient. While distributed graph routing is most suitable for distributed schedulers, our proposed distributed graph routing is also applicable for centralized schedulers. For centralized schedulers, the central manager can collect routes through management superframes and periodic maintenance communication.

\subsection{Example}\label{ICDCN_routing_sec:constructingGraphRoute}
Fig. \ref{ICDCN_routing_fig:WHRouting} illustrates the proposed distributed graph routing for a simple topology with two clusters $C_1$ (with $CH_1=v_2$) and $C_2$ (with $CH_2=v_7$). It shows the final routing table at node $v_1$,$v_4$ and $v_8$ created through our distributed approach. Each row in the table gives cost to a destination node or cluster head through each of the neighboring nodes. Route to destination node $v_3$ is only present at nodes $v_1$ and $v_4$ as $v_3$ is within their clusters. However, nodes $v_5$ and $v_8$ do not contain entries to node $v_3$ as $v_3$ is not in the same cluster as nodes $v_5$ and $v_8$. 

\begin{figure}[ht]
\centering
\includegraphics[width = 0.65\textwidth]{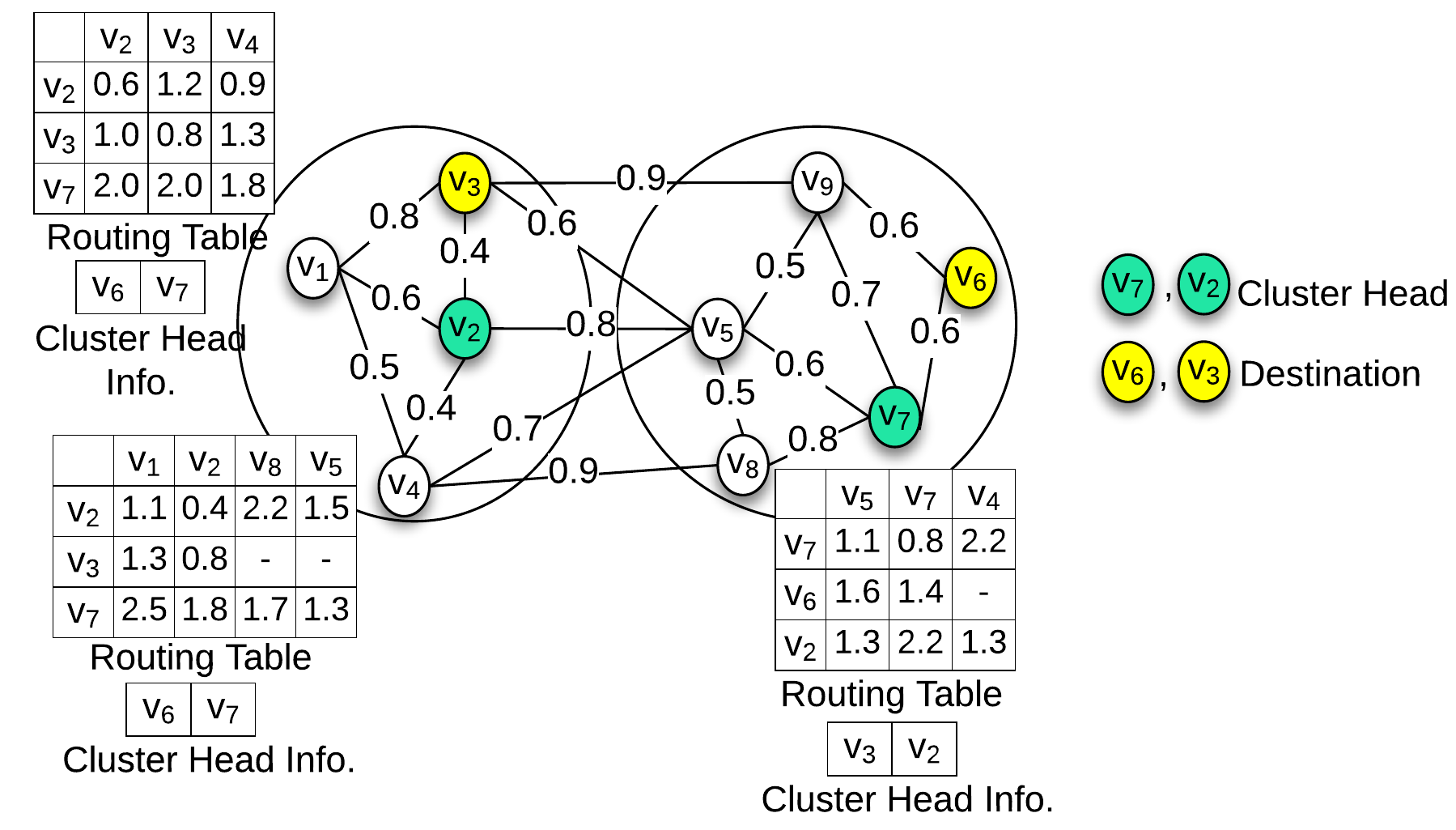}%
\caption{Example of Routing Tables Used in Distributed Graph Routing}
\label{ICDCN_routing_fig:WHRouting}%
\end{figure}

\begin{figure}[ht]
\centering
\includegraphics[width=0.55\textwidth]{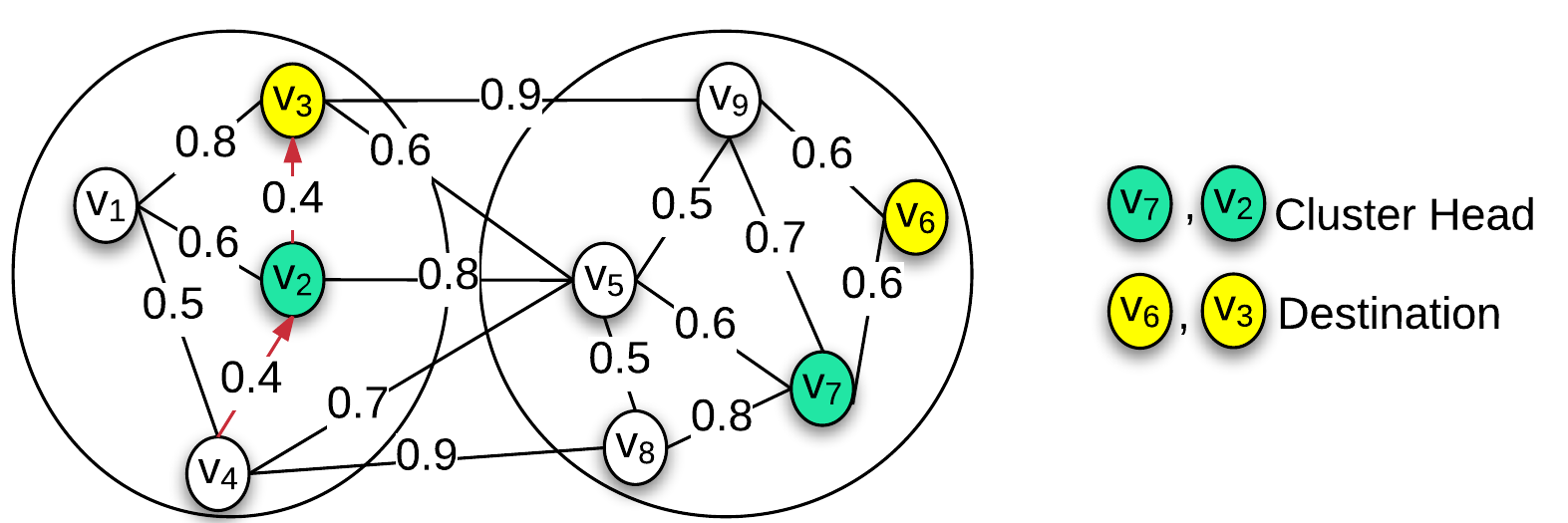}
\vspace{-0.1in}
\caption{Example of Intra-Cluster Routing Under Distributed Graph Routing}
\label{ICDCN_routing_fig:WHRoutingExample1}
\end{figure}

\begin{figure}[ht]
\centering
\includegraphics[width=0.55\textwidth]{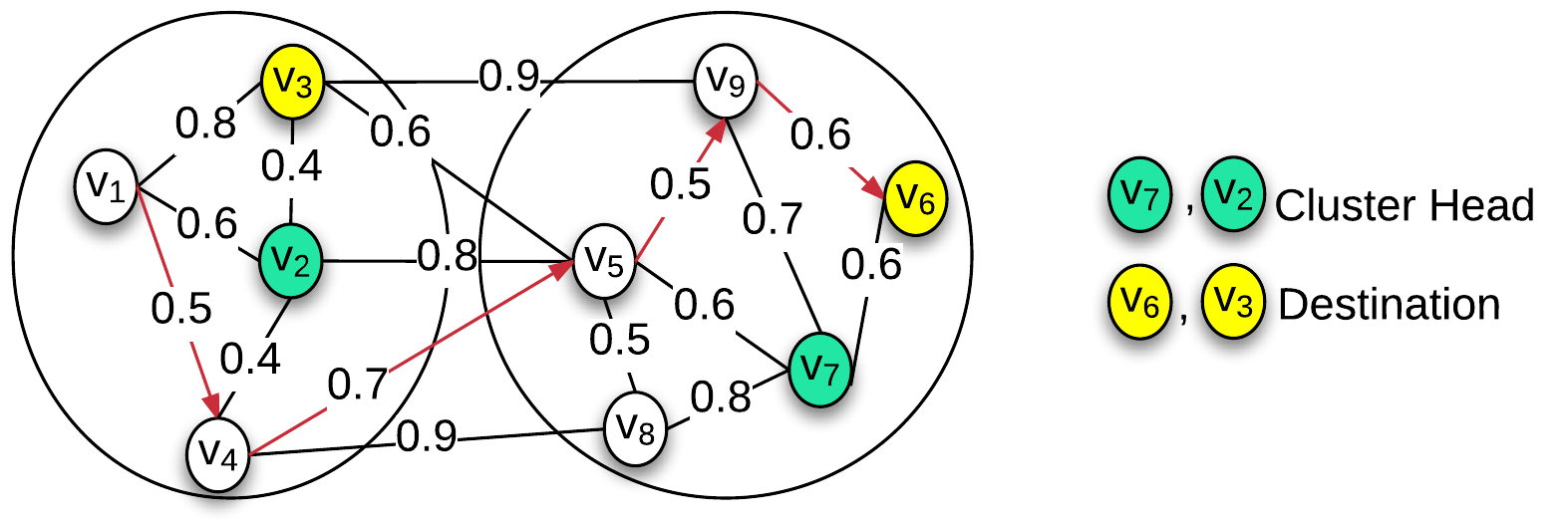}
\caption{Example of Inter-Cluster Routing Under Distributed Graph Routing}
\label{ICDCN_routing_fig:WHRoutingExample2}
\end{figure}

Consider $v_3$, as shown in Fig. \ref{ICDCN_routing_fig:WHRouting}, as the destination for packet at $v_4$. From its routing table, $v_4$ is aware of an optimal route through $v_2$ hence $v_4$ forwards the packet to $v_2$. If the link to $v_2$ or node $v_2$ fails, node $v_4$ can also forward the packet to $v_1$. Similarly, both nodes $v_1$ and $v_2$ can forward the packet through an optimal path to $v_3$ and in case of failure, they can transmit on an alternate path. In this example, least cost path (or primary) from source to destination is $v_4 \rightarrow v_2 \rightarrow v_3$ while other listed paths are back-up paths.
Consider $v_6$ as the destination for a packet originating at $v_1$. Since $v_6$ and $v_1$ are in different clusters, node $v_1$ does not have a route to $v_6$ in its routing table. However, from cluster head information stored at node $v_1$, it can identify the cluster head for node $v_6$ is $v_7$. From its routing table, node $v_1$ is aware of a path to $v_7$ and forwards the packet to $v_4$ and in case of failure of transmission it can forward on an alternate path through $v_2$ or $v_3$. Similarly node $v_4$ forwards packet on an optimal path to node $v_5$. Node $v_5$ however is aware of an optimal path to $v_6$ through node $v_9$ and chooses to forward the packet to node $v_9$. Node $v_9$ then forwards the packet on an optimal path to node $v_6$. In this example, the primary path of routing graph from $v_1$ to $v_6$ is $v_1 \rightarrow v_4 \rightarrow v_5 \rightarrow v_9 \rightarrow v_6$.

%% file: ICDCN_routing/Theoretical_evaluation.tex
\section{Theoretical Analysis}
This section presents the convergence and sub-optimality of the proposed approach and the performance analysis of the routes generated by the proposed algorithm when compared to an optimal solution generated by a centralized scheduler. 
Notations used in this section are defined in Table \ref{ICDCN_routing_table:notations}. 

\begin{table}
	\centering
	\begin{tabular}{ ccc }
		symbol & & Description \\ 
		\hline
		\hline
		V && set of nodes \\
		D && set of Destination nodes \\
		$C_i$ && Cluster $i$ \\
		$CH_i$ && Cluster head of $C_i$ \\
		$CH(v_j)$ && Cluster head for node $v_j$ \\
		$|C(v_j)|$ && Number of nodes in cluster for which cluster head is $CH(v_j)$ \\ 
		R($v_i$,$v_j$) && Optimal cost in number of hops from node $v_i$ to $v_j$ \\
		L($v_i$,$v_j$) && sub-Optimal cost in number of hops from node $v_i$ to $v_j$ \\
		
	\end{tabular}
	\vspace{0.05in}
	\caption{Notations Used in Chapter \ref{chapter:distributedgraphrouting}}
	\label{ICDCN_routing_table:notations}
\end{table}

\subsection{Convergence and Sub-optimality}
\label{ICDCN_routin_sec:convergence}
In the proposed distributed graph routing approach, we limit the Bellman-Ford algorithm to compute optimal routes to known destination nodes within each cluster. In this section, we show that limiting the number of nodes does not impact convergence and optimality of routes between a source and a destination.  We then present a discussion on the convergence and optimality for both the route generation phase and the update phase. During the route generation phase, optimal routes to a single source converge in $|V| - 1$ steps in \cite{cormen}. We extend this proof to multiple destinations to prove that the Bellman-Ford algorithm converges and generates optimal routes to a set of destination nodes.

The Bellman-Ford algorithm is known to generate all possible paths from a source to a destination. These paths obey the following properties: 1) they are acyclic, and 2) they are directed. Thus, the routes are analogous to a directed acyclic graph rooted at the source. Suppose an optimal path of the graph route is given as $v_1 \rightarrow v_2 \rightarrow \cdots v_k$, where $v_1$ and $v_k$ are source and destination nodes, respectively. Under a link failure between nodes $v_i$ and $v_j$, node $v_i$ updates the cost to $v_j$ as $\infty$ and computes the optimal path (local to $v_i$) to the destination. After the update, node $v_i$ re-broadcasts its routing table to all of its neighbors. Similarly, node $v_{i-1}$ updates its routing table and computes the new optimal path (local to $v_{i-1}$) to the destination. Each node in the network similarly updates its routing table and computes the optimal path local to itself. Since all paths in the directed acyclic graph are rooted at node $v_1$ and the local optimal cost is the lowest cost among all other paths, the locally optimal path at $v_1$ is the globally optimal path. Therefore, the optimality of the Bellman-Ford algorithm under generation and update phase is preserved when only destination nodes are considered.

\subsection{Performance Analysis}
\label{ICDCN_routin_sec:bounds}
For a source and destination pair, cost on the primary path of a routing graph generated by our distributed graph routing may not be optimal. This section presents the approximation ratio of primary path cost obtained by the distributed graph routing when compared to the optimal cost on the primary path of a routing graph. We consider the cost on a path for a pair of source and destination nodes as the number of hops between them. Note that the cost obtained on the primary path depends on the clustering algorithm used to cluster the nodes. Besides, cost also depends on the position of the nodes in the network. If $v_i$ and $v_j$ are in the same cluster, then our distributed graph routing generates the least cost path from $v_i$ to $v_j$, which gives the approximation ratio as 1. When $v_i$ and $v_j$ are in different clusters, $v_i$ forwards packets on a path to the cluster head ($CH(v_j)$) of $v_j$. If the shortest path from $v_i$ to cluster head $CH(v_j)$ passes through $v_j$, then the packet reaches the destination in the optimal cost as it is sent to a destination within the cluster. However, the cost is maximum when the shortest path to cluster head $CH(v_j)$ does not pass through destination $v_j$. We derive an approximation ratio for hop-count on the primary path in Theorem \ref{ICDCN_routing_thm:approximationRatio}.

\begin{thm} \label{ICDCN_routing_thm:approximationRatio}
Let $L(v_i, v_j)$ be the cost of the primary path of the routing graph between a source $v_i$ and a destination $v_j$ generated by our distributed algorithm, when cost is considered as the number of hops between $v_i$ and $v_j$. Then $L(v_i, v_j)$  is at most $R(v_i,v_j) + |CH(v_j)|  - 2$, where $CH(v_j)$ is the cluster head of $v_j$, $R(v_i,v_j)$ is the optimal path cost between $v_i$ and $v_j$, and $R(v_i,CH(v_j))$ is the optimal path cost between $v_i$ and $CH(v_j)$.
\end{thm}
\begin{proof}
	When $v_i$ and $v_j$ are in the same cluster, the hop-count on the primary path is the same as that of an optimal path between them. When $v_i$ and $v_j$ are in different clusters, the hop-count on primary path from $v_i$ to $CH(v_j)$ is less than the sum of hop-count on path from $v_i$ to $v_j$ and the hop-count on path from $v_j$ to $CH(v_j)$, that is $R(v_i,CH(v_j)) \le R(v_i,v_j) + R(v_j,CH(v_j))$. This is due to the fact that, the Bellman-Ford algorithm generates a least cost path from a source and destination. Therefore, the maximum value of hop-count from $v_i$ to $CH(v_j)$ is $R(v_i,CH(v_j)) = R(v_i,v_j) + R(v_j,CH(v_j))$. Hence,  the hop-count on the primary path of the routing graph between $v_i$ and $v_j$  
	\begin{equation}
	L(v_i, v_j) \le R(v_i,v_j) + (2 \times R(v_j,CH(v_j))).
	\label{ICDCN_routing_eq_subOptimality1}
	\end{equation}
	
	The maximum distance between any two nodes in a cluster is considered as the {\em diameter} of that cluster. Assuming that all nodes can generate a graph route, i.e. nodes inside a cluster can form a 2-connected graph, the maximum value of diameter can be expressed as $R(v_j, CH(v_j)) = \frac{|C(v_j)|}{2} - 1$. Therefore, Equation (\ref{ICDCN_routing_eq_subOptimality1}) can be expressed as 
	\begin{equation}
	L(v_i, v_j) \le R(v_i,v_j) + |C(v_j)|  - 2. \qedhere
	\label{ICDCN_routing_eq_subOptimality2}
	\end{equation} 
\end{proof}

Since hop-count is dependent on the clustering algorithm, we consider all nodes in a cluster to be one hop away from their cluster head. Under this assumption, corollary \ref{ICDCN_routing_cor:apprximationRatio} shows that the cost in terms of hop-count on the primary path is at most 3 times the optimal cost.

\begin{cor} \label{ICDCN_routing_cor:apprximationRatio}
The cost on the primary path, $L(v_i,v_j)$, is three times the optimal cost, i.e. $L(v_i,v_j) = 3 \times R(v_i,v_j)$,  when cost is considered as hop-count and clustering algorithm generates clusters with a one-hop radius, where $L(v_i,v_j)$ is the cost on primary path obtained by our distributed graph routing and $R(v_i,v_j)$ is the cost on optimal path between $v_i$ and $v_j$.
\end{cor}
\begin{proof}
	Under the assumption that all nodes in a cluster are one hop away from a cluster head, $\frac{R(v_j,CH_k)}{R(v_i,v_j)} \leq 2$. Note that the maximum value of diameter of the cluster is $2$ and the fraction is maximum when $R(v_i, v_j) = 1$. Hence, $R(v_j,CH_k) = 2 \times R(v_i,v_j)$. Therefore,  from Equation (\ref{ICDCN_routing_eq_subOptimality2}), cost on primary path of routing graphs generated by our approach is at most $3$ times the cost on the optimal path. 
\end{proof}

%% file: ICDCN_routing/evaluation.tex
\section{Evaluation}
\label{ICDCN_routing_sec:eval}
This section presents the evaluation of the proposed distributed graph routing protocol through experiments using TelosB \cite{energyTelosb} and simulations using Tossim \cite{tossim}  by employing a WirelessHART protocol stack implemented on TinyOS 2.1.2 \cite{sha2015implementation}. TelosB mote is equipped with a CC2420 radio which is compliant with 802.15.4. Note that WirelessHART adapts the physical layer of IEEE 802.15.4. We performed our evaluations to determine the effectiveness of the proposed routing in terms of scalability, energy, and network dynamics. We compared its performance against existing centralized graph routing approaches for WirelessHART.

\subsection{Experiment}
We implemented distributed graph routing protocol on TinyOS 2.2 \cite{tinyos} and evaluated on an indoor testbed using TelosB devices \cite{energyTelosb} for real experiments and TOSSIM \cite{tossim} for large scale simulations. Each TelsoB device is equipped with Chipcon CC2420 radios compliant with the IEEE 802.15.4 standard. Note that the physical layer of WirelessHART is based on 802.15.4 physical layer. For the network layer of implementation, we used LEACH\cite{leach2002} clustering algorithm, for the sake of simplicity, to randomly create cluster heads and request nodes to join a cluster. After the network is clustered, cluster heads and destination nodes initiate the Bellman-Ford algorithm using signal strength at the receiver as the cost between the two nodes. All nodes in the network generated routing graphs to all cluster heads and destination nodes within their cluster. After converging to a solution, each cluster head broadcasted a list of destination nodes that are in its cluster to all nodes in the network and each node updated its cluster head information accordingly. We assumed that all nodes compute the signal strength before the start of the execution of the distributed graph routing algorithm. The proposed distributed graph routing algorithm is executed using the default CSMA-CA MAC protocol in TinyOS 2.2. For lower layers of implementation, we use the implementation provided in  \cite{sha2015implementation}.

For the sake of simplicity of the experiment, we created a 3-hop network consisting of 11 nodes deployed in an office building. The network is deployed such that each node is in communication range with a minimum of 2 neighbors, and the average number of neighbors in the network is 3. We evaluated our protocol under a varying number of control loops and presented the result for an average of 10 experiments. For each experiment, we randomly selected a base-station and two cluster heads. The choice of a random selection stems from the fact that it represents a larger spectrum of scenarios. 

\begin{figure}
    \centering
    \begin{subfigure}[b]{0.35\textwidth}
        \includegraphics[width=\textwidth]{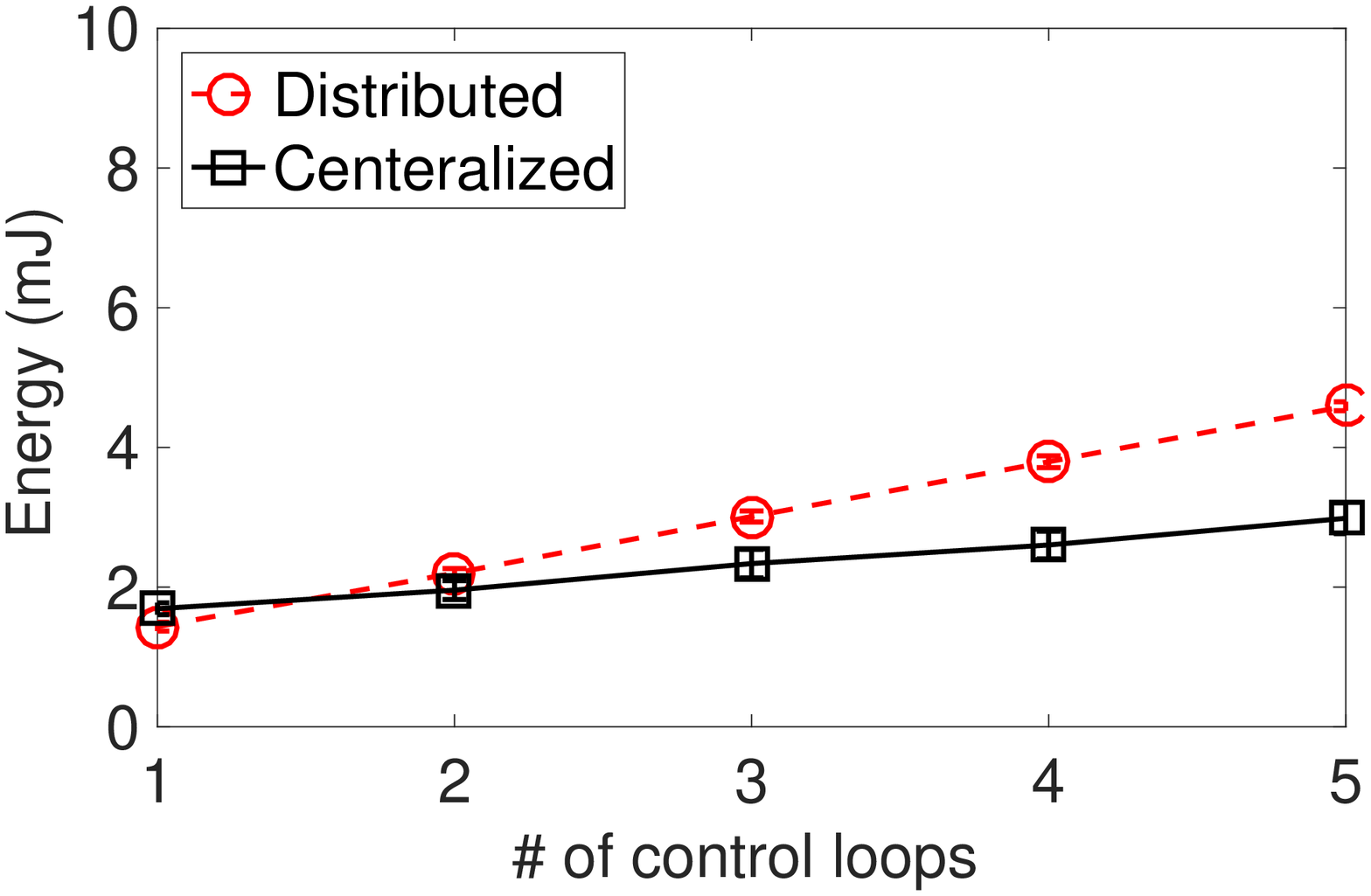}
        \caption{Energy consumption}
        \label{ICDCN_routing_fig:energy_exp}
    \end{subfigure}
    \quad
    \begin{subfigure}[b]{0.35\textwidth}
        \includegraphics[width=\textwidth]{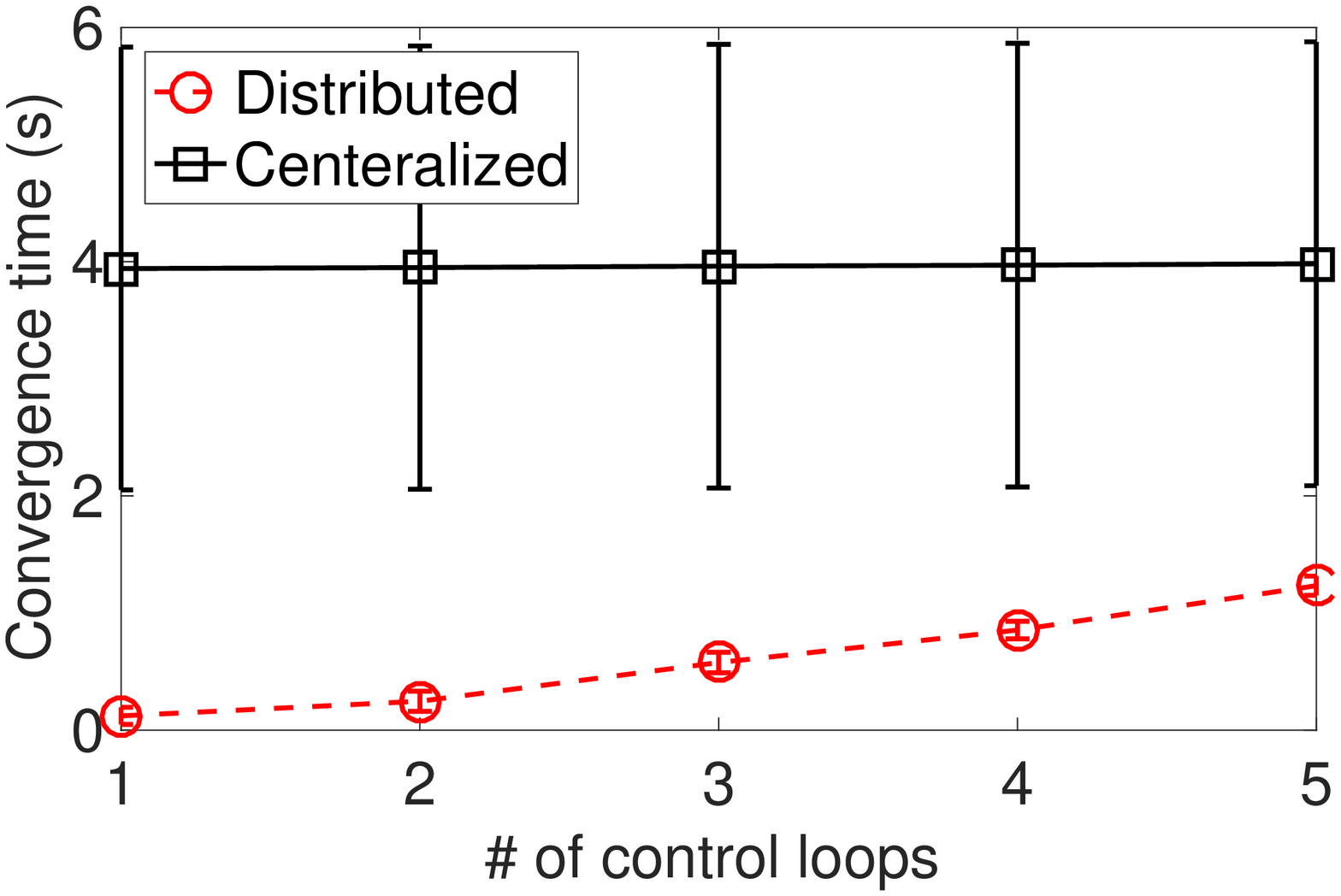}
        \caption{Convergence time}
        \label{ICDCN_routing_fig:time_exp}
    \end{subfigure}
    \caption{Experimental Evaluation of Distributed Graph Routing Under Varying Number of Clusters}
	\label{ICDCN_routing_fig:experiment}
\end{figure}

During the experiment (results are shown in Fig. \ref{ICDCN_routing_fig:energy_exp} and Fig. \ref{ICDCN_routing_fig:time_exp}), we have observed that the difference in the average number of messages communications per node between the two protocols to be small. However, we have observed that the proposed method consumes more energy rather than a centralized algorithm. Since the number of control loops is very small, the average number of message transmissions by the centralized algorithm is small. In the proposed method, each node needs to generate a route to two cluster heads and destination nodes within its cluster, which requires higher energy consumption. For large networks, centralized algorithms typically require more energy than the proposed approach, which can be observed in the simulation results. The distributed protocol can use simultaneous communication to minimize the convergence time, while a centralized algorithm has to wait until each message is passed to all nodes in the network.  In this experiment, we have also observed that memory consumption for all cases is approximately the same with a distributed algorithm consuming 4 bytes of more information(on average) when compared with a centralized algorithm \cite{songRTAS11}.

\subsection{Simulation Setup}
For large scale evaluations, we performed simulation on 148 nodes. We used the testbed model from \cite{sha2015implementation} of 74 nodes to generate the network. To scale with the number of nodes, we assumed all nodes were placed in a grid structure and replicated the grid. We added edges between neighboring grids to generate a connected graph. We used distributed vertex coloring algorithms to generate a schedule. We considered one percent of the node as access points, and nodes with the highest degree of neighbors were selected to be access points. Default value for parameters used in this simulation is given in Table \ref{ICDCN_routing_table:params}.

\begin{table}[h]
\centering
\begin{tabular}{ ccccc }
  Symbol && Description && Default Value \\
 \hline
 \hline
$|V|$ && \# of nodes && 148 \\
 $|D|$ && \% of control loops && 40 \\
 $\Phi$ && power level of transmission && $-5$ $dBm$ \\
 $q$ && \% of clusters && 5 \\
\end{tabular}
\caption{Parameters Values Used to Evaluate Distributed Graph Routing}
\label{ICDCN_routing_table:params}
\end{table}

 \textbf{Metrics.} We used \textit{energy consumed} by each protocol to generate and dissipate routes to each node in the network as a metric for comparison. We observed that a packet transmission requires an average of 6ms. From the data sheet \cite{energyTelosb}, we have calculated that CC2420 radio requires $0.38mJ$ of energy to transmit a packet.  We used $0.38mJ$ of energy per packet transmission and the number of packet transmissions from simulation to compute the energy consumed at a node. We used \textit{convergence time} of the algorithm as another metric to evaluate the proposed approach. We define convergence time as the time difference between the deployment and generation of routing tables at all nodes in the network. The convergence time of the centralized algorithm includes time taken to perform network discovery, generate routes centrally, and broadcast routes. The Convergence time of the distributed routing algorithm includes time taken to discover neighbors and generate all routes. We also use average \textit{memory} required at each node and number of broken routes as metrics for comparison. We define the number of broken routes as the number of disconnected paths in the graph route.

\subsection{Simulation Results}
We used reliable graph routing (Centralized) \cite{songRTAS11} and shortest path graph routing (Energy-aware) \cite{wu2016maximizing} to evaluate our distributed graph routing(Distributed) protocol. This section presents the performance of each protocol in terms of energy and convergence time under the scalability of nodes, scalability of control loops, varying transmission power levels, and link failures. We used an average of $50$ random test cases for each parameter to obtain our results. For each test case, we randomly generate the sensor, actuator, and priority for control loops. 

\subsubsection{Performance under Varying Number of Nodes}

\begin{figure}[th]
    \centering
    \begin{subfigure}[b]{0.35\textwidth}
        \includegraphics[width=\textwidth]{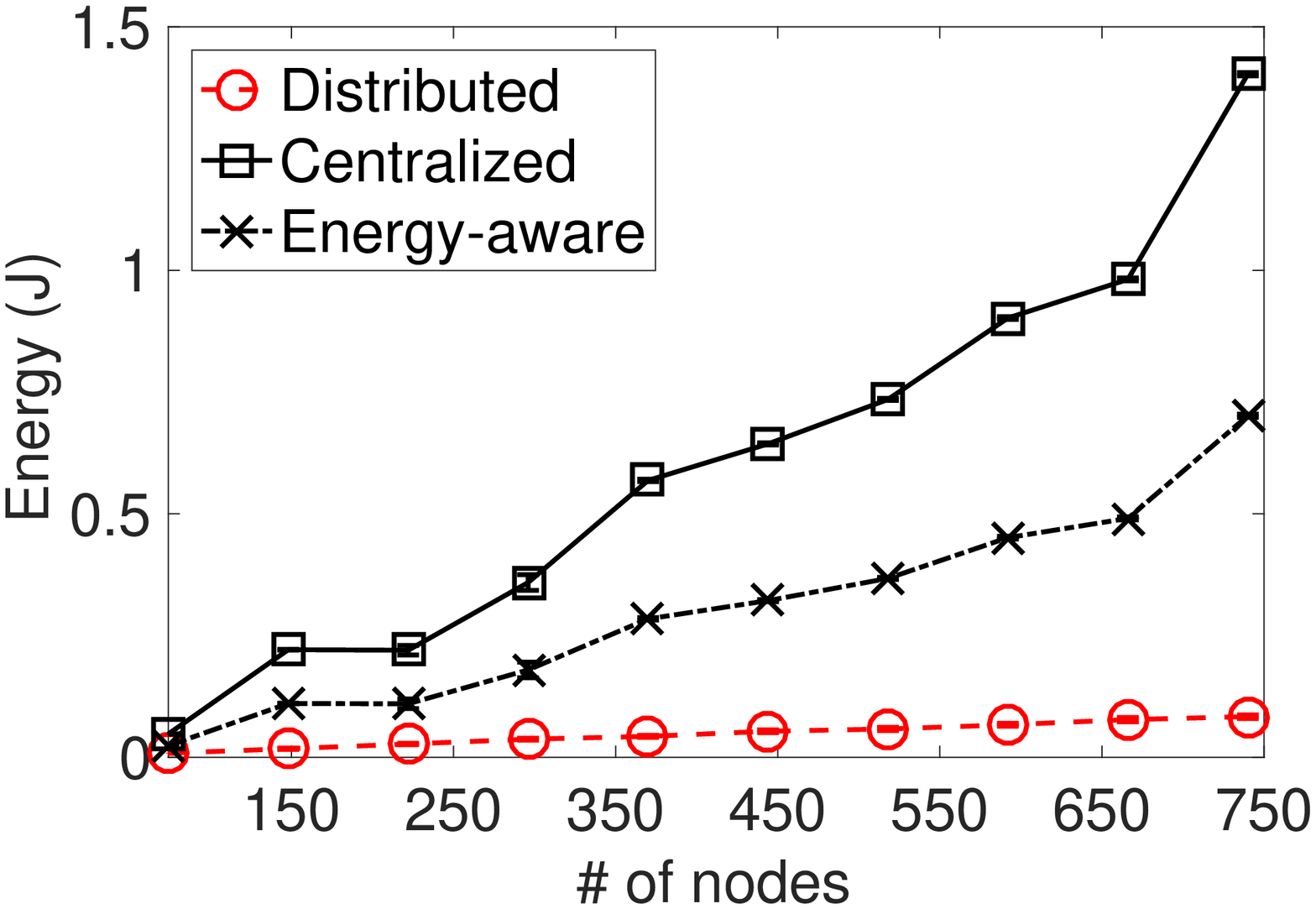}
        \caption{Energy consumption}
        \label{ICDCN_routing_fig:energy_noOfNodes}
    \end{subfigure}
    \quad
    \begin{subfigure}[b]{0.35\textwidth}
        \includegraphics[width=\textwidth]{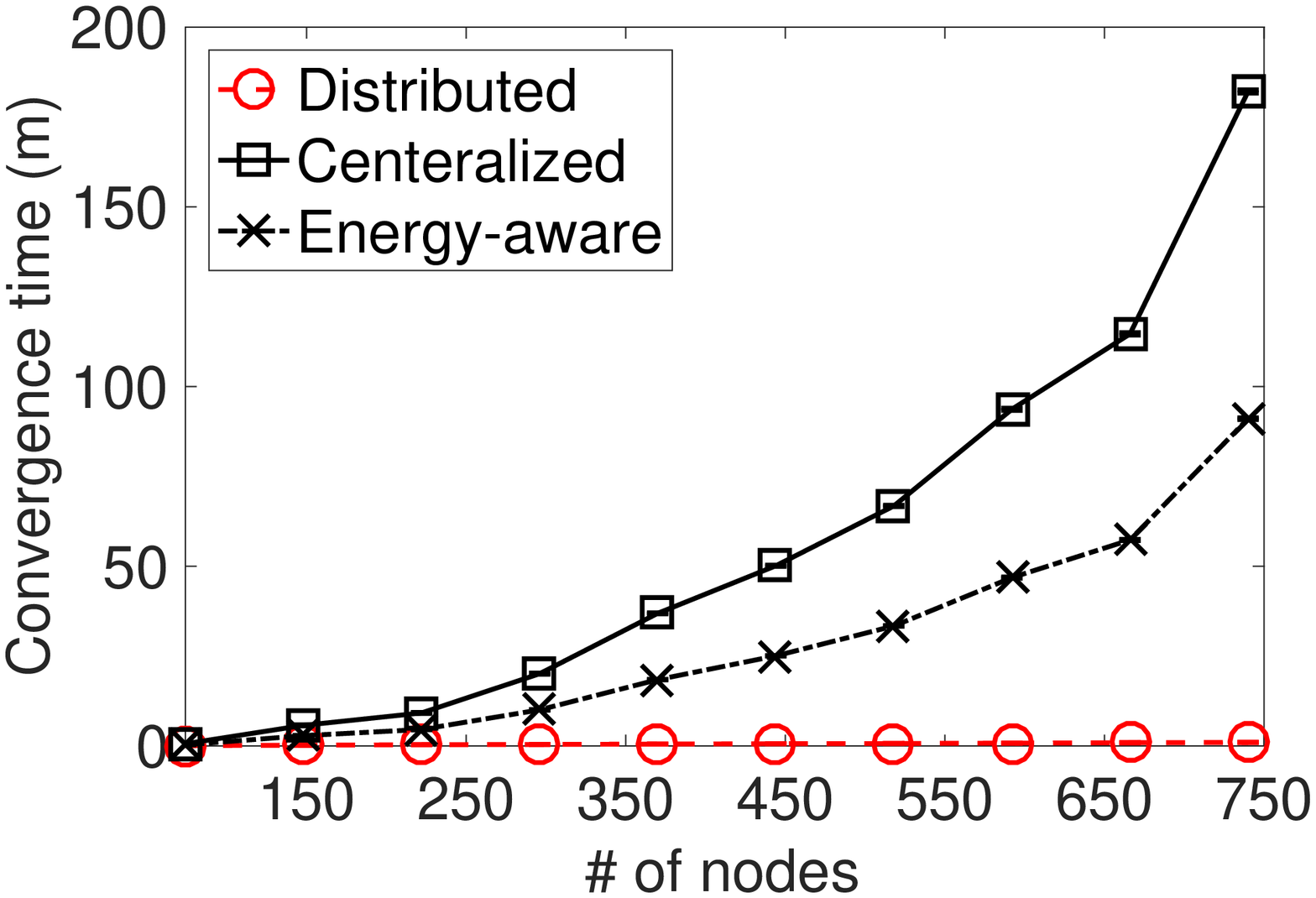}
        \caption{Convergence time}
        \label{ICDCN_routing_fig:time_noOfNodes}
    \end{subfigure}
     \quad
    \begin{subfigure}[b]{0.35\textwidth}
        \includegraphics[width=\textwidth]{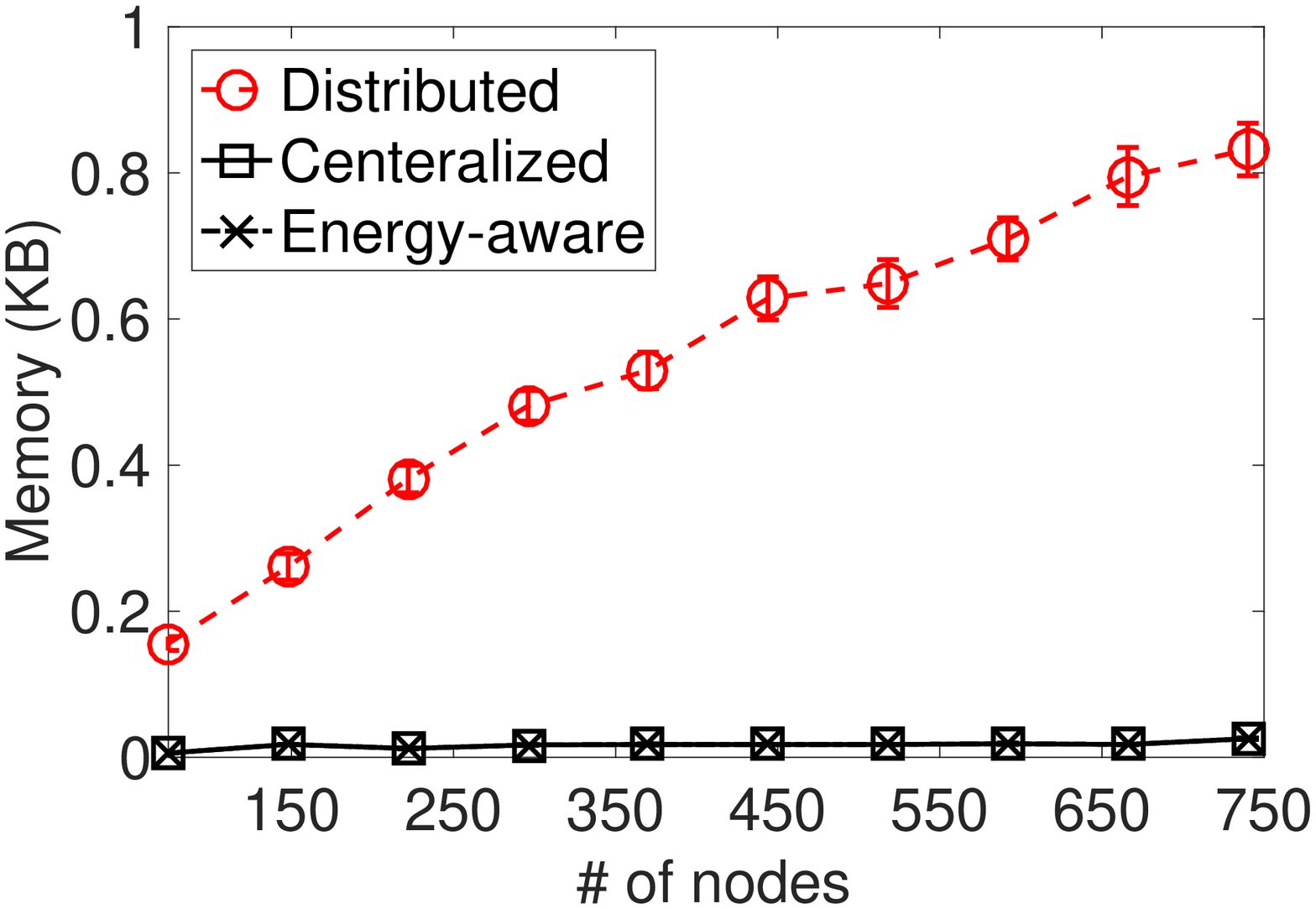}
        \caption{Memory consumption}
        \label{ICDCN_routing_fig:memory_noOfNodes}
    \end{subfigure}
   \caption{Performance of Distributed Graph Routing under Varying Number of Nodes}
   \label{ICDCN_routing_fig:noOfNodes}
\end{figure}

We now show the performance of the proposed distributed algorithm under the scalability of the number of nodes. In this simulation, we varied the number of nodes from $74$ to $740$ while keeping the density of the network constant. For a distributed algorithm, we observed that energy consumption is affected by an increase in the number of destination nodes and the number of cluster heads in the network. Since in this simulation percentage of destination and percentage of cluster heads are kept constant, the number of destination nodes and the number of cluster heads increase linearly with an increase in the number of nodes. Thus, there is a linear increase in energy consumption by a distributed algorithm. However, centralized algorithms need more energy as messages have to be propagated to all nodes before they reach their destination. Centralized algorithms theoretically need an exponential increase in energy consumption in the order of $|V|^2$, and our simulations verify the theoretical result. Moreover, our simulations (as shown in Fig. \ref{ICDCN_routing_fig:energy_noOfNodes}) show that our distributed algorithm consumes less energy than centralized or energy-aware algorithms.

We have observed that distributed graph routing saves a minimum of $65\%$ of energy consumption when compared to an energy-aware algorithm and $80\%$ of energy consumption when compared to the centralized algorithm. Similar to energy consumption, the execution time increases linearly by a factor of $|V|$ in distributed algorithm and $|V|^2$ in centralized algorithms as both are dependent on the number of message communications in the network. Thus, from Fig. \ref{ICDCN_routing_fig:time_noOfNodes} shows a similar trend as Fig. \ref{ICDCN_routing_fig:energy_noOfNodes} and we observed that distributed algorithm saves a minimum of $68\%\%$ in execution time. These results conclude that distributed algorithm is scalable under varying number of nodes in the network.

Fig. \ref{ICDCN_routing_fig:memory_noOfNodes} shows the average memory required per node in the network. The distributed algorithm generates all possible paths between a source and a destination, and each node maintains a route through all of its neighbors. Moreover, each node maintains cluster head information of all destination nodes in the network. Therefore, memory used at each node increases linearly with an increase in the number of nodes in the network. However, for centralized algorithms, only nodes that are a part of a graph route require memory.  Fig. \ref{ICDCN_routing_fig:memory_noOfNodes} shows that distributed algorithm consumes $800B$ additional memory than the centralized algorithms. Nevertheless, additional memory requirement posed by distributed routing is significantly lower than the available memory for the application program in WirelessHART or TelosB devices, which have a capacity of $16kB$ \cite{MemTeolsb}. These results show that the memory overhead of the proposed distributed algorithm is minimal when compared to the available memory at nodes.

\subsubsection{Performance under Varying Number of Control Loops}

\begin{figure}[th]
    \centering
    \begin{subfigure}[b]{0.35\textwidth}
        \includegraphics[width=\textwidth]{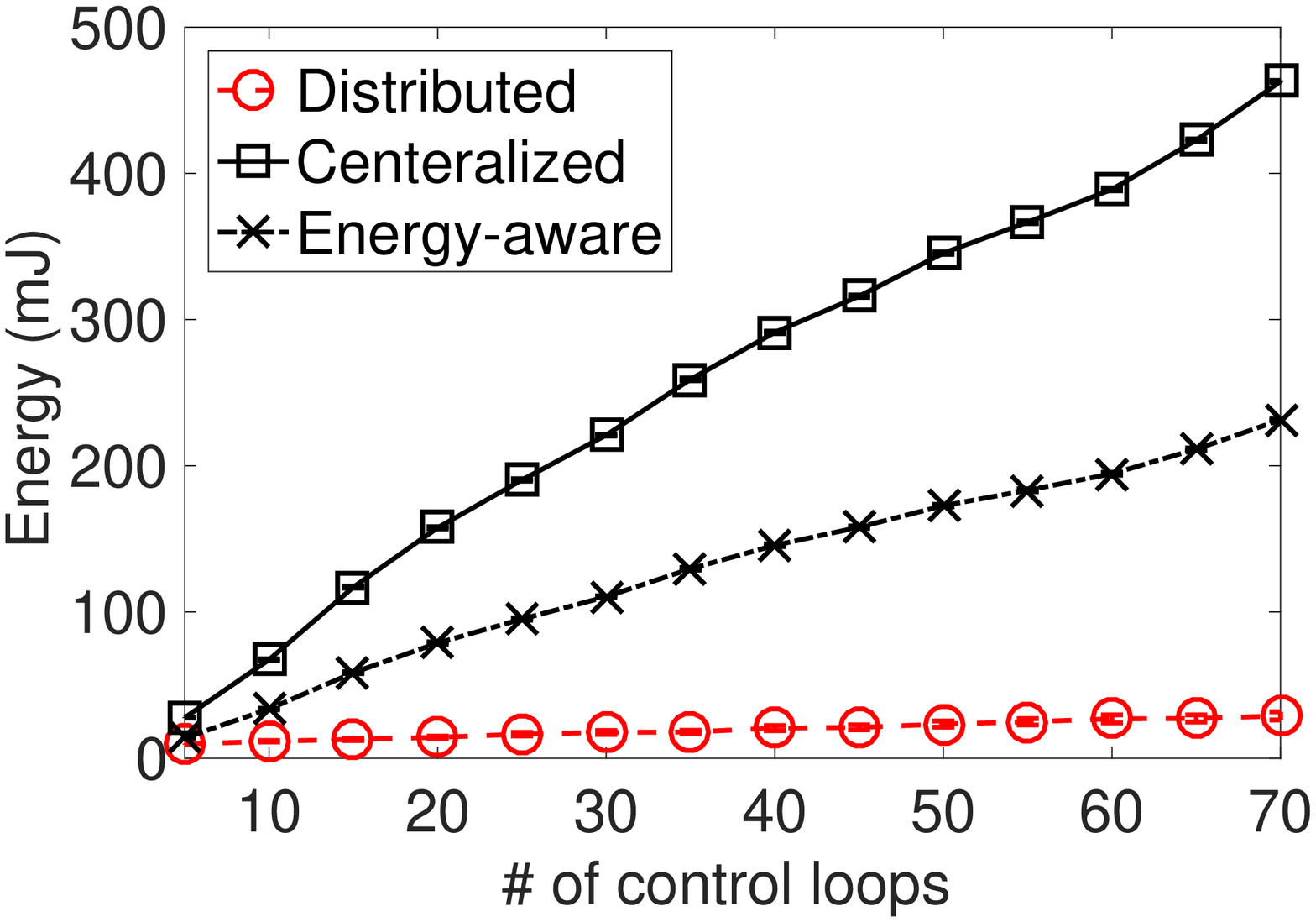}
        \caption{Energy consumption}
        \label{ICDCN_routing_fig:evergy_controlLoops}
    \end{subfigure}
    \quad
    \begin{subfigure}[b]{0.35\textwidth}
        \includegraphics[width=\textwidth]{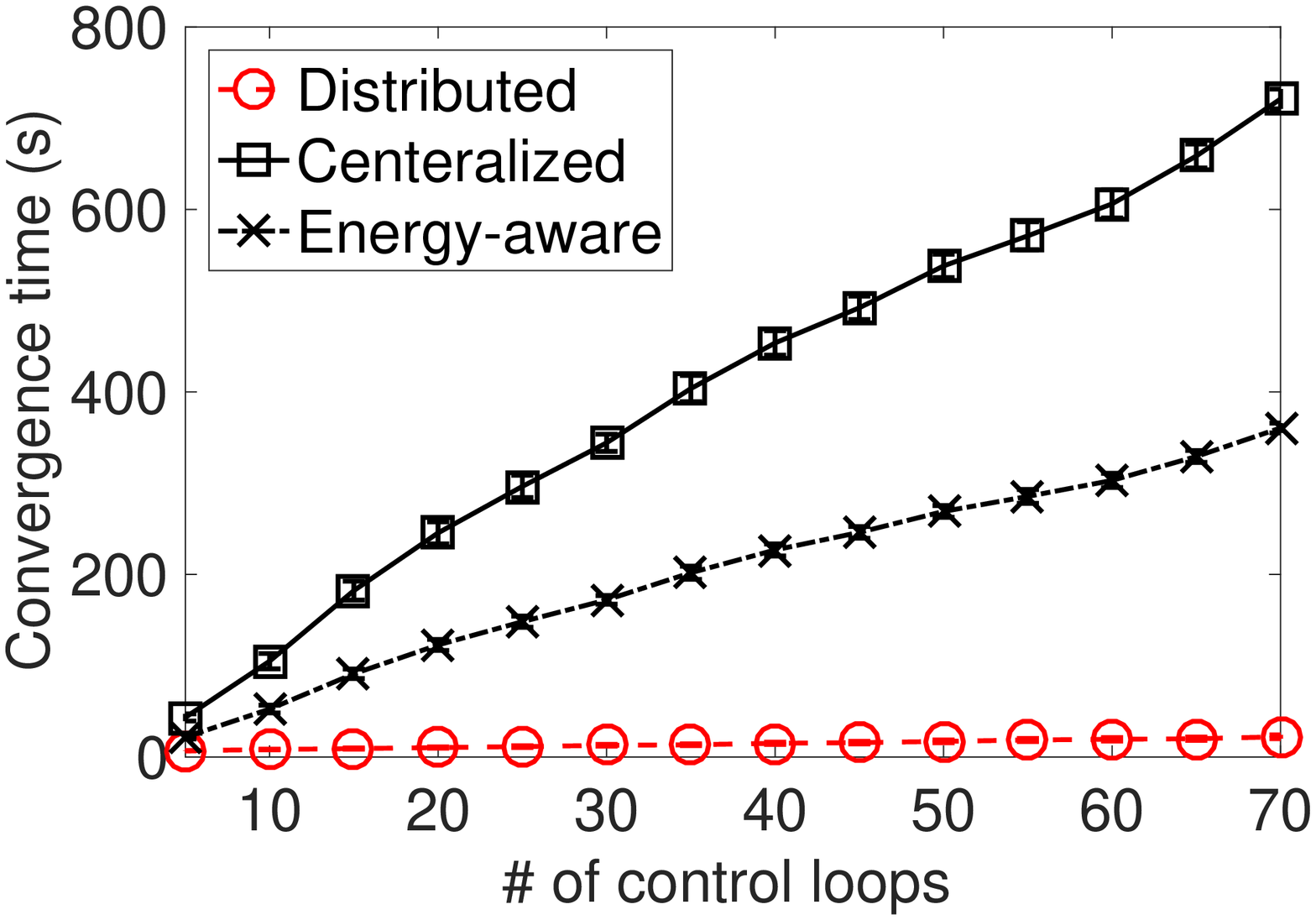}
        \caption{Convergence time}
        \label{ICDCN_routing_fig:time_controlLoops}
    \end{subfigure}
     \quad
    \begin{subfigure}[b]{0.35\textwidth}
        \includegraphics[width=\textwidth]{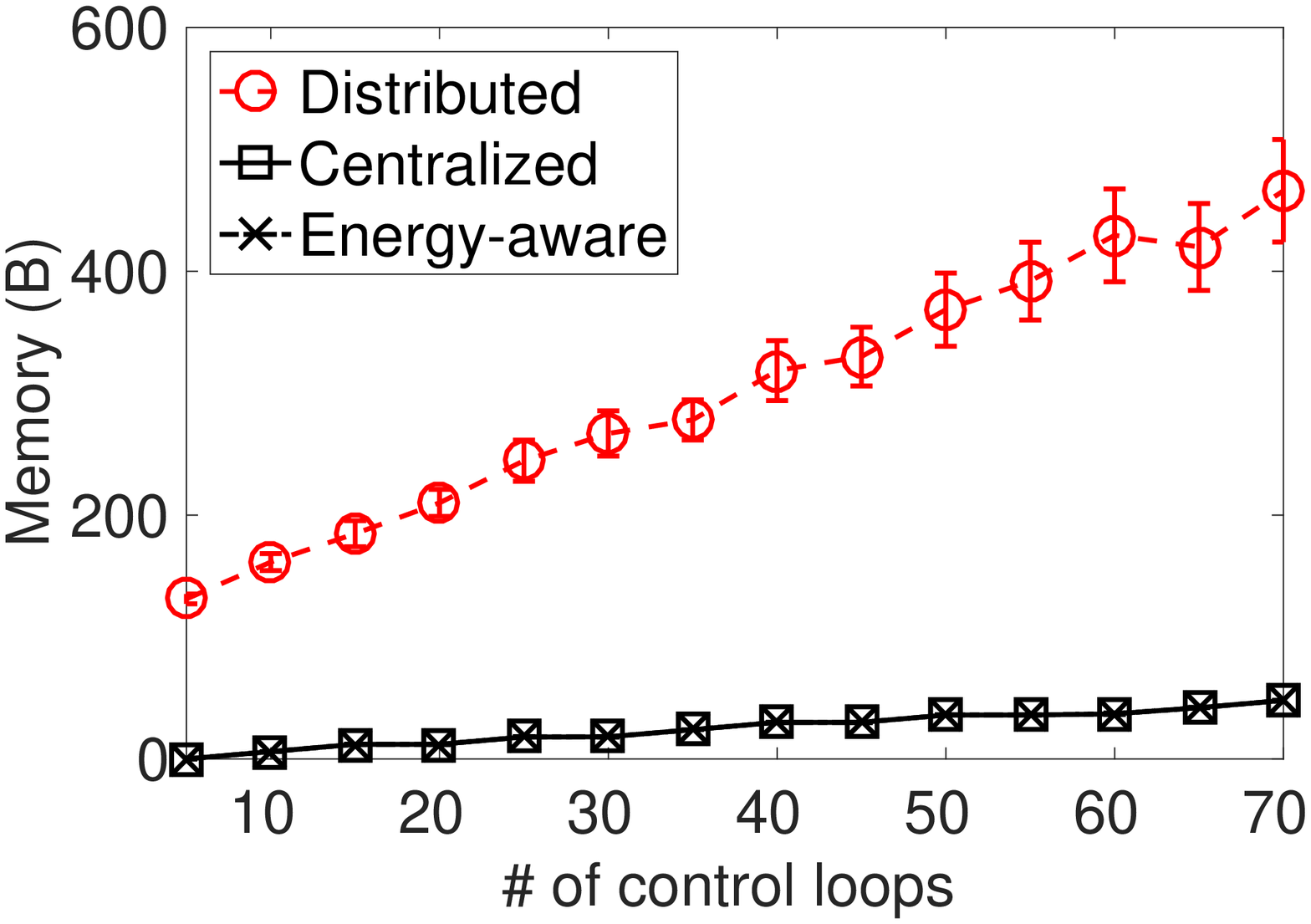}
        \caption{Memory consumption}
        \label{ICDCN_routing_fig:memory_contorlLoops}
    \end{subfigure}
   \caption{Performance of Distributed Graph Routing under varying Percentage of Control Loops}
   \label{ICDCN_routing_fig:controlLoops}
\end{figure}

This section presents the performance of the proposed approach under the scalability of control loops. We kept the number of nodes constant at $148$ and varied the number of control loops in the network from $5$ to $70$. For this simulation, the percentage of cluster heads and the number of nodes are kept constant. Thus, the energy consumption of a node is only dependent on the number of destinations in the network, which increases linearly with the number of control loops. However, this increase in the number of destinations only corresponds to a small increase (a maximum of $20mJ$ additional energy per node) in the average energy consumption since destination nodes are spread across different clusters. However, for the centralized and energy-aware approach, an increase in the number of control loops increases the number of routing tables at each node and thereby the number of broadcasts. This results in a sharp increase in energy consumption for the centralized algorithm, as shown in Fig. \ref{ICDCN_routing_fig:evergy_controlLoops}. We have observed that distributed graph routing saves a minimum of $40\%$ when compared to the energy-aware algorithm and $65\%$ when compared to the centralized algorithm. Fig. \ref{ICDCN_routing_fig:time_controlLoops} shows the convergence time required for distributed and centralized algorithms.  Similar to energy consumption, convergence time for the proposed approach is dependent on the number of cluster heads, which explains a steady increase in the execution time of the distributed algorithm when compared to a sharp increase in centralized algorithms. Fig. \ref{ICDCN_routing_fig:time_controlLoops} shows an average decrease of $40\%$ in execution time of distributed graph routing. These results conclude that distributed graph routing is scalable under a varying number of control loops.

The effect of the number of control loops on memory is shown in Fig. \ref{ICDCN_routing_fig:memory_contorlLoops}. At each node, the distributed graph routing algorithm generates routing graphs to all cluster heads in the network, while centralized generates routing graphs for the destination nodes. This accounts for the difference in memory consumption when the number of destination nodes is $10\%$ of the total number of nodes. As the number of destination nodes increase, the memory required by the distributed graph routing algorithm increases proportionally, since nodes running distributed algorithm require memory for destination nodes that are within in its cluster and cluster head information for destination nodes outside its cluster. Nevertheless, the centralized and energy-aware algorithm only consume memory for nodes that are a part of a routing graph. This accounts for the sharp increase in memory for the distributed algorithm but a steady increase in centralized algorithms. We observe that for $100\%$ source and destination nodes in the network, distributed graph routing consumes $1kB$ additional memory than centralized. This extra memory need is very less compared to available memory for the application program in WirelessHART or TelosB devices, which have a capacity of $16kB$ \cite{MemTeolsb}. These results conclude that distributed graph routing is scalable in terms of energy efficiency and execution time at the cost of small memory need for a varying number of source and destination nodes.

\subsubsection{Performance under Varying Power Level}

\begin{figure}[th]
    \centering
    \begin{subfigure}[b]{0.35\textwidth}
        \includegraphics[width=\textwidth]{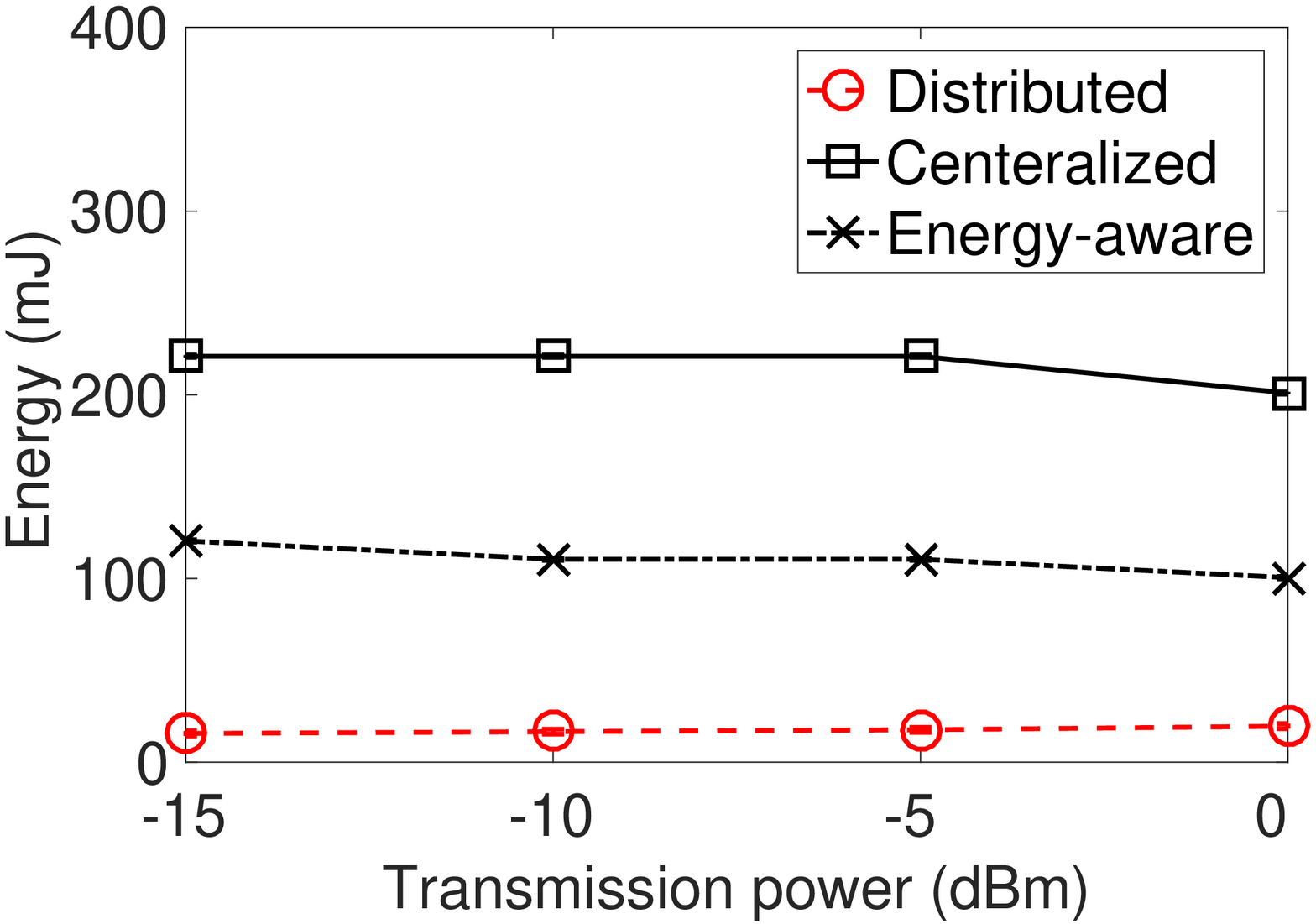}
        \caption{Energy consumption}
        \label{ICDCN_routing_fig:energy_power}
    \end{subfigure}
    \quad
    \begin{subfigure}[b]{0.35\textwidth}
        \includegraphics[width=\textwidth]{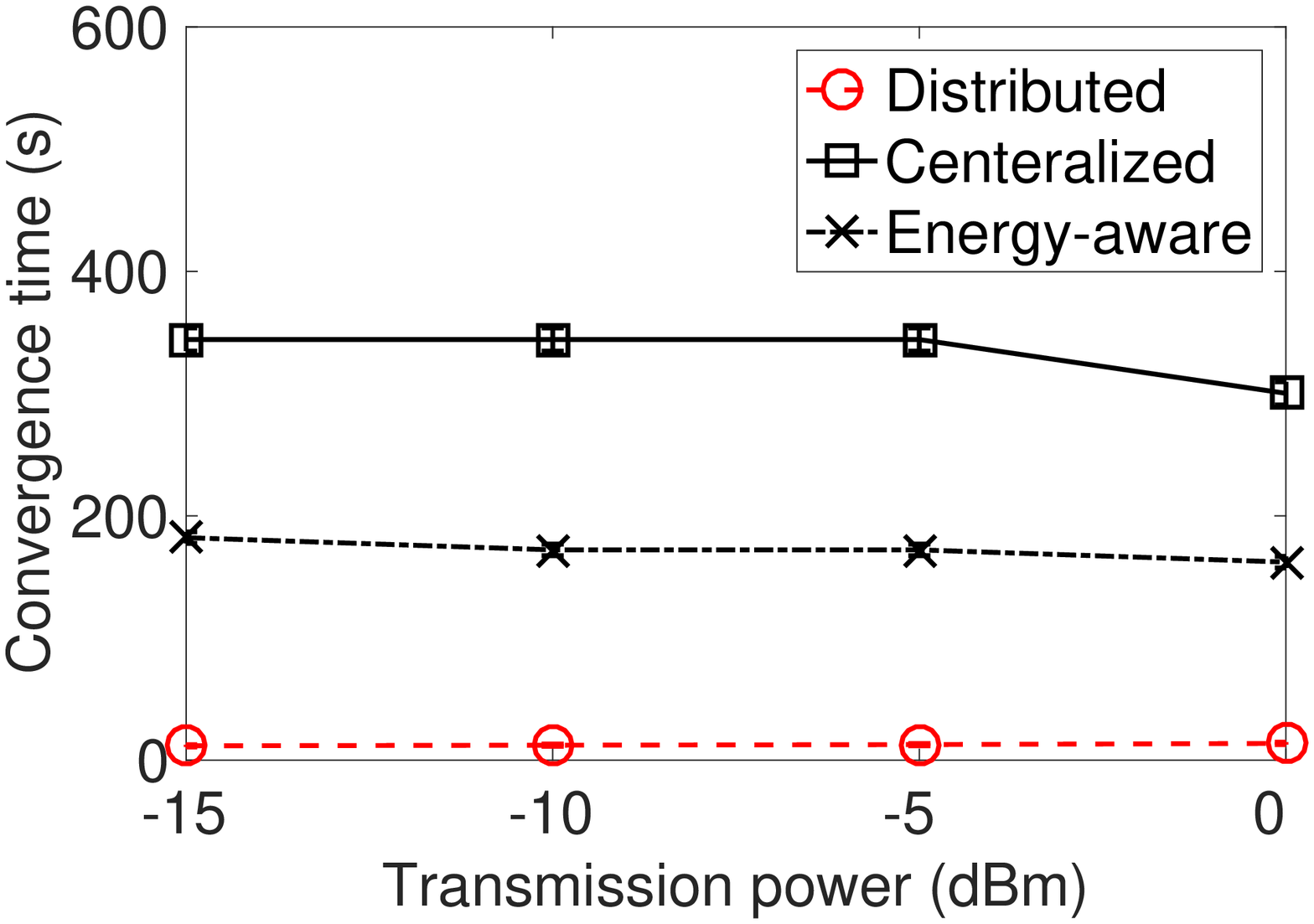}
        \caption{Convergence time}
        \label{ICDCN_routing_fig:time_power}
    \end{subfigure}
     \quad
    \begin{subfigure}[b]{0.35\textwidth}
        \includegraphics[width=\textwidth]{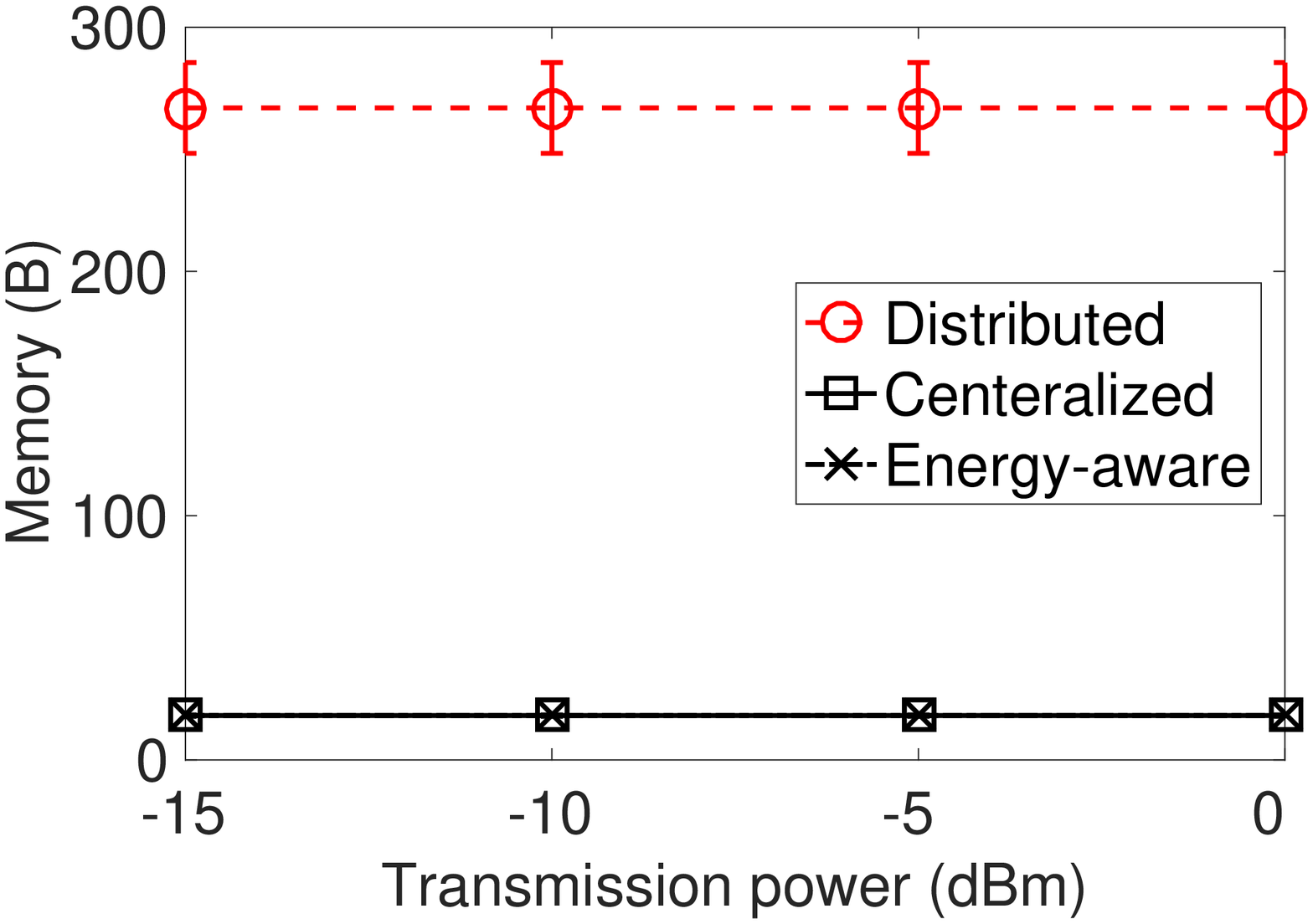}
        \caption{Memory consumption}
        \label{ICDCN_routing_fig:memory_power}
    \end{subfigure}
   \caption{Performance of Distributed Graph Routing under varying Transmission Power Levels}
   \label{ICDCN_routing_fig:VaryingPowerLevels}
\end{figure}

We have evaluated the performance of distributed graph routing under varying density in the network by varying transmission power levels. For this simulation, we considered 4 power levels $-15, -10, -5$ and $0 dBm$. Fig. \ref{ICDCN_routing_fig:energy_power} shows the energy consumption with increase in power levels. As expected, energy consumption increases linearly as the number of neighbors increase. However, energy consumption for centralized algorithms decreases due to the presence of new shorter paths that are included due to an increase in power level. This change in energy consumption for both centralized and distributed is around $100mJ$. Similar to energy consumption, the convergence time of distributed graph routing increases gradually with increasing power level for transmission due to an increase in the number of edges in the network, as shown in Fig. \ref{ICDCN_routing_fig:time_power}. However, for centralized algorithms, convergence time decreases as the number of messages in the route decrease. This decrease is minimal as the time required to collect the topology and compute the routes at the network manager is constant. These results conclude that the effect of network density on the distributed algorithm is negligible, and similar performance can be observed from centralized algorithms.

The effect of varying node density on memory is shown in Fig. \ref{ICDCN_routing_fig:memory_power}. With an increase in the number of neighbors, more paths can be generated from each node. Thus, memory consumption of the distributed algorithm increases linearly with an increase in power level or an increase in the number of neighbors. The effect of varying node density on memory is shown in Fig. \ref{ICDCN_routing_fig:memory_power}. With an increase in the number of neighbors, more paths can be generated from each node. Thus, memory consumption of the distributed algorithm increases linearly with increasing power level or increasing number of neighbors. However, for the centralized algorithms, memory consumption is fixed as the number of neighbors selected is always a maximum of $2$. These results show that despite the slight increase in the transmission power, the distributed algorithm performs much better than centralized in terms of energy and execution time at the cost of small additional memory need at nodes.

\subsubsection{Performance under Varying Number of Link Failures}

\begin{figure}[th]
    \centering
    \begin{subfigure}[b]{0.35\textwidth}
        \includegraphics[width=\textwidth]{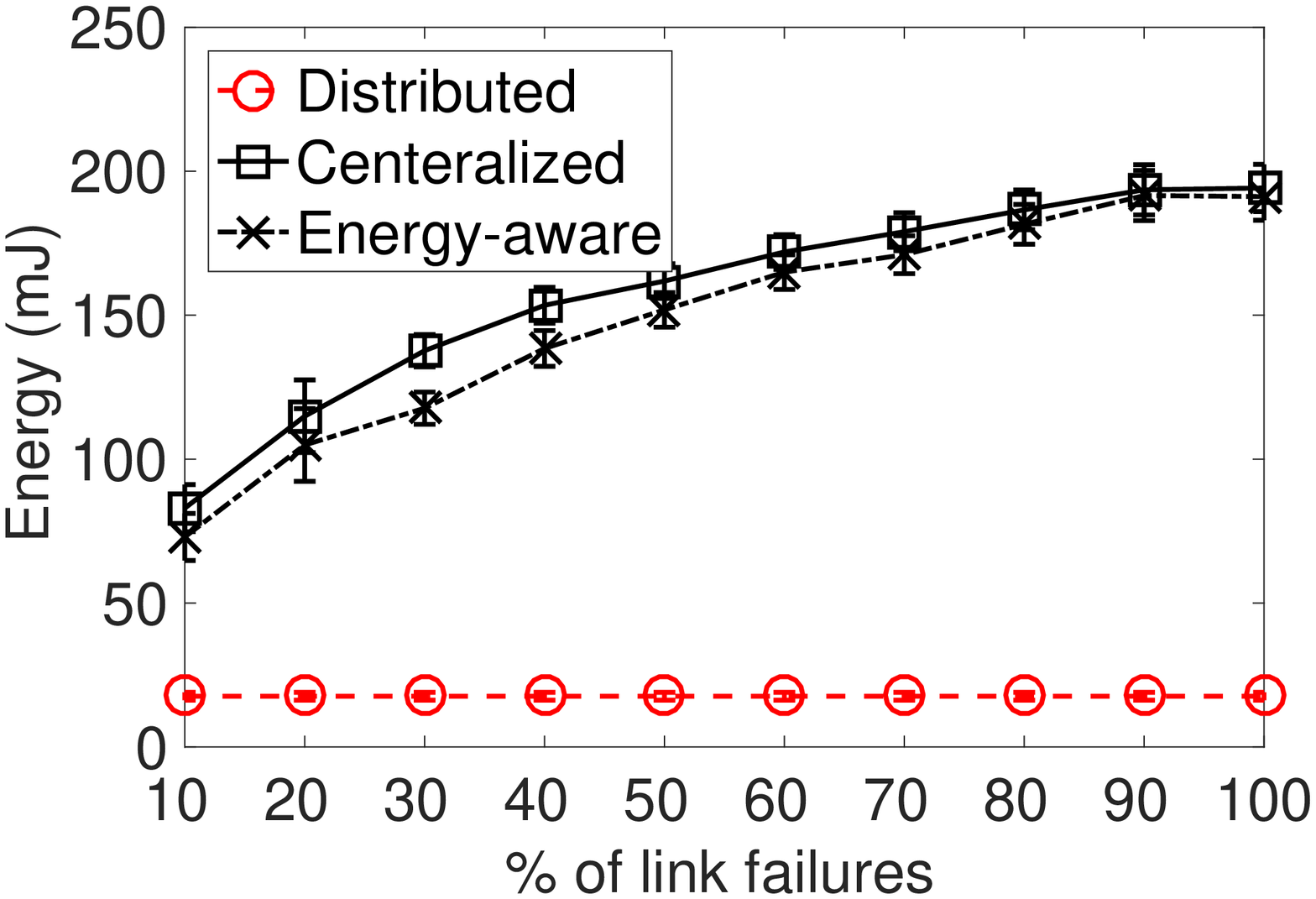}
        \caption{Energy consumption}
        \label{ICDCN_routing_fig:energy_linkQuality}
    \end{subfigure}
    \quad
    \begin{subfigure}[b]{0.35\textwidth}
        \includegraphics[width=\textwidth]{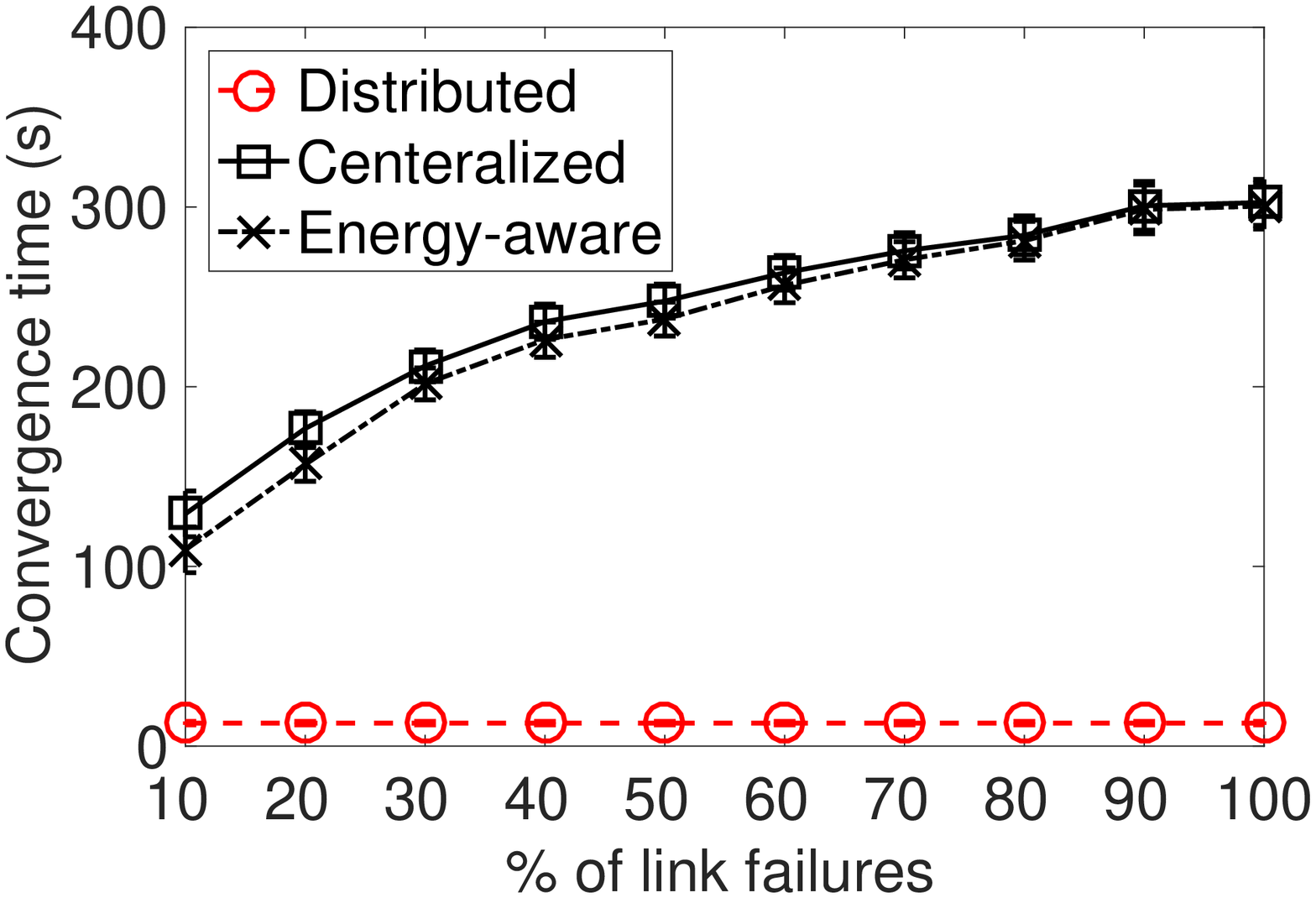}
        \caption{Convergence time}
        \label{ICDCN_routing_fig:time_linkQuality}
    \end{subfigure}
     \quad
    \begin{subfigure}[b]{0.35\textwidth}
        \includegraphics[width=\textwidth]{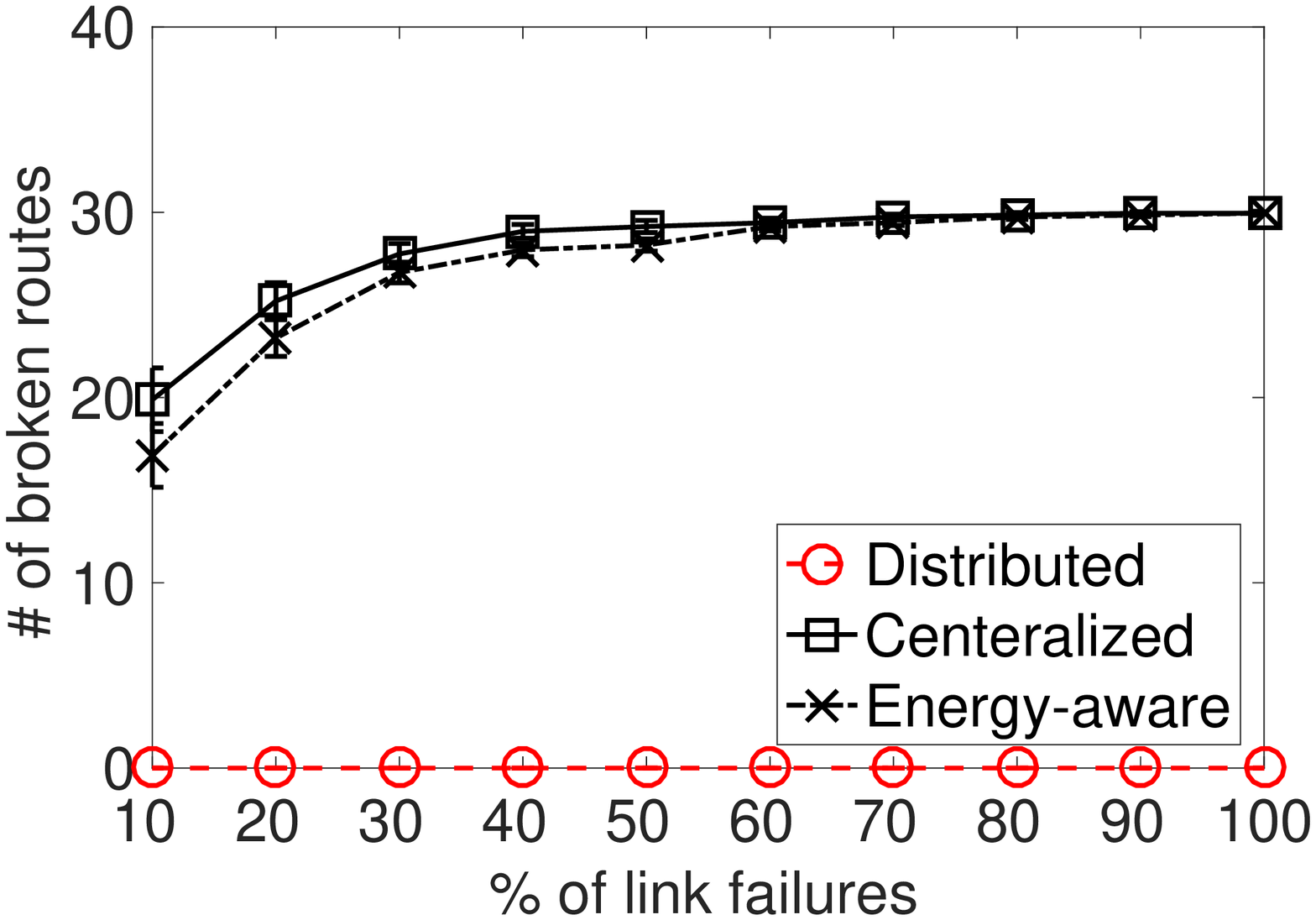}
        \caption{Reliability comparison}
        \label{ICDCN_routing_fig:routes_linkQuality}
    \end{subfigure}
   \caption{Performance of Distributed Graph Routing under varying Percentage of Link Failures}
   \label{ICDCN_routing_fig:linkQuality}
\end{figure}

The proposed distributed algorithm is naturally adaptive to network dynamics. For example, when links break, nodes recalculate routing tables and updates neighbors.  Under volatile link conditions, the node has the option to choose between using existing routes or recomputing the optimal route based on the duration of the failure. In this simulation, we have evaluated the performance of our distributed algorithm under link failures, resulting in the recomputation of routes, when all other parameters as constant. We measure the additional energy and time required to recompute the routes for the distributed algorithm and the centralized and energy-aware algorithms. We also measure the reliability of the protocols by comparing the number of broken routes that need to be recomputed. 

The effect of link failures on the number of broken routes is shown in Fig. \ref{ICDCN_routing_fig:routes_linkQuality}. For the centralized and energy-aware algorithm, the number of disconnected paths increases with the increase in link failures. This accounts for the sharp increase in the number of broken routes for the centralized and energy-aware algorithm. Since the distributed algorithm generates all possible paths from a source to destination, the number of broken routes remains almost constant. These results show that under link failure, distributed routing performs better in terms of reliability by offering multiple paths as back-up paths. 

Fig. \ref{ICDCN_routing_fig:energy_linkQuality} shows the average additional energy consumed at each node. A WirelessHART network manager will generate a new path to replace a broken path such that the reliability of packet transmission is maintained. Thus, additional energy consumed by a node is proportional to the number of paths that are broken in the network. This accounts for the sharp increase in additional energy consumption for the centralized and energy-aware approach. However, for the distributed algorithm, additional energy required was very less since the number of broken paths is very small in the distributed approach. We have observed that the distributed algorithm saves a minimum of $80\%$ energy when compared to the centralized and energy-aware algorithms. Similar to energy, convergence time is also dependent on the number of broken paths, and this value is high for centralized algorithms, as shown in Fig. \ref{ICDCN_routing_fig:time_linkQuality}. We observed a minimum of $80\%$ saving in convergence time.

\subsubsection{Performance under Varying Density of Node Deployment}

\begin{figure}
\centering
\includegraphics[width=0.35\textwidth]{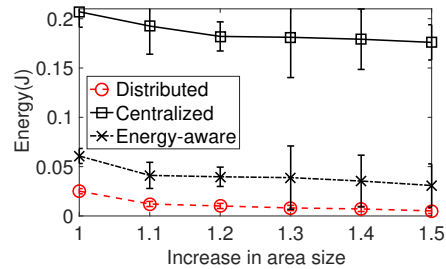}
\caption{Performance of Distributed Graph Routing under varying Deployment Area}
\label{ICDCN_routing_fig:energy_area} 

\end{figure}

We evaluated the performance of distributed graph routing under varying density for the same number of nodes by increasing the area. We used random placement in areas ranging from $1.1 - 1.5$ times of the original area to determine the performance of energy with a decrease in the density of nodes.  Fig. \ref{ICDCN_routing_fig:energy_area} shows a decrease in energy consumption with a decrease in the density of the network. In the distributed algorithm, energy decreases due to a decrease in the number of neighboring nodes.  However, in the centralized and energy-aware algorithms, the number of nodes in a routing graph remains almost the same as WirelessHART mandates that every node should have a minimum of two neighbors. Thus, there is a steady decrease in energy consumption for the distributed algorithm. For the centralized and energy-aware algorithms, energy consumption at a node remains the same.

\subsubsection{Performance under Varying Number of Clusters}

\begin{figure}[th]
    \centering
    \begin{subfigure}[b]{0.35\textwidth}
        \includegraphics[width=\textwidth]{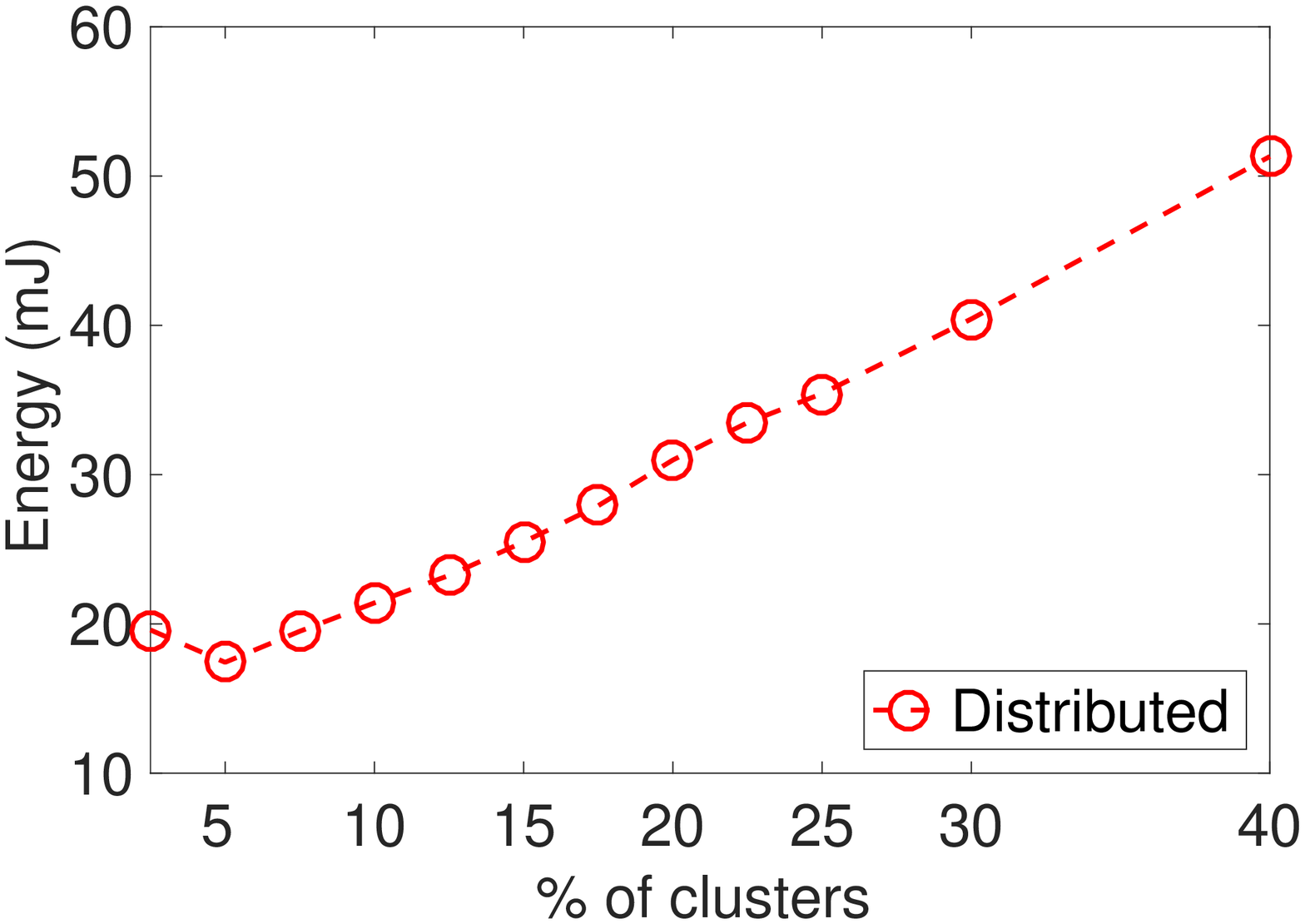}
        \caption{Energy consumption}
        \label{ICDCN_routing_fig:energy_cluster}
    \end{subfigure}
    \quad
    \begin{subfigure}[b]{0.35\textwidth}
        \includegraphics[width=\textwidth]{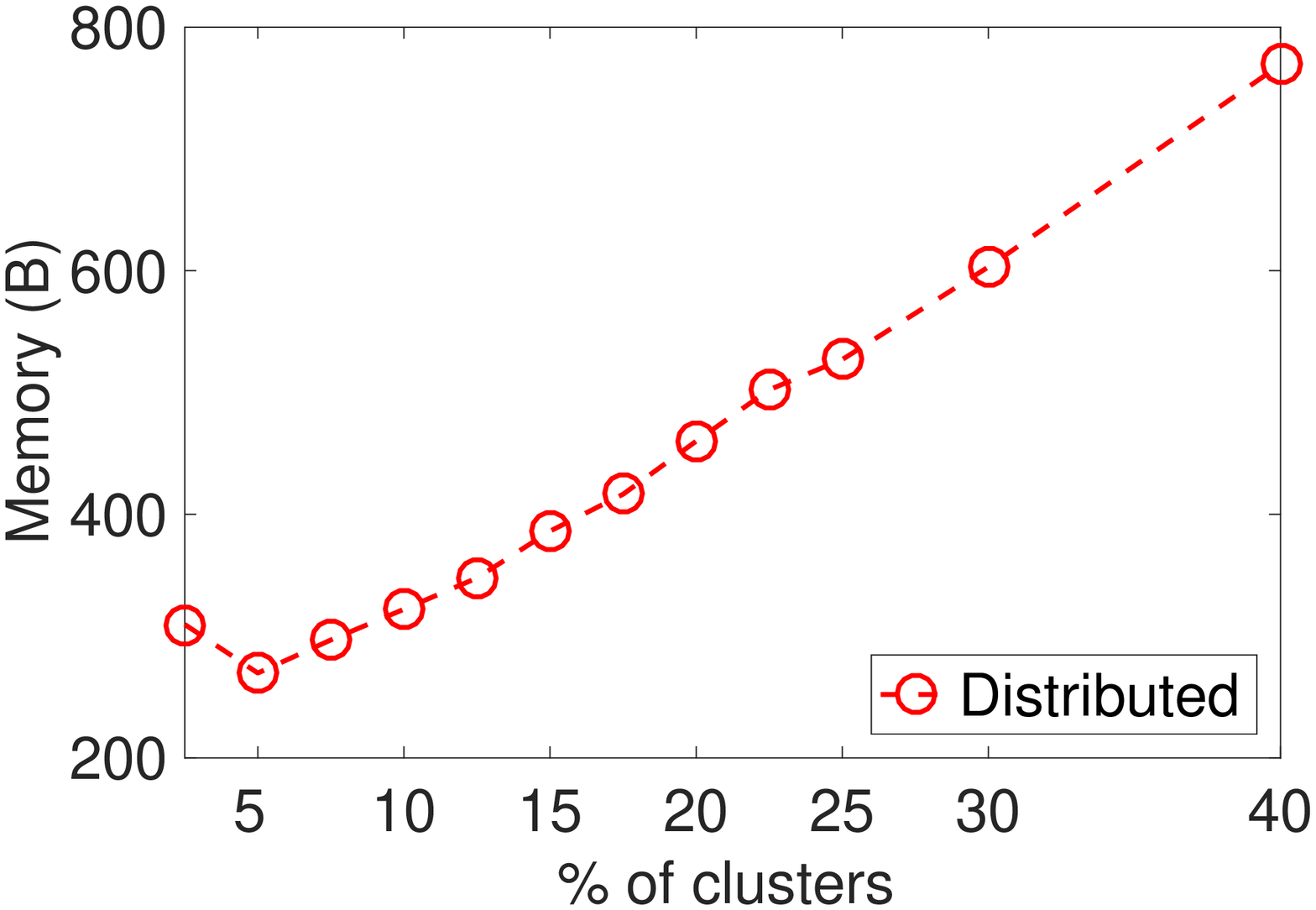}
        \caption{Convergence time}
        \label{ICDCN_routing_fig:memory_cluster}
    \end{subfigure}
   \caption{Performance of Distributed Graph Routing under varying Number of Clusters}
   \label{ICDCN_routing_fig:clusters}
\end{figure}

We evaluated the performance of distributed graph routing under a varying number of clusters. The memory consumption of a node depends on the number of routing graphs to cluster heads. Similarly, the average energy consumption of a node depends on the number of clusters. Fig. \ref{ICDCN_routing_fig:energy_cluster} and Fig. \ref{ICDCN_routing_fig:memory_cluster} show the energy and memory consumption for a varying number of cluster in the network. Initially, as the number of clusters increase, destination nodes are allocated to different clusters, thereby reducing the energy and memory consumed at each node. After the number of the cluster increases beyond $5\%$, energy, and memory consumption also increase in the network. We observe that distributed graph routing requires the least energy and memory when the number of clusters is around $5\%$ of total nodes in the network. This result gives an optimal cluster size for the clustering approach used for evaluation.

%% file: DistributedHart_TMC/TMC19.tex
\chapter{A Distributed Real-Time Scheduling System for Industrial Wireless Networks}
\label{ch:distributedhart}
		The concept of Industry 4.0 introduces the unification of industrial Internet-of-Things (IoT), cyber physical systems, and data-driven business modeling to improve production efficiency of the factories. To ensure high production efficiency, Industry 4.0 requires industrial IoT to be adaptable, scalable, real-time and reliable. Recent successful industrial wireless standards such as WirelessHART appeared as a feasible approach for such industrial IoT. For reliable and real-time communication in highly unreliable environments, they adopt a high degree of redundancy. While a high degree of redundancy is crucial to real-time control, it causes a huge waste of energy, bandwidth, and time under a centralized approach, and are therefore less suitable for scalability and handling network dynamics. To address these challenges, we propose DistributedHART - a distributed real-time scheduling system for WirelessHART networks. The essence of our approach is to adopt local (node-level) scheduling through a time window allocation among the nodes that allows each node to schedule its transmissions using a real-time scheduling policy locally and online. DistributedHART obviates the need of creating and disseminating a central global schedule in our approach, and thereby significantly reducing  resource usage and enhancing  the scalability.  To our knowledge, it is the first distributed real-time multi-channel scheduler for WirelessHART. We have implemented DistributedHART and experimented on a 130-node testbed. Our testbed experiments as well as simulations show at least $85\%$ less energy consumption in DistributedHART compared to existing centralized approach while ensuring similar schedulability.

	\input{DistributedHart_TMC/Introduction}
	\input{DistributedHart_TMC/related_work}
	\input{DistributedHart_TMC/NetworkModel}
	\input{DistributedHart_TMC/ProposedOnlineMAC}

	\input{DistributedHart_TMC/SchedulabilityAnalysis}

	\input{DistributedHart_TMC/NonUniform}
	\input{DistributedHart_TMC/evaluation}

    	\input{DistributedHart_TMC/simulations}

    \section{Summary}\label{hart_sec:Conclusion}
     We have proposed DistributedHART - a local and online real-time scheduling system for industrial IoT. DistributedHART enables local scheduling at nodes through time window allocation at the nodes. Within a time window, a node locally schedules a packet for transmission based on any existing real-time scheduling policy. Thus, DistributedHART obviates the need for creating and disseminating a central global schedule, thereby reducing resource waste and enhancing scalability. Furthermore, DistributedHART would lead to higher network utilization in the network at the expense of a slight increase in the end-to-end delay of all control loops. Through experiments on a 130-node testbed as well as large-scale simulations, we observe at least 85\% less energy consumption in DistributedHART compared to an existing centralized approach.
      	
	The performance of our schedulability test suggests that there is still room for improvement. In the future, we will derive an improved schedulability test by deriving tighter delay bounds. Nevertheless, in DistributedHART, we use existing centralized or distributed routing protocols to generate routes between sensors, actuators, and a controller. In DistributedHART, the schedule generation is not optimized for the routing protocol. A distributed joint routing and scheduling algorithm is highly challenging for WirelessHART networks and is an open problem.

%% file: DistributedHart_TMC/Introduction.tex
\section{Introduction}\label{hart_sec:introduction}
The concept of Industry 4.0 introduces the unification of industrial Internet-of-Things, cyber physical systems, and data-driven business modeling to improve production efficiency of the factories~\cite{iiot}. To ensure high production efficiency, Industry 4.0 requires industrial Internet-of-Things to be adaptable, scalable, real-time and reliable. Recent successful industrial wireless standards such as WirelessHART have shown their feasibility as a cost-efficient, real-time, and robust approach for industrial Internet-of-Things~\cite{distributedhart}. 

To make reliable and real-time communication in highly unreliable wireless environments, WirelessHART adopts a high degree of redundancy using a Time Division Multiple Access (TDMA) based Media Access Control (MAC) protocol. A time slot can be either {\em dedicated} (i.e., a time slot when at most one transmission is scheduled to a receiver) or {\em shared} (i.e., a time slot when multiple nodes may contend to send to a common receiver). To handle transmission failures, each node on a path from a sensor to an actuator is assigned two dedicated time slots and a third shared slot on a separate path for retransmission~\cite{WirelessHART2007_standard}. A network manager creates the transmission schedule {\bf centrally} and {in advance} for all nodes and then disseminates them. A centralized WirelessHART scheduler with high redundancy raises several practical challenges in achieving scalability as described below.

High level of redundancy in centralized algorithms~\cite{RTSS10paper, RTSS15paper} causes a huge waste of time and bandwidth, and hence is not scalable. For example, if the transmission of a packet along a particular link succeeds, all time slots (on the current link and redundant links) that were assigned to handle its failure remain unused. Similarly, if it fails along that particular link, all time slots that were assigned for its subsequent links to handle a successful transmission remain unused. Our experiments observed up to 70\% unused time slots in WirelessHART networks (see Section~\ref{hart_sec:ProposedMethod}). 

Furthermore, there can be events or emergencies that occur unpredictably or aperiodically. For example, a WirelessHART network in an oil-refinery may suddenly detect a safety valve displacement requiring immediate attention to avoid accidents. Existing solution handles emergencies by allocating time slots in the centrally created schedule and by stealing slots in the absence of emergencies~\cite{li2015incorporating}.  However, this approach leaves most of the slots of the periodic server unstolen, and hence unused. Thus, the network remains largely underutilized which affects the scalability of the system. 

Schedule dissemination in centralized algorithm consumes bandwidth, energy, and time, even for a smaller network or a smaller workload. Typically, hyper-period and length of the schedule increase exponentially with the increase in the number of flows or their periods, which hinders the scalability of the network. Note that, in general, periods can be non-harmonic to ensure stability or control performance~\cite{RTAS2012}. Furthermore, the mobility of nodes introduces discernible issues for a central scheduler due to the frequent changes to the network topology. In an industrial environment, moving objects like robotic arms or carts can affect link quality of nodes and change the topology of the network. Such frequent changes to the topology require frequent computation and re-dissemination of schedules. Nonetheless, the data-driven business model in Industry 4.0 introduces frequent changes to sampling rates, which also requires re-configuration and re-dissemination of schedules. Frequent re-dissemination of the schedule consumes high energy, time, and bandwidth.  Thus, fully centralized scheduling is less suitable for industrial Internet-of-Things, which considers the mobility of nodes and interfering objects. Besides, it is typically suitable for deterministic traffic patterns (like periodic traffic) arising from stationary nodes.

To address the above limitations, in this paper, we propose  a distributed real-time scheduling system for WirelessHART networks. Designing a distributed TDMA protocol with scheduling performance  close to a centralized one is highly challenging as the former has to achieve this without global knowledge. For a WirelessHART network, a distributed TDMA protocol also has to incorporate dedicated and shared slots in local scheduling.  We address these challenges by proposing DistributedHART. 
We make the following contributions in the paper. 
\begin{itemize}
\item We propose DistributedHART, the first {\bf \slshape distributed} real-time multi-channel scheduling for WirelessHART networks. DistributedHART adopts local (node-level) scheduling through a time window allocation among the nodes that allows each node to schedule its transmissions locally and online.   Thus, DistributedHART can handle any communication pattern (periodic or aperiodic) and any length of schedule. It obviates the need for creating and disseminating a global schedule.
\item We provide a schedulability test for DistributedHART that can be used to determine the real-time performance of a WirelessHART network with a high probability.
\item We have implemented DistributedHART in TinyOS~\cite{tinyos} for TelosB~\cite{energyTelosb} platform and performed experiments on a 130-node physical indoor testbed~\cite{testbedwsu} to show the effectiveness of DistributedHART. To consider more experimental scenarios, we also evaluated DistributedHART through simulations on TOSSIM~\cite{tossim} using the topology of another testbed~\cite{sha2015implementation}. In both experiments and simulations, we observe at least $85\%$ less energy consumption in  DistributedHART compared to existing centralized approach. 
\end{itemize}

DistributedHART enables local scheduling of packets at nodes through distributed time window allocation to the nodes. Thus, DistributedHART efficiently handles network and workload dynamics and obviates the need of creating a schedule centrally and disseminating it repeatedly. Furthermore, it significantly reduces the energy consumption of the devices in the network and provides scalability.

Section~\ref{hart_sec:RelatedWorks} reviews related work. Section~\ref{hart_sec:sysModel} describes the model. Section~\ref{hart_sec:ProposedMethod} describes the design of DistributedHART under the assumption that length of all time windows is constant and pre-determined. Section~\ref{hart_sec:Delay} presents the end-to-end delay analysis for DistributedHART with the same assumption. Section~\ref{hart_sec:NonUniform} describes non-uniform time window allocation for DistributedHART, and the changes to different protocols to ensure correct operation of DistributedHART. Section~\ref{hart_sec:NonUniform} also describes the changes to end-to-end latency due to non-uniform time window assignment. Section~\ref{hart_sec:latency} presents latency performance of DistributedHART and an algorithm to address the latency limitations. Sections~\ref{hart_sec:experiments} and~\ref{hart_sec:simulations} present experiments and simulations, respectively.  Section~\ref{hart_sec:Conclusion}  presents the conclusion and future work. 

%% file: DistributedHart_TMC/related_work.tex
\section{Related Work} \label{hart_sec:RelatedWorks}

Existing work in \cite{Stankovic2003_Realtime} explored the real-time scheduling for wireless networks. 
CSMA/CA based real-time scheduling has been studied in \cite{Karenos2006_real, li2005scheduling, Wang2009_FlowbasedRealTimeMultichannel, Kanodia2001_Distributed, Lu2002_RAP, He2007robust}. In contrast, WirelessHART adopts a TDMA-based protocol to achieve predictable latency bounds.
TDMA-based real-time scheduling without multi-channel communication or multi-path graph routing was studied in \cite{dujovne20146tisch, zimmerling2017adaptive, Liu2006_JiTS,   GuRTSS, Mangharam2006_Voice, palattella20146tisch, terraneo2018tdmh}.   
Real-time scheduling for data collection in WirelessHART network under tree topology was studied in \cite{Soldati2009_WirelessHARTevacuation, Zhang2009_WirelessHARTRapid}.
Real-time routing was studied in \cite{songRTAS11, wu2016maximizing, modekurthy2018distributed, wu2018real, kunal2019adaptive}. Schedule modeling for a WirelessHART network was studied in \cite{upenn}. Priority assignment of packets in WirelessHART network was studied in \cite{ECRTSpaper}. Channel assignment to nodes in a WirelessHART network was studied in \cite{channelallocation, gunatilaka2017impacts}. Security vulnerabilities of channel hopping sequence for WirelessHART was studied in \cite{cheng2019cracking}. Schedulability analysis for industrial wireless networks was studied in \cite{RTAS11paper, TC, RTSS15paper, modekurthy2018utilization}. These works did not focus on the real-time scheduling of packets.

Existing work in \cite{RTSS10paper} showed that the real-time scheduling for flows in WirelessHART networks is NP-hard and proposed real-time scheduling policies for WirelessHART. A flexible retransmission policy for WirelessHART networks was proposed in \cite{brummet2018flexible}. Scheduling under multiple co-existing wirelessHART networks was studied in \cite{jin2017reliability}. Mobility aware real-time scheduling of packets was studied in \cite{dezfouli2016mobility, dezfouli2016real}. These papers adopt a fully centralized scheduler that creates a schedule in advance, and they rely on the current WirelessHART scheduling approach with high redundancy. Such an approach causes a huge waste of time, bandwidth, energy,  and memory, making it less suitable for dynamics and scalability. In this paper, we aim to address these limitations and propose an online and distributed real-time scheduling system for WirelessHART.

Orchestra \cite{duquennoy2015orchestra}, D$^2$-PaS \cite{song2018distributed, zhang2018fd, song2018distributed} and DiGS \cite{shidigs, shi2019distributed}  are the recent distributed scheduling approaches for a multi-hop wireless network. However, they have the following limitations. First, they only consider a single channel protocol while WirelessHART uses multiple channels. Second, they do not consider shared slots while WirelessHART adopts graph routing with both dedicated and shared slot transmissions. In Orchestra and DiGS, the end-to-end communication latency of a flow is in the order of the number of nodes in the network. Such a high latency is less suitable for real-time communications. Due to these limitations, Orchestra, D$^2$-PaS, and DiGS are less suitable for WirelessHART. In contrast, DistributedHART is a practical scheduling system for WirelessHART that considers multichannel and graph routing, which are highly critical for wireless control applications in unreliable environments and is not limited to sparse traffic.

%% file: DistributedHart_TMC/NetworkModel.tex
\section{Background and System Model}  \label{hart_sec:sysModel}
WirelessHART networks operate in the 2.4GHz band and are built based on the physical layer of IEEE 802.15.4. They form a multi-hop mesh topology of nodes - field devices, multiple access points, and a Gateway. The {\slshape field devices} are wirelessly networked sensors and actuators. Each node contains a {\slshape half-duplex} omnidirectional radio transceiver that cannot both transmit and receive a packet at the same time and can receive from at most one sender at a time. {\slshape Access points} provide redundant paths between the wireless network and the Gateway. The {\slshape network manager} and the controller remain at the {\slshape Gateway}. The network employs feedback control loops between sensors and actuators.  Sensors measure process variables and deliver to a controller. The controller sends control commands to the actuators  which then operate the control components 
to adjust the physical processes.

Transmissions in a WirelessHART network are scheduled based on a multi-channel TDMA protocol. The network employs the channels defined in IEEE 802.15.4.  In large networks spread over a wide area, two distant nodes (which do not interfere with each other) can use the same channel in parallel, i.e., we allow spatial re-use of channels. Time in the network is globally synchronized. A receiver acknowledges each transmission from a sender. Both, a transmission and its acknowledgment should happen in one 10ms time slot. A transmission time slot can be dedicated for a receiver and a sender, or shared between multiple senders and a receiver. In a {\slshape dedicated slot}, only one sender is allowed to transmit to a receiver. In a {\slshape shared slot}, multiple senders can attempt to send to a common receiver. To mitigate collisions in a shared slot, a WirelessHART network adopts the random back-off policy according to the standard.

For enhanced reliability, the network adopts {\slshape graph routing}~\cite{WirelessHART2007_standard}. A {\em routing graph} is a directed list of loop-free paths between a source and a destination. Each node in a routing graph has at least two neighbors that provide redundant paths to a destination. Graph routing allows to schedule a packet using multiple channels on multiple time slots to deliver a packet through multiple paths, thereby ensuring high reliability in highly unreliable environments. A routing graph consists of an uplink graph and multiple downlink graphs. An uplink graph connects all sensors to controllers while a downlink graph connects a controller to an actuator. 

We consider there are $n$ real-time flows in the system denoted by $ F = \{F_1, F_2, . . . , F_n\}$. The period and deadline of a flow $F_i$ are denoted by $T_i$ and $D_i$, respectively, where $D_i\le T_i$. Our system is applicable to fixed or dynamic priority assignment.  In practice, flows may be prioritized based on deadlines, periods, or criticalities of the loops. In this chapter, we use {\em flow} and {\em control loop} interchangeably.

Here we give an outline of the current centralized scheduling approach adopted in WirelessHART networks. For a control loop scheduling in the uplink graph, the network manager allocates two dedicated slots for each device on the primary path of a graph route starting from the source. The second dedicated slot on the same path handles retransmissions in case of transmission failures on the first dedicated slot.  Then, to handle failures of both transmissions along a link, it allocates a third shared slot on a separate path to handle another retry. The links in the downlink graph are scheduled similarly. 

In a centralized scheduling approach, a central manager creates a global schedule in advance, which is split into superframes. A  {\slshape superframe} is a series of time slots representing the communication pattern of a set of nodes and it repeats after the completion of the series. The manager disseminates the superframes among the nodes.

\begin{figure}[h]
	\centering
	\includegraphics[width=0.32\textwidth]{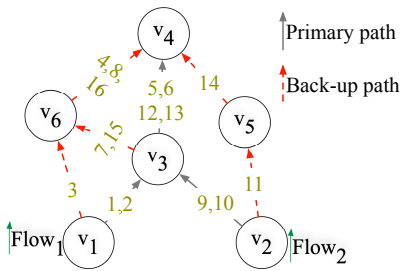}
	\caption{Example of Scheduling in WirelessHART}
	\label{hart_fig:scheduling}
\end{figure}

Fig.~\ref{hart_fig:scheduling} shows an example of transmission scheduling of two flows $\text{Flow}_\text{1}$ and $\text{Flow}_\text{2}$ from node $\text{v}_\text{1}$ and $\text{v}_\text{1}$, respectively, to an access point $\text{v}_\text{4}$ in a network of $6$ nodes. In Fig.~\ref{hart_fig:scheduling}, the label on a link refers to its transmission time slot. The primary path for $\text{Flow}_\text{1}$ is $\text{v}_\text{1} \rightarrow \text{v}_\text{3} \rightarrow \text{v}_\text{6}$, and the primary path for $\text{Flow}_\text{2}$ is $\text{v}_\text{2} \rightarrow \text{v}_\text{3} \rightarrow \text{v}_\text{6}$. Fig.~\ref{hart_fig:scheduling} shows dedicated transmission links using a solid gray line and shared slots by a dashed red line. Each link on the primary path is allocated two dedicated transmission time slots. Link $\text{v}_\text{1} \rightarrow \text{v}_\text{3}$ and node $\text{v}_\text{2} \rightarrow \text{v}_\text{3}$ are scheduled to use slots 1,2 and 9,10, respectively, for dedicated transmission of packets. Similarly, link $\text{v}_\text{3} \rightarrow \text{v}_\text{4}$ is scheduled time slots $5,6$ to transmit packets of flow $\text{Flow}_\text{1}$ and time slots $12, 13$ to transmit packets of flow $\text{Flow}_\text{2}$. A similar approach is used to schedule packets on the shared links. Here, the length of the superframe is 16 time slots, i.e., the communication schedule repeats after every 16 time slots.

In this work,  our objective is to develop a real-time distributed scheduling system where each node can locally schedule its transmissions. Generating routes is not our focus. We generate routes using the distributed graph routing algorithm proposed in~\cite{modekurthy2018distributed}, however DistributedHART works with any graph routing algorithm.

%% file: DistributedHart_TMC/ProposedOnlineMAC.tex
 \section{The Design of DistributedHART} \label{hart_sec:ProposedMethod}
In the existing centralized scheduling approaches, the network is largely underutilized. For example, in Fig. \ref{hart_fig:scheduling}, if the packets from both flows $\text{Flow}_\text{1}$ and $\text{Flow}_\text{2}$ are successful in the first attempt, 4 time slots (namely 1, 5, 9, and 10) are used out of 16 pre-allocated slots. Although the central manager assigns redundant slots for worst-case scenarios, it leaves 75\% of the slots  unusable under good network conditions. 

Since a transmission time slot is associated with a link and a flow, disseminating the schedule would require several messages. Furthermore, each message would require $20$ to $30$ms for dissemination through the network. In the event of network/workload dynamics, the schedule re-computation and re-dissemination can cause long delays and consume very high energy.
Fig. \ref{hart_fig:scheduling} shows an example of transmission scheduling from node $a$ to an access point (AP) for one control loop in a network of 4 nodes. When there is no transmission failure in the network, the packet will use slot 1 and 3 to reach $AP$. 

\begin{figure}[!h]
	\centering  
	\includegraphics[width=0.4\textwidth]{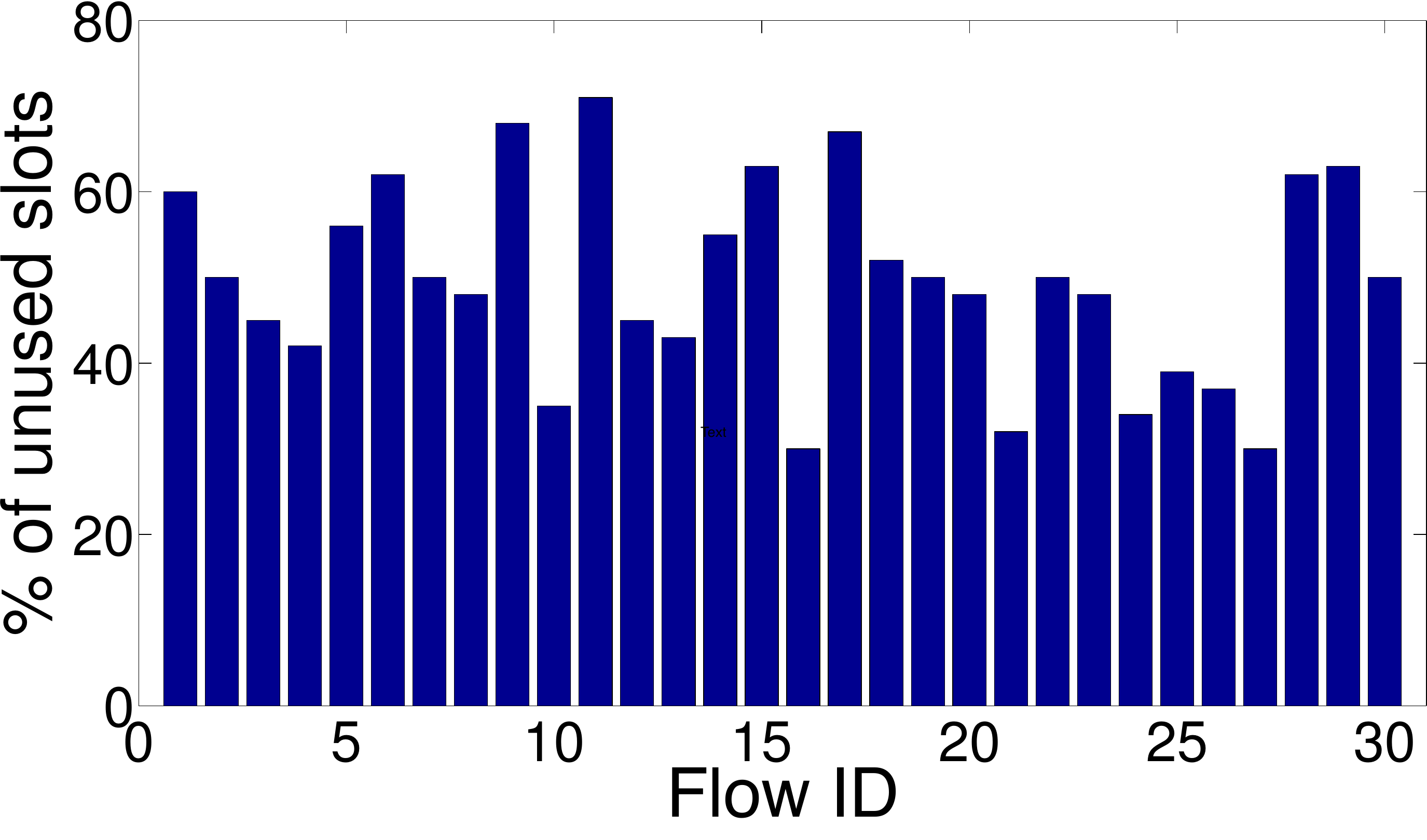}
   	\caption{Percentage of Unused Slots per Flow}
	\label{hart_fig:unused}
\end{figure}

In an experiment conducted on 30 flows on a testbed of 69 nodes, work in \cite{RTSS15paper} observed that 70\% of the slots were unused for a flow in a run. Thus, the network remains largely underutilized.  Although redundancy is crucial to real-time control for handling worst-case scenarios, such scheduling with high redundancy causes a huge waste of energy, bandwidth, time,  memory (to store schedule), and is less suitable for network/workload dynamics and scalability. 

To address these issues, we propose DistributedHART which offers a distributed scheduling system for WirelessHART networks. We describe the design of DistributedHART below.

\subsection{Distributed Scheduling}

The essence of our approach is to enable local and online scheduling at the nodes. To do so, in DistributedHART, we propose to assign time windows (a collection of time slots) to nodes rather than assigning transmission time slots to flows. In this Section, for the sake of simplicity in explanation, we use a uniform time window selection where each node selects a time window of length $w$ time slots. In DistributedHART, during a node's transmission time window, it locally selects and transmits an available packet, from its queue, for transmission based on a real-time scheduling policy.

In DistributedHART, a node locally chooses the number of redundant slots a packet needs for successful transmission to the next node. If a packet transmission is successful on the first attempt, a node can use the rest of the time slots in its window to transmit other packets in its queue (i.e., the number of redundant slots needed is $0$).  If a packet transmission is not successful in the first attempt, a node can re-transmit a packet on three time slots, as specified by the WirelessHART standard. Furthermore, if a node's queue is empty, it can re-transmit a packet on more than three time slots.  

In DistributedHART, transmission time windows repeat after a fixed interval. Thus, nodes in DistributedHART can periodically transmit all packets within its queue to their respective destinations. We define an {\slshape epoch}, as the interval after which the time windows repeat. Assuming that the network requires $\gamma$ unique time window allocations, the length of an epoch is given as $\gamma \times w$.

\begin{figure}[ht]
\centering
\includegraphics[width=0.3\textwidth]{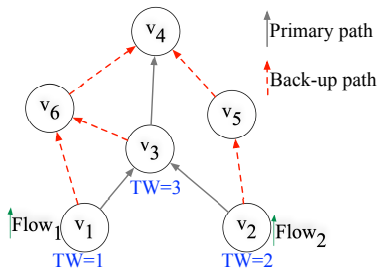}
\caption{Example of Time Window Allocation in DistributedHART}
\label{hart_fig:exDistHart}
\end{figure}

\begin{figure}[ht]
\centering
\includegraphics[width=0.49\textwidth]{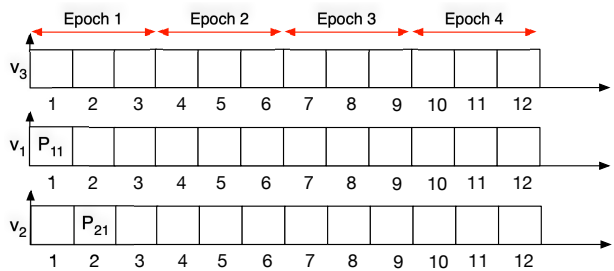}
\caption{Example of Packet Scheduling in DistributedHART}
\label{hart_fig:exDistHart_scheduling}
\end{figure}

Fig. \ref{hart_fig:exDistHart_scheduling} show a possible time window selection for nodes in the primary path of a flow for the network shown in Fig. \ref{hart_fig:exDistHart}. For this time window selection, $w = 1$ and $\gamma = 3$, and hence, the length of the epoch is 3. During time slot 1, node $\text{v}_\text{1}$ can transmit a packet of flow $\text{Flow}_{1}$ to node $\text{v}_\text{3}$. Similarly, during time slot 2, node $\text{v}_\text{2}$ can transmit a packet of flow $\text{Flow}_{2}$ to node $\text{v}_\text{3}$. If the packet transmission from $\text{v}_\text{2}$ to $\text{v}_\text{3}$ is not successful, $\text{v}_\text{2}$ can use time slot $5$ of epoch 2 to make a second dedicated transmission to $\text{v}_\text{2}$. If the packet transmission from $\text{v}_\text{2}$ to $\text{v}_\text{3}$ is successful, node $\text{v}_\text{3}$ has 2 packets in its queue, and $\text{v}_\text{3}$ can use any real-time scheduling policy to determine the next packet to transmit.

Since packet scheduling in not pre-determined in DistributedHART, a node can use any unused time slot within its transmission time window to transmit packet/s from new/existing flows. Thus, local and online scheduling in DistributedHART improves network utilization, scalability, and reliability. Furthermore, the local scheduling of packets (within a window) obviates the need for creating and distributing a schedule in advance. Thus, DistributedHART consumes less energy, memory, and bandwidth (even under frequent dynamics) when compared to centralized algorithms.

To ensure reliability in communication and minimize collisions in the network, DistributedHART generates a conflict-free transmission time window and channel allocation.
We consider that a set of transmissions on the same channel is  {\bf\slshape conflict-free} if the Signal-to-Noise plus Interference Ratio (SNIR) of all receivers exceeds a threshold. In such a model, we say that two nodes $a$ and $b$ are  conflict-free if both receivers can successfully receive a packer. Similarly, we say that two nodes $a$ and $b$ are  {\bf\slshape in conflict} if and only if simultaneous transmissions from $a$ and $b$ cause radio interference at a receiver. To minimize such conflicts, each node first collects an interference model of the network using Signal-to-Noise plus Interference Ratio (SNIR) such as the RID protocol \cite{zhou2005rid}.

Using the interference model, each node performs a receiver based channel allocation based on vertex coloring proposed in \cite{rhee2006drand}.
After channel allocation, each node performs time window allocation using distributed vertex coloring. In time window allocation, two nodes are assigned different time windows if they remain in conflict even after channel allocation. DistributedHART allows spatial reuse where many non-conflicting nodes will have the same time window and transmit simultaneously. 

After the time window allocation, each node is aware of its transmission time windows and all its neighbors transmission time windows. Thus, each node knows when to expect a packet and when to send a packet. This information is useful to determine when a node can go to sleep.

In DistributedHART, nodes execute distributed channel and time window allocation during network initialization and under some network dynamics (e.g., when routes are affected). Workload dynamics or some network dynamics (e.g., that does not affect routes) will not trigger these algorithms, keeping the overhead of DistributedHART low.

Note that, DistributedHART is a novel scheduling system for WirelessHART. DistributedHART proposes local and online scheduling of packets, where a node decides which packet to transmit within its transmission time window rather than a central manager. To generate nodes' transmission time window, DistributedHART uses existing algorithms such as RID (to generate a conflict graph) and DRAND (to perform vertex coloring on the conflict graph). Although these algorithms are important to the working for DistributedHART, they individually do not provide a real-time and reliable scheduling policy for industrial wireless networks, which is provided by DistributedHART. Furthermore, DistributedHART is compatible with any distributed conflict graph generation methods and distributed vertex coloring algorithms for the generation of transmission time windows.

 We describe channel and time window allocation, local scheduling policy, and online scheduling as a dedicated and shared slot in the following sections.  

\subsection*{Channel and Time Window Allocation:} 

The first step in channel and time window allocation is to generate an interference model. This interference model is used to generate two types of conflict graph a receiver conflict graph and a transmission conflict graph. On a receiver conflict graph, we perform a receiver based channel allocation. A channel allocation does not necessarily solve all conflicts. To remove all conflicts, we perform time window allocation on a transmitter conflict graph. This section describes these steps.

\textbf{Conflict Graph Generation. } The objective of generating a conflict graph is only to identify nodes that conflict with each other. DistributedHART uses RID protocol \cite{zhou2005rid}, which generates a conflict graph in a distributed approach without generating any tree structure. The first step in conflict graph generation is for each node to generate a physical interference model of the network using Signal-to-Noise plus Interference Ratio (SNIR). In a physical interference model, a node maintains the SNIR value of each transmitter in its neighborhood. Neighboring nodes exchange their physical interference model, to compute a conflict graph, as specified in the RID protocol. Note that both WirelessHART and DistributedHART network assume each node maintains a physical interference of the network to maintain neighbors and generate routes. Thus, the additional overhead in DistributedHART is to exchange the physical interference model with neighbors. Consequently, a conflict graph construction method is distributed and efficient in practice.

\textbf{Receiver Based Channel Allocation. } In a receiver based channel allocation, all nodes that receive a packet are assigned a channel. Any transmission to this node is made on the allocated channel. To maximize the number of simultaneous transmissions, two receiver nodes, where simultaneous transmission on the same channel interferes each other, are allocated different channels. An optimal channel allocation that maximizes the number of simultaneous transmissions is known to be NP-Hard \cite{ghosh2009multi}. To generate a channel assignment, we use a distributed heuristic where we use distributed vertex coloring on a receiver conflict graph. We define a \textit{receiver conflict graph} as a graph (over all nodes) in which two nodes are connected by an edge if and only if a packet transmission to one node interferes the other. We use DRAND \cite{rhee2006drand} to perform distributed vertex coloring.

\textbf{Time Window Allocation. } To remove all conflicts, we perform a time window allocation to each transmitting node. We assign time windows to nodes such that they can transmit to any of its neighbors without interfering (or being interfered by) other nodes. A node transmits a packet during its time window on the receiver's channel. To allocate time windows, we first represent all remaining conflicts using a transmitter conflict graph. In a \textit{transmitter conflict graph}, two transmitters have an edge if simultaneous transmissions by both transmitters can lead to a collision at one of the intended recipients.  We use distributed vertex coloring using DRAND \cite{rhee2006drand} on transmitter conflict graph to compute non-conflicting transmission windows at each transmitter node. 

\subsection*{Handling Dynamics:}
An industrial wireless network can be highly unpredictable, with frequent changes to the networks. In DistributedHART, when the packet transmission fails frequently, nodes use the management and maintenance cycle to generate the new physical interference model. The new physical interference model is reported to the central network manager. Nodes/network manager (in case of a centralized routing algorithm) compute/computes the new routes based on the collected physical interference model. Nodes exchange the collected physical interference model to generate new receiver and transmission conflict graphs. If the current channel and time window allocation generates a conflict-free transmission time window and channel allocation, then no action is taken. However, if some nodes conflict with each other with the current assignment, then these nodes execute the distributed vertex coloring algorithm to compute the conflict-free channel and transmission time windows.

In centralized routing, the nodes collect and maintain an interference model during management and maintenance cycles. The network manager collects the interference model to generate routes and schedules for nodes. The network manager then disseminates the scheduler to all the nodes.

Exchanging physical interference model between neighbors and executing distributed vertex coloring for channel and time window allocation is typically less energy consuming when compared to schedule re-dissemination in a large network. Thus, the overhead of generating a schedule in DistributedHART under network dynamics is smaller than that of existing centralized WirelessHART algorithms, and hence, handling network dynamics is a key advantage of DistributedHART.

\subsection{Scheduling Policy} 
DistributedHART can work with any type of priorities - fixed or dynamic.  To explain local scheduling policy for dynamic priority, we consider {\bf\slshape EDF (Earliest Deadline First)} as an example here. EDF assigns priorities dynamically to packets according to their absolute deadlines. 
Since, in our method, we adopt node-level scheduling, each node has to adopt EDF policy locally. Namely, among the packets that it has to transmit or forward, the one with the shortest absolute deadline will have the highest priority. 

\begin{figure}[ht]
    	\centering
	\includegraphics[width=0.49\textwidth]{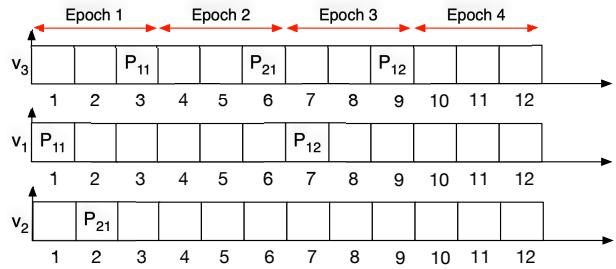}
	\caption{Example of Local Scheduling in DistributedHART} 
	\label{hart_fig:localScheduling}
\end{figure}

An example of local scheduling at each node is shown in Fig. \ref{hart_fig:localScheduling}, where $P_{11}, P_{12}$ represent the first and second packet of flow $F_1$, respectively. At time window $1$, $v_1$ transmits packet $P_{11}$ to node $v_3$. At time window $2$, $v_2$ transmits $P_{21}$ to node $v_3$. At time window $3$, node $v_3$ has packets $P_{11}$ and $P_{21}$ from flows $F_1$ and $F_2$. Based on the EDF policy, node $v_3$ selects packet $P_{11}$ for transmission on time window $3$ and packet $P_{21}$ remains in the until the next epoch.

To explain local scheduling policy for fixed priority scheduling, we consider {\bf\slshape Deadline Monotonic (DM)} policy as an example here. DM assigns priorities to flows according to their deadlines. The flow with the shortest deadline acquires the highest priority. If the deadline is equal to the period, the schedule generated by rate monotonic policy and DM are the same.

Since source nodes can update their sampling period/deadline and aperiodic events can occur (having their own deadlines), we do not rely on the network manager to assign priorities. Instead, the source node will append the period (or deadline) information in the packet.  Thus, every intermediate node will know the priorities of the packets at its buffer and schedule accordingly. Thus, the network can handle changes due to plant/workload dynamics locally, and the manager need not update the entire schedule. Management and diagnostic superframes can run in parallel with the highest priority.

\subsection{Online Scheduling as Dedicated and Shared Slots} 
A key \textbf{challenge} for DistributedHART is to incorporate both dedicated and shared slots. WirelessHART standard defines shared slot as a time slot, where many nodes transmit simultaneously to the same node. We adopt the following technique to handle shared and dedicated slots.  

A node can use a slot within its time window as a dedicated or a shared slot, and a node can also use time slots outside its time window as a shared slot.  A node $v_1$ starts the transmission at the beginning of the slot after channel setting if it intends to use it as a dedicated slot. If node $v_1$'s transmission on 2 dedicated slots fails, then it does not need to wait for an explicitly assigned shared slot. A shared slot can be either in its assigned window or outside it. In either case, it waits for some time $\theta$ in the slot during which it keeps sensing the channel.  If the channel is busy, then it concludes that the corresponding receiver is involved with a dedicated slot, and it leaves the slot. If it does not sense any busy channel within $\theta$ time, then it can use it as a shared slot. If the current slot is within its transmission window, then waiting for $\theta$  time allows other nodes to use it as shared. If the current slot is outside its transmission window, then waiting for $\theta$ allows it to know if other nodes are using it as a dedicated slot. In either case, the node makes a small random back-off before transmitting. 

In some cases, a hidden terminal $v_2$ may transmit to $v_3$ as dedicated slot without $v_1$ knowing it. In these scenarios, we need to ensure that $v_2$'s packet ($P_2$) is received correctly as it was transmitted in $v_2$'s dedicated slot. We enable {\bf capture effect} \cite{captureeffect}  of the radio at the receiver to avoid/handle collision between $P_2$ and $P_1$ (packet sent by $v_1$). 

\begin{figure}[ht]
	\centering
	\includegraphics[width=0.49\textwidth]{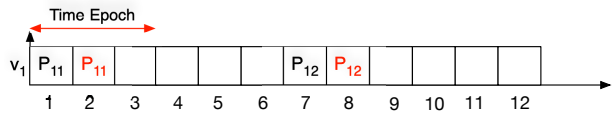}
	\caption{Example of Scheduling as Dedicated and Shared Slot}
	\label{hart_fig:sharedExample}
\end{figure}

An example of dedicated and shared slot scheduling at node $v_1$ (for network shown in Fig. \ref{hart_fig:exDistHart} with periods $T_1 = 12$ and $T_2 = 24$ slot) is shown in Fig. \ref{hart_fig:sharedExample}. If the transmission of packet $P_{11}$ fails on time window 1, then $v_1$ can use time slots in time window $2$ as shared slots even though they lie outside the transmission window of $v_1$.

\subsection*{Enabling Capture Effect} 

WirelessHART networks use IEEE 802.15.4 compliant radios \cite{WirelessHART2007_standard}. In such radios, during the header decoding (or synchronization), a node's radio searches for a preamble and a start frame delimiter with the strongest Received Signal Strength (RSS) \cite{WirelessHART2007_standard, dezfouli2014cama}.  After this, the radio generates an interrupt and locks to payload reception mode and does not search for preambles. Therefore,  {\bf\slshape capture effect} \cite{dezfouli2014cama} can recover the stronger packet if it comes before the radio locks to a weaker packet's payload reception mode, requiring no physical layer modification. Hence, our objective is to ensure that a receiver (node $v_3$) receives a packet transmission on a dedicated slot (packet $P_2$) before the packet transmission on a shared slot (packet $P_1$). Moreover, the strongest packet can be recovered if its RSS is higher (by 1--3dB based on modulation) than that of the other colliding signal/s.  Hence for successful reception of $P_2$, we adopt the following technique. When a node uses a slot as a dedicated slot, it will transmit immediately after the slot starts and will use the highest transmission (Tx)  power.  On the other hand, when a node uses a slot as a shared slot, it will transmit at a moderate Tx power to make the required RSS difference at the receiver. Also, a node transmits packets after $\theta$ time in a shared slot (while in a dedicated slot, it transmits packets in the very beginning of a slot). The transmission power difference and $\theta$ time difference ensures that the receiver's radio locks to the payload reception mode of $P_2$ and successfully receives $P_2$ even under collision.

\begin{figure}[h]
    \centering
	\includegraphics[width=0.32\textwidth]{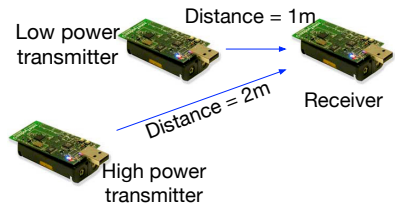}
	\caption{Capture Effect Experiment Setup}
	\label{hart_fig:captureEffectSetup}
\end{figure}

Existing work in \cite{whitehouse2005exploiting} demonstrates the effectiveness of capture effect for IEEE 802.15.4 based networks. Here, we experimentally determine values of $\theta$ (time difference) and Tx power difference to enable the capture effect. We use a setup consisting of 3 TelosB motes (that use radios based on 802.15.4), one receiver and two time synchronized transmitters as shown in Fig. \ref{hart_fig:captureEffectSetup}. We performed one experiment to determine $\theta$ and another to determine the power difference for enabling capture effect. For the first experiment, we used a transmission power of 0dBm for both the transmitters and varied $\theta$ and measured PRR (packet reception rate). We then found $3ms$ as a good value of $\theta$, and using this value, in the second experiment,  we decreased the power level of one transmitter while keeping the other transmitter's power at 0dBm. We then found $3dBm$ as a good value of Tx power difference and using this value (and $\theta = 3ms$), we performed another experiment to observe the performance of capture effect under varying distance. We varied the differences between the distances (from the receiver node) of the two transmitters by increasing the distance (from the receiver node) of the transmitter that used $0dBm$. Each node transmitted 1500 packets with a payload of 14 bytes at a period of 10ms on channel $15$. We present an aggregate result from 10 iterations. Due to the small variance in the obtained result, we limit the number of iterations to 10.

\begin{figure}[t]
    	\centering
	\begin{subfigure}[b]{0.35\textwidth}
		\includegraphics[width=\textwidth]{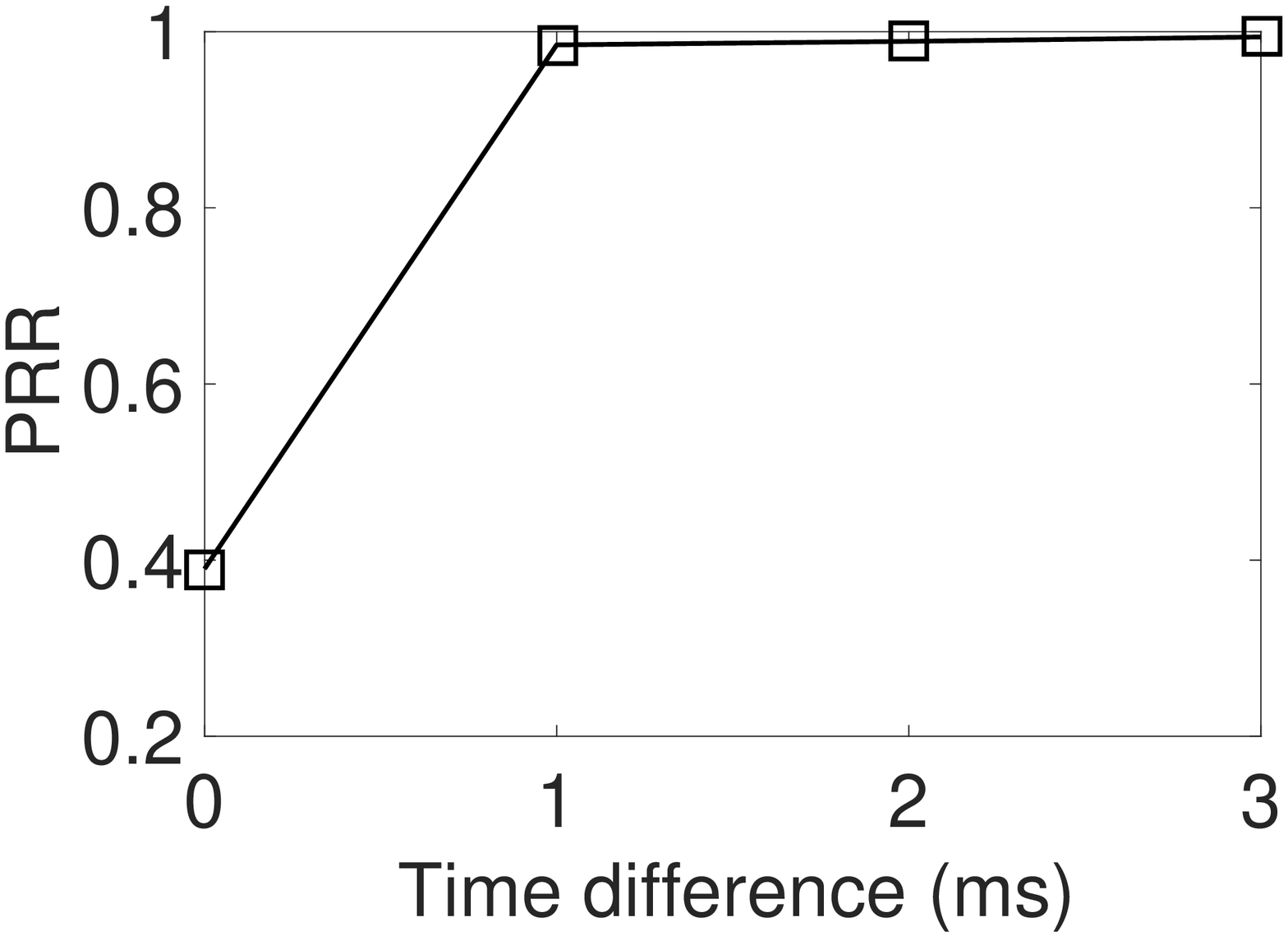}
		\caption{Time difference}
		\label{hart_fig:captureEffect_time}
	\end{subfigure}
	\quad
	\begin{subfigure}[b]{0.35\textwidth}
		\includegraphics[width=\textwidth]{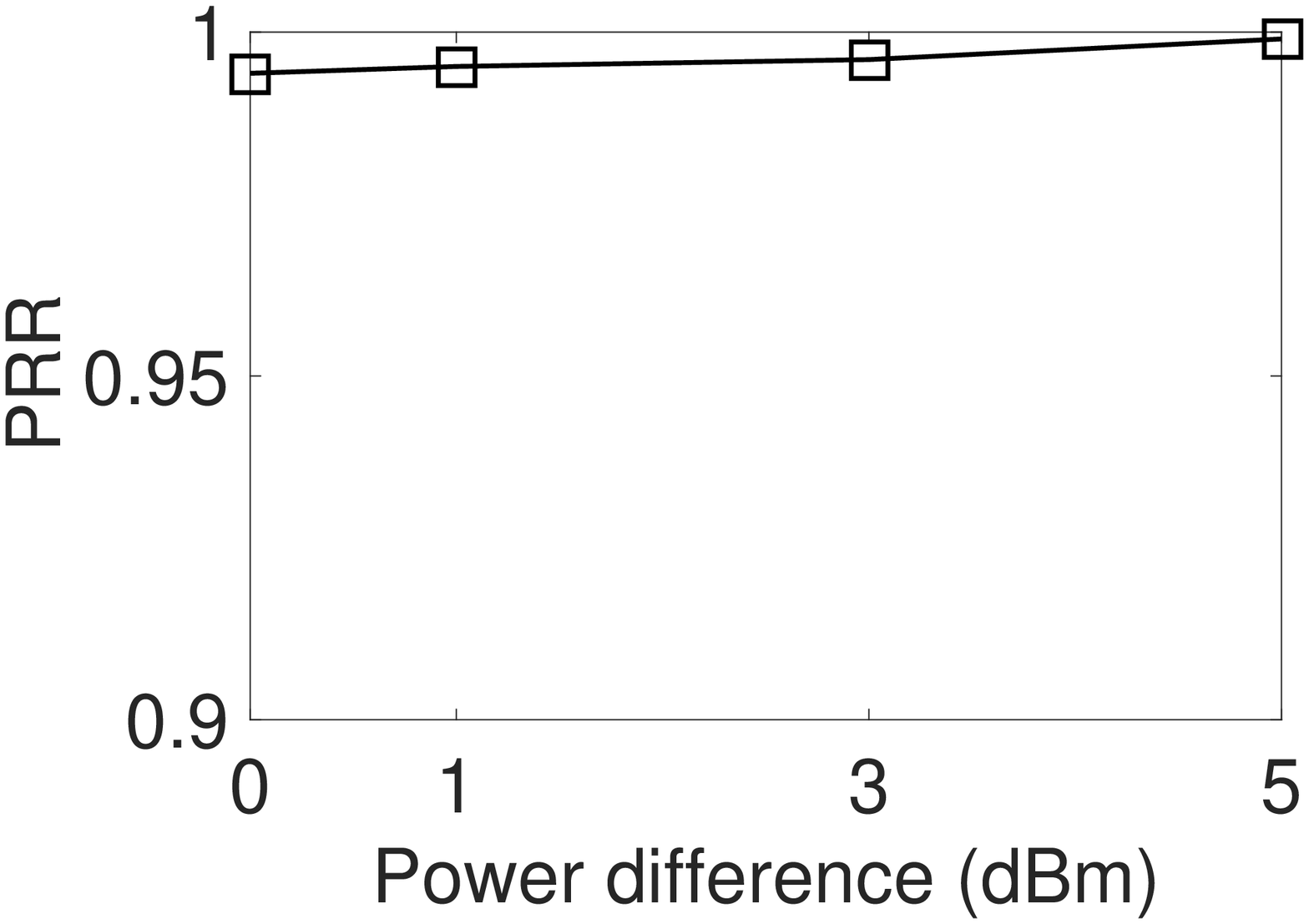}
		\caption{Power difference}
		\label{hart_fig:captureEffect_power}
	\end{subfigure}
       	\quad
	\begin{subfigure}[b]{0.35\textwidth}
		\includegraphics[width=\textwidth]{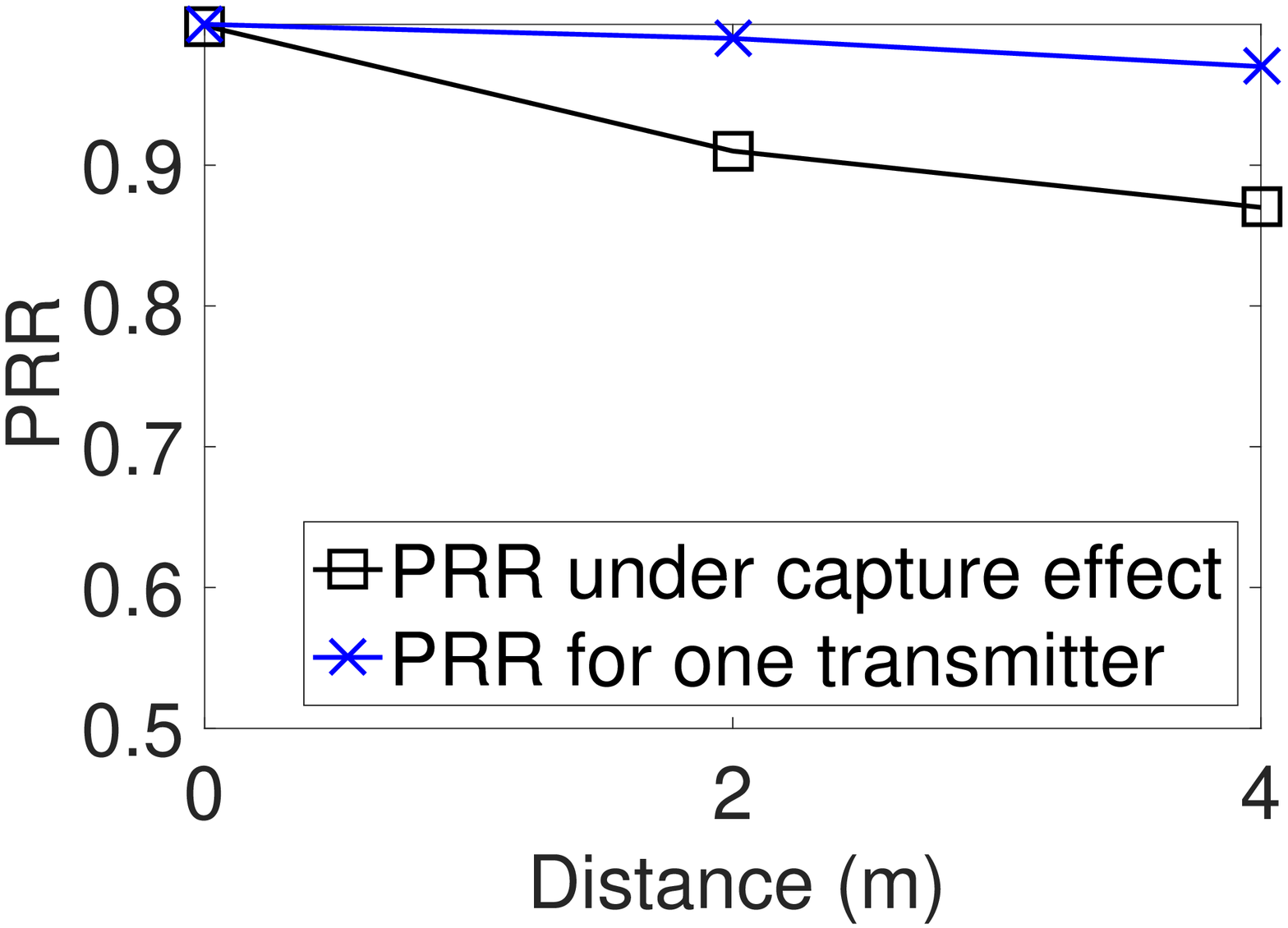}
		\caption{Distance}
		\label{hart_fig:captureEffect_distance}
	\end{subfigure}
	\caption{Probability of Successful Reception through Capture Effect with Varying Distance, Time and Power Differences between Two Transmitters}
	\label{hart_fig:EXPCaptureEffect}
\end{figure}

Fig. \ref{hart_fig:captureEffect_time} shows the average PRR under varying time difference. We observed a very small-time synchronization error between the two nodes which resulted in a PRR of $0.38$ when both nodes transmit a packet at the same time. However, when $\theta$ is increased, we observed that PRR of dedicated transmission significantly improved. This phenomenon was due to two factors 1) receiver's radio locked to the packet sent in a dedicated slot, and 2) there was small/no overlap between the transmission times (due to short packet lengths). Fig. \ref{hart_fig:captureEffect_power} shows the average PRR under varying power difference. As expected, with the increase in power difference, we observed an increase in PRR for a packet transmitted at higher power. Based on the results shown in Fig. \ref{hart_fig:EXPCaptureEffect}, we set the value of $\theta$ to $3ms$ and Tx power difference to $3dBm$. Note that, we do not change the time slot structure by delaying the transmission for 3ms. Rather, we capitalize on the remaining 7ms within the current time slot to successfully transmit a packet. Fig. \ref{hart_fig:captureEffect_distance} shows the average PRR under varying distance. We observed that the PRR through capture effect decreases with an increase in distance. However, the decrease in PRR (by 0.1) is very minimal given the distance. Note that, for this experiment, we used a pessimistic scenario where the high power transmitter is at a greater distance when compared to a low power transmitter and this may not be the case always. In DistributedHART, transmission power difference and $\theta$ can be adjusted based on the node placements, which is quite feasible as long as topology changes and/or mobility are not overwhelming.

%% file: DistributedHart_TMC/SchedulabilityAnalysis.tex
\section{End-to-end Delay Analysis for DistributedHART}
\label{hart_sec:Delay}

Here, we present a probabilistic end-to-end delay analysis for DistributedHART. In a probabilistic end-to-end delay analysis, the end-to-end delay experienced by a packet is a function of the required probability guarantees $\mathbb{P}_i$ and is given by $R_i(\mathbb{P}_i)$. Here, $\mathbb{P}_i$ represents the probability of a packet of flow $F_i$ experiencing a maximum delay of $R_i(\mathbb{P}_i)$. A network designer can determine a good value of $\mathbb{P}_i$ based on the application requirements. For the sake of simplicity in estimating the probability guarantees $\mathbb{P}_i$, we assume PRR of each link is independent of all other links in the network.

\begin{figure}[!h]
	\centering  
	\includegraphics[width=0.35\textwidth]{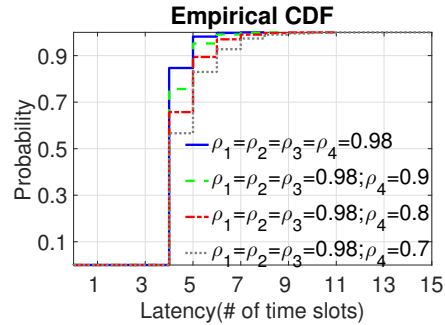}
   \caption{Latency of one Flow with 4 Hops under Different PRR}
	\label{hart_fig:latencyDem}
\end{figure}

Fig. \ref{hart_fig:latencyDem} shows the cumulative distribution function of latency observed from $10,000$ iterations by a packet of, a flow, on the primary path  In Fig. \ref{hart_fig:latencyDem}, $\rho_{a}$ represents the  Packet Reception Rate (PRR) at node $a$. When the PRR of each link is $0.98$, the worst-case delay observed was $8$ time slots, and the probability of a packet experiencing the worst-case delay was $0.001$. However, the probability of a packet experiencing a delay less than or equal to $5$ was $0.99$, i.e., in most cases, the maximum delay is $5$ time slots. A similar result can be observed for different PRR values. These results conclude that the worst-case delay experienced by a packet of a flow is very big, but the probability of a packet experiencing such a large delay is infinitesimal. This result motivates the use of a probabilistic end-to-end delay analysis technique on the primary path for DistributedHART. Thus, we only present the end-to-end delay analysis calculation along the primary path of each flow. Nevertheless, the analysis can be extended to a graph route by considering a packet experiences delay on all paths of the graph route.

In this section, we develop a delay bound analysis for DistributedHART where the nodes use DM scheduling policy. A similar approach can be used to develop a delay bound analysis for DistributedHART where the nodes use EDF scheduling policy.

To calculate the end-to-end delay analysis of a flow, the central network manager must first estimate the value of $\gamma$ prior to deployment. To estimate an upper bound on $\gamma$ (or the chromatic number), the network manager can use Brook's theorem \cite{reed1999strengthening}. From Brook's theorem, a safe upper bound on $\gamma$ is the sum of $1$ and maximum degree of a node. 

The next step in estimating end-to-end delay analysis of a flow is to estimate the delay experienced by a packet at a node $\upsilon$ on the primary path of flow $F_i$. At a node $\upsilon$, a packet experiences three sources of delay: (1) delay experienced between the arrival of a packet and the start of the first transmission time window ($\delta^{pre}_{\upsilon, i} (\mathbb{P}_{\upsilon, i})$), (2) number of time slots required to successfully transmit a packet ($c_{\upsilon, i}(\mathbb{P}_{\upsilon, i})$), and (3) delay caused by interrupting high priority flows ($\delta^{hp}_{\upsilon, i} (\mathbb{P}_{\upsilon, i})$). The total delay experienced by a node can be given by Equation (\ref{hart_eq:total_delay_initial}).
\begin{equation} \label{hart_eq:total_delay_initial}
\delta_{\upsilon, i} (\mathbb{P}_{\upsilon, i}) = \delta^{pre}_{\upsilon, i} (\mathbb{P}_{\upsilon, i}) + c_{\upsilon, i}(\mathbb{P}_{\upsilon, i}) + \delta^{hp}_{\upsilon, i} (\mathbb{P}_{\upsilon, i})
\end{equation}

The maximum delay experienced by a packet between its arrival at node and the start of the first transmission time window is $(\gamma - 1) w$. Since the transmission time window repeats after every $\gamma w$ time slots, the maximum delay a packet can experience is when the packet arrives immediately after the transmission time window. For example, the transmission time window of $\upsilon$ is time slots 3,4 of the epoch and $\gamma w = 10$. The maximum delay is experienced when a packet arrives at time slot 5 and waits till 13, i.e. 8 (which is given by $(\gamma - 1)w$) time slots.

To compute the number of time slots required to successfully transmit a packet, let $\rho_{\upsilon, i}$ be the probability of successful transmission of flow $F_i$ on a primary node $\upsilon$, and $\mathbb{P}_{\upsilon, i}$ be the probability of packet of $F_i$ experiencing a delay less than $\delta_{\upsilon}(F_i, \mathbb{P}_{\upsilon, i})$ at node $\upsilon$. The probability $\mathbb{P}_{\upsilon, i}$ can also be interpreted as the probability of successful transmission of packet from flow $F_i$ from node $\upsilon$. 
The probability of successful reception of a packet on two dedicated transmissions at node $\upsilon$ is expressed as $1 - (1- \rho_{\upsilon, i})^2$. Similarly, the probability of successful reception after $c$ transmissions is expressed as $1 - (1- \rho_{\upsilon, i})^c$.
For a given $\mathbb{P}_{\upsilon, i}$, the number of slots needed $c_{\upsilon, i}(\mathbb{P}_{\upsilon, i})$ to successfully transmit a packet can be computed by equating the probability of successful transmission after $c_{\upsilon, i}(\mathbb{P}_{\upsilon, i})$ transmissions to $\mathbb{P}_{\upsilon, i}$.

\begin{equation*} \begin{split}
\mathbb{P}_{\upsilon, i} &= 1 - (1- \rho_{\upsilon, i})^{c_{\upsilon, i}(\mathbb{P}_{\upsilon, i})} \\
\Rightarrow 1 -  \mathbb{P}_{\upsilon, i} &= (1- \rho_{\upsilon, i})^{c_{\upsilon, i}(\mathbb{P}_{\upsilon, i})} \\
\Rightarrow log(1 -  \mathbb{P}_{\upsilon, i}) &= c_{\upsilon, i}(\mathbb{P}_{\upsilon, i}) \times log (1- \rho_{\upsilon, i}).
\end{split}\end{equation*} 

The transmission requirement at node $\upsilon$ for flow $F_i$ for a probability requirement of $\mathbb{P}_{\upsilon, i}$ is given by Equation (\ref{hart_eq:c(p)}).
\begin{equation} \label{hart_eq:c(p)} c_{\upsilon, i}(\mathbb{P}_{\upsilon, i})  = \bigg\lceil \frac{log (1 - \mathbb{P}_{\upsilon, i})}{log (1 - \rho_{\upsilon, i})} \bigg\rceil \end{equation}

For example, let the requirement on the probability of successful transmission along a link be $0.99$ (i.e., $\mathbb{P}_{\upsilon, i} = 0.99$) and PRR ($\rho_{\upsilon, i}$) is $0.9$, then $c_{\upsilon, i}(\mathbb{P}_{\upsilon, i})$ is computed as shown below. $$c_{\upsilon, i}(\mathbb{P}_{\upsilon, i})  = \frac{log (1 - 0.99)}{log (1 - 0.9)} = \frac{log(0.01)}{log(0.1)} = \frac{-2}{-1} = 2$$

Computing the delay experienced by a packet due to a high priority packet in DistributedHART is similar to the response time computation of a task in a uniprocessor scheduler. In DistributedHART, for an epoch of length $\gamma w$, a node can only transmit a packet within the $w$ time slots of its transmission time windows. However, in a uniprocessor scheduling, a processor can execute tasks consecutively without halting.

Computing the delay experienced by a packet due to high priority packets at a node in DistributedHART is similar to the delay experienced by a task due to high priority tasks in a uniprocessor scheduler. A detailed description of the similarity between a packet transmission and task scheduling can be found in \cite{modekurthy2018utilization}. Thus, we use the delay computation of a task in a uniprocessor DM scheduling as the foundation for the $\delta^{hp}_{\upsilon, i} (\mathbb{P}_{\upsilon, i})$ computation in DistributedHART.

In a uniprocessor scheduling, a processor can execute tasks consecutively without halting. However, in DistributedHART, for an epoch of length $\gamma w$, a node can only transmit a packet within the $w$ time slots of its transmission time windows. That is, if a high priority packet takes the entire transmission time window, a packet has to wait until the next epoch to attempt a transmission. For example, consider a node $\upsilon$ has a transmission time window of two time slots 2, 3, and the length on epoch is ten times slots. Assume at time slot 2, if a node $\upsilon$ has two packets $P_1$ (with higher priority) and $P_2$ (with lower priority) in its queue. If the packet $P_1$ requires two time slots (2 and 3) for successful transmission,  then packet $P_2$ has to wait until time slot $12$ since the node does not have any other transmission time slots in between. Note that, during the epoch within which packet $P_2$ is successfully transmitted, packet $P_2$ is not delayed for $\gamma w$ time slots.

Accounting for the difference in architecture, the $\delta^{hp}_{\upsilon, i} (\mathbb{P}_{\upsilon, i})$ can be computed as shown in Equation (\ref{hart_eq:delay_hp}), where $HP_\upsilon (F_i)$ denotes the high priority flows of $F_i$.

\begin{equation} \label{hart_eq:delay_hp}
\begin{split}
 & \delta^{hp}_{\upsilon, i} (\mathbb{P}_{\upsilon, i}) = \\& \sum_{F_j \in HP_\upsilon (F_i)}  \Bigg\{\Bigg \lfloor \bigg \lceil \frac{\delta_{\upsilon}(F_i) - D_j}{T_j} \bigg \rceil \frac{C_{\upsilon, j}(\mathbb{P}_{\upsilon, j})}{w} \Bigg \rfloor (\gamma w  - 1) +  \\& \; \; \; \; \; \; \; \; \; \; \; \; \; \; \; \; \; \; \; \; \;   \bigg \lceil \frac{\delta_{\upsilon}(F_i) - D_j}{T_j} \bigg \rceil \frac{C_{\upsilon, j}(\mathbb{P}_{\upsilon, j})}{w} \Bigg\}
\end{split}
\end{equation}

Note that, $\Bigg \lfloor \bigg \lceil \frac{\delta_{\upsilon}(F_i) - D_j}{T_j} \bigg \rceil \frac{C_{\upsilon, j}(\mathbb{P}_{\upsilon, j})}{w} \Bigg \rfloor \gamma w$ denotes the delay of a packet in the all epochs excluding the last epoch, and $\bigg \lceil \frac{\delta_{\upsilon}(F_i) - D_j}{T_j} \bigg \rceil \frac{C_{\upsilon, j}(\mathbb{P}_{\upsilon, j})}{w} - \Bigg \lfloor \bigg \lceil \frac{\delta_{\upsilon}(F_i) - D_j}{T_j} \bigg \rceil \frac{C_{\upsilon, j}(\mathbb{P}_{\upsilon, j})}{w} \Bigg \rfloor$ represents the delay of a packet in the last epoch.

Here, $C_{\upsilon}(\mathbb{P}_{\upsilon, j})$ denotes the worse-case transmission requirement of flow $F_j$ that interferes $F_i$ at node $\upsilon$. For a fixed priority local scheduler, the value of $C_{\upsilon}(\mathbb{P}_{\upsilon, j})$ can be computed during the response time calculation for flow $F_j$. For a dynamic priority local scheduler, an estimate of $C_{\upsilon}(\mathbb{P}_{j})$ can be used.

Since nodes use a real-time scheduling policy to schedule packets locally, the worst case delay experienced by a packet at a node is independent of the delay experienced by the same packet at other nodes in the route. Thus, we express the total delay (experienced by a flow) as the sum of delays experienced at each node on the primary path. The total delay experienced by a packet at a node $\upsilon$ is shown in Equation (\ref{hart_eq:delayAtOneNode}).
\begin{equation}
 \label{hart_eq:delayAtOneNode} \begin{split}
	 &\delta_{\upsilon}(F_i, \mathbb{P}_{\upsilon, i}) =  (\gamma - 1) \times w + C_{\upsilon, i}(\mathbb{P}_{\upsilon, i}) + \\&  \sum_{F_j \in HP_\upsilon (F_i)}  \Bigg\{\Bigg \lfloor \bigg \lceil \frac{\delta_{\upsilon}(F_i) - D_j}{T_j} \bigg \rceil \frac{C_{\upsilon, j}(\mathbb{P}_{\upsilon, j})}{w} \Bigg \rfloor \times (\gamma  \times w  - 1)+ \\&  \; \; \; \; \; \; \; \; \; \; \; \; \; \; \; \; \; \; \; \; \;  \bigg \lceil \frac{\delta_{\upsilon}(F_i) - D_j}{T_j} \bigg \rceil \frac{C_{\upsilon, j}(\mathbb{P}_{\upsilon, j})}{w} \Bigg\}
	 \end{split}
\end{equation}

For implicit deadline flows, the end-to-end delay $R_i$ experienced by a control loop $F_i$ under DistributedHART with DM scheduling is given by Equation (\ref{hart_eq:totalDelay_uniform}) and  Equation (\ref{hart_eq:delayAtOneNode}) with a probability $\mathbb{P'}_{i} =  \prod_{\upsilon \in V_i} \mathbb{P}_{\upsilon, i}$.

\begin{equation}  \label{hart_eq:totalDelay_uniform} \begin{split}
	 R_i (\mathbb{P'}_i) = \sum_{\upsilon \in V_i} \delta_{\upsilon}(F_i)  \end{split} \end{equation}

The delay bound computation, described above, generates the end-to-end delay for a probability of  $\mathbb{P'}_i$. To compute the end-to-end delay for a selected value of $\mathbb{P}_i$, we first compute the end-to-end delay for all possible combinations of $\mathbb{P'}_i$ by generating all possible combinations of $\mathbb{P}_{\upsilon, i}$. From the list of all possible combinations, we select the maximum delay for a probability $\mathbb{P}_i$ as the end-to-end delay for flow $F_i$. All possible combinations of $\mathbb{P}_{\upsilon, i}$ is given by $\{\rho_{\upsilon, i}, 1 - (1- \rho_{\upsilon, i})^2\}$, this is because a node can use at most 2-time slots to make a transmission along the primary path. Since the possible combinations for each link are bounded by a fixed constant $2$, the running time of the delay bound calculation for one flow is $O(n^2)$, where $n$ is the number of nodes in the network.

%% file: DistributedHart_TMC/NonUniform.tex
 \section{Non-Uniform Time Window Assignment in DistributedHART} \label{hart_sec:NonUniform}
The previous sections present the design and analysis of DistributedHART with uniform time window allocation. Intuitively, assigning a fixed number of time slots to each node can cause large delays at nodes, like the access point, which have high traffic flowing through them. Assigning non-uniform time windows, or more time slots per window, for these nodes can potentially reduce the delays experienced by the flows. In this section, we present the design of DistributedHART with non-uniform time window allocation and extend the end-to-end delay analysis for non-uniform time window allocation. Furthermore, we show a comparison between uniform and non-uniform time window allocation.

\subsection{Algorithm to Select Non-Uniform Time Window Length}
Here, we present an algorithm to select the number of time slots per window at a node. We initially assign $\alpha$ time slots to each node that has a primary path of a route passing through it. After the initial assignment, we add $\beta$ time slots to the window for ever $\phi$ units of rate through the node. In this approach, $\alpha$, $\beta$ and $\phi$ are design parameters.  We define $\sigma^\upsilon$ as the sum of rates through the node $\upsilon$ and is calculated as shown in Equation (\ref{hart_eq:sigma}),  where $V_i$ is the set of nodes in the route of $F_i$.
\begin{equation} \label{hart_eq:sigma}
\sigma^{\upsilon} = \sum_{F_i \textit{ if } \upsilon \in V_i} \frac{1}{T_i}
\end{equation}
Then, the length of transmission time window for a node $\upsilon$ is given by Equation (\ref{hart_eq:nonuniform}).
\begin{equation}\label{hart_eq:nonuniform}w_{\upsilon} = \alpha + \bigg\lceil \frac{\sigma^\upsilon}{\phi} \bigg\rceil \times \beta\end{equation}

\subsection{Non-Uniform Time Window Allocation: }
In a non-uniform time window assignment, we represent each time slot by a unique color. Each node then uses vertex multi-coloring to select multiple transmission time slots and locally define their transmission time window. For example, one node can select time slots $1$, $4$, and $7$ as its transmission time window, and another can select time slots $3$ and $7$ as its transmission time window.  Since the definition of a time window is local to a node, it avoids potential issues of synchronization between nodes about how each node defines its transmission time window.

For distributed vertex multi-coloring, we choose a variant of DRAND \cite{rhee2006drand} where each node selects multiple colors instead of one.  Although there exist many distributed vertex multi-coloring algorithms, we pick DRAND for its simplicity in implementation. 

One key issue with using DRAND is that it requires the maximum number of available colors (or $\gamma$) beforehand. The central manager can estimate $\gamma$ from the product of the maximum degree of a node and maximum window length of all nodes, similar to uniform time window assignment. However, this estimation is very loose as the central manager assumes each node selects a fixed/maximum window length.  Such a loose estimation can result in idle time slots within an epoch, which significantly increases the latency. 

To overcome this issue, we propose to use a leader election algorithm after DRAND to determine the leader or the node with the last used time window within the epoch. The leader can reduce the length of the epoch and re-broadcast the new epoch value to all nodes. Typically, the leader election is a polynomial time algorithm for a mesh network and needs to be executed only once, i.e., during network initialization. We represent the total number of unique time slots required for non-uniform time window allocation or length of epoch as $\hat{\gamma}_n$.

\subsection{End-to-End Delay Analysis for Non-Uniform Time Window Assignment}
To compute the end-to-end delay analysis for non-uniform time window assignment, we use a similar approach as described in Section \ref{hart_sec:Delay}. The network manager can estimate an upper bound on $\hat{\gamma}_n$, which facilitates the test of schedulability.

The total delay experienced by a flow $F_i$ given the probability requirement as $\mathbb{P}_{\upsilon, i}$ at node $\upsilon$ is given by Equation (\ref{hart_eq:nonUniformdelayAtOneNode}), where $C_{\upsilon, i}(\mathbb{P}_{\upsilon, i})$ is computed from Equation (\ref{hart_eq:c(p)}).
\begin{equation}
 \label{hart_eq:nonUniformdelayAtOneNode} \begin{split}
	&\delta_{\upsilon}(F_i, \mathbb{P}_{\upsilon, i}) =  (\gamma - 1) \times w + C_{\upsilon, i}(\mathbb{P}_{\upsilon, i}) + \\&  \sum_{F_j \in HP_\upsilon (F_i)}  \Bigg\{\Bigg \lfloor \bigg \lceil \frac{\delta_{\upsilon}(F_i) - D_j}{T_j} \bigg \rceil \frac{C_{\upsilon, j}(\mathbb{P}_{\upsilon, j})}{w} \Bigg \rfloor \times (\hat{\gamma}_n  - 1)+ \\&  \; \; \; \; \; \; \; \; \; \; \; \; \; \; \; \; \; \; \; \; \;  \bigg \lceil \frac{\delta_{\upsilon}(F_i) - D_j}{T_j} \bigg \rceil \frac{C_{\upsilon, j}(\mathbb{P}_{\upsilon, j})}{w} \Bigg\}
	 \end{split}
\end{equation}

For implicit deadline flows, the response time $R_i$ experienced by a control loop $F_i$ under DistributedHART with DM scheduling is given by Equation (\ref{hart_eq:totalDelay_nu}) and  Equation (\ref{hart_eq:nonUniformdelayAtOneNode}) with a probability $\mathbb{P'}_{i} =  \prod_{\upsilon \in V_i} \mathbb{P}_{\upsilon, i}$.
\begin{equation}  \label{hart_eq:totalDelay_nu} \begin{split}
	 R_i (\mathbb{P'}_i) = \sum_{\upsilon \in V_i} \delta_{\upsilon}(F_i)  \end{split} \end{equation}

\subsection{Performance Evaluation of Non-Uniform Time Window Assignment}
To evaluate the performance of DistributedHART with EDF local scheduling and non-uniform time window assignment, we performed simulation in TOSSIM. We used a $148$ node network topology and varied the number of control loops from $5$ to $35$. For further details about the simulation setup, please refer to Section \ref{hart_sec:simulations}. 

We present the aggregate result from 50 random test cases.  For each test case, we randomly selected sensor and actuators and assigned random harmonic periods in the range of $2^{11\sim13}$ time slots. To decrease the workload on the network, we doubled the range after adding every $10$ flows. For this simulation, we used the $1$ and $1$ for $\alpha$ and $\beta$, respectively. We selected $\phi = \frac{(\min{\frac{1}{T_i}} + \max{\frac{1}{T_i}}) \times 7}{8}$. For the test cases with $5$ control loops, the maximum rate was $8$ times the minimum rate, and hence we chose $\frac{7}{8}$ as the smallest increase to the non-uniform time window assignment.

\begin{figure}[h]
	\centering  
	\includegraphics[width=0.35\textwidth]{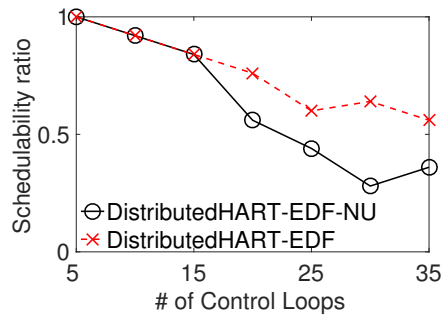}
   \caption{Performance Comparison between Non-Uniform and Uniform Time Window Assignment in DistributedHART}
	\label{hart_fig:nonuniformSched}
\end{figure}

Fig. \ref{hart_fig:nonuniformSched} shows the schedulability ratio comparison between non-uniform (DistributedHART-EDF-NU) and uniform time window assignment (DistributedHART-EDF). In these simulations, we have observed that for some test cases non-uniform time window assignment performs better and performs poorly in other test cases. We have also observed that when the number of control loops is $15$, some test cases schedulable under uniform time window assignment were not schedulable under non-uniform time window assignment, and vice versa. When the number of control loops is greater than $20$, we observed that the average schedulability ratio of non-uniform time window assignment was poor when compared to uniform time window assignment.  We observed that this result is due to the increase in the epoch $\hat{\gamma}_n$, which results in a significant increase in the end-to-end delay observed on a path. From this simulation result, we can conclude that for a smaller number of control loops, non-uniform time window assignment performs better in some cases and uniform time window assignment performs better under some other. For larger number of control loops uniform time window assignment performs better in more cases than non-uniform time window assignment.

\section{Latency under DistributedHART}
\label{hart_sec:latency}
To evaluate the performance of DistributedHART with EDF local scheduling and non-uniform time window assignment, we performed simulation in TOSSIM. We used a $148$ node network topology with $25$ flows. For further details about the topology, please refer to Section~\ref{hart_sec:simulations}. We assigned the same period of $4s$ for all the flows, and the same period assignment is for this simulation alone. For all other simulations and experiments, we use different but harmonic periods.  We choose to assign the same periods because it allows us to compare the latency of a flow with every other flow.

\begin{figure}[ht]
\centering
\includegraphics[width=0.35\textwidth]{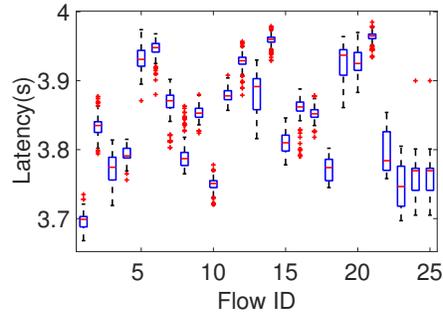}
   \caption{Latency under DistributedHART}
        \label{hart_fig:latency} 
\end{figure}

\begin{figure}[ht]
\centering
\includegraphics[width=0.35\textwidth]{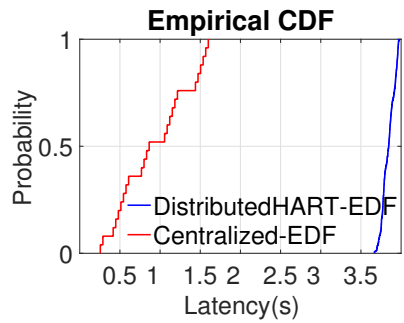}
   \caption{Latency Comparison between Centralized-DM and DistributedHART-DM}
 	\label{hart_fig:latencyCDF}
\end{figure}

Fig. \ref{hart_fig:latency}shows the latency experienced by each flow in one test case under DistributedHART with EDF local scheduling. We observed an average latency of $3.825s$ for all flows. We observed that flow $14$ experienced the maximum latency in the network of $3.98s$. Thus,  Fig. \ref{hart_fig:latencyCDF} shows the cumulative distribution function of latency observed under DistributedHART-EDF to be $1$ at $3.98s$ while that of centralized is $1.7s$.  These results show that the latency experienced by a flow is high compared to the centralized algorithms.

In DistributedHART, DRAND facilitates time window selection through a random selection of colors. Intuitively, a random selection of colors induces long latencies, as shown in the following example. Assume a flow $F_i$ has a primary path $v_i -> v_j -> v_k$, where the source is $v_i$ and destination is $v_k$. For this flow, if the time window assignment for $v_i, v_j, \text{and} v_k$ is $1, 2, \text{and} 3$, respectively, then latency of the path is $ 3w $ (note that $w$ is the length of each window). On the contrary, if the time window assignment for $v_i, v_j, \text{and} v_k$ is $3, 2, \text{and} 1$, respectively, then at each node a packet has to wait for an epoch, i.e., $\gamma$ time units and latency along the path is $3w \times \gamma$. 

To improve the latency, we prioritize nodes during distributed time window selection such that the first node of the highest priority flow selects the first time slot. The centralized manager assigns a different probability of selecting a vertex to each node. The probability assignment is based on the priority of flows passing through the node and position of the node in the route. If a node has two or more flows passing through it, the probability is assigned based on the higher priority flow. We refer to this approach as DistributedHART-IL, which stands for DistributedHART with improved latency.  

We ran simulations to evaluate the performance of DistributedHART-IL. Fig. \ref{hart_fig:latency-IL} shows the latency values experienced by each flow for DistributedHART-IL. We ran simulations on TOSSIM on a $148$ node network with $25$ control loops with a $4$s period. We observed that DistributedHART-IL results in similar performance as DistributedHART. We observed the average latency is $3.99$s. To further evaluate the impact on schedulability, we ran more simulations by varying the number of flows in the network. Fig. \ref{hart_fig:latencySched} shows the comparison between DistributedHART-IL and DistributedHART under the schedulability ratio. We observed that the DistributedHART-IL performs better than DistributedHART. However, we have observed that some cases that were schedulable under DistributedHART were not schedulable under DistributedHART-IL. Although DistributedHART-IL offers better schedulability, DistributedHART is easy to implement and converges faster. From this simulation, we can conclude that both DistributedHART and DistributedHART-IL are good choices. Since both DistributedHART and DistributedHART-IL perform similar, we only evaluate DistributedHART through experiments and simulations, Section \ref{hart_sec:experiments} and Section \ref{hart_sec:simulations}, respectively.

\begin{figure}[ht]
\centering
\includegraphics[width=0.35\textwidth]{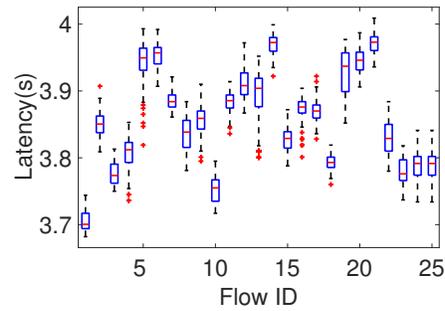}
   \caption{Latency under DistributedHART-IL}
	\label{hart_fig:latency-IL}
\end{figure}

\begin{figure}[ht]
\centering
\includegraphics[width=0.35\textwidth]{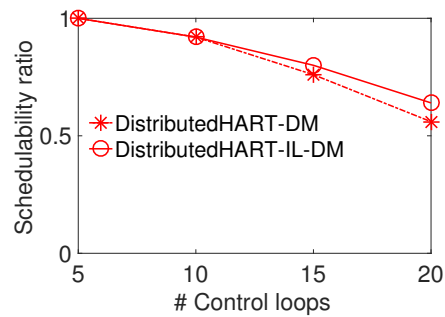}
   \caption{Comparison between DistributedHART-IL and DistributedHART}
	\label{hart_fig:latencySched}
\end{figure}

When compared to centralized algorithm, both DistributedHART may perform poorly under latency. However, it achieves a similar schedulability ratio (which is a more important metric for real-time networks) as centralized algorithms in most of the cases. Furthermore, the centralized approach relies on global information to achieve low latency. Acquiring global information is challenging, and hence centralized algorithm is not scalable. However, DistributedHART uses local information to make scheduling decisions, thereby providing advantages like supporting any type of traffic (periodic/aperiodic), handling network-dynamics (which is frequent in industrial environments) and scalability (which is important for Industrial Internet of Things). Under network-dynamics, a centralized approach re-calculates and re-distributes schedules among nodes. Since network-dynamics are frequent in industrial environments, frequently calculating and re-distributing the schedules degrades the performance of centralized approach. Thus, over a long period of time (with frequent network-dynamics), the overall performance of the centralized approach can be worse than DistributedHART (even with long latencies).

%% file: DistributedHart_TMC/evaluation.tex
\section{Testbed  Experiments}
\label{hart_sec:experiments}
We implemented DistributedHART on TinyOS 2.2~\cite{tinyos} and evaluated on a 130 node testbed~\cite{testbedwsu} of TelosB mote for real experiments. TelsoB devices use Chipcon CC2420 radios, which are compliant with the IEEE 802.15.4 standard. Note that the physical layer of WirelessHART is similar to 802.15.4 physical layer. We deployed the nodes in one room of the Maccabees building. To create a multi-hop network, each node used a transmission power of $-28.7dBm$. The topology of the testbed is shown in~\cite{testbedwsu}. 
Our DistributedHART implementation consists of multi-channel TDMA MAC protocol. 

Time is divided into 10ms slots, and clocks are synchronized using the Flooding Time Synchronization Protocol (FTSP)~\cite{FTSP}.  We ran FTSP algorithm frequently to avoid issues with capture effect. For a fair comparison, we used centralized version of the graph routing algorithm proposed in~\cite{modekurthy2018distributed}. For simplicity in implementation and experimentation, network manager computed the channel and time window allocations and disseminated them using TinyOS dissemination protocol library. We then evaluated the online scheduling of DistributedHART. 

\subsection{ Evaluation Metrics}

We evaluated DistributedHART using four metrics 1) energy, 2) memory, 3) convergence time, and 4) schedulability performance and then compared the performance with centralized EDF~\cite{RTSS10paper} and DM~\cite{WirelessHART2007_standard} (with spatial re-use).

 {\em Schedulable ratio} is defined as the fraction of test cases that were schedulable among all cases. Each test case corresponds to a set of flows and is said to be {\em schedulable} if all packets from all flows met their deadlines (i.e., max latency $\le$ deadline). The deadlines of the flows were set equal to their periods. 
 \textit{Memory consumption} is the average memory consumed per node to store a schedule.
\textit{Convergence Time} is the average time taken for all nodes to obtain a schedule excluding the time taken to generate routes. 
\textit{Energy consumption} is the average energy consumed per node in Joules to generate/disseminate a schedule. In our testbed, we use USB cables to power the nodes and hence, we can not obtain the actual energy consumed. We use a product of the number of transmissions and energy consumed for each communication ($0.22mJ$, calculated from TelosB datasheer~\cite{energyTelosb} with $5.5ms$ Tx time) to estimate energy consumption per node. 
We used a CSMA/CA protocol to generate/disseminate schedule for DistributedHART and centralized algorithms. Thus, duty-cycle of operation (equal to convergence time) was very large and hence, was not a good metric for comparison.

\subsection{Results}

\begin{figure*}[t]
    	\centering
	\begin{subfigure}[b]{0.35\textwidth}
		\includegraphics[width=\textwidth]{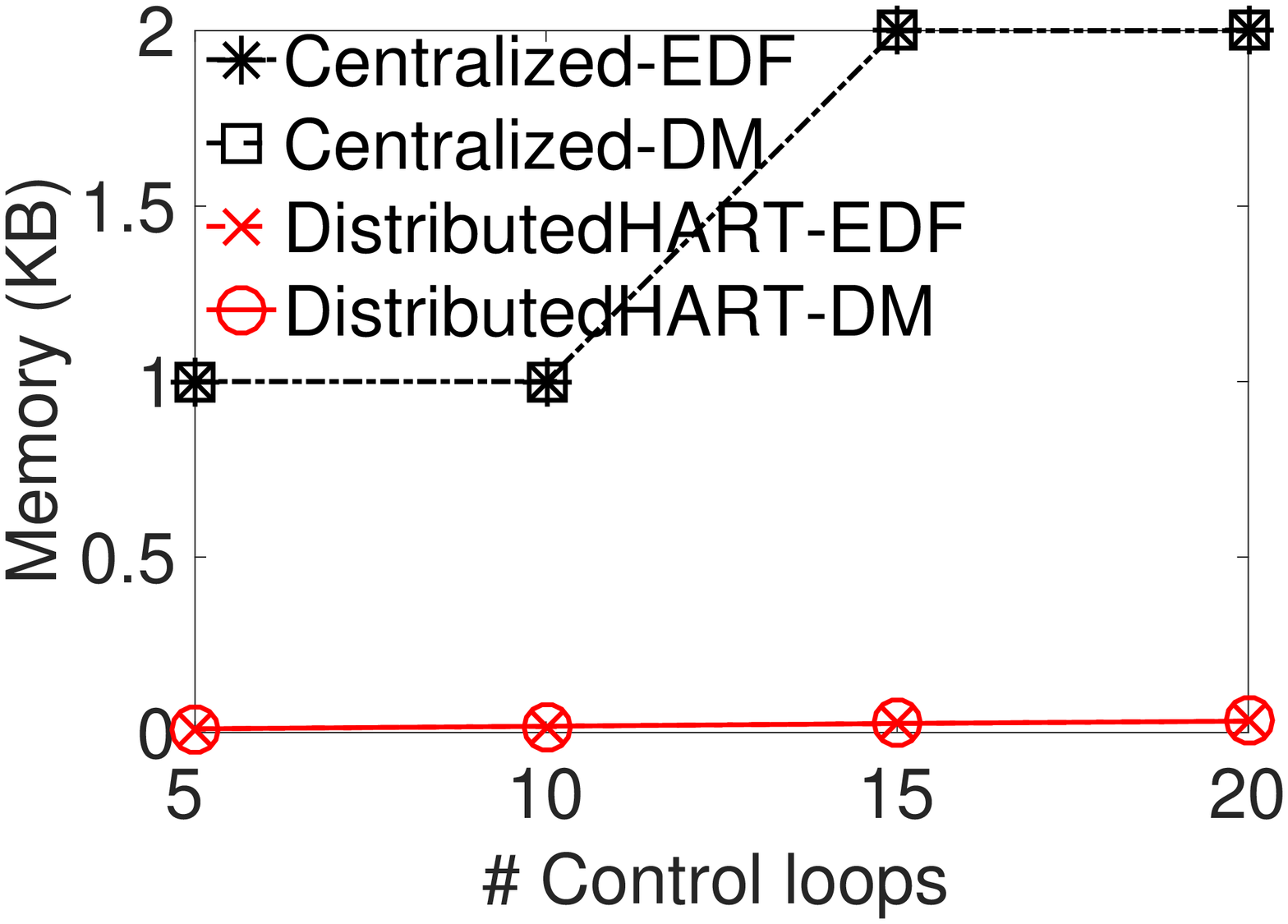}
		\caption{Memory Consumption}
		\label{hart_fig:EXPloops_memory}
	\end{subfigure}
	\quad
	\begin{subfigure}[b]{0.35\textwidth}
		\includegraphics[width=\textwidth]{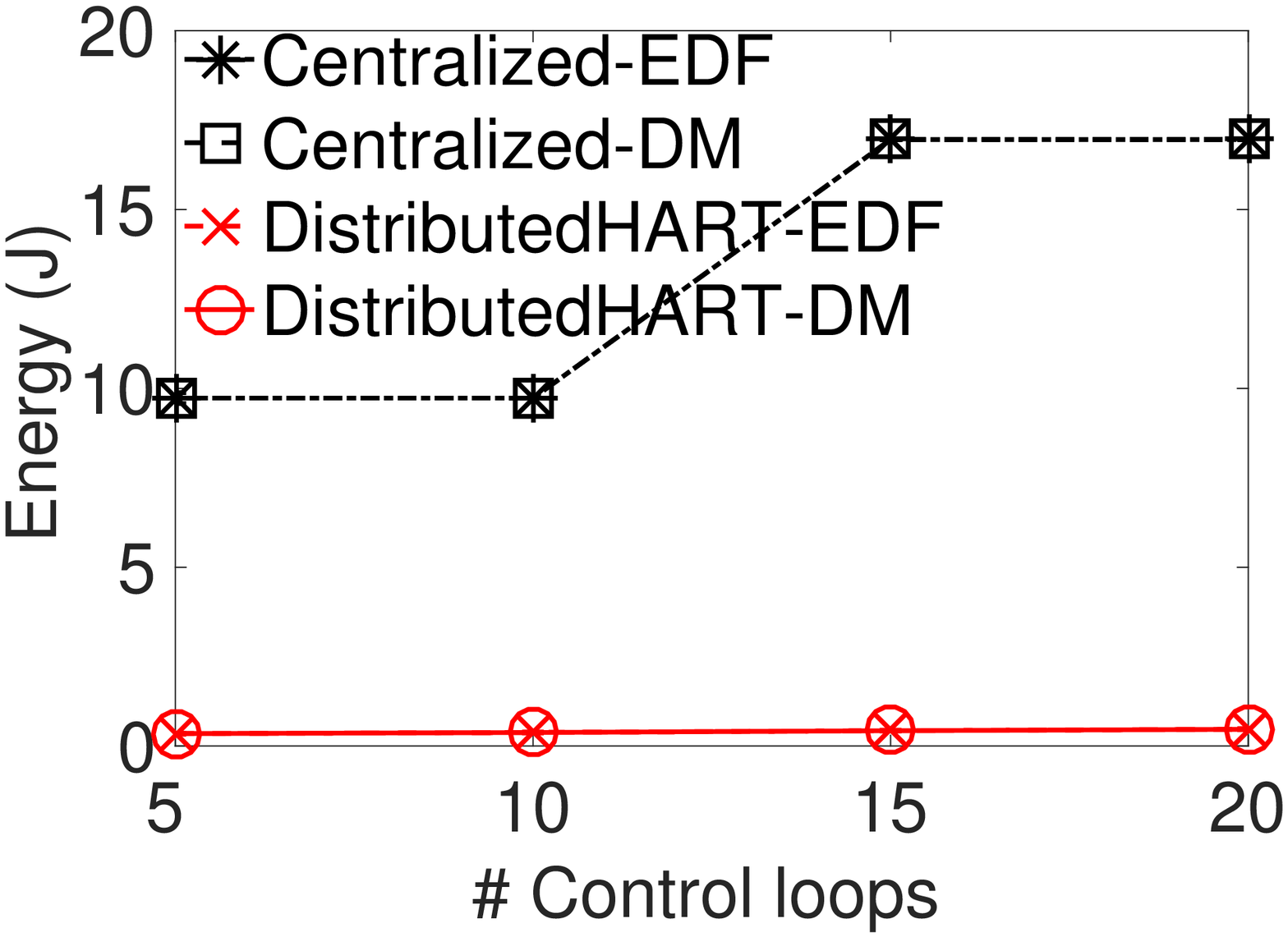}
		\caption{Energy consumption}
		\label{hart_fig:EXPloops_energy}%
	\end{subfigure}
        	\quad
	\begin{subfigure}[b]{0.35\textwidth}
		\includegraphics[width=\textwidth]{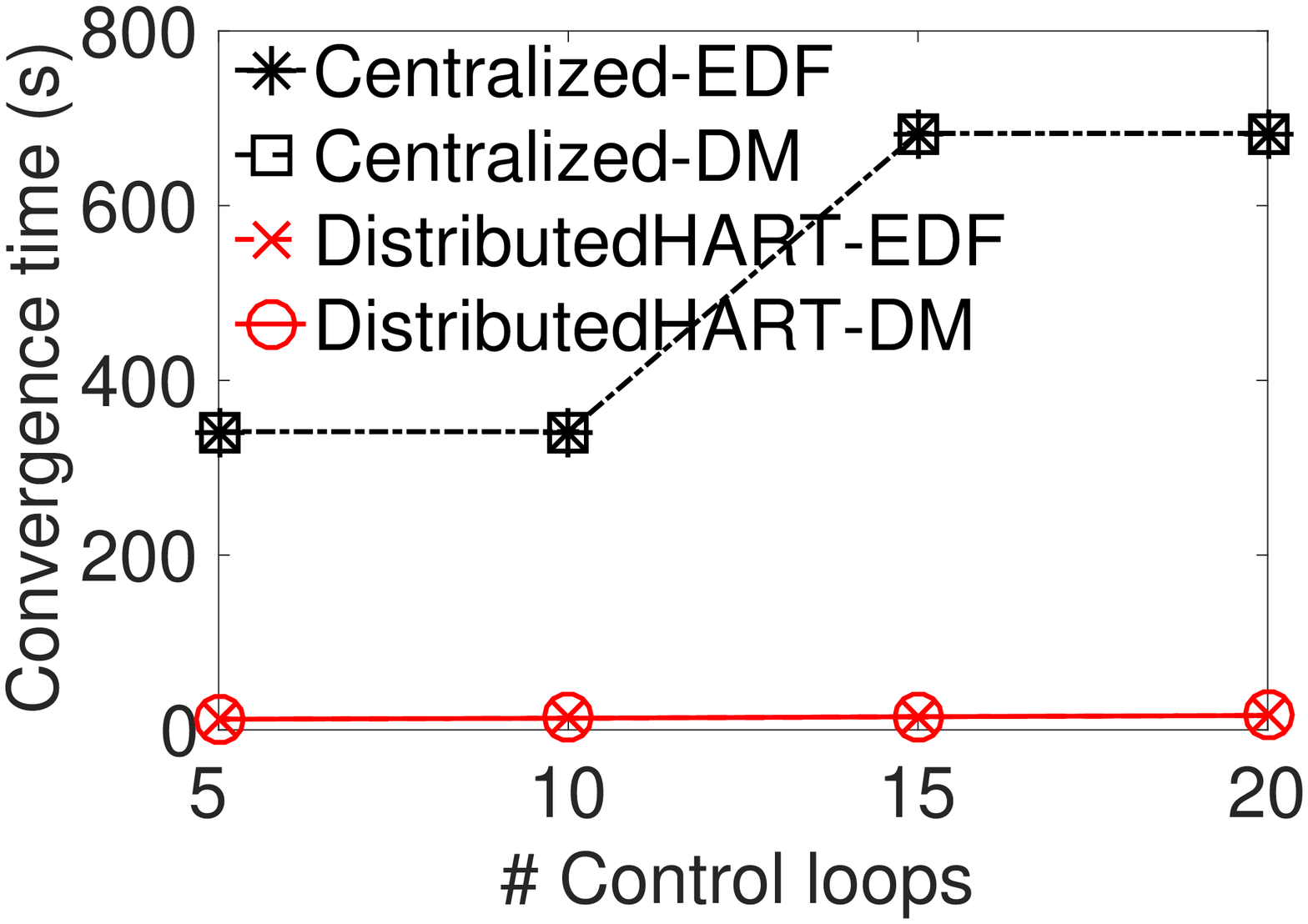}
		\caption{Convergence Time}
		\label{hart_fig:EXPloops_time}
	\end{subfigure}
	\quad
	\begin{subfigure}[b]{0.35\textwidth}
		\includegraphics[width=\textwidth]{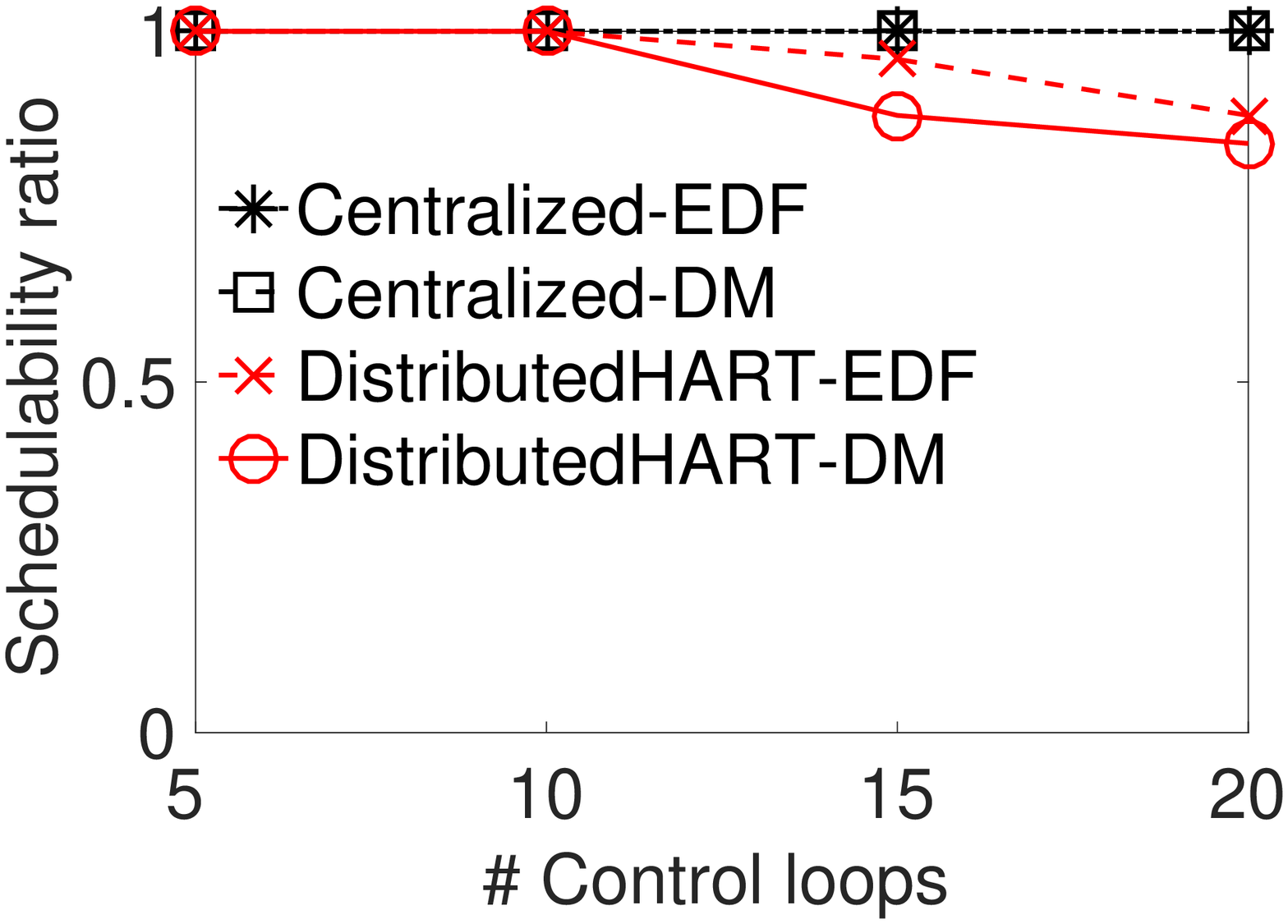}
		\caption{Schedulability Ratio}
		\label{hart_fig:EXPloops_sched}
	\end{subfigure}
	\caption{Experimental Result under Varying Number of Flows Considering Harmonic Periods}
	\label{hart_fig:EXPloops}
\end{figure*}

Fig.~\ref{hart_fig:EXPloops} and Fig.~\ref{hart_fig:EXPloops_sched} shows the performance comparison between DistributedHART and  centralized scheduler proposed in WirelessHART standard. The fixed priority and dynamic priority centralized scheduler are labeled "Centralized-DM" and "Centralized-EDF", respectively. Similarly, DistributedHART with fixed priority and dynamic priority local schedulin are labeled "DistributedHART-DM" and "DistributedHART-EDF", respectively.  As the performance of Orchestra~\cite{duquennoy2015orchestra} and DiGS~\cite{shidigs} is expected to be worse, we evaluated them only in simulation.

Fig.~\ref{hart_fig:EXPloops} and Fig.~\ref{hart_fig:EXPloops_sched} shows the aggregate result from $25$ test cases under varying number of flows in the network. In this experiment, we varied the number of flows between $5$ and $40$. For each test case, we generated flows by randomly selecting source and destination nodes. For test cases with $5$ flows, we assign harmonic periods in the range $2^{13\sim 16}$ time slots. To decrease the workload on the network, we double the range after adding every $10$ flows. Note that the average energy, time and memory consumption results shown in this chapter, specifically Fig.~\ref{hart_fig:EXPloops}, report the confidence interval. Since the confidence interval is very close to the average, the confidence intervals are not visible in the figures.

\subsubsection{Memory Consumption}
Typically, memory consumption in centralized algorithms is proportional to the hyper-period. In some special cases, a compact schedule may be feasible. However, we considered a general scenario where memory consumption is proportional to the length of the hyper-period. Fig.~\ref{hart_fig:EXPloops_memory} shows a step increase in memory consumption since we double the hyper-period for every 10 control loops. Since the transmission schedule repeats after every time epoch, time window information during the first epoch and time epoch length is sufficient. This information is subsequently smaller than centralized transmission schedule. We observed a small increase in worst case chromatic number with the increase number of control loops. We observed that DistributedHART consumes at least $75\%$ less memory than both centralized EDF and DM.

\subsubsection{Energy Consumption and Convergence Time}
The centralized algorithms use a dissemination protocol to broadcast schedules to all nodes in the network. Hence, average energy consumption at a node is dependent on the length of the schedule. Thus,  Fig.~\ref{hart_fig:EXPloops_energy} shows a step increase in average energy consumption similar to memory consumption result. In this experiment, for the sake of simplicity, we computed channel and time window allocation at the central manager for DistributedHART and disseminate the information. In DistributedHART, the length of the schedule was only dependent on the worst case chromatic number and hence, Fig.~\ref{hart_fig:EXPloops_energy} shows the energy consumption is close to constant.  We observed that DistributedHART consumes at least $95\%$ less energy than EDF. Similar to energy consumption, convergence time for centralized algorithms is also dependent on hyper-period (has a step increase) while convergence time for DistributedHART is almost constant. Fig.~\ref{hart_fig:EXPloops_time} shows that DistributedHART consumes at least $95\%$ less convergence time than EDF. 

\subsubsection{Schedulability Ratio}
Centralized algorithms rely on global knowledge of channel/link quality information and harmonic periods which made it feasible for EDF to achieve high schedulability ratio (this may not be feasible with arbitrary periods), as shown in Fig.~\ref{hart_fig:EXPloops_sched}. In this experiment, we used a very dense deployment. Thus, the worst case chromatic number or $\gamma$ for DistributedHART was very high. We also observed that $\gamma$ increases with an increase in the number of flows since more nodes require time windows. This increase in $\gamma$ increases the end-to-end delay and decreases schedulability ratio.

While a distributed scheduler handles networks/workload dynamics and save energy, it is expected to perform poorly under schedulability ratio due to the lack of global information. However, DistributedHART is highly competitive in terms of schedulability ratio when compared to centralized algorithms. From this experiment, we can conclude that DistributedHART is a practical choice for wirelessHART as it offers a competitive schedulability ratio and consumes less energy, convergence time and takes very low memory to store a schedule.

%% file: DistributedHart_TMC/simulations.tex
\section{Simulation}
\label{hart_sec:simulations}
\subsection{Simulation Setup}
We perform evaluations considering a topology \cite{sha2015implementation} of 148 nodes through simulations in TOSSIM \cite{tossim}. We use a testbed topology of 74 nodes \cite{sha2015implementation} deployed over a wider area to compliment the experiment results. To scale with the number of nodes, we assume all nodes of
that topology are placed in a grid structure and replicate this grid. We add edges between neighboring grids to generate a connected bigger topology for large scale simulation.
For the simulations, we follow the fully distributed approach for allocating channel and time windows, as mentioned in DistributedHART. 
We evaluated the performance of DistributedHART under varying number of control loops with harmonic periods, number of nodes. We present the performance of DistributedHART under varying time window lengths, number of control loops with non-harmonic periods, hyper-periods, and workload dynamics. We evaluated the performance of DistributedHART when compared to a compact WirelessHART scheduler.  We also evaluated the performance of the proposed schedulability analysis. 
For the simulation results, we presented the aggregate result from 50 random test cases.  For each test case, we randomly selected sensor and actuators and assigned random harmonic periods in the range of $2^{11\sim13}$ time slots. To decrease the workload on the network, we doubled the range after adding every $10$ flows.

\subsection{Performance under Varying Window Lengths}
\label{hart_sec:sim_window}

 \begin{figure}[h]
	\centering  
	\includegraphics[width=0.35\textwidth]{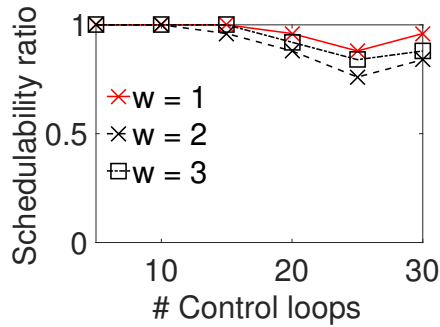}
   	\caption{Schedulability ratio under Varying Window Lengths}
	\label{hart_fig:wSched}
\end{figure}

Fig. \ref{hart_fig:wSched} shows the schedulability ratio under different $w$. Our results showed that $w=1$ offered ta better schedulability ratio when compared to $w=2$ or $w=3$. For $w=2$ (or $w=3$), a packet waited for $2\gamma$ (or $3\gamma$) time units between its arrival and the availability of the first transmission window, at each node. Such long delays increased the total latency of the packet and decreased the schedulability ratio. We also observed that, most often, packets from all flows arrived at $v_i$ during different time epochs. Therefore, the second slot in a transmission window (in most cases) remained unused. We determine $w=1$ as the good setting for this network and traffic pattern. This result also implies that longer time windows may not necessarily increase the schedulability ratio.

\subsection{Performance under Varying Number of Control Loops considering Harmonic Periods}
\label{hart_sec:harmonicperiods}

\begin{figure*}[t]
    	\centering
	\begin{subfigure}[b]{0.35\textwidth}
		\includegraphics[width=\textwidth]{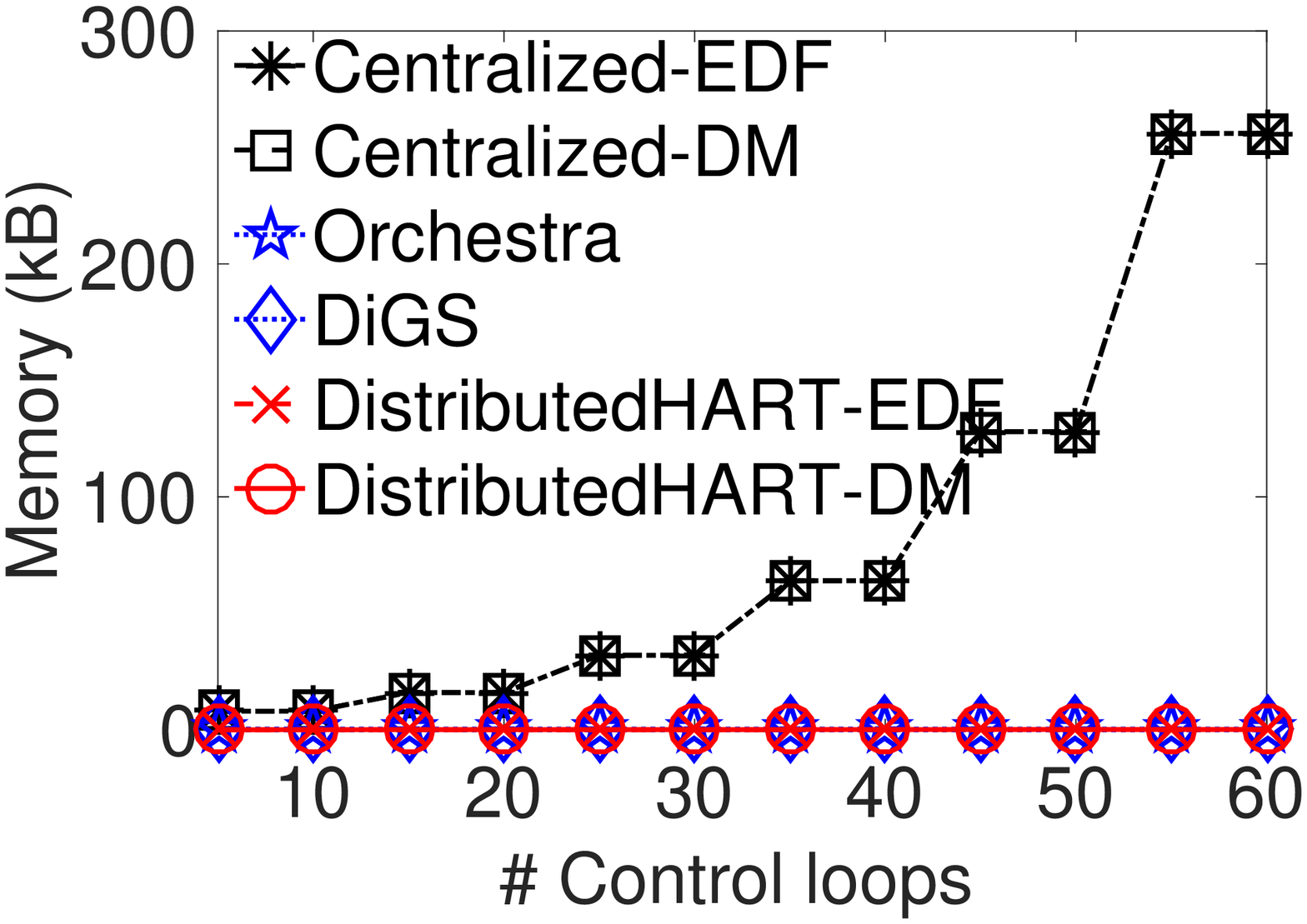}
		\caption{Memory Consumption}
		\label{hart_fig:loops_memory}
	\end{subfigure}
	\quad
	\begin{subfigure}[b]{0.35\textwidth}
		\includegraphics[width=\textwidth]{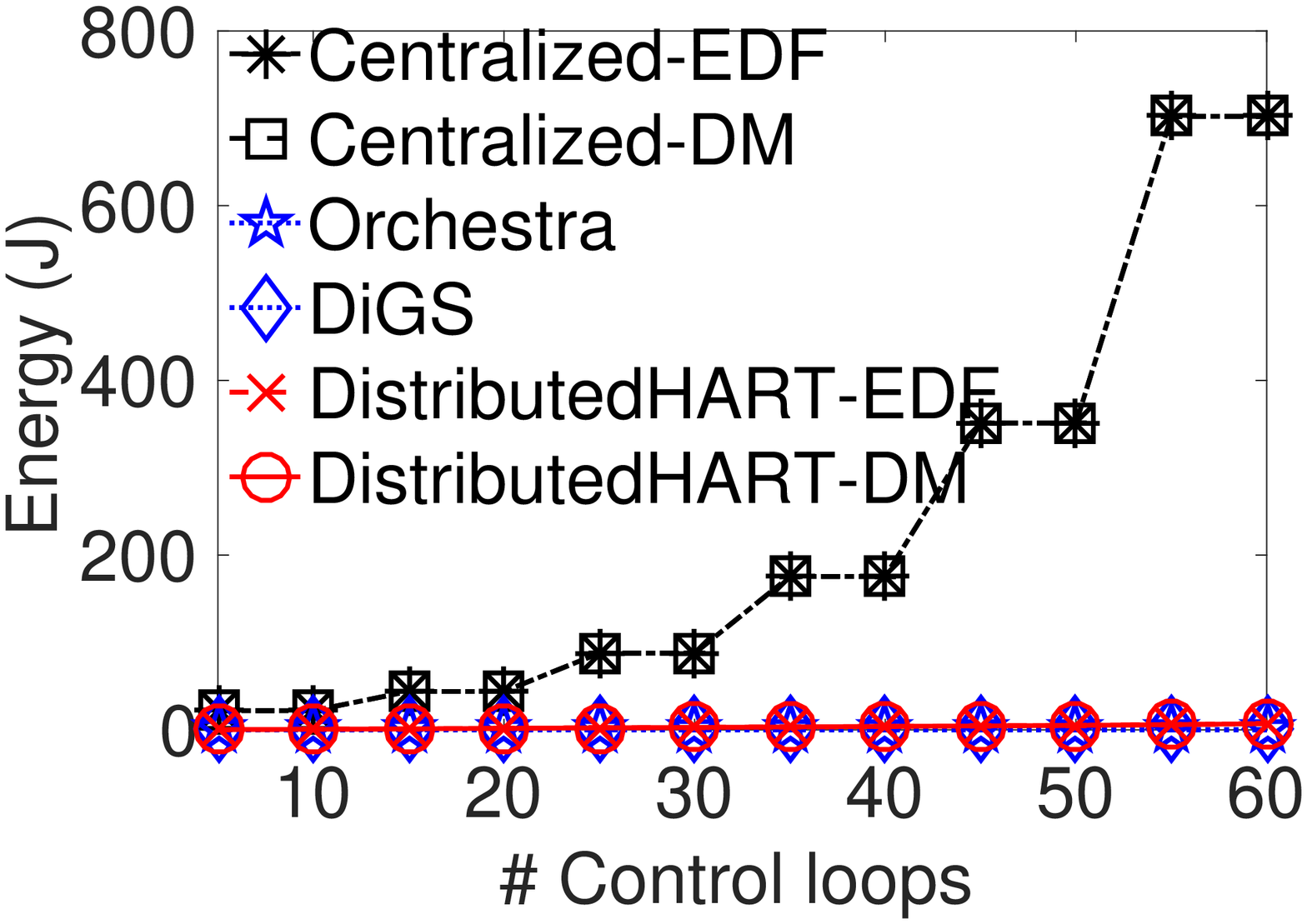}
		\caption{Energy consumption}
		\label{hart_fig:loops_energy}%
	\end{subfigure}
        	\quad
	\begin{subfigure}[b]{0.35\textwidth}
		\includegraphics[width=\textwidth]{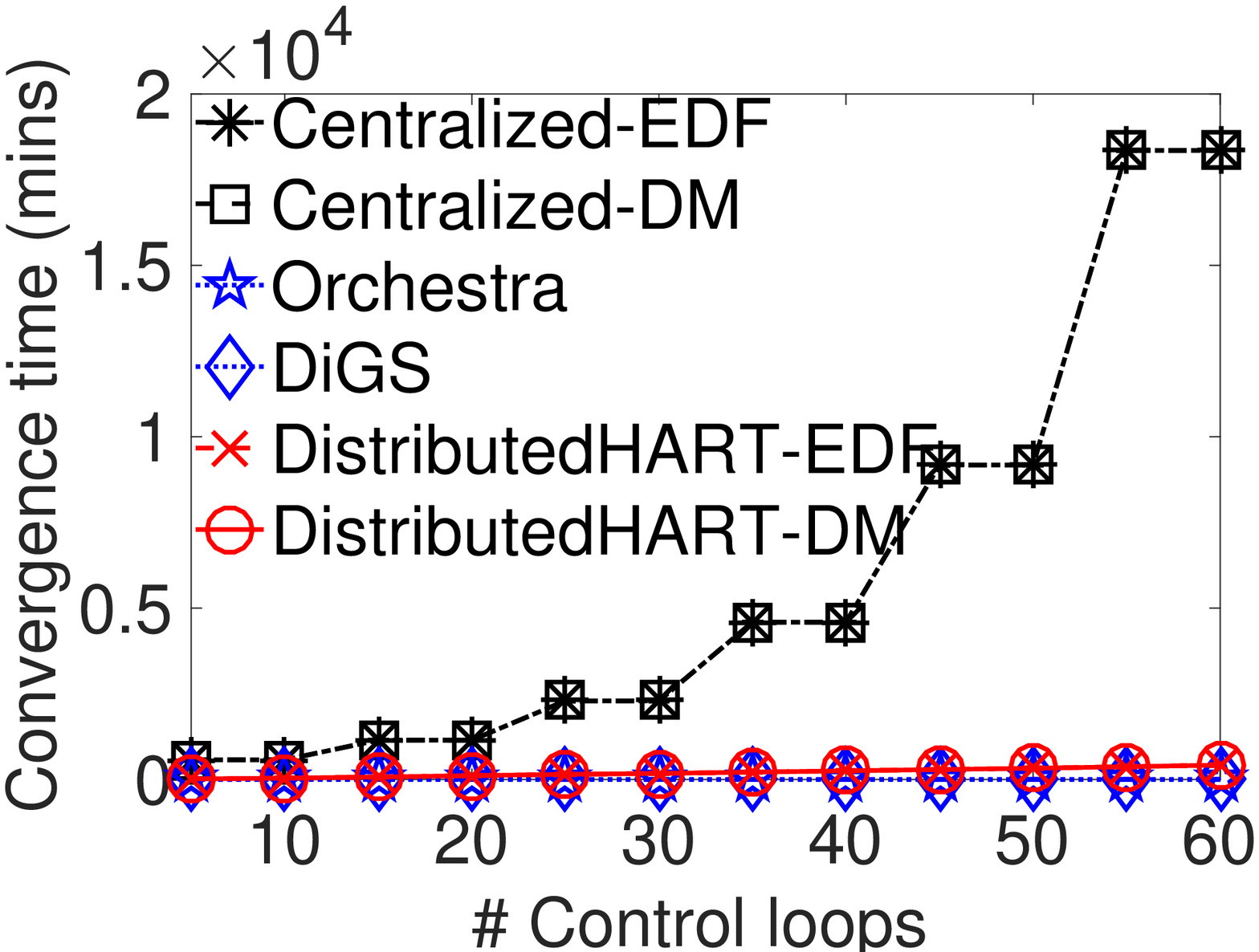}
		\caption{Convergence Time}
		\label{hart_fig:loops_time}
	\end{subfigure}
	\quad
	\begin{subfigure}[b]{0.35\textwidth}
		\includegraphics[width=\textwidth]{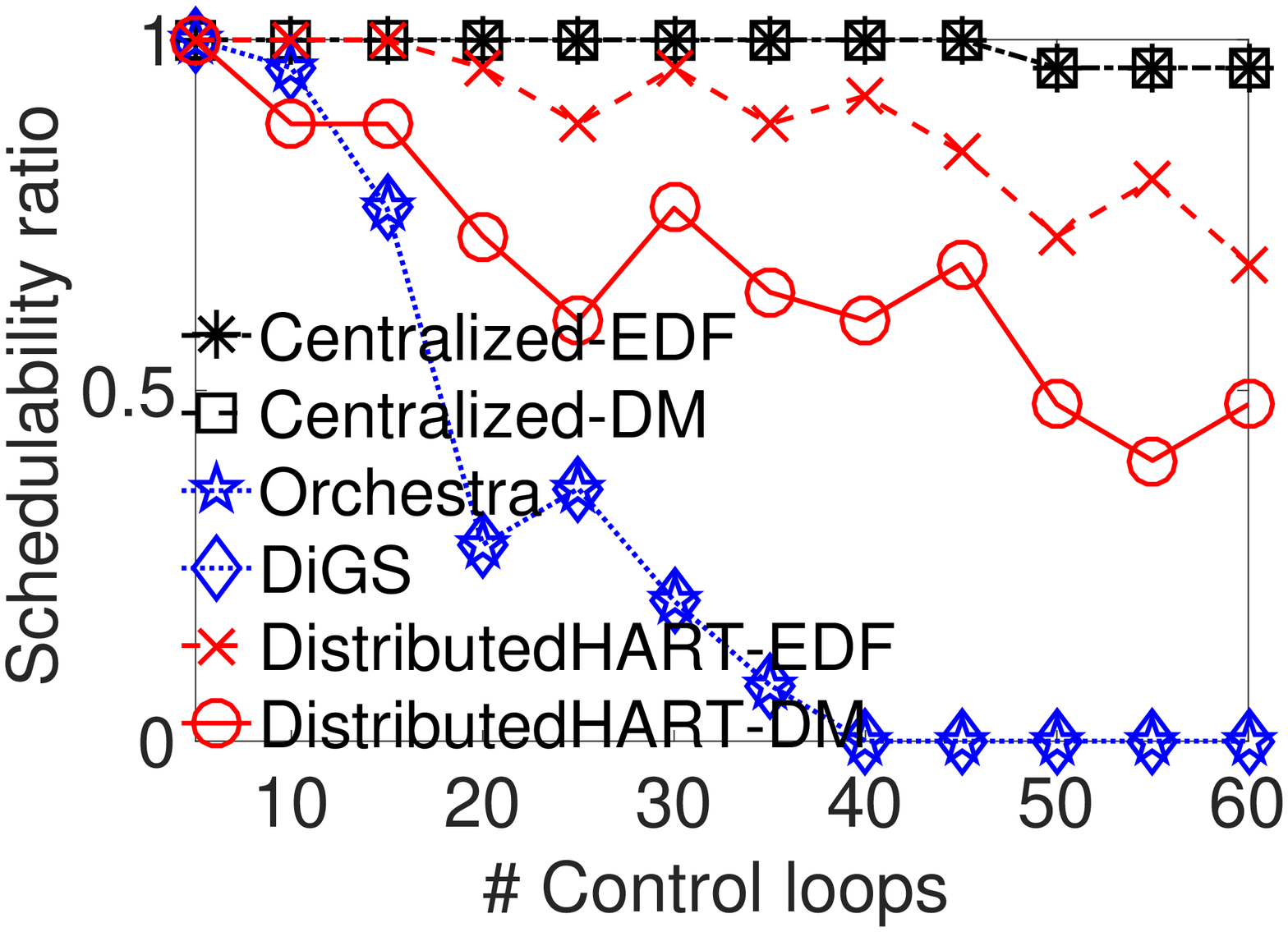}
		\caption{Schedulability Ratio}
		\label{hart_fig:loops_sched}
	\end{subfigure}
	\caption{Performance under Varying Number of Control Loops Considering Harmonic Periods}
	\label{hart_fig:loops}
\end{figure*}

Fig. \ref{hart_fig:loops} shows the performance of DistributedHART (under both fixed \& dynamic priority scheduling) and compares it with centralized EDF and DM under varying number of control loops with harmonic period assignments. We also compare the performance of DistributedHART with Orchestra \cite{duquennoy2015orchestra} and DiGS \cite{shidigs}. Orchestra and DiGS assign time window based on the nodeId. We use the same routing protocol and local scheduling algorithm (EDF) for Orchestra and DiGS as DistributedHART, for a fair comparison. Since Orchestra does not have a dedicated and shared slot assignment, we consider that both transmissions happen within the time window. We used $w = 1$ for DistributedHART, Orchestra, and DiGS. We varied the number of control loops from $5$ to $60$. Simulation results for this setup are shown in Fig. \ref{hart_fig:loops}.

  \textbf{Memory Consumption.} 
For centralized EDF and DM, memory consumption at a node is dependent on the hyper-period. Fig. \ref{hart_fig:loops_memory} shows a step increase in memory consumption since we double the hyper-period for every 10 control loops.  For DistributedHART, memory consumption depends only on worst case chromatic number $\gamma$. Thus, there is a very small increase in the memory consumption of DistributedHART. In this simulation, we have observed that DistributedHART consumes, a minimum of, $95\%$ less memory than centralized algorithms.

  \textbf{Energy Consumption and Convergence Time. } 
In centralized EDF and DM, nodes consume energy during schedule dissemination. Since the number of messages transmitted by each node in centralized algorithms is proportional to the hyper-period, Fig. \ref{hart_fig:loops_energy} shows an exponential increase in average energy consumption.
However, for DistributedHART, each node has to communicate only with its neighboring nodes in the conflict graph. We use controlled flooding to communicate with them since routes to all nodes are not available. Thus, the average energy consumption of a node only increases linearly. From these simulations, we observed that DistributedHART consumes $95\%$ less energy when compared to centralized algorithms. Similar to energy, convergence time for DistributedHART also increases linearly. DistributedHART consumes $90\%$ less time than centralized algorithms. Orchestra and DiGS use an autonomous approach where each node computes its schedule locally and does not require any communication between nodes.

  \textbf{Schedulability Ratio. }
In this simulation, we consider harmonic periods which make it feasible for centralized EDF and DM to achieve very high schedulability ratio, close to optimal, because they assume all local information is available at the network manager. Thus, Fig. \ref{hart_fig:loops_sched} shows a high schedulability ratio for centralized algorithms. For DistributedHART, smaller $\gamma$ for initial conditions results in similar schedulability as centralized algorithms.  With the increase in the number of control loops, $\gamma$ increases linearly, which increases end-to-end delay and decreases schedulability ratio. We selected sensors, actuators, and periods randomly which attributes to an increase in schedulability ratio with an increase in the number of control loops. Although a distributed scheduler is expected to perform poorly under schedulability ratio due to the lack of global information, DistributedHART is highly competitive in terms of schedulability ratio when compared to centralized algorithms.
For Orchestra and DiGS,  $\gamma$ is equal to the number of nodes in the network, which causes a large delay at each node of the flow. Thus, the schedulability ratio is very low compared to DistributedHART. Due to the poor schedulability ratio, we do not present results of Orchestra and DiGS in other evaluations. 
From these results, we can conclude that, DistributedHART outperforms Orchestra/DiGS under schedulability ratio, and centralized algorithms under energy, memory and convergence time while achieving similar schedulability ratio. 

\subsection{Performance under Varying Number of Control Loops considering Non-Harmonic Periods}
\label{hart_sec:sim_nonharmonic}

 \begin{figure*}[t]
    	\centering
	\begin{subfigure}[b]{0.35\textwidth}
		\includegraphics[width=\textwidth]{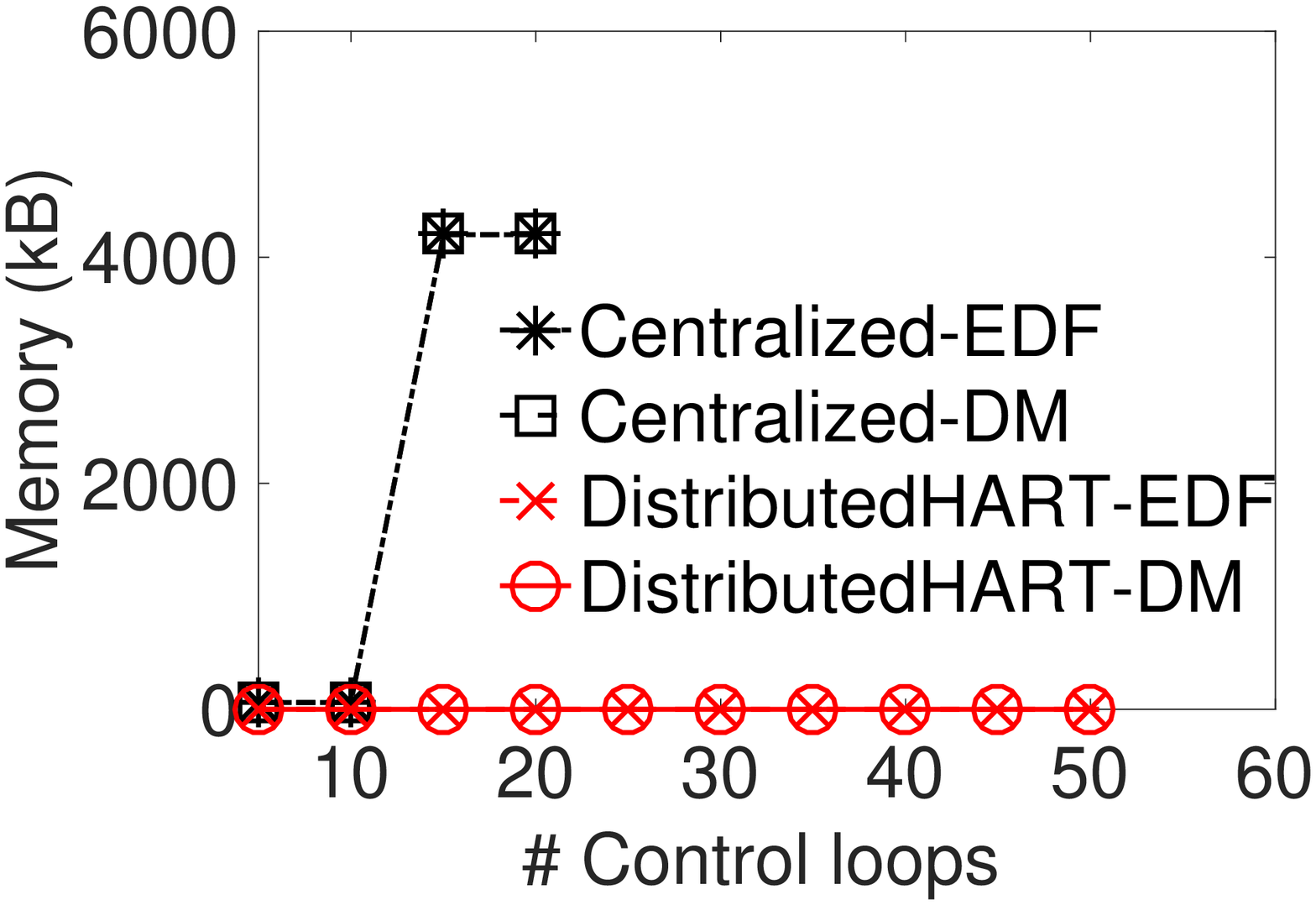}
		\caption{Memory Consumption}
		\label{hart_fig:NHloops_memory}
      	\end{subfigure}
	\quad
       \begin{subfigure}[b]{0.35\textwidth}
		\includegraphics[width=\textwidth]{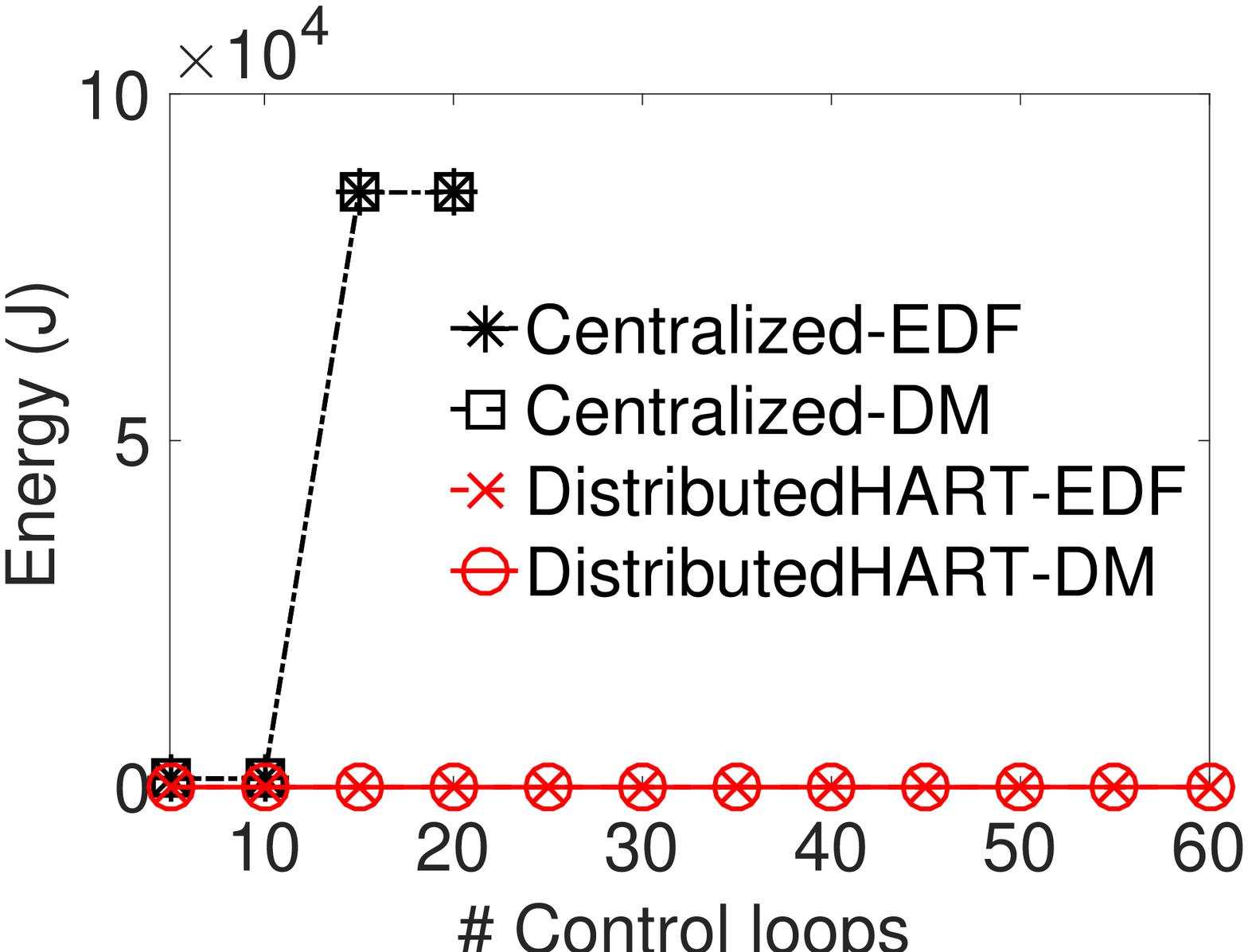}
		\caption{Energy consumption}
		\label{hart_fig:NHloops_energy}%
	\end{subfigure}
	\quad
	\begin{subfigure}[b]{0.35\textwidth}
		\includegraphics[width=\textwidth]{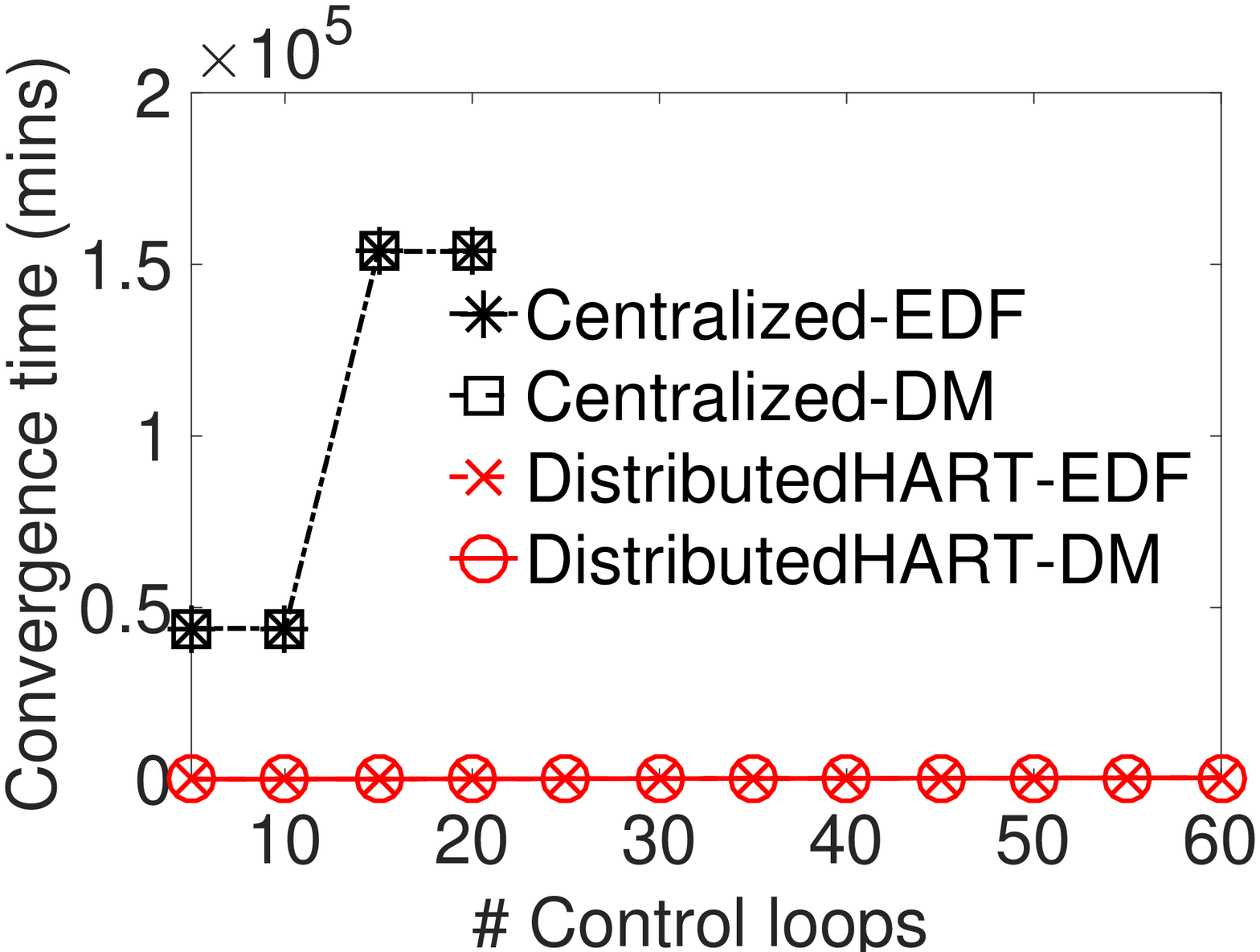}
		\caption{Convergence Time}
		\label{hart_fig:NHloops_time}
      	\end{subfigure}
	\quad
	\begin{subfigure}[b]{0.35\textwidth}
		\includegraphics[width=\textwidth]{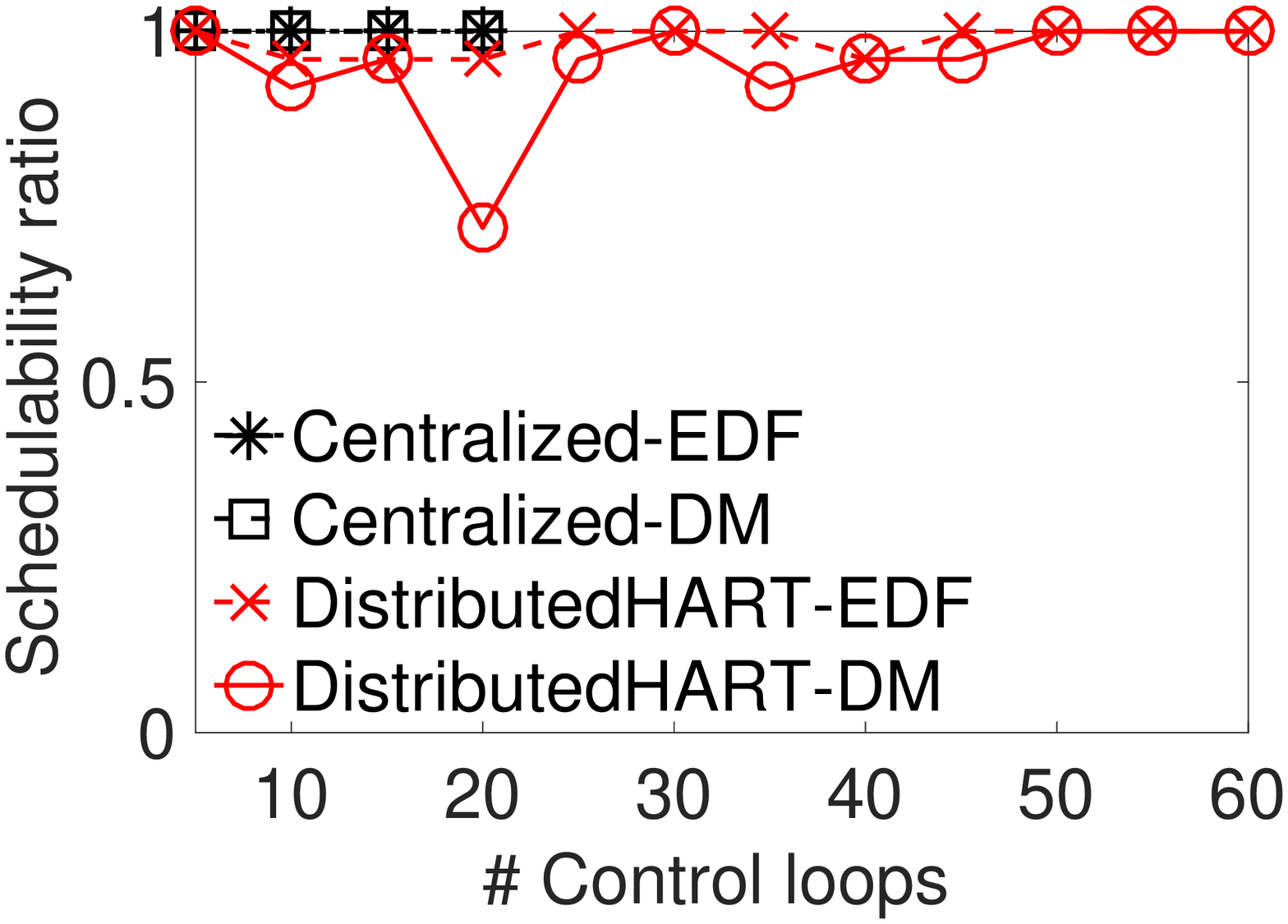}
		\caption{Schedulability Ratio}
		\label{hart_fig:NHloops_sched}	
      	\end{subfigure}
	\caption{Performance under Varying Number of Control Loops Considering Non-Harmonic Periods}
	\label{hart_fig:NHloops}
\end{figure*}

Fig. \ref{hart_fig:NHloops} shows the performance of DistributedHART (under both fixed \& dynamic priority scheduling) and compares it with centralized EDF and DM under varying number of control loops with non-harmonic period assignments. We used $w = 1$ for DistributedHART. We varied the number of control loops from $5$ to $60$. We presented the aggregate result from 50 random test cases.  For each test case, we randomly selected sensor and actuators and assigned random harmonic periods in the range of $2^{11\sim13}$ time slots. To decrease the workload on the network, we doubled the range after adding every $10$ flows.

 \textbf{Memory Consumption.} 
Memory consumption for centralized EDF and DM depends on the hyper-period of the test case. For non-harmonic periods, the hyper-period grows exponentially. Thus, as shown in Fig. \ref{hart_fig:NHloops_memory}, memory consumption of EDF and DM increases exponentially. For $25$ control loops, memory consumption surpasses the limited storage capacity of WirelessHART and TelosB devices \cite{energyTelosb}. However, for DistributedHART, memory consumption is not dependent on the hyper period. This simulation shows DistributedHART consumes at least $99\%$ less memory.

 \textbf{Energy Consumption and Convergence Time. } 
Since energy consumption and convergence time for EDF and DM are proportional to the hyper-period, they increase exponentially, as shown in Fig. \ref{hart_fig:NHloops_energy}, \ref{hart_fig:NHloops_time}. However, for DistributedHART, energy consumption and convergence time increase linearly with an increase in the number of nodes. We observed that DistributedHART saves a minimum of $99\%$ of energy and $99\%$ of convergence time for control loops with non-harmonic periods. We observed that DistributedHART saves a minimum of $99\%$ of energy and $99\%$ of convergence time for control loops with non-harmonic periods. 

 \textbf{Schedulability Ratio. } 
For centralized algorithms, we could compute schedulability ratio up to $20$ control loops. Beyond $20$ control loops, the running time of schedule generation was too large let alone simulation. For DistributedHART, we observed that schedulability ratio is better than that of harmonic periods. From this simulation, we can conclude that for non-harmonic periods DistributedHART outperforms centralized algorithms.

\subsection{Performance under Varying Number of Nodes}
\label{hart_sec:numberOfNodes}

\begin{figure*}[t]
    	\centering
	\begin{subfigure}[b]{0.35\textwidth}
		\includegraphics[width=\textwidth]{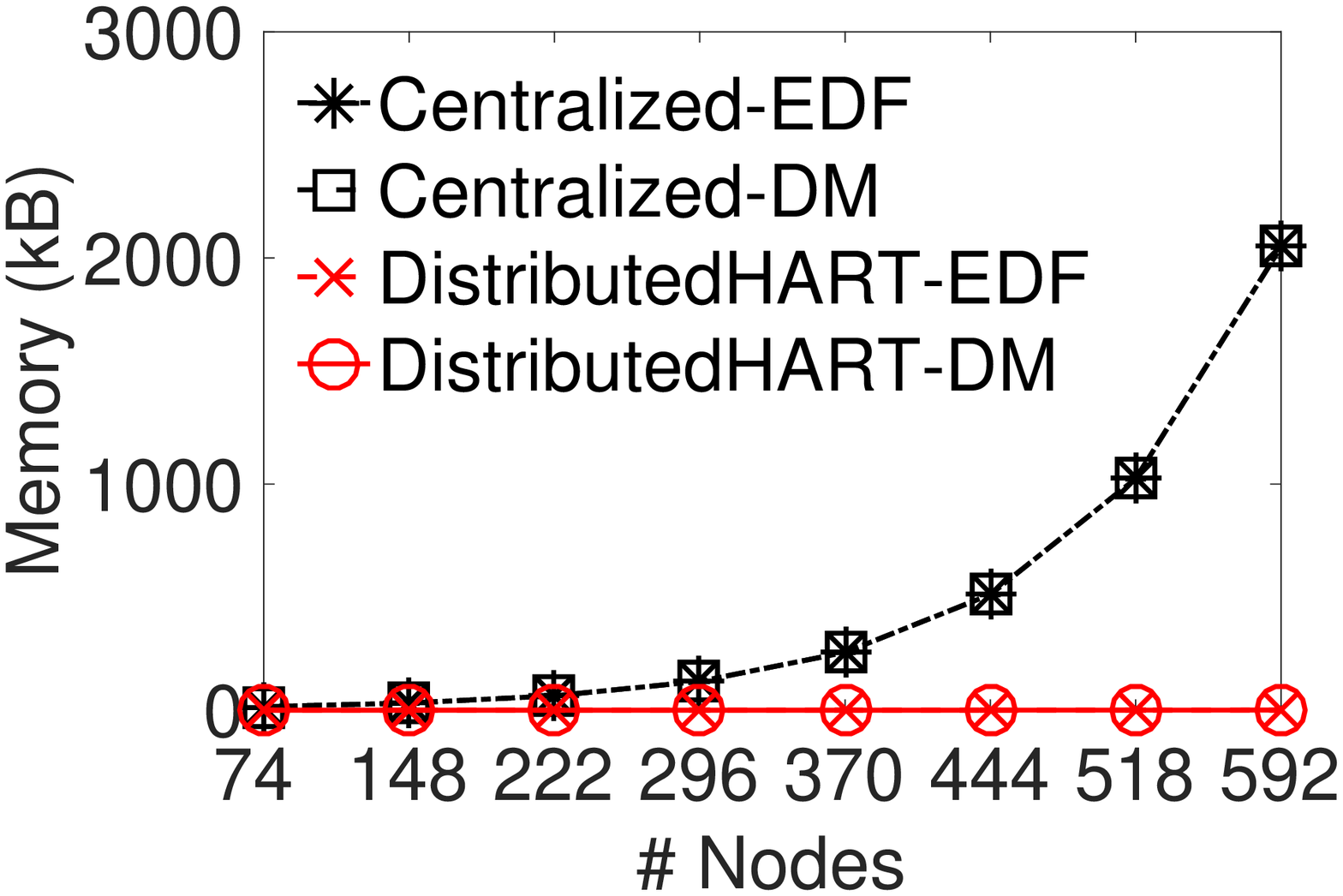}
		\caption{Memory Consumption}
		\label{hart_fig:nodes_memory}
	\end{subfigure}
	\quad
	\begin{subfigure}[b]{0.35\textwidth}
		\includegraphics[width=\textwidth]{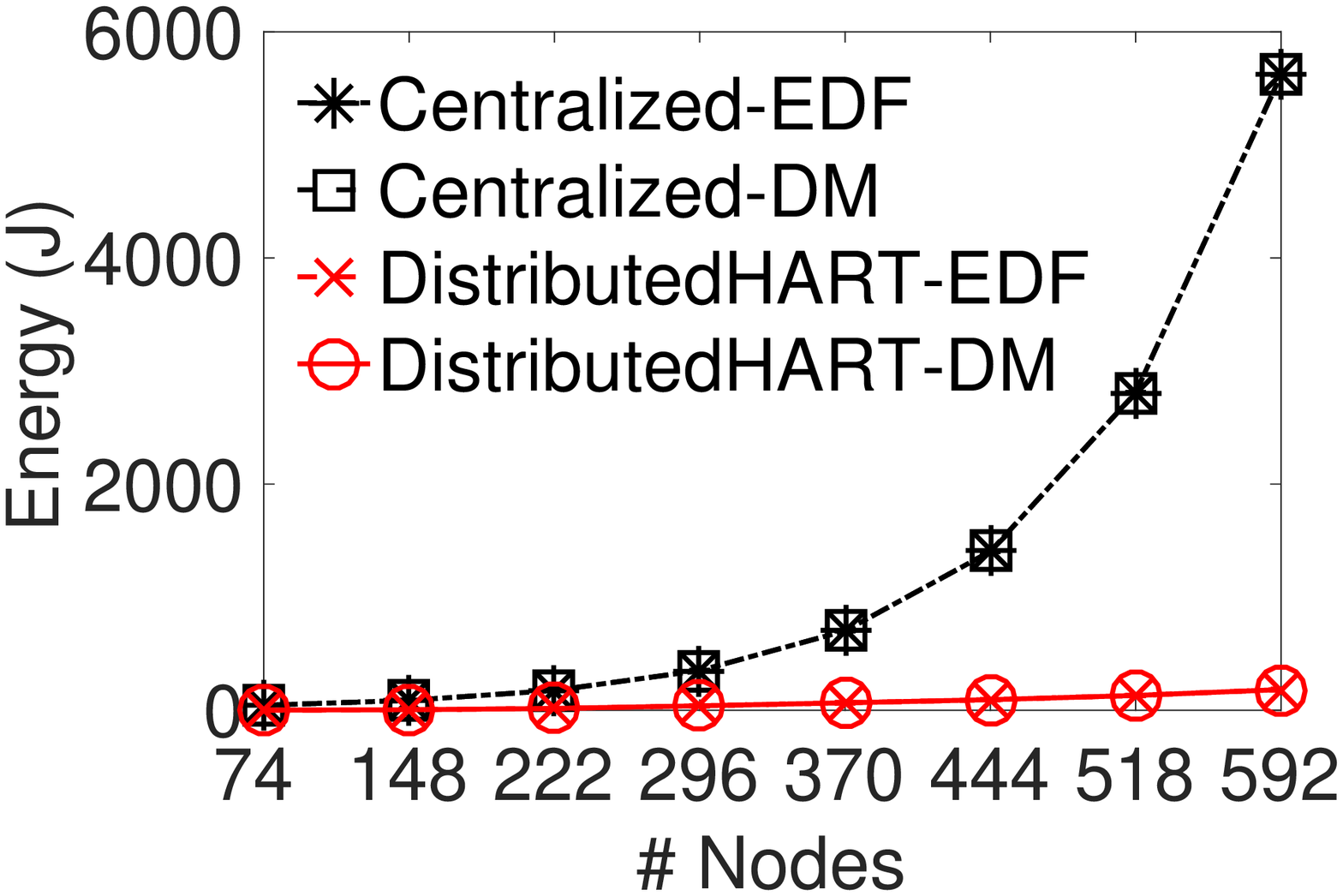}
		\caption{Energy consumption}
		\label{hart_fig:nodes_energy}%
	\end{subfigure}
        	\quad
	\begin{subfigure}[b]{0.35\textwidth}
		\includegraphics[width=\textwidth]{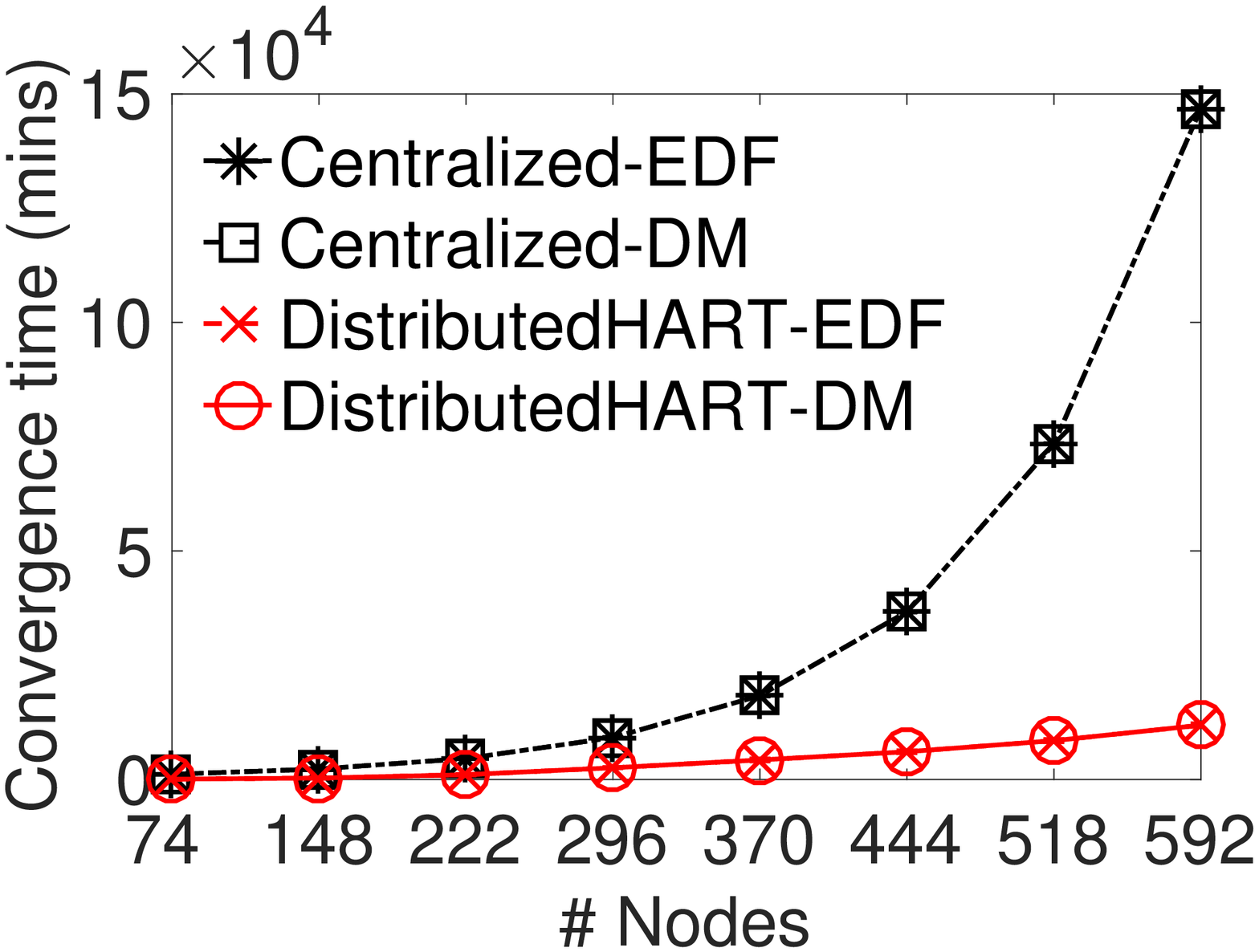}
		\caption{Convergence Time}
		\label{hart_fig:nodes_time}
	\end{subfigure}
	\quad
	\begin{subfigure}[b]{0.35\textwidth}
		\includegraphics[width=\textwidth]{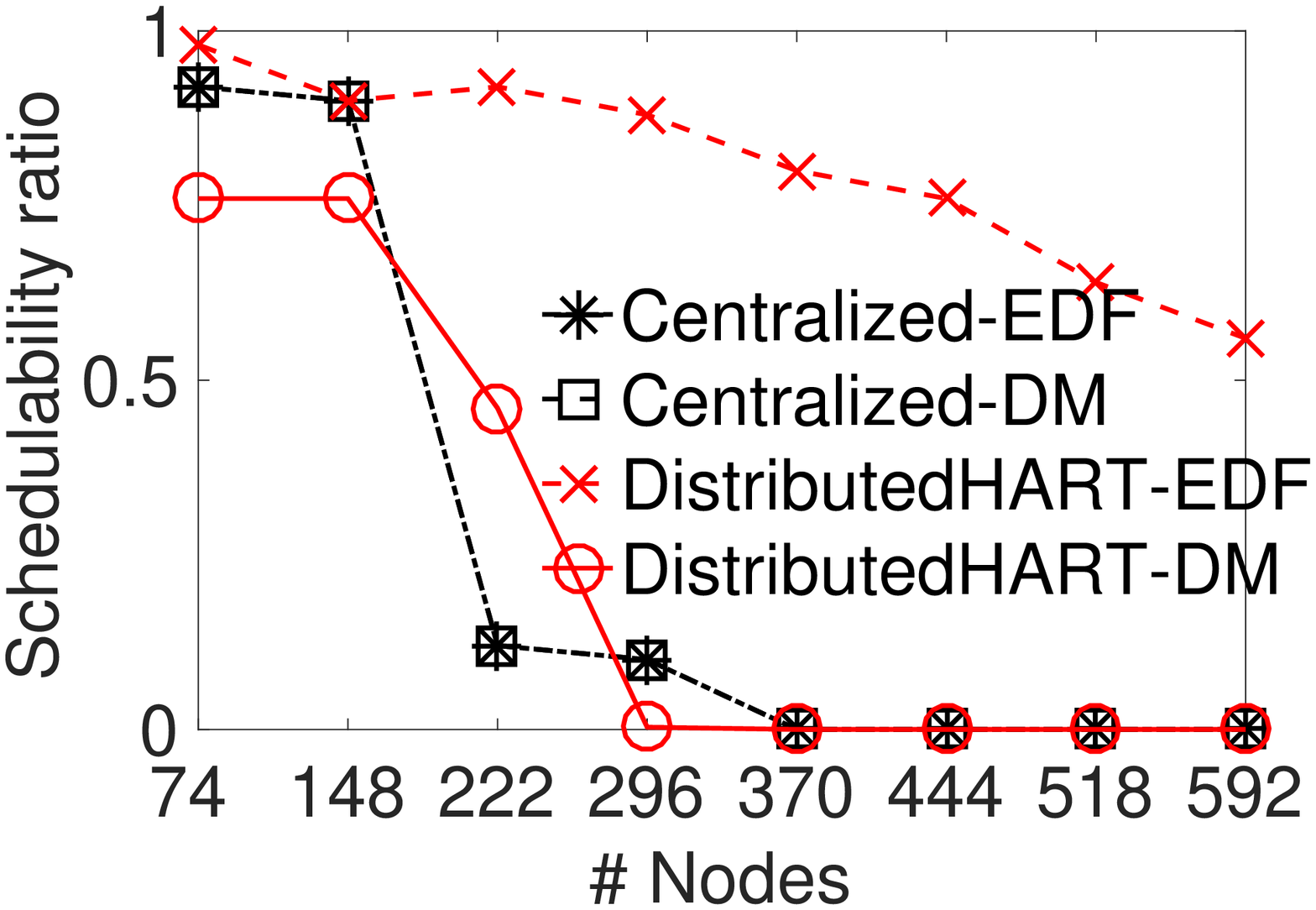}
		\caption{Schedulability Ratio under Varying Number of Nodes}
		\label{hart_fig:nodes_sched}
	\end{subfigure}
	\caption{Performance under Varying Number of Nodes}
	\label{hart_fig:nodes}
\end{figure*}

Here, we show the performance of DistributedHART under varying number of nodes. Fig. \ref{hart_fig:nodes} shows the simulation results when number of control loops is $20\%$ of the number of nodes.

  \textbf{Memory Consumption.} 
In this simulation, the number of control loops increases linearly with the increase in the number of nodes, which exponentially increases the hyper-period. Thus, Memory consumption for both centralized and distributed algorithms follow the same result as varying number of control loops with harmonic periods.  Fig. \ref{hart_fig:nodes_memory} shows DistributedHART consumes $99\%$ less memory than EDF.

  \textbf{Energy Consumption and Convergence Time. }
Both energy and convergence time follow the similar curves as memory consumption as they are also dependent on hyper-period. Fig. \ref{hart_fig:nodes_energy} and Fig. \ref{hart_fig:nodes_time} show DistributedHART consumes at least $94\%$ less energy and $85\%$ less convergence time.

  \textbf{Schedulability Ratio.}
As shown in Fig \ref{hart_fig:nodes_sched}, centralized algorithms are not scalable with the number of nodes due to the huge wastage of time slots. In contrast, DistributedHART offers better schedulability ratio when compared to centralized algorithms. From this result, we can conclude that DistributedHART scales with number of nodes.

\subsection{Performance under Varying Hyper-Periods}
\label{hart_sec:varyingHyperPeriods}

\begin{figure*}[t]
    	\centering
	\begin{subfigure}[b]{0.35\textwidth}
		\includegraphics[width=\textwidth]{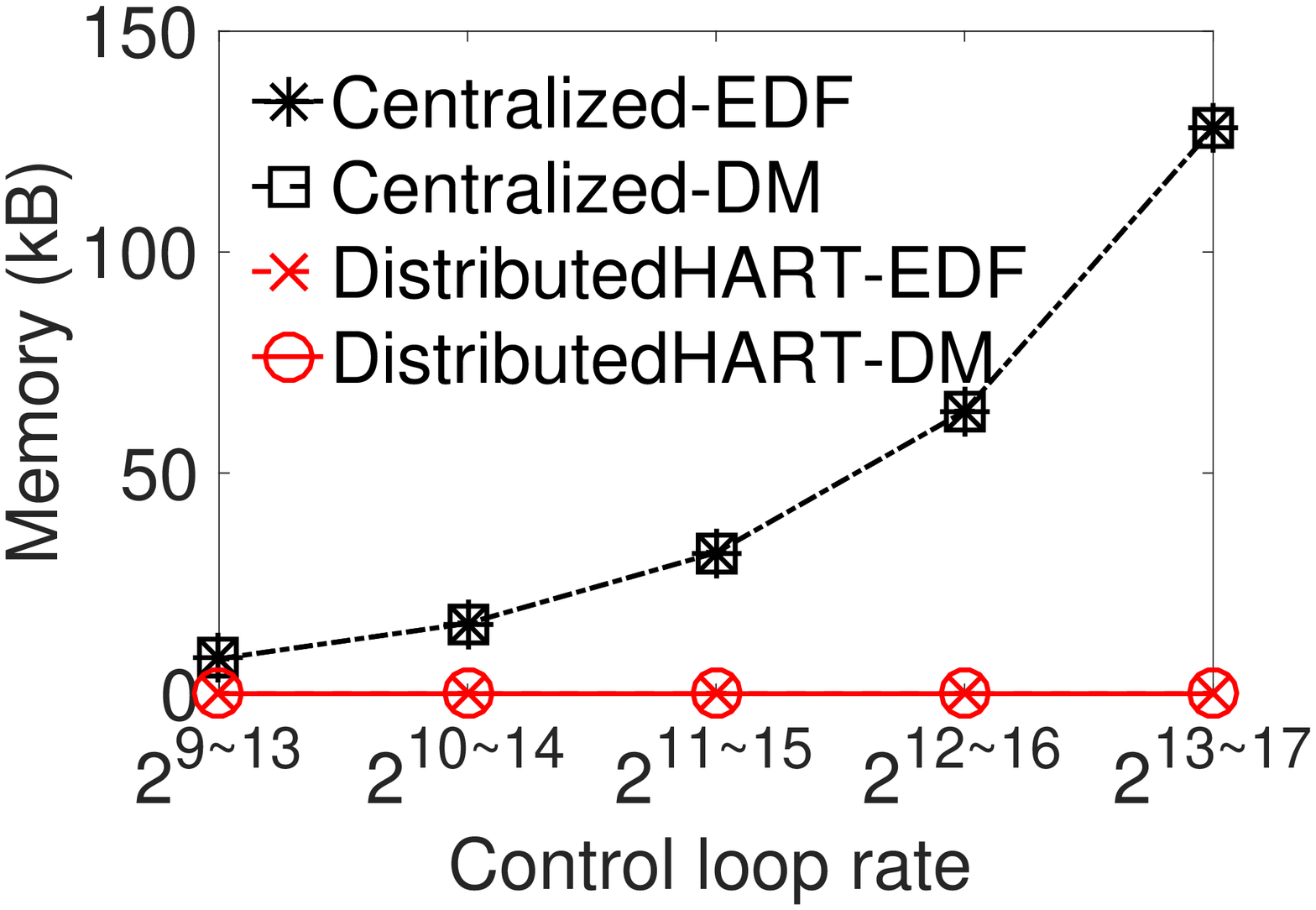}
		\caption{Memory Consumption}
		\label{hart_fig:flow_memory}
      	\end{subfigure}
	\quad
      	\begin{subfigure}[b]{0.35\textwidth}
		\includegraphics[width=\textwidth]{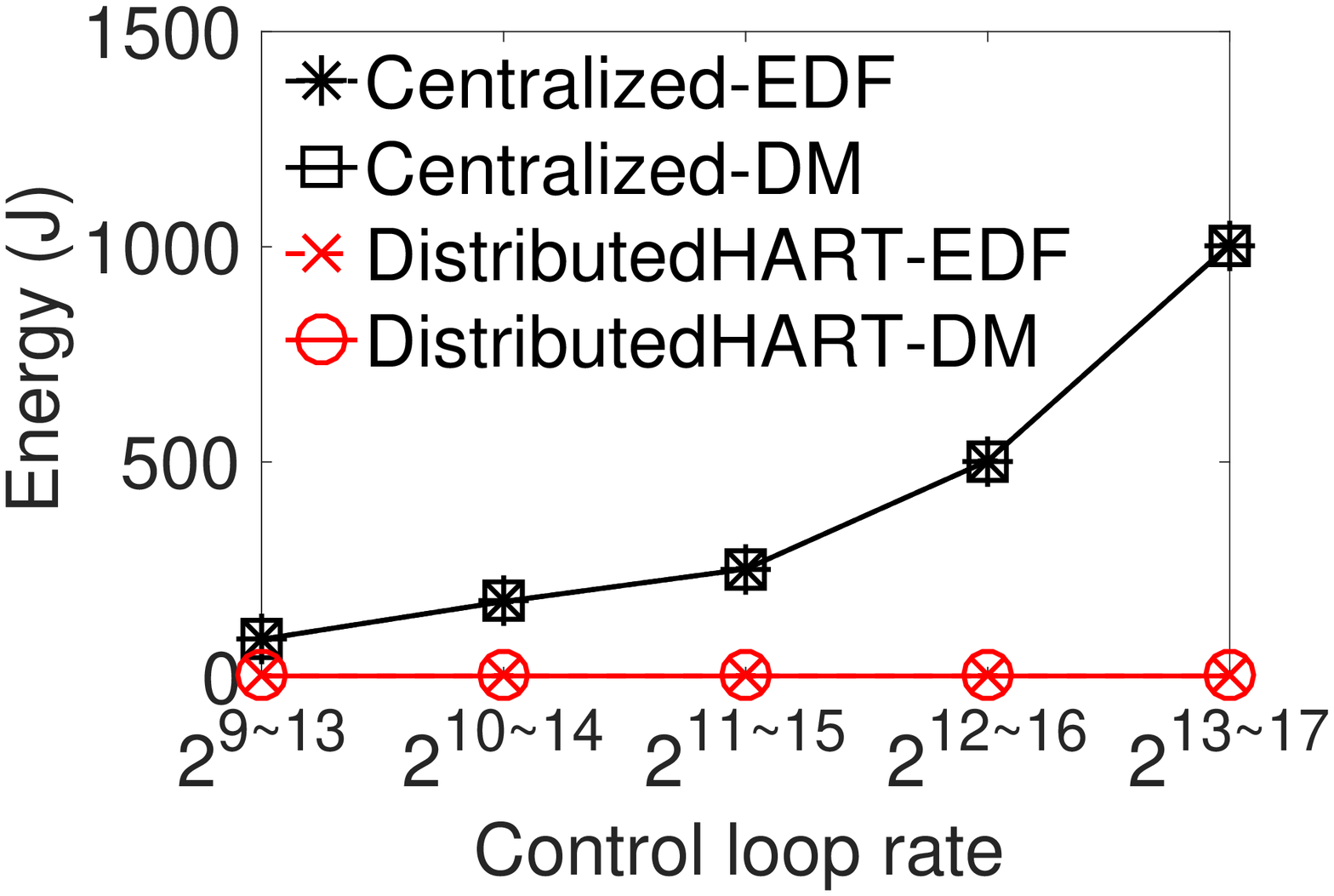}
		\caption{Energy consumption}
		\label{hart_fig:flow_energy}%
	\end{subfigure}
	\quad
	\begin{subfigure}[b]{0.35\textwidth}
		\includegraphics[width=\textwidth]{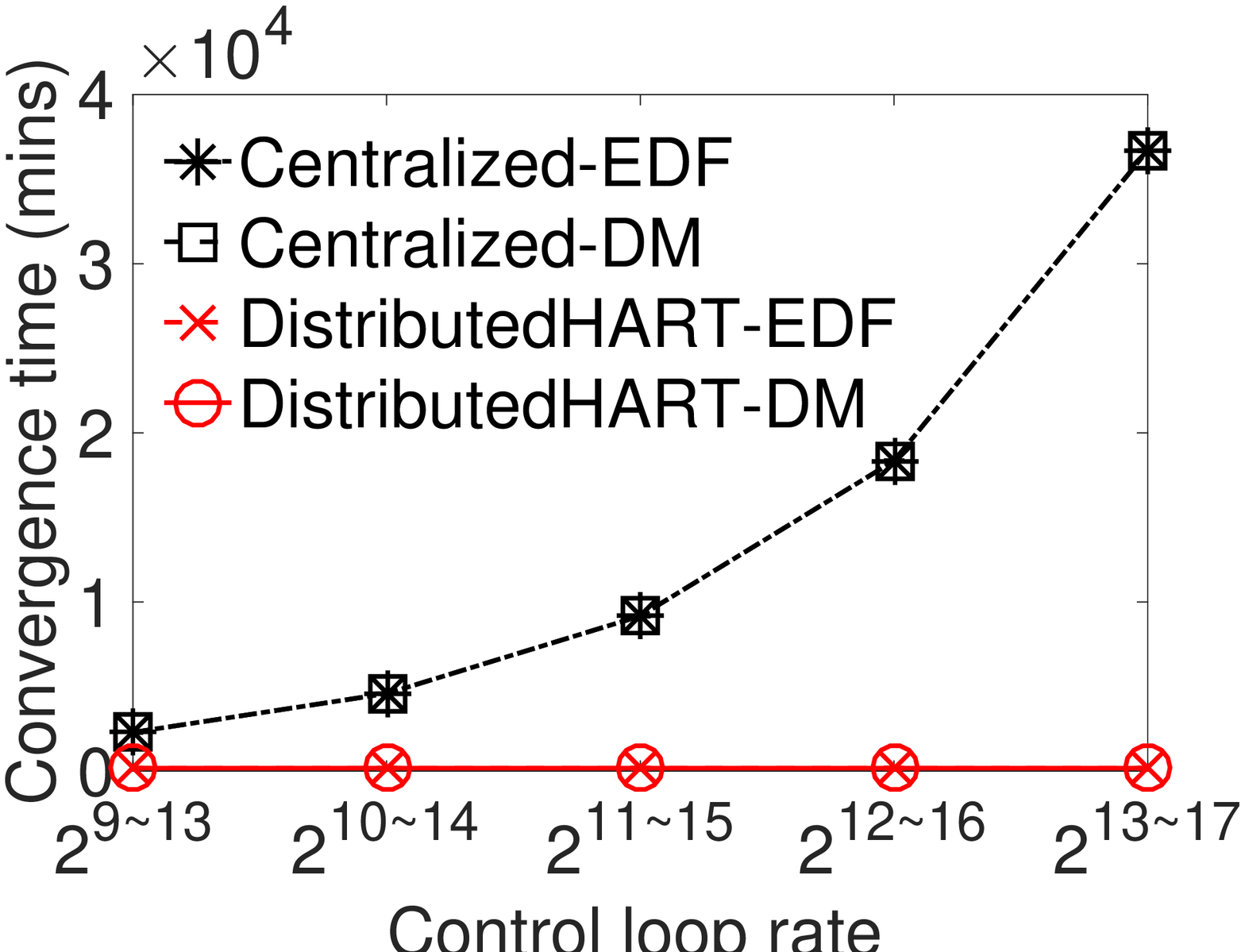}
		\caption{Convergence Time}
		\label{hart_fig:flow_time}
      	\end{subfigure}
	\quad
	\begin{subfigure}[b]{0.35\textwidth}
		\includegraphics[width=\textwidth]{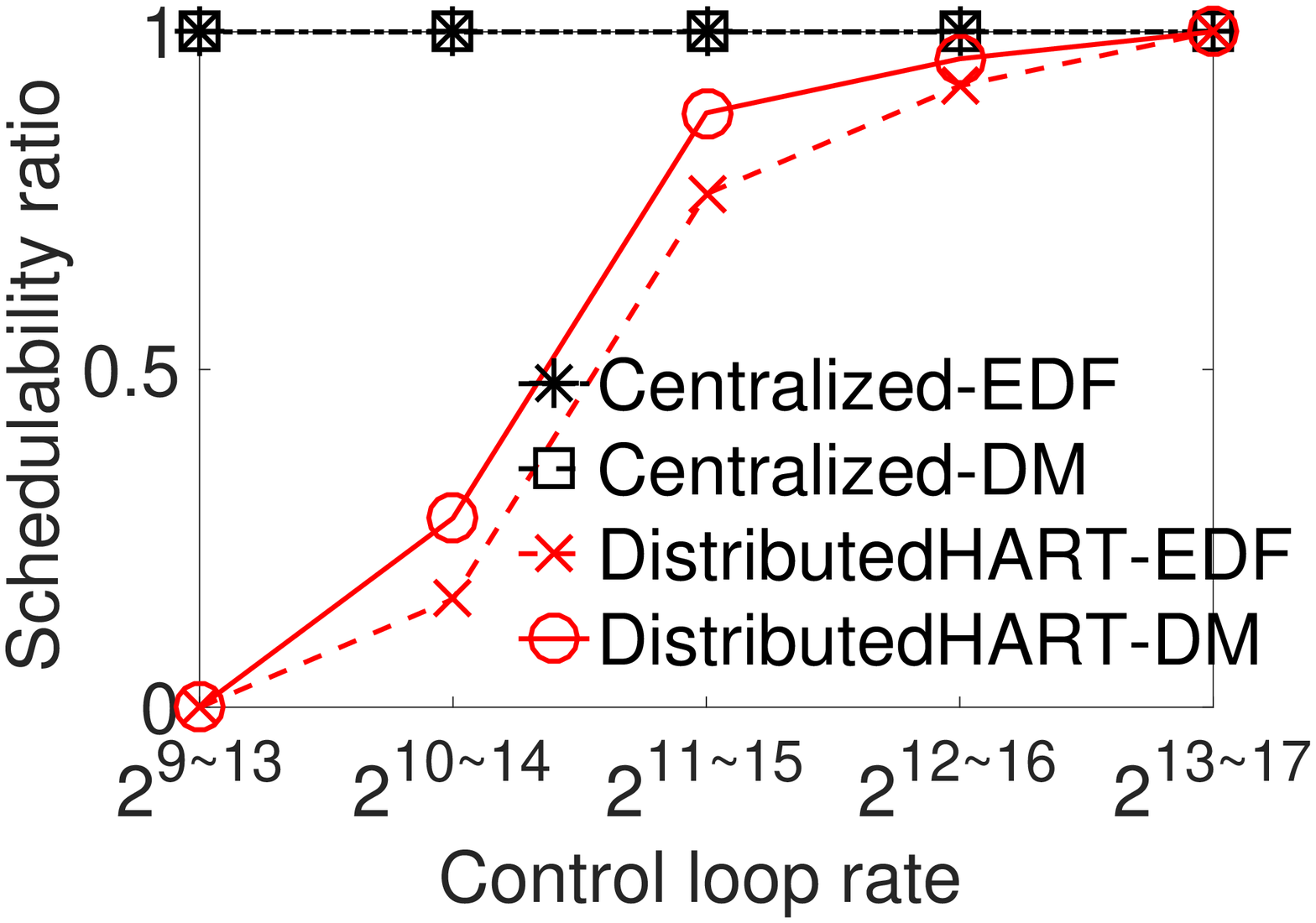}
		\caption{Schedulability Ratio}
		\label{hart_fig:flow_sched}
      	\end{subfigure}
	\caption{Performance under Varying Hyper-Periods}
	\label{hart_fig:flow}
\end{figure*}

Here, we show the performance of DistributedHART under varying ranges of hyper-periods. We used the same setup as performance under varying number of control loops with non harmonic periods (Appendix \ref{hart_sec:sim_nonharmonic}). However, we kept the number of control loops constant at $30$ and varied the range of periods from $2^{9\sim 13}$ to $2^{13\sim 17}$.

  \textbf{Memory Consumption. }
Fig. \ref{hart_fig:flow_memory} shows memory consumption increases exponentially for centralized EDF and DM. This exponential increase in memory consumption is due to the exponential increase in the hyper-period. However,  memory consumption for DistributedHART remains constant since the memory consumption is not a function of the hyper-period. We have observed that DistirbutedHART consumes $99\%$ less memory than centralized algorithms.  

  \textbf{Energy Consumption and Convergence Time. }
As expected, energy and convergence time of centralized EDF and DM increases exponentially. However, energy (as shown in Fig. \ref{hart_fig:flow_energy}) and convergence time (as shown in Fig. \ref{hart_fig:flow_time}) for DistributedHART is constant and remain unaffected by varying hyper-period lengths.

  \textbf{Schedulability Ratio. }
As shown in Fig. \ref{hart_fig:flow_sched}, with an increase in the range of periods, schedulability ratio remains constant at $1$ for centralized algorithms and DistributedHART. For DistributedHART, we observed that increasing the period improved the schedulability ratio due to the increase in deadline, while total delay remained constant. From these results, we can conclude that under DistributedHART outperforms centralized algorithms in terms of energy, convergence time and memory while achieving similar schedulability ratio.

\subsection{Performance under Varying Workload Dynamics}
\label{hart_sec:varyingWorkload}

\begin{figure*}[t]
    	\centering
	\begin{subfigure}[b]{0.35\textwidth}
		\includegraphics[width=\textwidth]{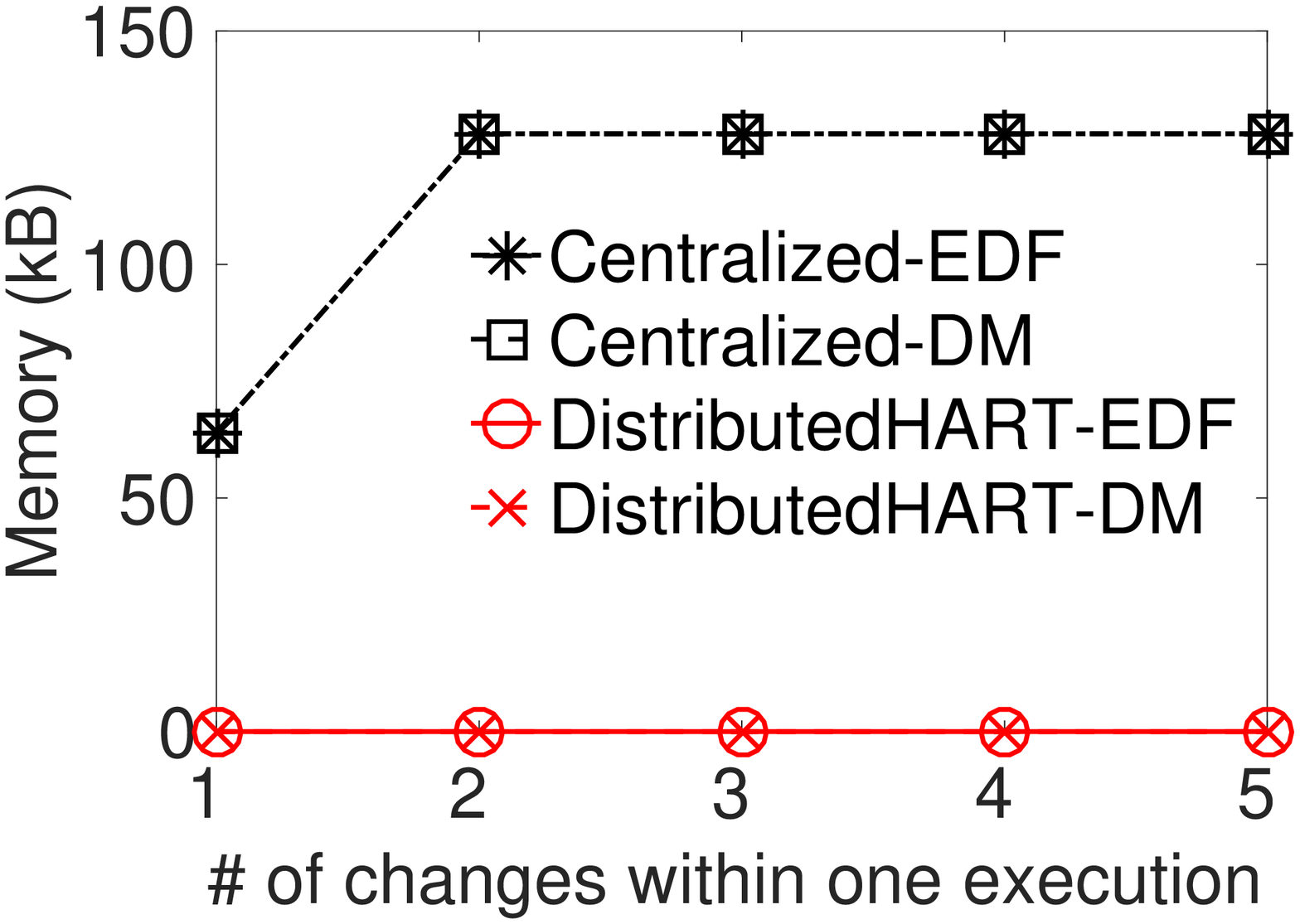}
		\caption{Memory Consumption}
		\label{hart_fig:VaryingFlow_memory}
      	\end{subfigure}
	\quad
       \begin{subfigure}[b]{0.35\textwidth}
       		\includegraphics[width=\textwidth]{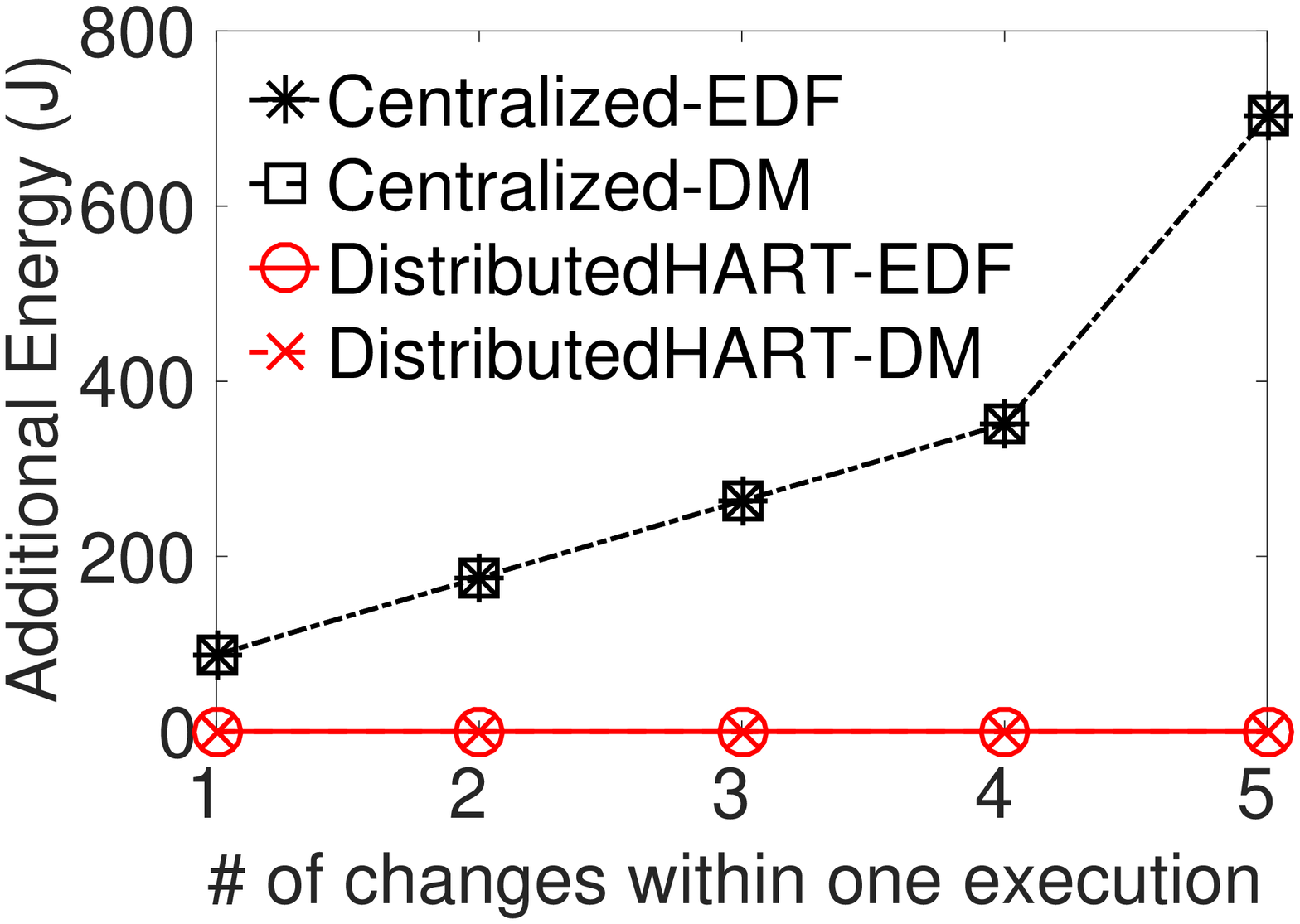}
		\caption{Energy consumption}
		\label{hart_fig:VaryingFlow_energy}%
	\end{subfigure}
	\quad
	\begin{subfigure}[b]{0.35\textwidth}
		\includegraphics[width=\textwidth]{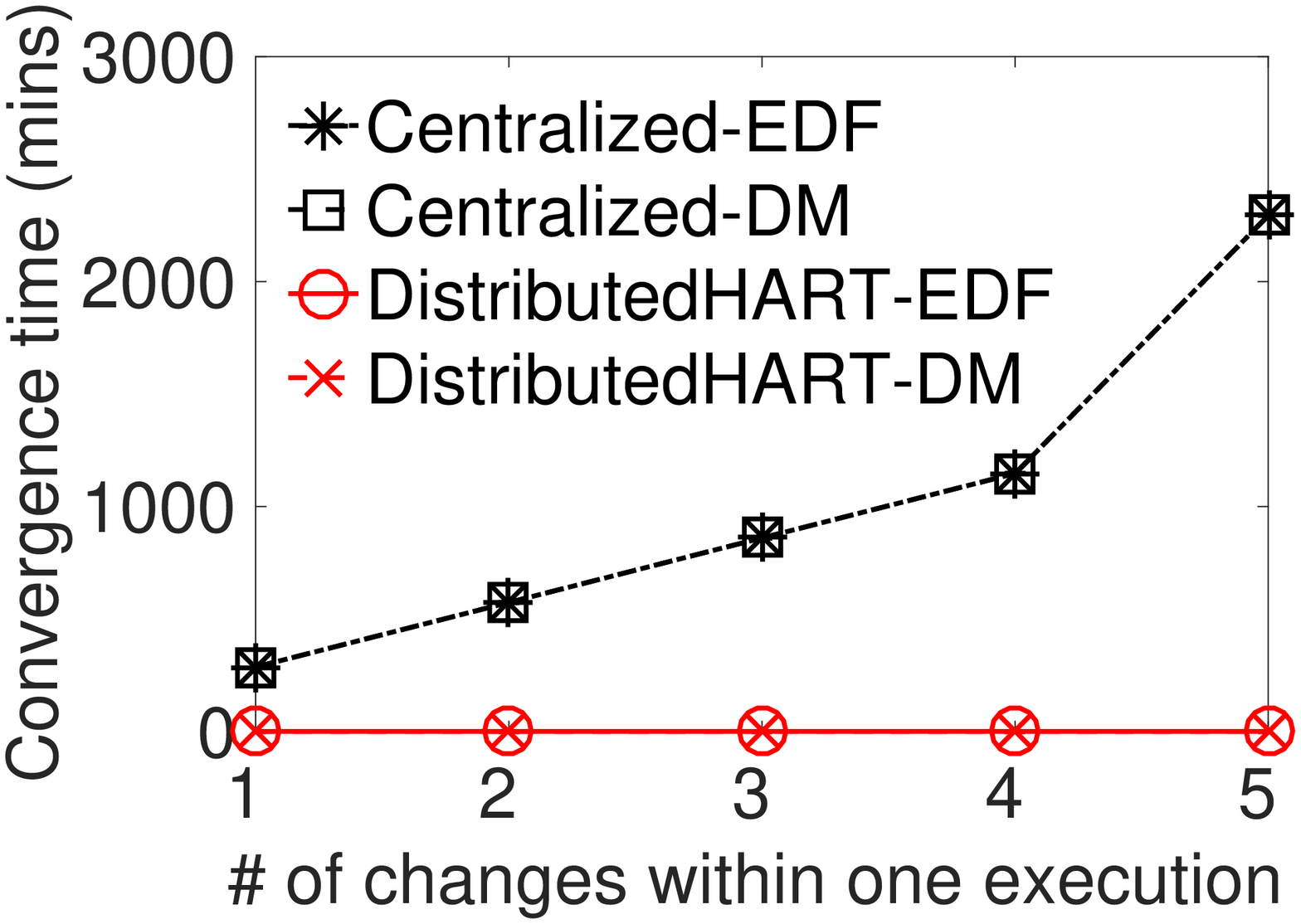}
		\caption{Convergence Time}
		\label{hart_fig:VaryingFlow_time}
      	\end{subfigure}
	\begin{subfigure}[b]{0.35\textwidth}
		\includegraphics[width=\textwidth]{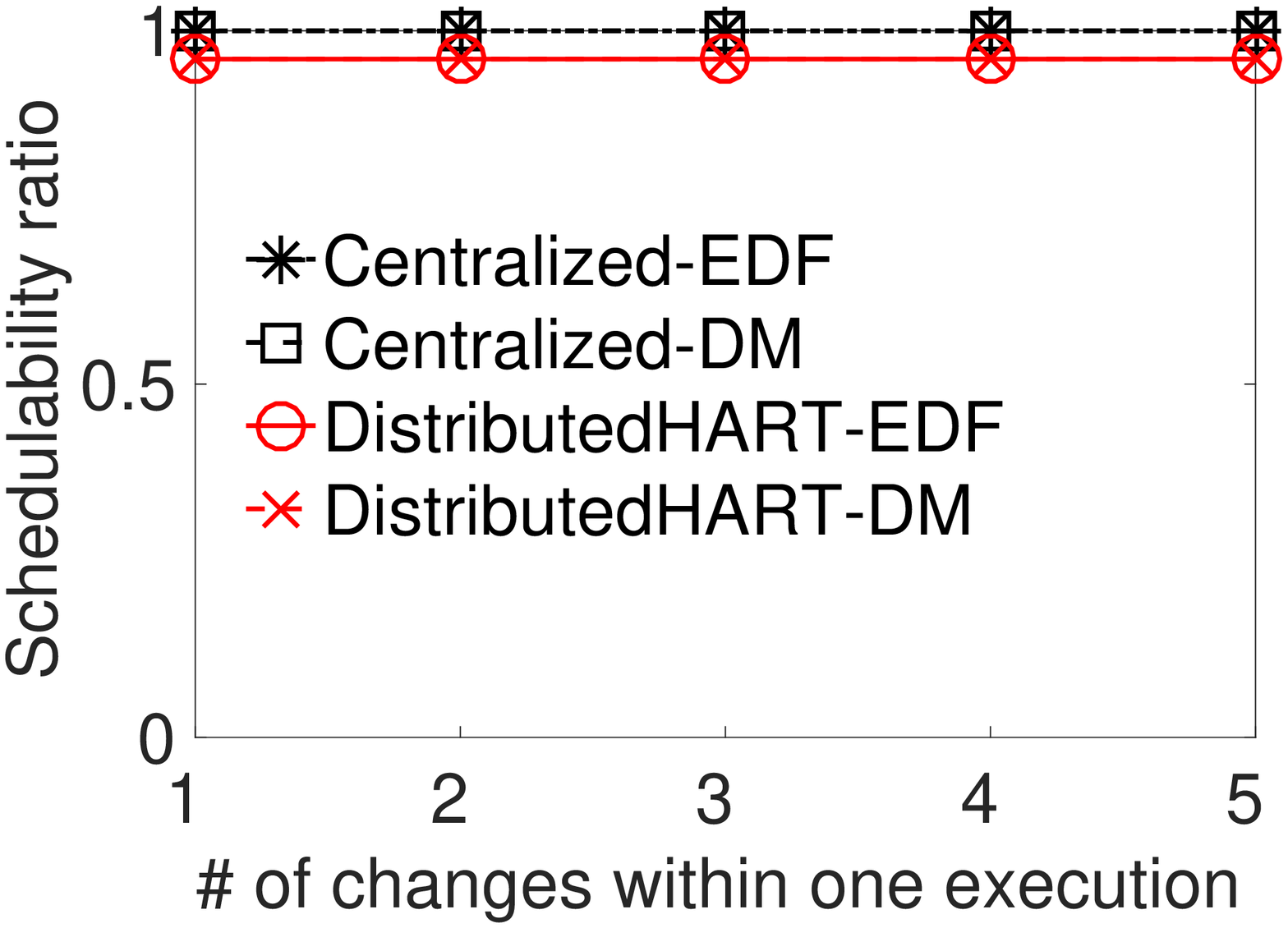}
		\caption{Schedulability ratio}
		\label{hart_fig:VaryingFlow_sched}
	\end{subfigure}
	\caption{Performance under Varying Workload Dynamic}
	\label{hart_fig:VaryingFlow}
	\vspace{-0.2in}
\end{figure*}
 
Fig. \ref{hart_fig:VaryingFlow} shows the performance of DistributedHART under different network dynamics. In this simulation, we kept the number of flows constant at $30$ and varied (increased or decreased) the period of random $5$ flows per workload change, while ensuring hyper-period is in between $2^{14 \sim18}$ time slots. We use the number of workload changes in one execution as a parameter for comparing the performance of DistributedHART and centralized algorithms. 

  \textbf{Memory Consumption.}
For $1$ workload change, hyper-period at most doubles. For more than one change, hyper-period increases by at most 4 times as shown in Fig \ref{hart_fig:VaryingFlow_memory}. For DistributedHART, change in the hyper-period does not affect the memory consumption and hence, it remains constant.

  \textbf{Energy Consumption and Convergence Time. }
For centralized algorithms, every change in workload necessitates an update in the schedule. A centralized approach has to collect the entire topology, re-create schedules and distribute the new schedules to all nodes in the network for each network/workload dynamic. This new schedule has to re-disseminated to the nodes. As shown in Fig. \ref{hart_fig:VaryingFlow_energy}, energy consumed by centralized algorithms (for changes in workload) increases linearly with the increase in the number of workload changes. Thus, centralized approaches are inefficient in large networks with frequent network dynamics. In DistributedHART, the local scheduler at each node handles workload dynamics. Thus, DistributedHART requires $0J$ of additional energy and $0s$ of convergence time to generate a schedule (for every workload change). Therefore, Fig \ref{hart_fig:VaryingFlow_time} shows a linear increase for centralized algorithms and $0$ for DistributedHART.

  \textbf{Schedulability Ratio. }
In the event of a workload change, we define a test case to be schedulable if all flows in the test case meet the deadline before and after the workload change.
In our simulation, we observed that centralized algorithms and DistributedHART have similar schedulability ratio, where centralized algorithms perform $4\%$ better than DistributedHART. A centralized approach has to collect the topology, re-create and distribute new schedules to all nodes upon dynamics. Network dynamics are quite frequent in industrial environments, making a centralized approach inefficient in large-scale networks. From this simulation, we can conclude that for varying workload requirements DistributedHART outperforms centralized algorithms in terms of energy and convergence time while offering similar schedulability ratio.

 \subsection{Performance of End-to-End Delay Analysis}
 \label{hart_sec:sim_delayanalysis}
Here, we show the performance of the proposed schedulability analysis. We compare the schedulability ratio obtained by the schedulability analysis and the simulation result (which provides a conservative upper bound). We use a similar setup as Appendix \ref{hart_sec:sim_nonharmonic} with periods in range $2^{12\sim 16}$ time slots. We present the schedulability ratio when end-to-end delay is estimated with a probability of $0.95$.

\begin{figure}[h]
\centering
\includegraphics[width=0.35\textwidth]{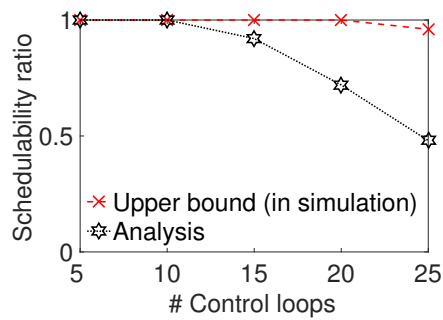}
\caption{Schedulability Ratio Comparison of Delay Analysis with an Upper Bound}
\label{hart_fig:SchedTest}
\end{figure}

As shown in Fig. \ref{hart_fig:SchedTest}, when the number of control loops is $\leq$ 10, all test cases were deemed to be schedulable under our schedulability analysis and simulations.  When the number of control loops was greater than 10, fewer test cases were deemed schedulable by our schedulability analysis.  This difference in schedulability ratio is because simulation results show a conservative upper bound on schedulability ratio while our schedulability analysis considers a pessimistic scenario (where each node requires two transmission time slots). Note that, our analysis is only a sufficient test and not an exact test. From this result, we can conclude that our schedulability analysis is close to the upper bound and can be used to determine schedulability.

\subsection{Comparison in terms of Percentage of Idle Slots }
\label{hart_sec:unused}

Here, we show the performance comparison between DistributedHART and centralized EDF in terms of percentage of idle slots and percentage of idle slots available for new flows. We define {\slshape percentage of idle slots} as the percentage of slots in a hyper-period that are unassigned or not used for packet transmission. We define {\slshape percentage of idle slots available for flows} as the percentage of slots in a hyper-period that unassigned or not used for packet transmission, which can be used to schedule packets of existing/new flows.

\begin{figure}
    \centering
    \includegraphics[width=0.35\textwidth]{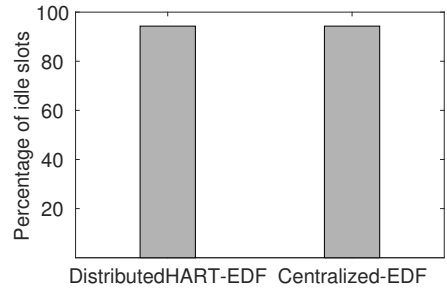}
    \caption{Performance Comparison in Terms of Percentage if Idle Slots}
    \label{hart_fig:unusedSlots}
\end{figure}

\begin{figure}
    \centering
    \includegraphics[width=0.35\textwidth]{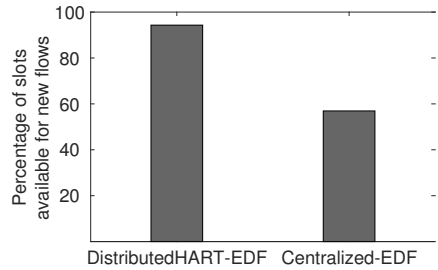}
    \caption{Performance Comparison in Terms of Percentage of Idle Slots Available for Flows}
    \label{hart_fig:usableSlots}
\end{figure}

To evaluate DistributedHART, we used a $148$ node network with $30$ flows and used a period assignment of $2.5s$ for all the flows. For this setup, Fig. \ref{hart_fig:unusedSlots} compares the percentage of idle slots between DistributedHART and centralized EDF algorithm. We observed that, in a single hyper-period, DistributedHART used approximately $250$ transmissions to transmit all packets to the destination, which leaves all other slots ($95\%$ approximately) idle. Under the assumption that the number of transmission failures is the same, both centralized WirelessHART and DistributedWirelessHART use the same number of time slots to transmit a packet, i.e., $95\%$ the slots are idle.  For DistributedHART, nodes can locally choose to transmit a packet from any flow during its time window. Thus, nodes can reuse all the idle time slots of DistributedHART for transmission, as shown in Fig. \ref{hart_fig:usableSlots}. However, for the centralized algorithms, the scheduler assigns time slots to flows, and nodes cannot use these time slots to transmit packets from any existing/new flows. We have observed from our evaluation shows that $56\%$ of the total slots are available for scheduling packets from existing/new flows. From this simulation, we can conclude that DistributedHART scales well by utilizing its network resources.

 \subsection{Performance Comparison between DistributedHART and Centralized Scheduler with Compact Schedules}
 \label{hart_sec:sim_compact}
 \begin{figure}[h]
	\centering  
	\includegraphics[width=0.35\textwidth]{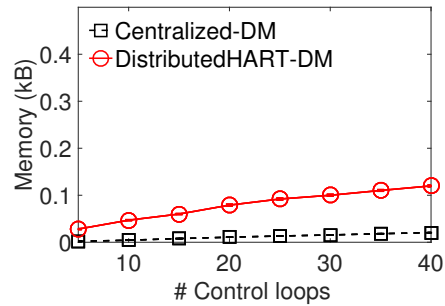}
	\caption{Memory Consumption Comparison between DistributedHART and Centralized - compact scheduler}
	\label{hart_fig:MemoryCompact}
\end{figure}

Here, we show the performance of DistributedHART when compared to the centralized scheduler with compact schedules. For special scenarios, centralized schedulers can create a compact schedule when periods are harmonic, and scheduling policy is rate monotonic. In the compact schedules, a node can use the first packet's time slot to generate time slots for subsequent packets within a hyper-period.  Note that, a central scheduler should assign each node the time slot and period of the first packet for each flow passing through it. For generic scenarios, centralized schedulers have to consider interference from all packets hence this approach will not work. Fig \ref{hart_fig:MemoryCompact} shows the memory consumption comparison between DistributedHART-DM and centralized-DM with a compact schedule. Our simulation results show that DistributedHART and centralized-EDF compact scheduler consume a similar amount of memory and we observed that DistributedHART consumes 0.08KB of additional memory compared to compact centralized schedulers.

%% file: snow_latency/rtss_integration.tex
\chapter{Low-Latency In-Band Integration of Multiple Low-Power Wide-Area Networks}
\label{chapter:low_latency_integration}
Today, industrial and agricultural Internet of Things (IoT) are emerging in very large-scale  and wide-area applications (e.g., oil-field management, smart farming) that may spread over hundreds of square kms (e.g., 74x8km$^2$ East Texas Oil-field). Although a single Low-Power Wide-Area Network (LPWAN)  covers several kms, it faces coverage challenge in such extremely large-area IoT applications, specially in rural or remote areas with no/limited infrastructure, requiring an in-band integration of multiple LPWANs. 
To avoid the crowd in the limited ISM band and the cost of licensed band and infrastructure,  {\slshape  SNOW (Sensor Network Over White spaces)} is an LPWAN architecture  over the TV white spaces.  It offers high scalability through concurrent asynchronous and bi-directional communication between a base station and numerous nodes. 

We consider a seamless integration of multiple SNOWs. Existing approach does not consider minimizing network latency and is less suitable for delay-sensitive and real-time applications. We propose the first scalable in-band integration of  multiple SNOWs that minimizes network latency. By taking into account the impact of  bandwidth on latency and  base station power dissipation, we formulate low-latency integration of multiple SNOWs as a constrained optimal spectrum allocation problem.  The formulated problem is solved through a greedy algorithm by analyzing network latency and by adopting a traffic-aware and latency-aware bandwidth allocation along the links to achieve an integrated network. We have implemented the proposed integration both on SNOW hardware and in NS-3 simulator. Both physical experiments and simulations show a significant reduction ($50\%$  and $84\%$, resp.) in network latency under our approach compared to existing approach.  

\input{snow_latency/introduction}

\input{snow_latency/snow_overview}

\input{snow_latency/integration}

\input{snow_latency/proposed_approach}

\input{snow_latency/evaluation}

\section{Summary}\label{sec:conclusion}
LPWAN is a promising IoT technology for communicating over long distances at low power. Despite their promise,  LPWANs face challenges in covering very wide areas making their adoption challenging for IoT applications in agriculture, oil, and gas fields that may extend over hundreds of square kms and are in rural/remote areas with no/limited infrastructure. To cover such large-areas, we have proposed to scale up LPWAN through a seamless in-band integration of multiple SNOWs. SNOW is an LPWAN architecture over the TV white spaces. We have proposed the first latency minimizing in-band integration of  multiple SNOWs. We have implemented the proposed integration on SNOW hardware platform and observed, through physical experiments, up to $50\%$ decrease in network latency compared to the existing approach. We have also performed simulations in NS-3 and have observed up to $84\%$ decrease in network latency in an integrated network of  5000 nodes (of 5 SNOWs) based on our approach compared to the existing approach. In the future, we plan to extend this work to support a closed loop communication between sensors and actuators for industrial IoT. we plan on leveraging on edge-based control of plants in wide area deployments to provide a low-latency control for applications while also collecting data to a central location for monitoring the control performance of the entire system.

%% file: snow_latency/introduction.tex
\section{Introduction}
\label{sec:integration_intro}

As an emerging Internet-of-Things (IoT) technology, {\slshape Low-Power Wide-Area Network (LPWAN)} enables low-power (milliwatts) wireless devices (sensors) to transmit at low data rates (kbps) over long distances (kms) using narrowband (kHz) \cite{oursurvey, lorawan_iiot}. With the fast growth of IoT, multiple  LPWAN technologies have recently emerged such as LoRa~\cite{lorawan}, SigFox~\cite{sigfox}, IQRF \cite{iqrf}, RPMA \cite{rpma}, DASH7~\cite{dash7}, Weightless-N/P~\cite{weightless}, Telensa~\cite{telensa} in the ISM band,  and  EC-GSM-IoT~\cite{ecgsmiot}, NB-IoT~\cite{nbiot},  LTE Cat M1 \cite{cat, LTE_advancedpro} in the licensed cellular band. To avoid the {\slshape crowd} in the {\slshape limited} ISM band and the {\slshape cost} of licensed band,  {\slshape  SNOW (Sensor Network Over White spaces)} is an LPWAN architecture  to support scalable wide-area IoT  over the TV white spaces \cite{ton_snow, saifullah2016snow, saifullah2017snow2}.  {\slshape White spaces} are the allocated but locally unused TV spectrums where unlicensed devices can operate as secondary users~\cite{FCC_first_order, fcc_second_order}. To learn about white spaces at a location,  a device needs to either sense the medium before transmitting, or consult with a cloud-hosted geo-location database, either periodically or every time it moves 100 meters~\cite{fcc_second_order}. Compared to the ISM band, the TV white spaces have lower frequencies (e.g., 54 -- 862kHz in USA) and much wider, less crowded spectrum in both rural and most urban areas, with an abundance in the former \cite{ws_sigcomm09}.

Although a single LPWAN covers several kilometers (kms),  it faces coverage and scalability challenges in very large-area (e.g., city-wide) deployment \cite{Charm, Choir, lorascale, lorascale2, loralimitation, bor2016lora}. Today, industrial and agricultural IoT and cyber-physical systems are emerging in large-scale  and wide-area applications. Specifically, agricultural fields \cite{agri1, agri2, oregonFarming}   and oil/gas fields \cite{oilgas}  may extend over hundreds of square kms.  For example,  the East Texas Oil-field  extends over an area of $74\times 8$km$^2$   requiring tens of thousands of sensors for management \cite{texasof}. Emerson  is targeting to deploy 10,000 nodes for managing an oil-field in Texas \cite{WH10000, WirelessHART_Emerson}.  To cover such large-areas for agricultural/industrial IoT, we need to integrate multiple LPWANs.  
 LPWANs are usually limited to star topology, and rely mostly on wired infrastructure (e.g., cellular LPWANs) or Internet (e.g., LoRa) to integrate multiple networks to cover large areas. However, lack of infrastructure hinders their adoption to rural/remote area applications such as agricultural/industrial IoT.  In this chapter, we propose a scalable and low-latency in-band integration of multiple LPWANs under no/limited infrastructure. 

We consider integrating multiple SNOWs whose conceptual notion can also be extended to other LPWANs. SNOW offers high scalability due to concurrent asynchronous and bi-directional communication between a base station (BS) and numerous nodes \cite{saifullah2016snow, ton_snow, saifullah2017snow2}.  Its implementation is available as open-source  \cite{opensource}. Due to their rapid growth,  LPWANs in the ISM band will suffer from crowded spectrum, making it  critical to exploit white spaces. 
Compared to cellular LPWANs, SNOW does not need wired infrastructure making it suitable in both rural and urban areas. Due to abundant white spaces, it is a promising platform for  smart farming,  a global need and recommendation by the United Nations to meet the 70\% more food demand by 2050~\cite{forbes, Beecham}.   Study shows that smart farming powered by IoT can double the produce at low cost by better measuring soil nutrients, moisture, fertilizer, seeds, and storage temperature through dense sensor deployment~\cite{farm, chemistry, nature}.  Industries like Microsoft \cite{vasisht17farmbeats, farmbeat} and  Monsanto \cite{monsanto} are hence promoting agricultural IoT.

\begin{figure}[!htb]
\centering 
\includegraphics[width=0.35\textwidth]{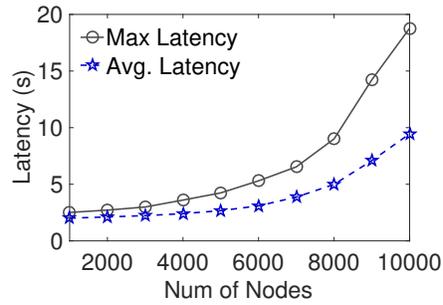}
\caption{Latency  under Existing SNOW Integration.}
\label{fig:motivation_results}
\end{figure} 
In-band integration considered in this chapter is conceptually different from and more challenging than traditional tiered or clustered wireless sensor network (WSN) \cite{leach2002} or 802.11 mesh \cite{80211mesh}. Unlike these networks,  integrating multiple SNOWs needs to find proper bandwidths for all links which are inter-dependent and subject to BS power dissipation.  Besides, the unique features of SNOW including massive parallel communication require a new approach for a seamless integration that will enable concurrent inter- and intra-SNOW communication. While a seamless integration of SNOWs was studied in a recent work \cite{integration, ton_integration}, its objective is to trade the scale (in terms of number of total nodes) for inter-SNOW interference. Its bandwidth allocation for inter-SNOW communication does not take into account the link traffics. As a result, it performs poorly in terms of network latency, eventually affecting the scale. A simulation using NS-3 \cite{ns3} in such an integrated network of five SNOWs shows that the maximum latency for collecting a packet increases exponentially and deviates sharply from average latency with the number of nodes (Fig. \ref{fig:motivation_results}). Thus, the existing approach is less suitable for delay-sensitive and real-time applications.  Note that many WSN applications are time-sensitive. A recent survey on 311 industries conducted by ON World and the International Society of Automation shows that 57\% of industrial IoT professionals are targeting LPWAN for industrial WSN applications  \cite{industrialLPWAN1, industrialLPWAN2}.

We propose the first scalable in-band integration of  multiple SNOWs that minimizes network latency. The integration is media access control (MAC) protocol-agnostic, i.e., in the integrated network any MAC protocol can be used in each single SNOW. Our contributions are listed as follows. 
\begin{itemize}
\item By taking into account the impact of  bandwidth on latency and  BS power dissipation, we formulate latency minimizing integration of multiple SNOWs as a constrained optimal spectrum allocation problem.  
\item The formulated problem is solved through a greedy algorithm by analyzing network latency under two different MAC protocols and by adopting a traffic-aware and latency-aware bandwidth allocation along the links to achieve an integrated network. This approach can be extended to other future LPWANs in white spaces (e.g., based on upcoming 802.15.4m standard \cite{802154m}).

\item We have implemented the proposed integration on SNOW hardware platform and observed, through physical experiments, up to $50\%$ decrease in network latency compared to the existing approach \cite{integration, ton_integration}. We have also performed simulations in NS-3 \cite{ns3} and have observed up to $84\%$ decrease in network latency in an integrated network of  5000 nodes (of 5 SNOWs) based on our approach compared to the existing approach. 
\end{itemize}

The rest of the chapter is organized as follows. Section \ref{sec:snow_overview} presents an overview of SNOW. Section \ref{sec:integration} gives a high-level overview  of our proposed approach for integrating multiple SNOWs along with related work, challenges, and the problem formulation. Section \ref{sec:proposed_method} details the proposed algorithm for solving the latency-minimizing integration problem. Section \ref{sec:experiment} presents  hardware experimental results. Section \ref{sec:simulation} presents simulation results. Section \ref{sec:conclusion} summarizes the chapter.

%% file: snow_latency/snow_overview.tex
\section{The SNOW Architecture}
\label{sec:snow_overview}
 
 Here we provide a brief overview of the SNOW LPWAN architecture. Its full description is available in \cite{ton_snow}. Due to long transmission (Tx) range, the nodes (sensor nodes) in SNOW are directly connected to the BS over up to several kms, forming a star topology as shown in Fig.~\ref{fig:arch}.  The BS periodically determines white spaces by providing locations of the devices in a cloud-hosted database through the Internet. The nodes are power-constrained, and do not determine white spaces.  The BS uses wide white space spectrum as a single wide channel that is split into narrowband orthogonal subcarriers ($f_1$, $f_2$, $\cdots$), each of equal spectrum width (bandwidth). Each node has a single half-duplex narrowband radio. It sends/receives on a subcarrier.  An example of a SNOW node is a Texas Instruments CC1310 board \cite{snow_implementation}. As shown in Fig.~\ref{fig:arch}, the BS uses two radios --  one for only transmission (called {\bf\slshape Tx radio})  and the other for only reception (called {\bf\slshape Rx radio}) -- to facilitate concurrent bidirectional communication.  A BS is a line-powered device with sufficient processing capabilities like a Raspberry Pi equipped with a dual-radio USRP (universal software radio peripheral) \cite{usrp} device. 
 The SNOW implementation is available as open-source \cite{opensource}. 
 
  \begin{figure}[!htb]
\centering
\includegraphics[width=0.5\textwidth]{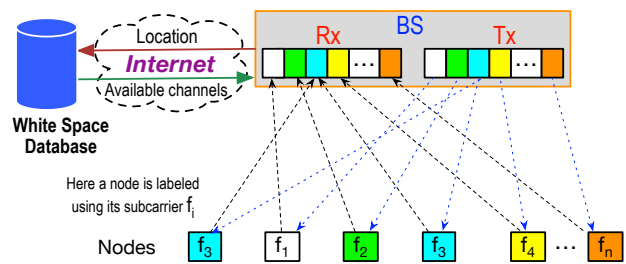}
\caption{\small SNOW architecture with dual radio BS and subcarriers \cite{ton_snow}.}
\label{fig:arch}
\end{figure}

To enable many simultaneous transmissions/receptions, the physical layer is designed based on a 
{\bf D}istributed implementation of {\bf OFDM} (Orthogonal Frequency Division Multiplexing), named {\bf D-OFDM},  for multi-user access. D-OFDM  splits a wide spectrum into many narrowband orthogonal subcarriers enabling parallel data streams to/from numerous distributed nodes from/to the BS.  A subcarrier bandwidth is in kHz (e.g., 50kHz, 100kHz, 200kHz, or so depending on packet size and needed bit rate). The nodes transmit/receive on orthogonal subcarriers, each using one. A subcarrier is modulated using Binary Phase Shift Keying (BPSK) or Amplitude Shift Keying (ASK).  If the BS spectrum is split into  $m$  subcarriers,  it can receive from $m$ nodes simultaneously using a single antenna. Similarly, it can transmit different data on different subcarriers through a single transmission. 

SNOW can adapt to any MAC protocol that is suitable for the energy-constrained nodes. In its default design, nodes use a lightweight CSMA/CA (carrier sense multiple access with collision avoidance) based MAC protocol for transmission that uses a static interval for random back-off like the one used in TinyOS~\cite{tinyos}. In this chapter, we adopt several energy-efficient MAC protocols for SNOW. The nodes can autonomously transmit, remain in receive (Rx) mode, or sleep. Since D-OFDM allows handling asynchronous Tx and Rx, the link-layer can send an acknowledgment (ACK) for any transmission in either direction. Both radios of the BS use the same bandwidth and subcarriers - the subcarriers in the Rx radio are for receiving while those in the Tx radio are for transmitting.

%% file: snow_latency/integration.tex
\section{Integrating Multiple SNOWs: An Overview of the Proposed Approach}
\label{sec:integration}

{\bf System model. }
Our objective is to enable scalability and extended coverage  in a wide area through a seamless integration of multiple SNOWs whose BSs are interconnected. We consider that the SNOWs form an in-band SNOW-tree (Fig. \ref{fig:integration_mesh}) like a  cluster tree  of IEEE 802.15.4m~\cite{802154m}. That is, their BSs  are connected through white spaces as a tree.  Let there be $N$ BSs, denoted by BS$_0$, BS$_1$, $\cdots$, BS$_{N-1}$, where BS$_i$ is the base station of SNOW$_i$. BS$_0$ is the {\slshape root BS} and is connected to the white space database through the Internet. BS$_0$  finds white spaces for all SNOWs and allocates spectrum among all SNOWs and tree links. To enable agricultural IoT (in farms) or industrial IoT (in oil/gas field), an Internet-connected root BS may so connect other BSs to cover the large area having no Internet. In urban areas, such an integration can overcome the reduced range problem  experienced by most LPWANs  \cite{Charm, Choir, lorascale, lorascale2, loralimitation, bor2016lora}. An integration will span a large number of nodes across large areas using multiple SNOWs, while all data need to be delivered to root base station BS$_0$. This delivery {\slshape latency} impacts the scalability through integration. Hence, we propose to minimize such latency.

\begin{figure}[!htb]
\centering 
\includegraphics[width=0.45\textwidth]{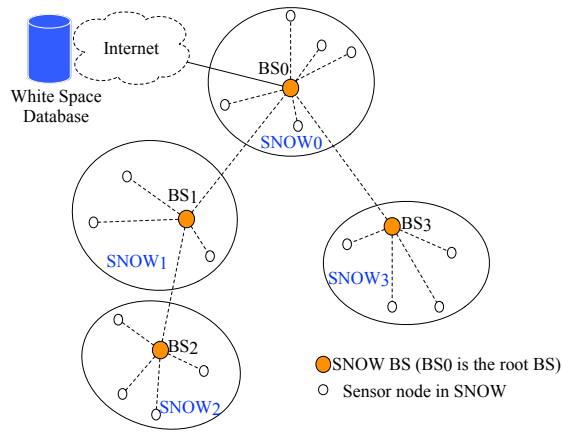}
\caption{Integrating Multiple SNOWs to Build a SNOW-tree.}
\label{fig:integration_mesh}
\end{figure}

BS$_0$ collects data from its nodes directly while data from other SNOWs are delivered to it through the tree links.  Through a BS's parent link, data from its descendant SNOWs are delivered. Latency along a tree link can be minimized by assigning a large number of subcarriers for it. However,  subcarrier assignment on a tree link affects that on other links and SNOWs. Also, by transmitting on a large number of subcarriers simultaneously the BS suffers from a traditional OFDM problem called {\bf\slshape peak to average power ratio (PAPR)}. PAPR is the ratio of the maximum power of a signal to its average power, and is equal to the number of subcarriers~\cite{ofdmbook2, ofdmbook}. PAPR can be high after inverse FFT (Fast Fourier Transform) done at the Tx radio of the BS during downlink transmission due to a large number of subcarriers forcing the non-ideal power amplifier operate non-linearly. A high PAPR increases power dissipated on the amplifiers resulting in higher power consumption and lower efficiency. Therefore, in integration, we have to minimize both latency and PAPR. Note that the (sensor) nodes do not suffer from PAPR as each of them transmits on a single subcarrier.

We minimize both latency and PAPR by assigning an appropriate number of subcarriers to each tree link (for BS-BS communication) and to each SNOW for intra-SNOW communication. This is done at BS$_0$ that knows the tree topology. We assign subcarriers such that communication along a  tree link is not interfered by others.   To reduce PAPR, we exclude ACK of transmissions along the tree links.  In intra-SNOW communication, when a large number of nodes transmit simultaneously to the BS, the BS needs to send ACK to all nodes which may cause high PAPR. To avoid this problem, we combine multiple ACKs into one frame. To handle multi-objective (latency and PAPR) of our approach, we will minimize latency while keeping PAPR as a constraint.

\subsection{Related Work and New Challenges} 
To cover a wide area, LoRa integrates multiple gateways through the Internet \cite{lorawan}. Cellular networks do the same relying on wired infrastructure~\cite{channel_cellular}. Rural and remote areas lack such infrastructure. Our proposed in-band integration is conceptually different from traditional tiered or clustered WSN \cite{leach2002} or 802.11 mesh \cite{80211mesh}. First, in these traditional networks, every link has an equal and known bandwidth. In contrast, link bandwidths in SNOW integration  are variable and inter-dependent. 
Second, SNOW integration needs to handle the PAPR of BS transmitter which is absent in those networks. Third,  the unique features of SNOW including massive parallel communication require a new approach for a seamless integration that will enable concurrent inter- and intra-SNOW communication. Hence, traditional channel allocation for wireless networks \cite{audhya2011survey, incel2011survey}) or cognitive radio networks \cite{chen2012game, tragos2013spectrum, cognitive_channel_survey2} cannot be used for bandwidth allocation in SNOW integration.


While SNOW integration was studied in a recent work \cite{integration, ton_integration}, it trades scale for inter-SNOW interference. Its bandwidth allocation for inter-SNOW communication does not take into account the link traffics. As a result, it performs poorly in terms of network latency, eventually affecting the scale. We propose the first scalable in-band integration of  multiple SNOWs that minimizes network latency. 


\subsection{Formulation}
\label{subsec:prob_formulation}
We want to minimize the maximum latency among all SNOWs for data delivery to root base station BS$_0$. 
Considering a uniform subcarrier bandwidth and spacing across all SNOWs, let $Z_i$ be the set of orthogonal subcarriers available at BS$_i$. Let $V_i$ be the set of nodes in SNOW$_i$ and 
$n_i= |V_i|$. Each node $u \in V_i$ generates packets periodically with a period $T_{i}(u)$.
Time needed for one packet transmission  is considered a time unit. 
A subset $S_i \subset Z_i$ will be assigned for SNOW$_i$ for its nodes $V_i$ to deliver their data to BS$_i$. Let $\rho(i) \in \{0, 1, \cdots, N-1\}$ be such that BS$_{\rho(i)}$ is the parent of BS$_i$ in the tree. The subcarriers allocated for communication along tree link BS$_i$ $\rightarrow$ BS$_j$, where BS$_j$ is the parent of BS$_i$, is denoted by set $S_{i,j}$. 
Thanks to the capability of D-OFDM to enable parallel data streams to/from numerous distributed nodes from/to the BS, the proposed integration becomes  {\slshape seamless} meaning that 
  inter-SNOW and intra-SNOW communication can happen in parallel. To achieve this, subcarriers allocated on the tree link do not overlap with those for intra-SNOW communication, i.e., $S_{i,j} \subset Z_i$ and $S_{i,j} \subset Z_j$, and $S_{i,j} \cap S_{i} = S_{i,j} \cap S_{j} = \emptyset$.

Let $A(i), D(i) \subseteq  \{0, 1, \cdots, N-1\}$ be such that each $B_k, k\in A(i)$, is an ancestor and each $B_j, j\in D(i)$, is a descendant of BS$_i$ in the SNOW-tree. 
For a packet of $u$ of SNOW$_i$, let $\lambda(u, u, \text{BS}_i)$  be its estimated intra-SNOW latency under the MAC protocol used in SNOW$_i$ (i.e., the latency to collect the packet at  $\text{BS}_i$),  and  $\lambda(u, \text{BS}_j, \text{BS}_{\rho(j)})$ be its estimated  latency along tree link BS$_j$ $\rightarrow$ BS$_{\rho(j)}$. Then, its total estimated latency  $\Lambda(u)$ is given by Equation (\ref{eq:total_lat}).
\begin{equation}
\label{eq:total_lat}
\Lambda(u) = \lambda(u, u, \text{BS}_i)+ \sum_{j \in A(i) - \{0\}}  \lambda(u, \text{BS}_j, \text{BS}_{\rho(j)})
\end{equation}
Let $L_i$ denote the maximum $\Lambda(u)$ in SNOW$_i$. That is, 
$$L(i) = \max \{ \Lambda(u) | \forall u \in V_i \}.$$
Note that $L(i)$ is an estimate of maximum latency for collecting a packet at root base station BS$_0$ from SNOW$_i$.  For fairness, it is important to minimize   $\max  \{ L(i) |  0\le i  < N \} $ 
so that data collection from a SNOW does not take overly long time in an integrated network. It also implies minimizing data collection time (e.g., in convergecast scenario) in the integrated network. 
Therefore, we want to minimize this metric.

\begin{figure}[ht]
\centering 
%
\includegraphics[width=0.42\textwidth]{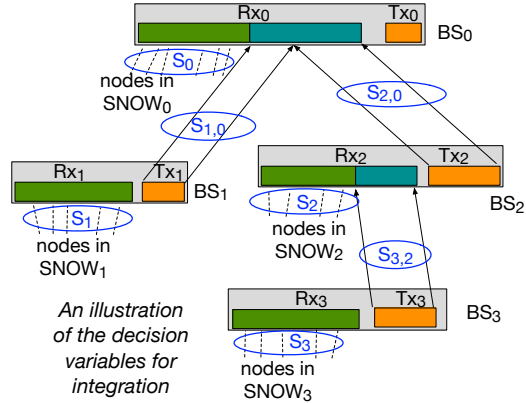} 
\caption{Integration using BS dual-radios.}
\label{fig:integration}
\end{figure}

If some communication in SNOW$_i$ is interfered by another in SNOW$_j$, then SNOW$_j$ is its {\bf \slshape interferer.} Let $I_i\subset \{0, 1, \cdots, N-1\}$ be such that each SNOW$_j, j\in I_i$,  is an interferer of SNOW$_i$. Let $J_i\subset \{0, 1, \cdots, N-1\}$ be such that a transmission along link BS$_i \rightarrow$BS$_{\rho(i)}$ can be interfered by that along BS$_k \rightarrow$BS$_{\rho(k)}$ or by some node's transmission in SNOW$_k, k\in J_i$.  In the following, Constraints (\ref{const:nointerfere1}) and (\ref{const:nointerfere2}) ensure that tree links are not interfered by any other transmission. We set constraint $\sum_{j\in I_i} | S_i \cap  S_j |  \le   \sigma_i |S_i|$ to define any allowed spectrum overlap between two interfering SNOWs where $0\le \sigma_i \le 1$.  Let $\beta_i$ be the maximum number of subcarriers on which BS$_i$ can transmit simultaneously having tolerable PAPR. The value of $\beta_i$ depends on the saturation point of the BS's transmission amplifier. Considering the BS acknowledges multiple transmissions on a single channel, $\beta_i - 1$ limits the number of concurrent transmissions by a BS$_i$ on the tree link. Thus, we use  constraint  $1 \le  |S_{i, \rho(i)}| < \beta_i $.
Our objective is to determine $S_i$ for $0\le i < N$  and  $S_{i, \rho(i)}$ for $0< i < N$  so as to 
\begin{align} 
\centering
\text {Minimize~~~}  & \max  \{ L(i) |  0\le i  < N \}   \nonumber \\
 \text{subject to~~~}& S_i \subseteq Z_i, ~~  S_{i, \rho(i)}  \subseteq Z_i, ~~  S_{i, \rho(i)}  \subseteq Z_{\rho(i)}     \label{const:availability}\\
                     & S_{i, \rho(i)}  \cap  S_j   = \emptyset, \forall j \in I_{\rho(i)}   \cup \{ i, \rho(i) \}  \label{const:nointerfere1} \\       
                     & S_{i, \rho(i)}  \cap  S_{j, \rho(j)}   = \emptyset, \forall j \in J_i  \cup   \{ \rho(i)\}  \label{const:nointerfere2}        \\
		    &\sum_{j\in I_i} | S_i \cap  S_j |  \le   \sigma_i |S_i|; ~~ 1 \le   |S_{i, \rho(i)}|     < \beta_i \label{const:papr}      
\end{align}

Fig. \ref{fig:integration} shows an illustration of decision variables for the tree of Fig.~\ref{fig:integration_mesh}.  
Note that the above formulation is applicable for any MAC protocol used inside each SNOW. However, network latency expressions usually take after multi-processor response-time analysis and become non-linear and often non-differentiable and so does the above optimization problem \cite{TC, RTAS2012}.  Adopting a subgradient approach or some global optimization framework such as  
Simulated Annealing based penalty method \cite{Chen2010} for an optimal solution can be extremely time consuming specially due to spectrum dynamics of white spaces which would require to frequently re-run the optimization.  Instead, we propose some highly intuitive greedy approach that provides a fast heuristic solution through a traffic-aware and latency-aware spectrum allocation along the links to achieve an integrated network. In the following section, we first derive latency expressions considering several MAC protocols and then describe the spectrum allocation algorithm.

%% file: snow_latency/proposed_approach.tex

\section{Description of the Low-Latency Integration of Multiple SNOWs}
\newcommand{\stage}[1]{stage{#1}} 
\label{sec:proposed_method}

In this section, we first derive the expressions for latency and then describe our proposed greedy heuristic  to perform a traffic-aware and latency-aware spectrum allocation along the links to achieve an integrated network.   he intuition behind the greedy heuristic is to iteratively assign subcarriers to links that cause a maximum delay on the packet.

The greedy heuristic relies on a good estimate of the maximum latency for a MAC. In this section, we first derive an expression for the maximum latency for two MAC protocols. Note that  the existing SNOW integration \cite{integration, ton_integration} uses a CSMA/CA based  MAC protocol, where a node transmits packets when clear channel assessment (CCA) determines the channel is idle. Otherwise, nodes use a random backoff in a fixed range to try at a later instant. However, CCA is less effective in avoiding collisions due to the following reasons, (1) the unprecedented number of hidden terminals in SNOW, (2) the increased transmission time of packets due to the use of a narrowband subcarriers for communication which increases the chances of collisions, and (3) the use of harmonic periods for simplicity warrants many time-triggered nodes in the network may generate packets around the same time (with minor differences arising due to clock offsets). Thus, a CSMA/CA based MAC protocol is prone to an unprecedented number of collisions, and predicting an upper bound on the number of collisions is less practical. Furthermore, it is less suitable for low-latency communication in SNOW integration. Besides, in many cases, a latency expression may not be derived as the latency is non-deterministic under a CSMA/CA based MAC protocol. 
Hence we consider two non-CSMA/CA MAC protocols for each SNOW and derive their latency expressions as described in the following section.

\subsection{Deriving an Expression of Latency}
 Here, we derive an expression of latency for a generic scenario where nodes in the SNOW integration can have different periods. Since deriving a latency expression for CSMA is less plausible, we first develop a time division multiple access (TDMA) version of a receiver-initiated (RI) MAC called RI-TDMA MAC. We then derive the latency expression RI-TDMA MAC.

\subsubsection{RI-TDMA MAC for SNOW} In this section, we propose a RI-TDMA MAC, which addresses the energy and latency limitations for the SNOW integration. RI-TDMA is a TDMA adaptation of the prominent receiver-initiated MAC protocols. In a receiver-initiated MAC, the receiver (or base station in SNOW) requests for a message from the nodes. Upon receiving the request, Nodes transmit the data to the base station. Since nodes specifically request for data transmissions, receiver-initiated MAC mitigates packet collisions and thereby decreases latency. However, it is less pertinent in sharing subcarriers between substations, handling event-triggered nodes, achieving a predictable latency, and conserving energy at the nodes. To achieve a predictable latency, we propose RI-TDMA, a TDMA adaptation of the receiver-initiated MAC. Under RI-TDMA, the SNOW base station operates in time slots. Each time slot contains two \stage{s}: a request \stage, and a data transmission \stage.  During the request \stage, the base station schedules a subsets of nodes for transmission on the available subcarriers by transmitting a request message. During the data transmission \stage, nodes that receive a request, respond with data transmission.

An example of data transmission in RI-TDMA MAC between a base station and four nodes is shown in Fig \ref{integration_fig:ri-tdma-mac-example}. Node $u, v, w, x$ have packets to transmit at time $0$, and the base station has two sub-carriers on which it can receive simultaneously. During the request \stage of time slot 1, the base station requests data from $u$ and $v$. During the data transmission \stage{}  of time slot 1, $u$, and $v$ respond with the data. Similarly, during time slot 2, the base station requests data from $w$ and $x$, and within the same time slot, $w$ and $x$ respond with the data. Since the nodes are not time synchronized, the nodes may start listening or transmitting at a different time, as shown in Fig. \ref{integration_fig:ri-tdma-schedule-example}.

\begin{figure}[ht!]
 	\centering
	\begin{subfigure}[b]{0.2\textwidth}
		\includegraphics[width=\textwidth]{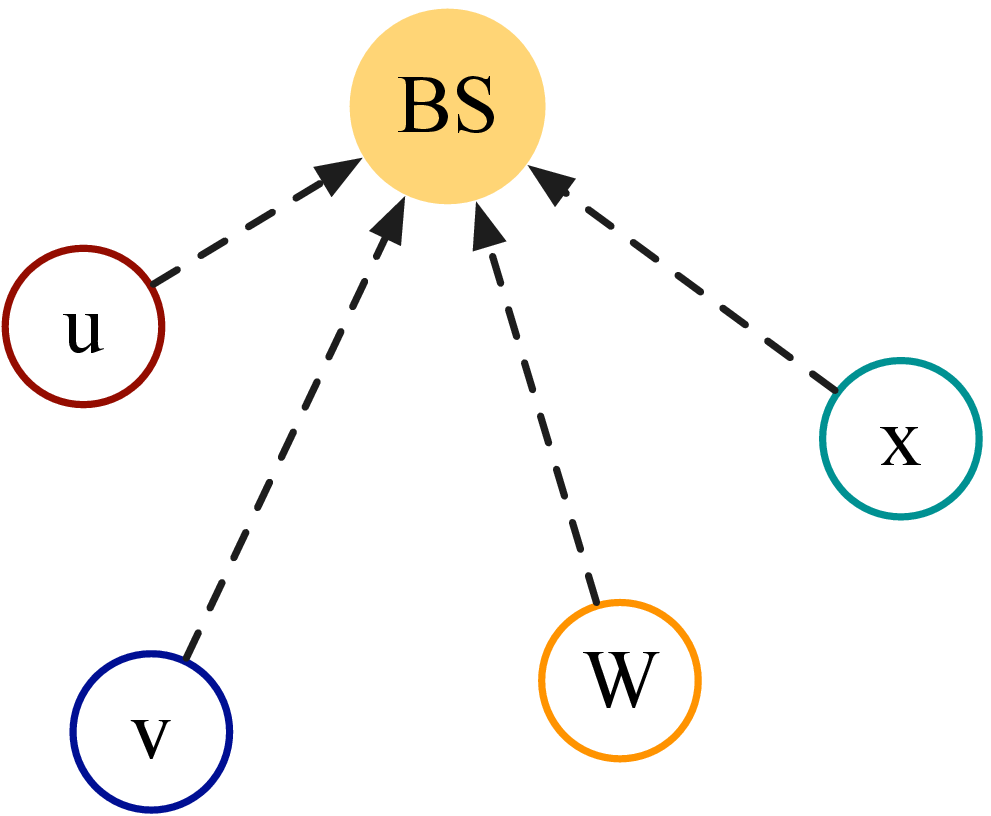}
		\caption{Example of a Single SNOW Network.}
		\label{integration_fig:ri-tdma-network-example}
	\end{subfigure}
	\quad
	\begin{subfigure}[b]{0.45\textwidth}
		\includegraphics[width=\textwidth]{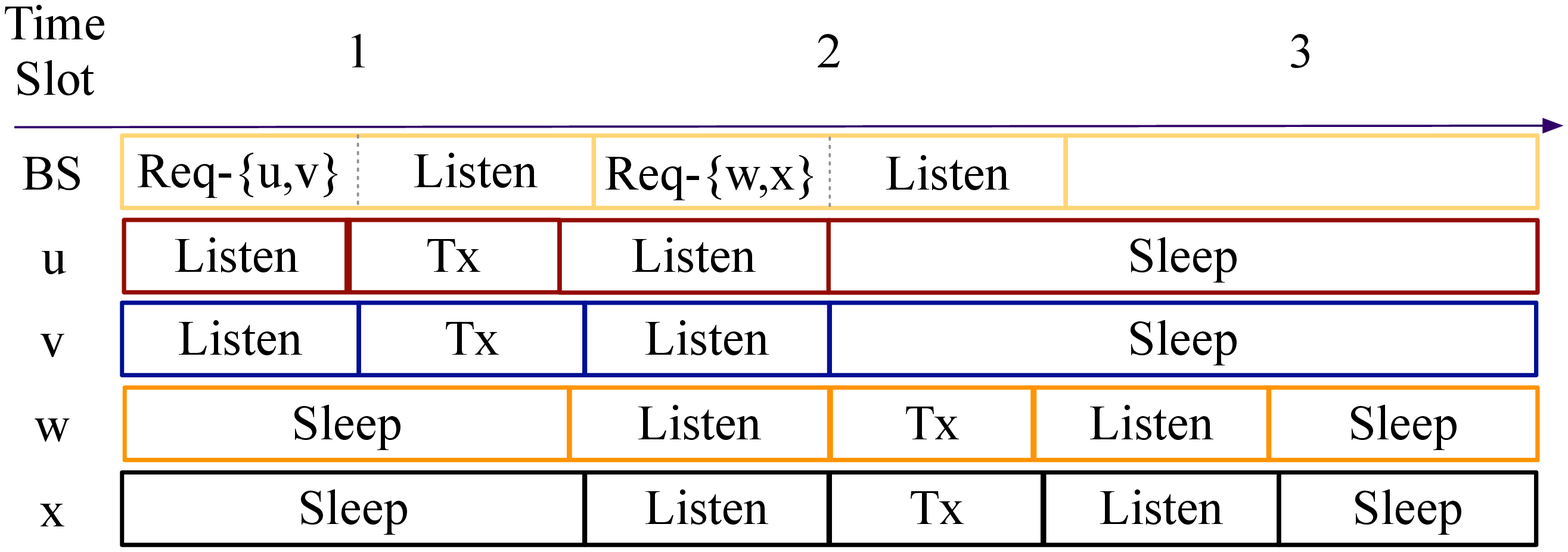}
		\caption{Transmission Schedule for the Example Network.}
		\label{integration_fig:ri-tdma-schedule-example}%
	\end{subfigure}
	\caption{An example of packet scheduling in RI-TDMA MAC.}
	\label{integration_fig:ri-tdma-mac-example}
\end{figure}

To schedule nodes for transmission in a time slot, a SNOW base station uses existing real-time multi-core CPU scheduling algorithms, such as RM (rate monotonic) scheduler. In an RM scheduler, the base station prioritizes packets with a smaller period for transmission. Base station breaks ties in priority using node ID, i.e., for two nodes with the same period, the node with the smaller node ID is prioritized for transmission. Note that the schedule is generated dynamically at the beginning of the time slot.

Upon generating the schedule for a time slot, the base station requests information from many nodes on a single subcarrier called the downlink subcarrier. During the request \stage, nodes listen to the downlink subcarrier for the request. Furthermore, the base station can use a single request packet to request packets from multiple nodes. The base station uses a bit vector of size equal to the number of nodes to identify the nodes transmitting in a time slot and another bit vector of size equal to the number of subcarriers to represent the available subcarriers, as shown in Fig. \ref{fig:payload_example}. In Fig. \ref{fig:payload_example}, $|V($BS$_i)|$ represents the number of nodes in SNOW$_i$ and $|S_i|$ represents the number of subcarriers assigned to SNOW$_i$. If the base station intends to request data from node with an ID $j$, it will set the $j^th$ bit of the bit vector to 1. One receiving the bit vector, node with ID $j$ checks the $j^th$ bit for transmission. If the bit is $1$, it calculates $\alpha$ the sum of values at indices in the range [0, i-1]. It then chooses the $\alpha^{th}$ subcarrier from the available subcarriers bit vector. For example, if base station sets bits  $0 - 4$ of the bit vector to 1 and $0 - 3$bits and $5^{th}$bit of the subcarrier bit vector are set to 1, then node $4$ sends the packet on the subcarrier with ID $5$. Note that the mapping between the subcarrier ID and the actual frequency is communicated during network deployment, similar to SNOW MAC in \cite{ton_integration}.

\begin{figure}[ht]
\centering
\includegraphics[width=0.3\textwidth]{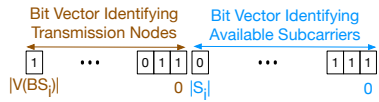}
\caption{Payload from BS$_i$'s Request Message.}
\label{fig:payload_example}
\end{figure}

Although RI-TDMA MAC minimizes packet collisions, the base station may not receive some packets due to bad link qualities. To increase the reliability in communication, we use request message of the next slot as an ACK.  If the base station requests a packet from a node in the next slot, then the node assumes that the base station did not receive the packet and re-transmits it.  Otherwise, it assumes the packet transmission was successful, as shown in Fig. \ref{integration_fig:ri-tdma-schedule-example}. Since the receiver initiation message also serves as a ACK and uses one subcarrier for transmission, it does not pose any additional impact on the PAPR of the BS.

To conserve energy, nodes in RI-TDMA MAC wake up for packet generation/transmission and sleep for the rest. Since RM generates a compact schedule consisting of a transmission time slot for the first packet, nodes can use the first transmission time slot and period to predict the next transmission time slot. Nodes in RI-TDMA MAC are not synchronized with the base station, and hence, the nodes can suffer from variable clock drifts that can interfere with the communication. Precisely, nodes may wake up early, thereby draining energy, or late, thereby wasting time slots and energy. To compensate for the clock drifts, each node locally computes the clock drift based on the time it receives a request message from the base station. The goal of the clock drift estimation is only to maximize the nodes' sleeping duration but not for clock synchronization. Note that the use of receiver-initiation eliminates any need for time synchronization in RI-TDMA.

\vspace{0.1in}
\subsubsection{Scheduling Event-Triggered Nodes}
Typically, nodes periodically sense the state and transmits the sensed information to its BS. However, some nodes in a SNOW may generate packets on the occurrence of a specific event. We refer to these nodes as {\slshape event-triggered nodes}. When an event is generated, nodes transmit the information to the BS. In this chapter, we consider that most nodes use time-triggered sensing, and some nodes use event-triggered sensing. Typically, a base station cannot predict such occurrences, and hence, cannot schedule an event-triggered node for a packet transmission. 

To address this challenge, we leverage on the concurrent transmission and reception capabilities of the base station. Specifically, during the request \stage{} of a time slot, the base station is not listening for any packet transmission, and all subcarriers are unused except the downlink subcarriers. Event-triggered nodes transmit packets to the base station on any unused subcarriers during the request \stage{}. 

Since nodes are not time-synchronized with the base station, event-triggered nodes can not predict the start of a request \stage{} of a time slot. To overcome this challenge, event-triggered node first listens to one packet transmission request from the base station and uses the length of the time slot to predict the start of the next request \stage{}. Thus, an event-triggered node can communicate with the base station without interfering with the network's time-triggered transmission. Furthermore, the base station can acknowledge the packet transmission during the data transmission segment while time-triggered nodes transmit packets to it. 

Typically, a single SNOW base station may communicate with multiple event-triggered nodes. However, the probability of all event-triggered nodes transmitting a packet at the same time is very low. If they transmit at the same time and collide, the nodes use random back-off for a re-transmission.

\vspace{0.1in}
\subsubsection{Scheduling Communication between Two Base Stations}
To facilitate a seamless integration, base station forwards packet to its parent on dedicated subcarriers, that are not interfered by any other communication in the network. Since a base station is equipped with two radios, it can transmit and receive packets simultaneously. In RI-TDMA MAC, a base station is using at most one subcarrier on both request \stage{} and data transmission \stage{}. Thus, it can transmit packets on the tree link subcarriers during both \stage{s} of the time slot. Note that, do not use receiver-initiated communication on the tree links. Instead, the base station forwards the packets on the next available subcarrier and time slot. 

Similar to intra-SNOW communication, packets in inter-SNOW communication are prioritized based on the RM scheduling policy and scheduled for transmission locally by a base station. A base station locally determines the next available packets for transmission based on the packets' priority in its queue. Note that, a packet generated by an event-triggered node has the lowest priority compared to a packet generated by a time-triggered node.

Since base stations operate in time slots, all base stations in the topology need to time-synchronized. Base stations are line-powered devices, and hence, synchronizing time at all base stations does not cause additional energy overhead.  Note that, only the base stations are time-synchronized while the nodes are not. 

Typically, a SNOW base station shares a subset of subcarriers with its neighboring SNOW. An effective sharing of subcarriers between SNOWs increases the availability of subcarriers and potentially decreases the maximum latency. To enable effective sharing, we propose to use a round-robin approach. In the round-robin approach, each node is allocated one time slot for a subcarrier usage in an epoch of length equal to the number of nodes sharing a subcarrier.  For example, if two base stations (BS$_0$ and BS$_1$)  are sharing one subcarrier, then BS$_0$ will use the subcarrier for time slots ${1, 3, 5, 7, \cdots}$ and BS$_1$ will use the subcarrier for time slots ${2, 4, 6, 8, \cdots}$. For the sake of simplicity and correctness of the RM transmission time slot prediction, we use the closet upper bound on the epoch length that is harmonic with the packet generation periods.

\vspace{0.1in}
\subsubsection{Estimating the Latency in RI-TDMA}
Typically, estimating an upper bound on the maximum latency of a packet in a wireless environment is less plausible since poor link conditions may cause many consecutive packet drops at a receiver. Existing works on delay analysis in a wireless network assume a fixed number of re-transmissions \cite{distributedhart}. In this section, we assume that a node/BS transmits each packet at most once. 

The first step in estimating the maximum latency of packet from $u$ of SNOW$_i$ is to determine the intra-SNOW latency from $u$ to BS$_i$, $\lambda(u, u, \text{BS}_i)$. To determine the intra-SNOW latency, we rely on the concept of critical instant, similar to processor scheduling. The critical instant, in RI-TDMA, is the arrival of a packet such that it experiences the maximum delay from all higher priority packets. Note that the high priority packets in RI-TDMA MAC are packets generated by a node $v$ with periods smaller than or equal to the period of $u$, i.e., $v \in hp_i(u) \text{ iff }  v \in V_i \text{ and } T_i(v) \le T_i(u)$, where $T_i(v)$ represents the period of node $v, v \in$ SNOW$_i$. Since BSs use RM scheduler, the maximum delay occurs when all high priority packets arrive at the same time as the packet from $u$. During a critical instant, multiple packets of a high priority node can delay one packet of $u$. The maximum number of packets generated by a high priority node with an interval of length $\lambda(u, u, \text{BS}_i)$ is given by the following expression: $$\sum_{v \in hp_i(u)} \bigg \lceil \frac{\lambda(u, u, \text{BS}_i)}{T_i(v)} \bigg\rceil.$$
Each packet of a high priority node delays the packet of node $u$ by one time slot. 

In RI-TDMA, BS$_i$ can receive from at most $|S_i|$ nodes concurrently, where $S_i$ represents the set of subcarriers assigned to BS$_i$. However, $\text{BS}_i$ may share a subset of $S_i$ subcarriers with other base stations, and sharing allows $\text{BS}_i$ to transmit packet on these shared subcarriers for a fraction of the time. We represent the sum of availabilities of all subcarriers in $S_i$ as $\psi_i$. In an interval of length $L$, BS$_i$ can receive $ \lceil \frac{L}{\psi_i} \rceil$ packets. 
Furthermore, the concept of subcarrier sharing between base stations adds an additional latency, as shown in Fig \ref{integration_fig:ri-channel-availability}. In Fig. \ref{integration_fig:ri-latency-cont-channel}, one subcarrier is available for use by four nodes. In Fig. \ref{integration_fig:latency-fractional-channel},  four subcarriers are available for one in four time slots, i.e., the repetition epoch is 4. The subcarrier availability in both cases is 1, but the latency is not. The worst case latency is observed when all the fractional channels are available at the end of their repetition epoch. Thus, fractional availability of a channel adds $\varsigma$ units of additional latency, where $\varsigma$ represents the maximum number of SNOW BS that share one subcarrier.
\begin{figure}[ht!]
    	\centering
	\begin{subfigure}[b]{0.2\textwidth}
		\includegraphics[width=\textwidth]{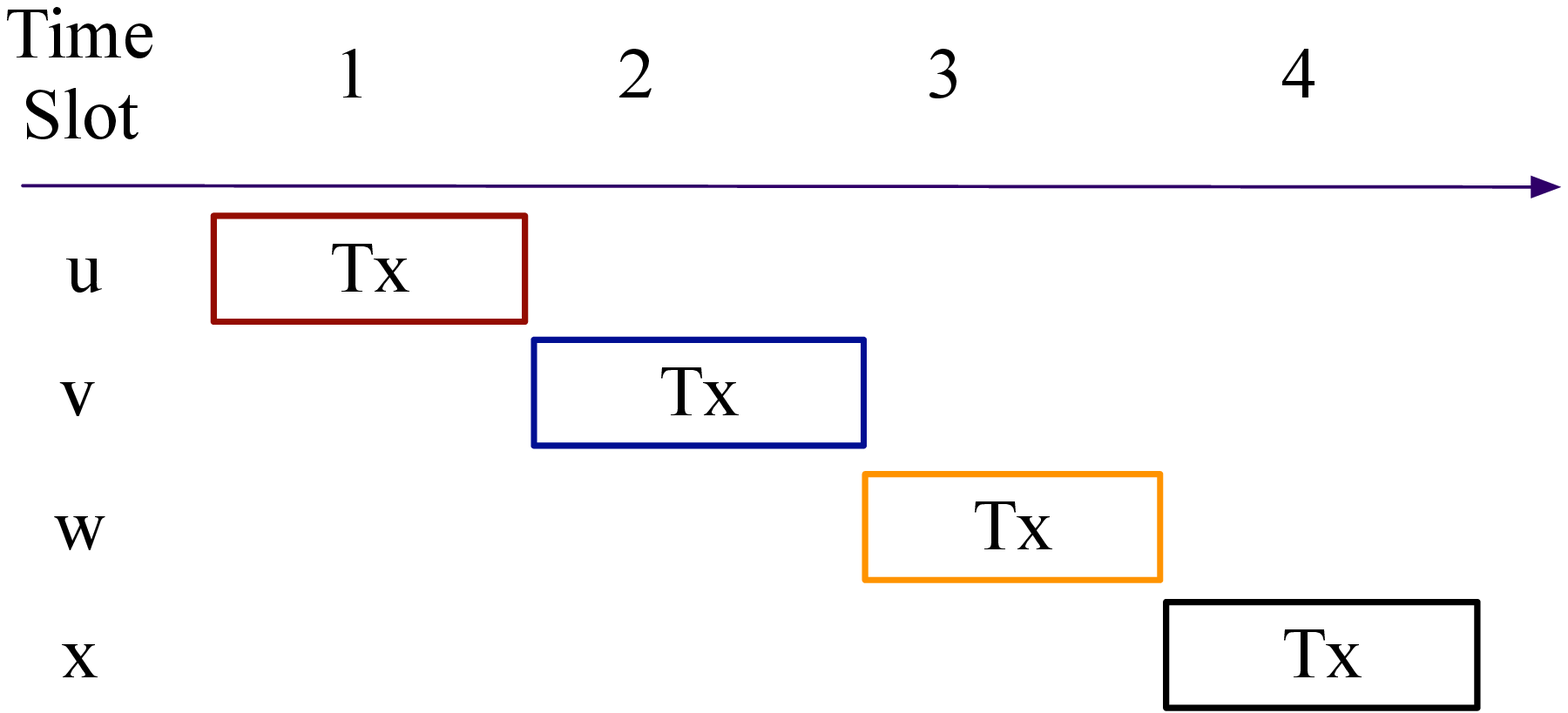}
		\caption{Continuous Availability of 1 Channel}
		\label{integration_fig:ri-latency-cont-channel}
	\end{subfigure}
	\quad
	\begin{subfigure}[b]{0.2\textwidth}
		\includegraphics[width=\textwidth]{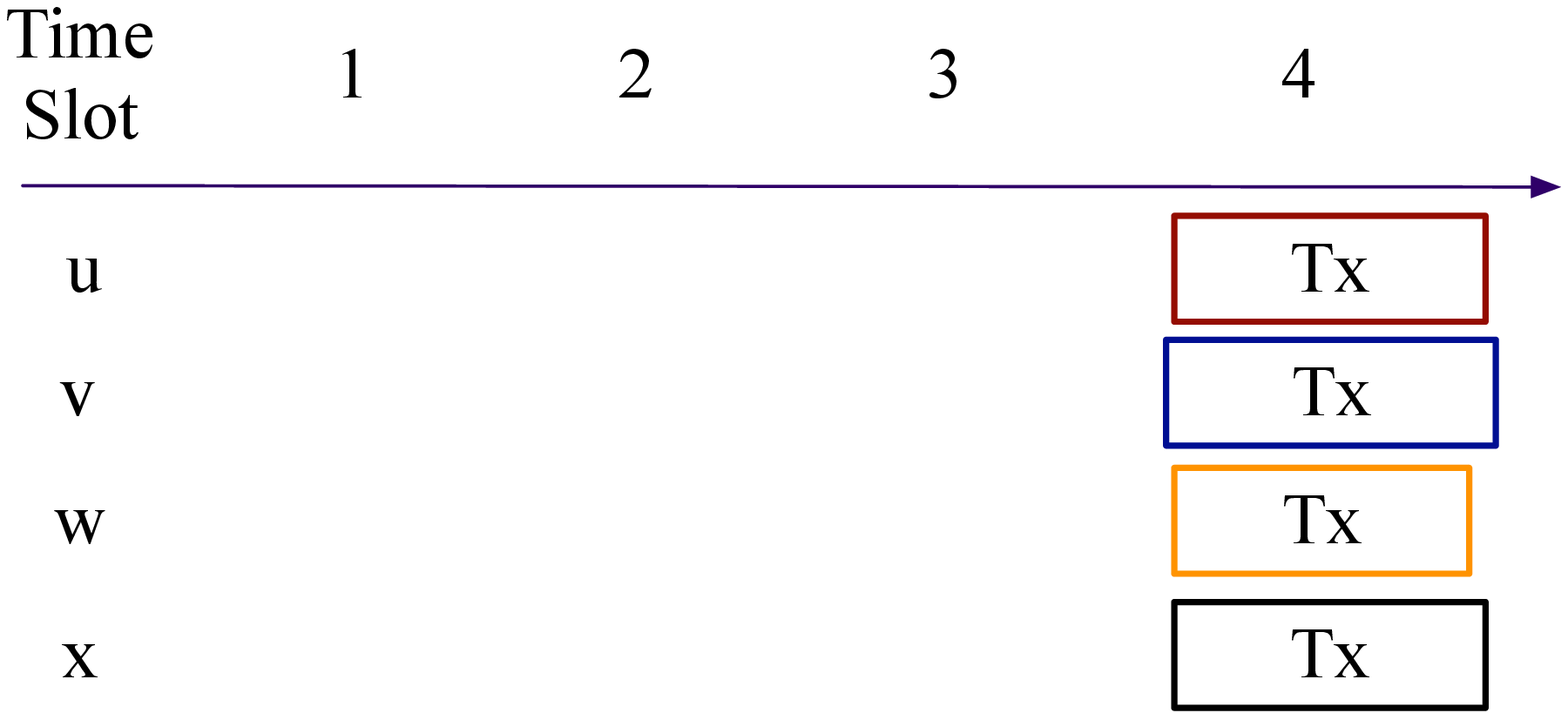}
		\caption{Fractional - $\frac{1}{4}^{th}$ - Availability of 4 Channels}
		\label{integration_fig:latency-fractional-channel}%
	\end{subfigure}
	\caption{Comparison of Latency under Continuous and Fractional Availability of a Channel }
	\label{integration_fig:ri-channel-availability}
\end{figure}

The maximum latency experienced by a packet of $u$ is the sum of the number of transmission attempts from $u$ (which is one) and the total delay caused by all high priority packets. The solution to the expression in Equation (\ref{eq:latency_node_BS}) gives the maximum latency experienced by a packet from node $u$.
\begin{equation}
\label{eq:latency_node_BS}
\lambda(u, u, \text{BS}_i) = 1 + \varsigma + \Bigg\lceil \frac{1}{\psi_i} \sum_{v \in hp_i(u)} \bigg\lceil \frac{\lambda(u, u, \text{BS}_i)}{T_i(v)} \bigg\rceil \Bigg \rceil
\end{equation} 

Note that scheduling packets for intra-SNOW communication in RI-TDMA is a special case of a task scheduling on a multi-core processor, where each packet transmission (or equivalent task) executes for precisely one time slot. Here, increasing the period of a node $u$ by one time slot, facilitates a lower priority packet from $v$ to start and complete transmission within this time slot. The transmission of a packet from $v$ causes a ripple effect on other lower priority tasks to finish earlier or at the same time as before. Thus, increasing the period of a node does not increase the latency of any lower priority packet, and scheduling anomalies discussed in \cite{andersson2000some} are not applicable here. Similarly, packet scheduling on a fractional availability of a processor, does not introduce any scheduling anomalies for packet scheduling in RI-TDMA MAC. Thus, the expression in Equation (\ref{eq:latency_node_BS}) provides a the maximum latency experienced by a packet from $u$.

The next step in estimating the latency of a packet is to compute the inter-SNOW latency. We first describe the latency experienced by a packet (originating at node $u$) to reach BS$_{\rho(i)}$ from BS$_i$, where node $u$ is in SNOW$_i$ (i.e., $u \in V_i$). 
We then extend the description for any base stations BS$_j$ and BS$_{\rho(j)}$, along the tree path of BS$_i$. 

In inter-SNOW communication, a packet experiences delays from two sources: (1) high priority packets that previously delayed this packet, and (2) new high priority packets that arrive at BS$_i$. 
The first source of delay emanates from the mismatch in the subcarrier assignment for intra-SNOW communication ($|S_i|$) and inter-SNOW communication ($|S_{i,\rho(i)}|$). Specifically, $|S_i|$ can be higher than $|S_{i,\rho(i)}|$, i.e., BS$_i$ can receive packets at a higher rate than it can transmit to its parent. Thus, high priority packets that have delayed the packet from $u$ can further delay the packet at BS$_i$.  The latency caused by high priority packets that previously delayed a packet of $u$ is given by the following expression:
$$ \lambda(u, u, \text{BS}_i)\bigg \lceil \frac{\psi_{i}}{2|S_{i, \rho(i)}|}\bigg \rceil $$

The second source of delay emanates from the arrival of new high priority packets at BS$_i$. The number of high priority packets arriving at BS$_i$ depends on the generation rate of the high priority packets and the packet reception rate of these packets at BS$_i$. For example, if BS$_k$ receives 10 high priority packets but can transmit only 6 in one time slot, implying that only 6 packets can interfere a packet from $u$ per time slot. 
Let $C(i) \subseteq \{0,1,\cdots ,N-1\}$ be such that each $B_k, k \in C(i)$, is a child of BS$_i$ in the SNOW tree. An upper bound on the number of new high priority packets arriving at BS$_i$ during the interval $\lambda(u, \text{BS}_i, \text{BS}_{\rho(i)})$ is given by the following Equation:
\begin{equation*}
\begin{split}
\bigg\lceil   \sum_{\substack{v \in hp_g(u) \, \& \, g \in D(k) \\ k \in C(i)}} & \frac{\lambda(u, \text{BS}_i, \text{BS}_{\rho(i)})}{T_g(v) \, |S_{k,i}|}  \\& +  \sum_{\substack{ w \in hp_i(u)}} \frac{\lambda(u, \text{BS}_i, \text{BS}_{\rho(i)})}{T_i(w) \, \psi_{i}} \bigg\rceil. 
\end{split}
\end{equation*}
In the above expression, the first summation takes into account all high priority packets that are generated by a successor of BS$_i$ and second summation takes into account all the high priority packets generated by nodes of SNOW$_i$.

In RI-TDMA, BS$_i$ concurrently transmits on $2|S_{i,\rho(i)}|$ subcarriers in each times slot since it transmits in both \stage{s} of the time slot. However, if the total number of subcarriers on which it receives a packet in one time slot is less than $2|S_{i,\rho(i)}|$, then it cannot transmit $2|S_{i,\rho(i)}|$ packets in a time slot. Thus, we represent the maximum number of concurrent inter-SNOW transmissions by BS$_i$ as $\phi_{i,\rho(i)}$, and is given by Equation (\ref{eq:max_concurrent_tx_BS}).
\begin{equation}
\label{eq:max_concurrent_tx_BS}
\phi_{i,\rho(i)} = \min (2|S_{i,\rho(i)}| \,, \, \psi_i + \sum_{k \in C(i)} |S_{k,i}| )
\end{equation}
The maximum latency experienced by a packet of $u$ to reach BS$_{\rho(i)}$ from BS$_i$ is given by the sum of two sources, and the solution to the expression in Equation (\ref{eq:latency_BSparent_BS}) gives the maximum latency.
\begin{equation}
\label{eq:latency_BSparent_BS}
\begin{split}
 \lambda  (u, \text{BS}_i, &\text{BS}_{\rho(i)}) =\lambda(u, u, \text{BS}_i)\bigg \lceil \frac{\psi_{i}}{2|S_{i,{\rho(i)}}|}\bigg \rceil  + \\& \Bigg\lceil \frac{1}{\phi_{i,\rho(i)}}  \bigg\lceil   \sum_{\substack{v \in hp_g(u) \, \& \, g \in D(k) \\ k \in C(i)}} \frac{\lambda(u, \text{BS}_i, \text{BS}_{\rho(i)})}{T_g(v) \, |S_{k,i}|}  \\& +  \sum_{\substack{ w \in hp_i(u)}} \frac{\lambda(u, \text{BS}_i, \text{BS}_{\rho(i)})}{T_i(w) \, \psi_{i}} \bigg\rceil \Bigg\rceil
\end{split}
\end{equation}

The maximum latency experienced by at any BS$_j$ (that receives the packet of $u$ from child BS$_i$) in the tree structure is given by the following expression.
\begin{equation*}
\begin{split}
 \lambda  (u, \text{BS}_j, &\text{BS}_{\rho(j)}) =\lambda(u, u, \text{BS}_j)\bigg \lceil \frac{\psi_{j}}{2|S_{j,{\rho(j)}}|}\bigg \rceil  + \\& \Bigg\lceil \frac{1}{\phi_{j,\rho(j)}}  \bigg\lceil   \sum_{\substack{v \in hp_g(u) \, \& \, g \in D(k) \\ k \in C(j)}} \frac{\lambda(u, \text{BS}_j, \text{BS}_{\rho(j)})}{T_g(v) \, |S_{k,j}|}  \\& +  \sum_{\substack{ w \in hp_j(u)}} \frac{\lambda(u, \text{BS}_j, \text{BS}_{\rho(j)})}{T_j(w) \, \psi_{j}} \bigg\rceil \Bigg\rceil
 \end{split}
\end{equation*}
Note that the above expression is a generalization of the expression in Equation (\ref{eq:latency_BSparent_BS}), where the transmission from node $u$ to reach BS$_i$ was replaced with with the transmission from BS$_i$ to BS$_j$.

The maximum latency experienced by a packet generated by node $u$ to reach the root BS is the summation of latency experienced to reach $u$'s base station BS$_i$ and all links from on the tree from BS$_i$ to root BS. It is given by the following Equation.
\begin{equation*}
\Lambda(u) = \lambda(u, u, \text{BS}_i)+ \sum_{j \in A(i) - 0}  \lambda(u, \text{BS}_j, \text{BS}_{\rho(j)})
\end{equation*}

The above latency formulation can be used to generate a solution for the bandwidth allocation problem formulation in Section \ref{subsec:prob_formulation}. 
Since the latency computation takes a pseudo-polynomial time to generate an accurate result, generating an optimal solution using any existing optimazation solvers like genetic algorithm takes a significantly long time. Hence, we propose a heuristic solution to generate an efficient solution.
Note that the latency derivation for any MAC can be used with the bandwidth allocation problem formulated in Section \ref{subsec:prob_formulation}. To show the feasibility of using different MAC protocol, we derive the latency formulation for a simple MAC in the following subsection.

\subsection{Deriving a Latency Expression for a special case}
Here, we consider a simple TDMA MAC protocol, where time in the network is divided into time slots. Each node in the SNOW is assigned a time slot and a subcarrier for transmission. At the end of the time slot, the BS acknowledges all packet transmissions via a single ACK message. The ACK message consists of a bit vector of size equal to the number of available subcarriers. In the next time slot, the BS broadcasts the received messages to its parent BS in the tree, without interrupting the transmissions from other nodes.  In TDMA scheduling, the transmission schedule is generated well in advance and broadcasted to all nodes during network deployment. 

Considering that all the nodes have the same periods and BSs do not share subcarriers, a packet from node $u$ experiences the maximum delay when packets from all nodes in SNOW$_i$ interfere it. The maximum latency experienced by a packet from node $u$ to reach BS$_i$ is given as follows.
$$\left\lceil \frac{n_i}{|S_i|} \right\rceil$$
A packet generated by $u$ experiences the maximum delay along the tree link  BS$_j$ $\rightarrow$ BS$_{\rho(j)}$ when packets from all nodes that pass through the BS$_j$ interfere it.  The maximum latency experienced by a packet generated by $u$ along the tree link  BS$_j$ $\rightarrow$ BS$_{\rho(j)}$ is given by the following equation:
 $$\left \lceil \frac{ \sum_{k\in D(j)} n_k}{|S_{j, \rho(j)}|} \right\rceil.$$

The maximum latency packet origination from SNOW$_i$ is given by Equation (\ref{eq:latency_tdma}).
\begin{equation}
\label{eq:latency_tdma}
 L(i) =  
\begin{cases}
~\left \lceil \frac{n_i}{|S_i|} \right\rceil, & \text{ if } i =0; \\
\left \lceil \frac{n_i}{|S_i|} \right\rceil  +  \sum_{j \in A(i) - \{0\}}  \left \lceil \frac{ \sum_{k\in D(j)} n_k}{|S_{j, \rho(j)}|} \right\rceil, &\text{ otherwise.}
\end{cases}
\end{equation}

The above latency formulation is a special case of the maximum latency derived in the previous section, without an efficient method of subcarrier sharing and assuming all nodes have the same periods. Although these assumptions restrict the applicability, it provides a tighter bound on the maximum latency computation and provides a better solution to the subcarrier allocation problem.

\newcommand{\gbasi}[1]{LT-SASI{#1}}
\subsection{Latency and Traffic aware Spectrum Allocation for SNOW Integration (\gbasi{})}
Here, we present heuristic solution to the subcarrier optimization problem for latency minimizing integration. The intuition behind the greedy heuristic is to assign additional subcarriers to a bottleneck link. Here, we refer to any link that causes the largest delay on the packet as the bottle neck link. Since allocating additional subcarriers to a link decreases the latency experienced by all packets along that link, iteratively repeating the intuition can result in a subcarrier assignment that minimizes the maximum latency in the network. We refer to this greedy heuristic as \gbasi{} and is presented in Algorithm \ref{alg:integration_gbasi}.

The greedy heuristic starts by assigning two unique subcarriers (one for intra-SNOW and another for inter-SNOW) to each SNOW base station, as shown in Lines 1-4. The function \texttt{Unique\_Subcarrier} first computes the list of subcarriers not used by any interfering SNOW and returns one subcarrier from the unused subcarrier list. Then, the code maximum latency of all packets in the network. Computing the latency of a packet via a fixed point algorithm is a pseudo-polynomial time algorithm. \gbasi{} uses a polynomial-time simplification by replacing the critical window length by the period of the packet. Aspirant links are then calculated by eliminating links that do not have any feasible subcarrier assignment left that minimize the latency, as shown in Line 7. Among the aspirant links, the bottleneck link is identified and assigned a subcarrier, as shown in Lines 11-16. The function \texttt{Add\_Feasible\_Channel} finds a an available subcarrier that minimizes latency and meets the Constraints (\ref{const:availability}), (\ref{const:nointerfere1}), (\ref{const:nointerfere2}), and (\ref{const:papr}).

\begin{algorithm}[h]
\small
\caption{\gbasi{}} 
\label{alg:integration_gbasi}
  \SetAlgoLined 
  \SetKwInOut{Input}{Input}
  \SetKwInOut{Output}{Output}
  \SetKwRepeat{Do}{do}{while}
  \DontPrintSemicolon
  \Input{$\text{Tree Structure}, Z_i, I_i, J_i, \, \forall i \in \{0, N\} $ \\ $ T_u \forall \, u \in \text{BS}_i, \, 0 \le i < N$}
  \Output{$S = \cup_{ \forall i} \{S_i, S_{i,\rho(i)} \} $}
  \vspace{0.03in}
  
  \For{$i \in \{0, N\}$} { 
  \vspace{0.03in}
   $S_{i} = \text{Unique\_Subcarrier} (Z_i, I_i, S)$ \\
   \vspace{0.03in}
   $S_{i, \rho(i)} = \text{Unique\_Subcarrier} (Z_i, Z_{\rho(i)}, J_i, S)$ \\
   }
   
   \Do{$\text{True}$}{
	\vspace{0.03in}
	$L = \text{Compute\_Latency} () $  \\
	$A = \text{Exclude\_infeasible\_links} (L, S) $ \\
	\If{$|A| = 0$} { 
		break
	}
	$\upsilon, \nu = \text{Identify\_Bottleneck\_Link} (A)  $\\
	\vspace{0.03in}
	\eIf{$\upsilon  \in \text{BS}$} { 
		 $S_{\upsilon, \nu} = \text{Add\_Feasible\_Channel} (S, Z)$
	}
	{
		$S_{\nu} = \text{Add\_Feasible\_Channel} (S, Z)$ 
	}
  }
  
\end{algorithm}

The time complexity of the \texttt{Unique\_Subcarrier} function is $O(|Z|)$, where $|Z|$ represents the maximum number of subcarriers available in a SNOW. Iterating over $N$ SNOW base stations, the time complexity for the initial subcarrier assignment is $O(N|Z|)$. The time complexity of computing the maximum latency for all nodes is $O(N |V|)$, where $|V|$ represents the maximum number of nodes in a SNOW. The time complexity of identifying the bottleneck link is $O(N)$ since the maximum length of the tree structure is $N$. The time complexity of \texttt{Exclude\_infeasible\_links} function is O(1), considering that \texttt{Add\_Feasible\_Channel} function stores the list of links with infeasible assignments in the. The time complexity of a brute force algorithm to find an available subcarrier that minimizes the latency of the bottleneck link is $O(|Z|)$. Iterating over $N |Z|$ times to exhaustively search through all aspirant links, the time complexity of \gbasi{} algorithm is $O(N|Z| + N|Z| (N + N|V|)) = O(N^2 |V| |Z|)$.

%% file: snow_latency/evaluation.tex
\section{Experiment}
\label{sec:experiment}
Here we evaluate the latency-minimizing spectrum allocation capabilities of \gbasi{} through real experiments. We compare the performance of the proposed \gbasi{} with the greedy heuristic for integration proposed in~\cite{integration, ton_integration}, which we refer to in this section as Greedy-SOP. Note that, generating a latency bound for a CSMA MAC protocol is less feasible due to the unprecedented number of collisions in an LPWAN network \cite{oursurvey}. Thus, we use a SNOW network with TDMA MAC for \gbasi{} approach and CSMA/CA MAC for the Greedy-SOP (as proposed in \cite{integration, ton_integration}).

\begin{figure}[h]
	\centering
	\includegraphics[width=0.45\textwidth]{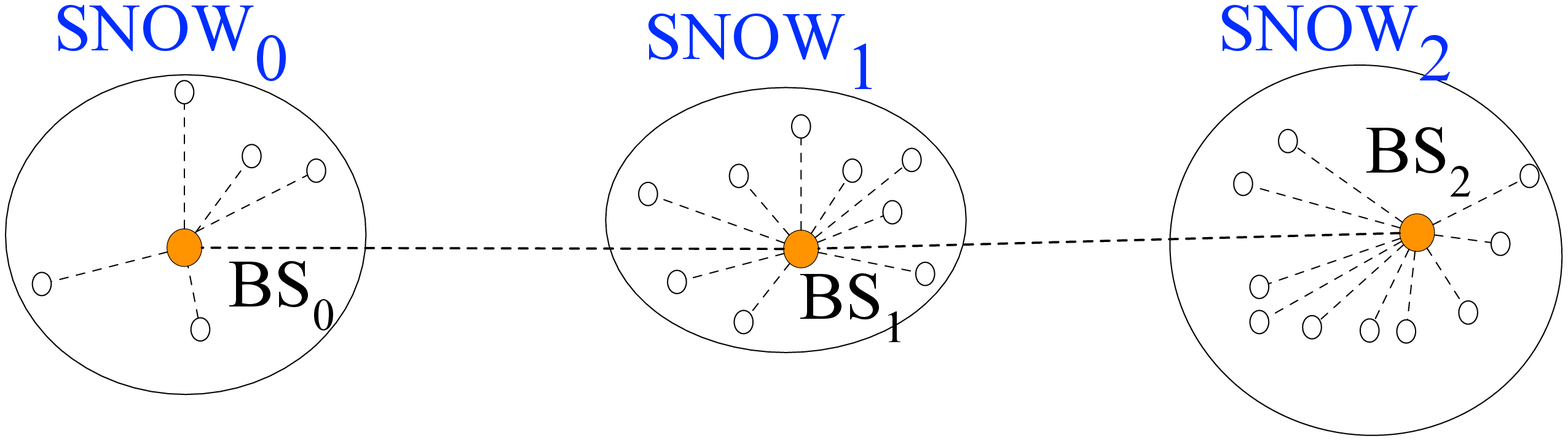}
	\caption{Experiment Setup}
	\label{fig:experiment_setup}
	\vspace{-0.1in}
\end{figure}

We create a SNOW integration with 3BSs, as shown in Fig. \ref{fig:experiment_setup}. Each base station is connected with 2 USRP devices, one for Tx radio and another for Rx radio. We use 25 CC1310 devices to emulate the communication pattern of 500 devices. Each CC1310 device is emulating the traffic of 20 nodes by hopping into the next subcarrier and transmitting a packet at the next available time. Note that the emulation of nodes shows the working of \gbasi{} under a large scale integration, as proposed in this chapter. Of the available white space spectrum, we used 2 TV channels in the range $500$MHz to $512$MHz. Both CC1310 and USRP devices use a transmission power of 0dBm. We use OOK modulation for communication between any two devices. Since CC1310 devices are built for a maximum bandwidth of $39$KHz for OOK modulation, nodes use a $39$kHz bandwidth for transmission. The base stations use oversampling at $400$KHz bandwidth to receive the packet successfully from a node. Furthermore, base stations use a $400$KHz bandwidth for inter-SNOW communication. For the experiment, we use a random payload of 10bytes and a packet generation interval of 2s. For the TDMA MAC protocol, the time in the network is slotted, and each slot is 50ms long.

We evaluate the performance of \gbasi{} using 2 metrics (1) maximum latency and (2) average latency. We consider {\slshape maximum latency} as the longest latency experienced any packet in the SNOW integration during a 10-minute interval. We consider {\slshape average latency} as the average latency of all packets generated in the SNOW integration during a 10-minute interval. 

The experiment result in Fig. \ref{ntegration_fig:results-experiment} shows that the spectrum allocation with \gbasi{} results in the maximum and average latencies for all the BSs to be similar to each other. Notably, the difference between the maximum latency of two BSs at most $0.08$s, while that of the Greedy-SOP is $0.4$s. This result shows that \gbasi{} ensures fairness in the spectrum usage of all BSs. It also demonstrates that \gbasi{} algorithm decreases the latency up to $50\%$.
 
\begin{figure}[t]
 	\centering
	\begin{subfigure}[b]{0.35\textwidth}
		\includegraphics[width=\textwidth]{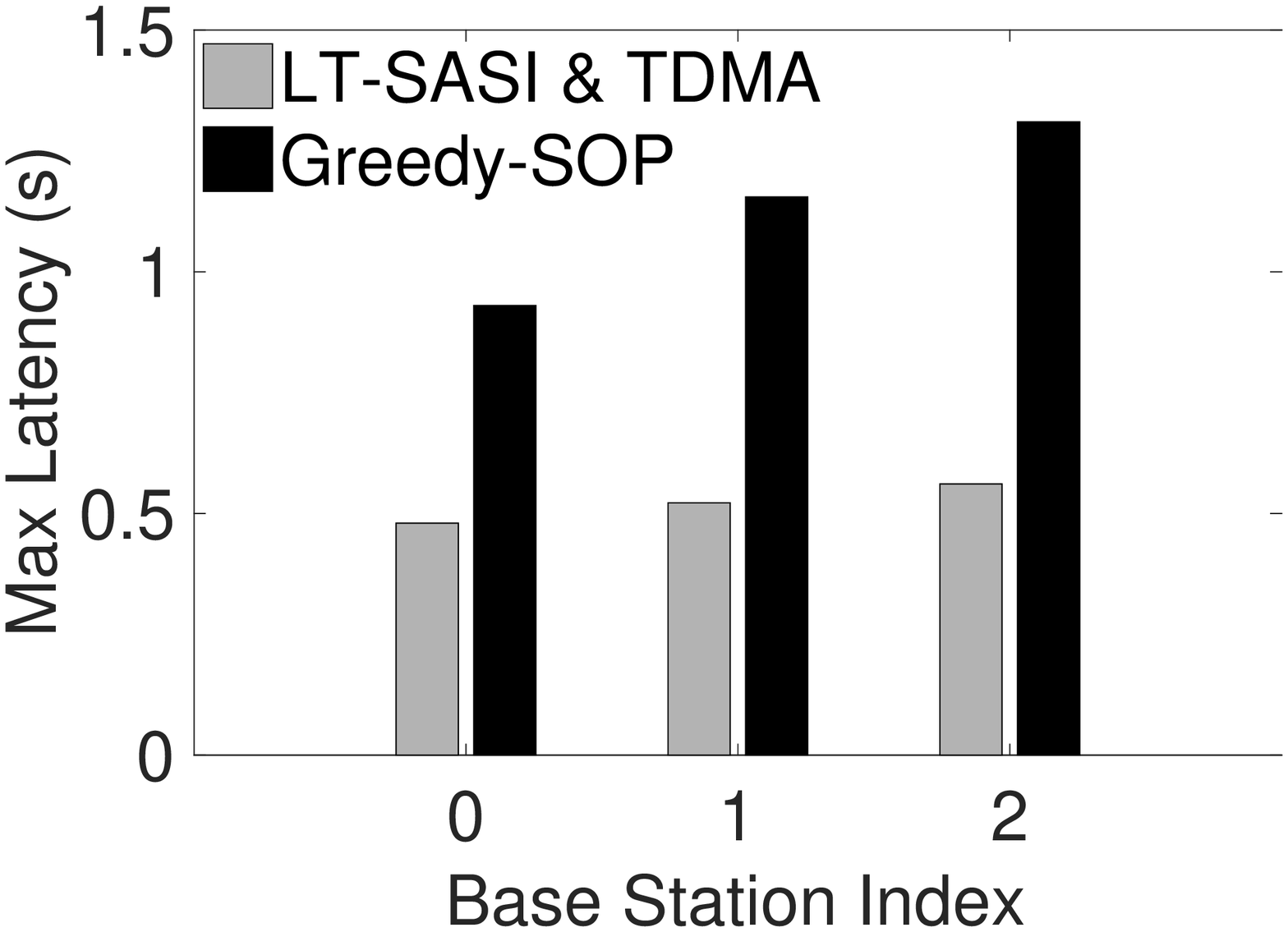}
		\caption{Max Latency of Each BS}
		\label{integration_fig:results-experiment-max-latency}
	\end{subfigure}
	\quad
	\begin{subfigure}[b]{0.35\textwidth}
		\includegraphics[width=\textwidth]{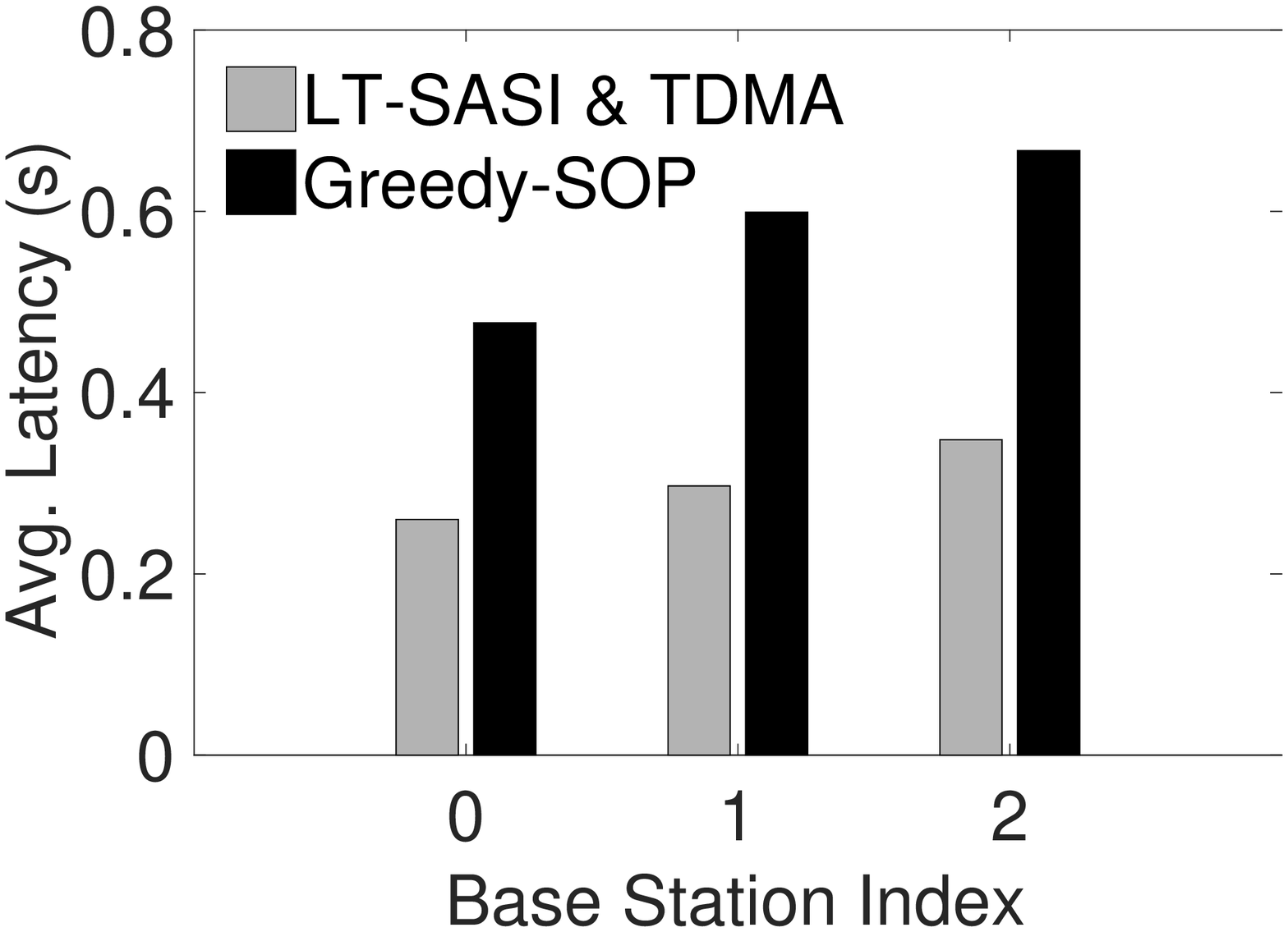}
		\caption{Average Latency of Each BS.}
		\label{integration_fig:results-experiment-avg-latency}
	\end{subfigure}
	\caption{Experimental Result of \gbasi{}  showing Maximum and Average Latency for Each SNOW}
	\label{ntegration_fig:results-experiment}
\end{figure}

\section{Simulation}
\label{sec:simulation}

We perform simulations in NS3\cite{ns3} to evaluate the performance of \gbasi{} under large scale integrations with many nodes and BS. We use a similar setup as the experiment. In the simulation, we create a random tree topology. We use a $400$KHz subcarrier bandwidth for both intra-SNOW and inter-SNOW communication. In the simulation, we evaluate the performance of  \gbasi{} under both RI-TDMA MAC and TDMA MAC, which we refer to in this section as \gbasi{} \& RI-TDMA and \gbasi{} \& TDMA, respectively.

\subsection{Maximum and Average Latencies for Each SNOW}
In this simulation, we compare the maximum and average latencies for each SNOW BS. We perform simulation on a $2000$ node SNOW integration with nodes distributed across $5$ BSs. We consider 50 subcarriers are available for use by a BS. We use a fixed period of $4$s for all nodes in the period.

\begin{figure}[ht]
	\centering
	\begin{subfigure}[b]{0.35\textwidth}
		\includegraphics[width=\textwidth]{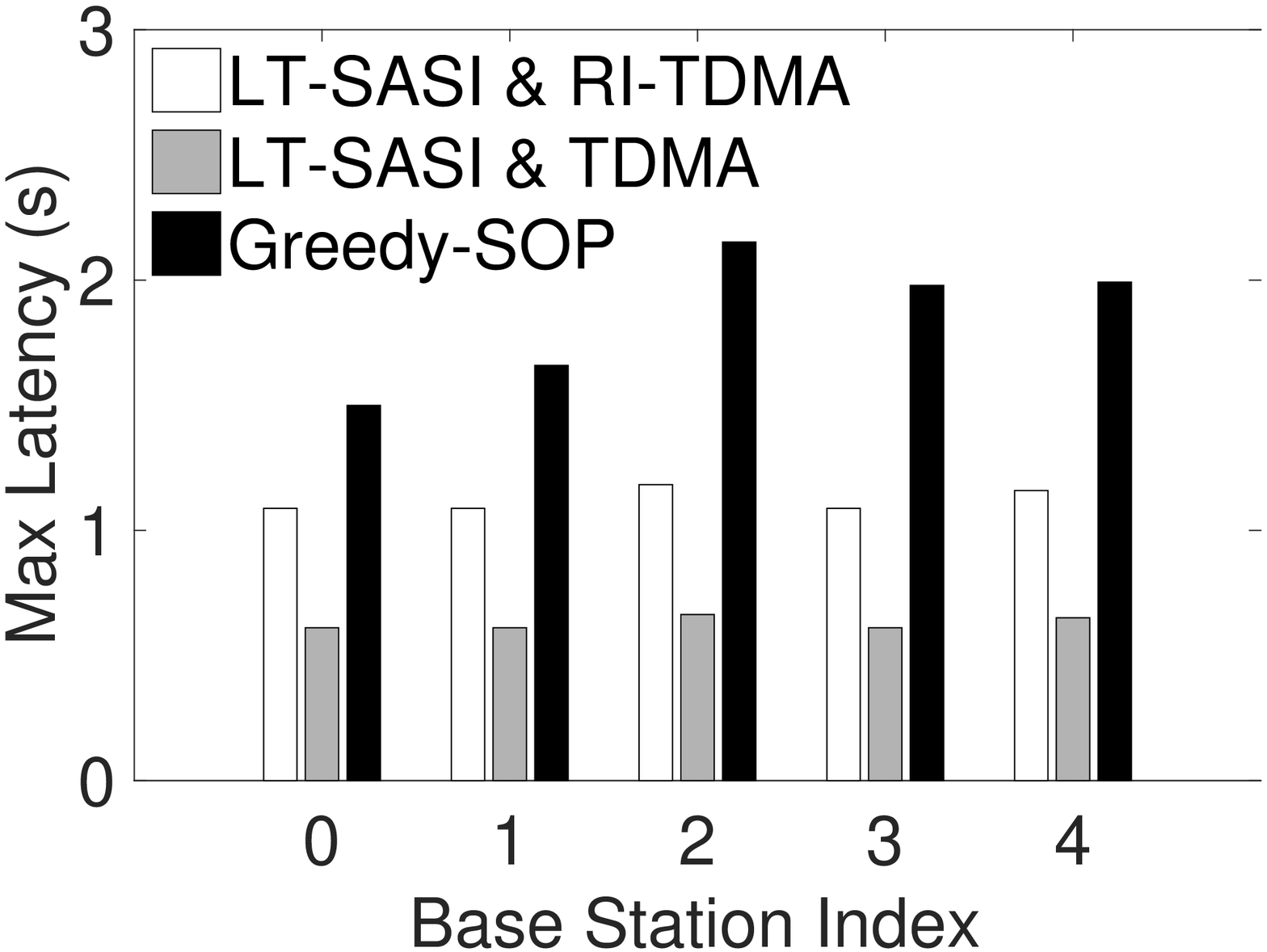}
		\caption{Maximum Latency of Each BS}
		\label{integration_fig:results-distribution-max-latency}
	\end{subfigure}
	\quad
	\begin{subfigure}[b]{0.35\textwidth}
		\includegraphics[width=\textwidth]{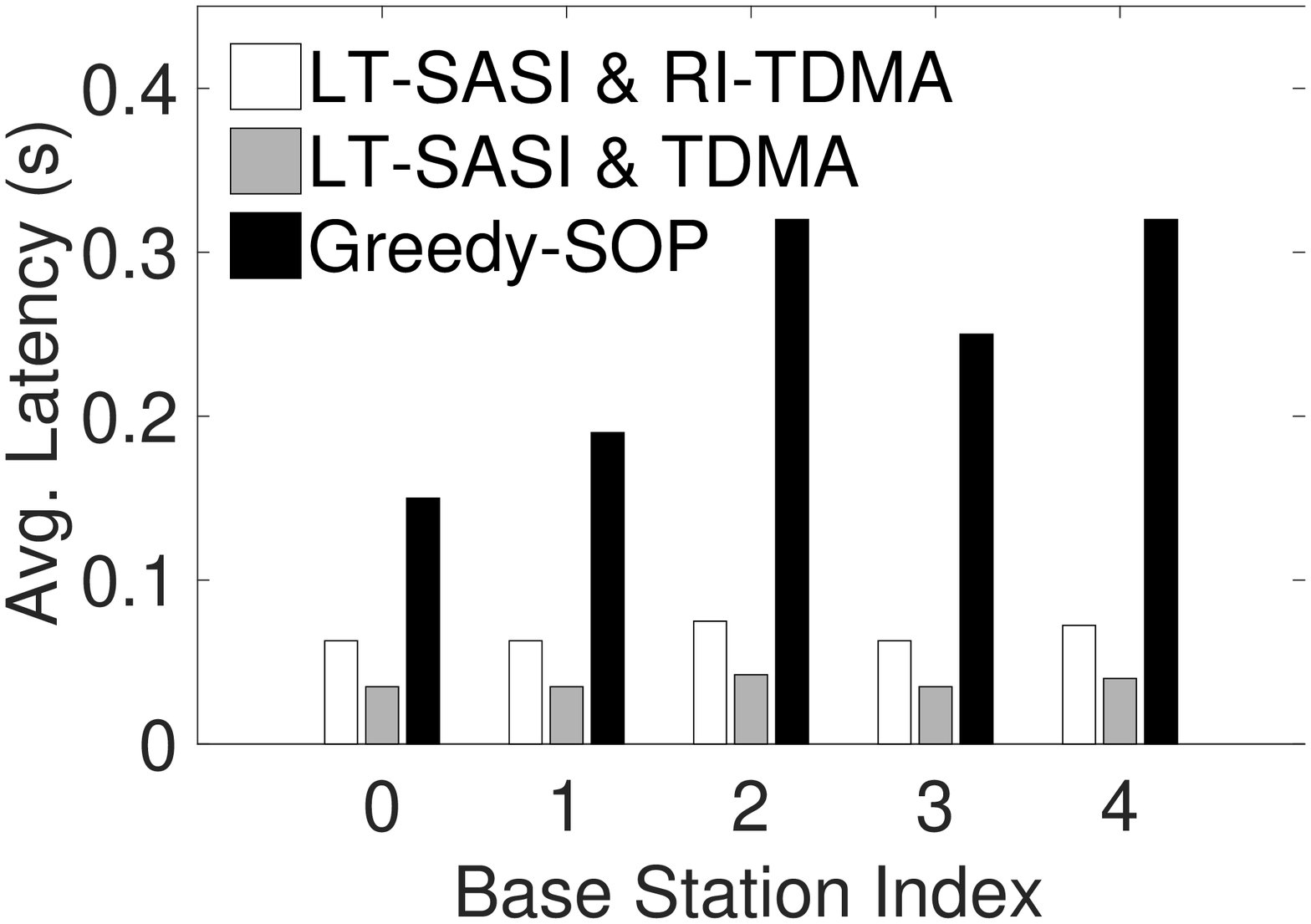}
		\caption{Avg. Latency of Each BS}
		\label{integration_fig:results-distribution-avg-latency}%
	\end{subfigure}
	\caption{Performance Evaluation of \gbasi{} under Maximum and Average Latencies for Each SNOW}
	\label{interation_fig:results_distribution}
\end{figure}

This simulation result complements the experiment result. Under \gbasi{}, the maximum and average latency are similar for each BS, while the maximum and average latency of the Greedy-SOP algorithm varies significantly. In this simulation, we observed that the maximum latency for the \gbasi{} \& TDMA was lower than that of \gbasi{} \& RI-TDMA. However, the spectrum allocation under both settings resulted in a similar assignment policy. We also observed that the difference in the latency arises from the length of the time slot, with RI-TDMA MAC's time slot being double the length of TDMA's time slot. Furthermore, a fixed period of all node generates a uniform traffic, and favors the latency computation of the TDMA MAC. However, for a harmonic period assignment, the latency computation for the TDMA MAC is not accurate since it does not take into account that multiple packets of a flow can delay the same packet.

 \begin{figure}[ht]
 	\centering
		\includegraphics[width=0.35\textwidth]{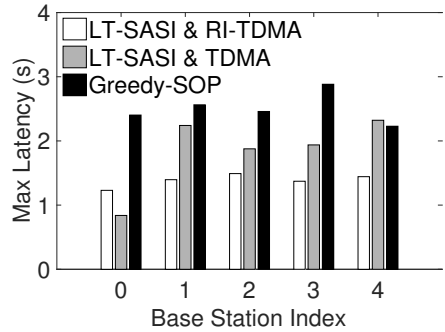}
		\caption{Maximum Latency of Each BS for 2000 nodes.}
		\label{integration_fig:results-distribution-max-latency-flowrate}
\end{figure}

To observe the impact of latency estimation accuracy on \gbasi{}, we perform another simulation with the same setup but with harmonic period assignments in the range $[1s, 4s]$. This simulation shows that the maximum latency for \gbasi{} \& RI-TDMA is similar across all BSs, while the latency of the \gbasi{} \& TDMA varies significantly, as shown in Fig. \ref{integration_fig:results-distribution-max-latency-flowrate}. From this result, we can conclude that the \gbasi{} algorithm generates good results for an accurate estimate in latency. Furthermore, we can also conclude that \gbasi{} decreases the latency up to $50\%$ when compared to Greedy-SOP.

\subsection{Performance of \gbasi{} under Scalability of Number of Nodes}
In this simulation, we evaluate the performance of \gbasi{} under the scalability of the number of nodes. For a BS tree topology consisting of $5$BSs, we vary the number of nodes from $500$ to $5000$. For each node, we select harmonic periods in the range $[1s, 4s]$. We consider that 50 subcarriers are available for use by each BS. 

\begin{figure}[ht]
	\centering
	\includegraphics[width=0.35\textwidth]{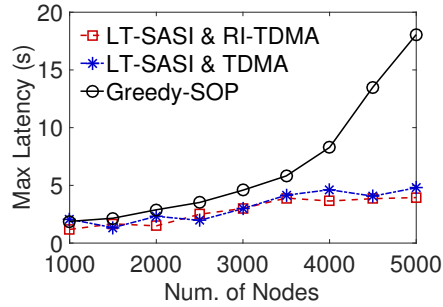}
	\caption{Maximum Latency under Scalability of Number of Nodes.}
	\label{integration_fig:Max Latency under Scalability of Number of Nodes}
\end{figure}

In this simulation, we observed that the increase in latency for \gbasi{} \& RI-TDMA is mainly due to the increase in time slot length. However, the increase in latency for Greedy-SOP algorithm arises from subcarrier contention. We observed that the performance of \gbasi{} \& TDMA fluctuates due to the randomness of the period assignment. \gbasi{} \& TDMA approach performs better when the periods are almost similar to each other, but performs poorly otherwise. From this simulation, we can conclude that the \gbasi{} performs well with the increase in the number of nodes in a SNOW. We have observed that \gbasi{} achieves up to $84\%$ decrease in latency when compared to Greedy-SOP.

\subsection{Performance of \gbasi{} under Scalability of Number of Base Station}
In this simulation, we evaluate the performance of \gbasi{} under the scalability of the number of BSs. We randomly generate a tree structure but keep the number of nodes in each tree constant at $500$ nodes. For each node, we select harmonic periods in the range $[1s, 4s]$. We consider that 100 subcarriers are available for use by each BS. 
\begin{figure}
	\centering
	\includegraphics[width=0.35\textwidth]{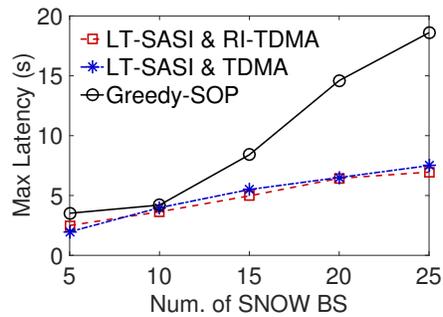}
	\caption{Peformance Evaluation of \gbasi{} under Scalability of Number of SNOW BSs with Harmonic Period Assignment}
	\label{fig:bs_max_latency}
\end{figure}
In this simulation, we observed that an increase in distance from the root BS significantly increases the maximum latency observed under Greedy-SOP. However, \gbasi{} minimizes the impact of distance by allocating a higher number of subcarriers. From this simulation, we can conclude that the \gbasi{} scales well with the increase in the number of BSs in SNOW integration. We have observed that \gbasi{} achieves up to $60\%$ decrease in latency when compared to Greedy-SOP.

%% file: utilization/ICII.tex
\chapter{A Utilization-Based Approach for Schedulability Analysis in Wireless Control Systems}
\label{ch:utilization}

Recent advancements in industrial Internet-of-Things (IoT), more specifically, the development of industrial wireless standards such as WirelessHART and ISA100, are paving the way for the fourth industrial revolution, Industry 4.0. These wireless standards specify highly reliable and real-time communications as key requirements in industrial wireless sensor-actuator networks. Schedulability analysis remains the cornerstone for  analyzing the real-time performance of these networks. While it is well-explored in the domain of CPU scheduling, schedulability analysis for multi-hop wireless networks has seen little progress till date. Existing work mostly focuses on worst-case delay analysis that runs in pseudo-polynomial time, making it is less suitable under frequent network dynamics which are quite common in industrial IoT. To address this, in this chapter, we develop a schedulability analysis based on utilization bound for multi-hop, multi-channel industrial wireless sensor-actuator networks. Because of its extremely low runtime overhead, utilization-based schedulability test is considered to be one of the highly efficient and effective schedulability analyses. However, no work has been done yet on utilization-based analysis for multi-hop wireless network. The key challenge for a utilization-based test for multi-hop wireless network arises from the fact that wireless network is subject to transmission conflict and network dynamics which are not present in CPU scheduling. We address this challenge by bridging the gap between wireless domain and CPU task scheduling. We have evaluated our result through simulations using TOSSIM that shows that our schedulability analysis is safe and effective in practice.

\input{utilization/Introduction}
\input{utilization/Background}

\input{utilization/Proposed_Research}

\input{utilization/evaluation}

\section{Summary}\label{utilization_sec:conclude}
We have developed a schedulability analysis based on \emph{utilization bound} for multi-hop wireless networks. This approach determines the maximum total utilization of all flows in the network and determines those as \emph{schedulable} if the total utilization does not exceed the maximum possible utilization in the network. Because of its extremely low runtime overhead, a utilization-based schedulability test is considered one of the most efficient and effective schedulability tests.  In this chapter, we show the utilization bound for a WirelessHART application with tree routing. We have also discussed the computation of a utilization bound for WirelessHART applications with source and graph routing. 

This work is the inception of a new horizon on utilization-based analysis for industrial IoT, which can direct the wireless community in the same way the real-time systems research today evolved from Liu and Layland's utilization bound.  Our result can trigger many research directions in the line of real-time scheduling, scheduling-control co-design, control performance optimization, routing, priority assignment, and mixed-criticality real-time wireless sensor and actuator networks. Our future work involves analyzing the effects of assigning sub-deadlines for large networks, packet loss, and the trade-offs among various control performance metrics.

%% file: utilization/Introduction.tex
\section{Introduction}
\label{utilization_sec:introduction} 

Recent advancements in industrial Internet-of-Things (IoT), more specifically the development of industrial wireless standards such as WirelessHART \cite{WirelessHART2007_standard} and ISA100 \cite{ISA}, are paving the way for the fourth industrial revolution, Industry 4.0 \cite{iiot}. These wireless standards offer a closed loop communication between the sensors and actuators, where sensors measure process variables and deliver to a controller. The controller generates control commands based on the measured process variables and then sends the control commands to the actuators through the network. In order to ensure the stability of the industry, these wireless standards should offer reliable and real-time communication, i.e., a control command should reach the actuator before a given deadline. For example, in oil refineries,  the spilling of oil tanks has to be avoided by controlling the level measurement in real-time. However, industry settings pose a harsh environment for wireless communication causing frequent transmission failures due to noisy channels, limited bandwidth,  obstacles, multi-path fading, and interference that make it difficult to meet these requirements \cite{lu2016real}.  

Industrial wireless standards such as WirelessHART mitigate frequent transmission failures through channel hopping and multi-channel communication. These networks, therefore, provide the feasibility of achieving reliable and real-time communication over wireless for critical process control applications. Nevertheless, unlike the wired counterpart, real-time scheduling theory for the wireless network is still not well-developed. 

\emph{Schedulability analysis} remains the cornerstone of real-time scheduling theory for industrial IoT \cite{iiot}. Schedulability analysis is used to determine, both at design time and for online admission control, whether a set of real-time control loops/flows (i.e., end-to-end communication between a sensor and an actuator) can meet their deadlines. It thus helps the network manager to plan in advance and adjust workloads in response to network dynamics for real-time process control applications. For example, during channel blacklisting or a route change, the analysis is used to promptly decide the rate of a control loop/s to maintain real-time guarantee. In addition to design time and online admission control, a schedulability analysis is used in scheduling-control codesign \cite{RTAS12_extended}, real-time routing, and priority assignment \cite{ECRTSpaper}. However, existing work on schedulability analysis focuses on worst-case delay analysis \cite{RTAS11paper, RTSS15paper} which runs in exponential time. Hence, these techniques are less suitable for Industry 4.0 architectures, using industrial IoT, which require frequent checking due to channel/link/node failures and changes to plant operating conditions.

In this chapter, we develop a schedulability analysis based on \emph{utilization bound} which is yet an unexplored problem for multi-hop wireless networks. In this approach, we ascertain the maximum possible utilization of all flows in the network and determine the flows as \emph{schedulable} if the total utilization does not exceed the maximum possible utilization in the network. Because of its extremely low runtime overhead, a utilization-bound based schedulability test is considered one of the most efficient and effective schedulability tests. Therefore, it was extensively studied in CPU scheduling \cite{buttazo}. However, no work has been done yet on utilization-based analysis of multi-hop wireless network.   The key challenge  arises from the fact that wireless networks are subject to transmission conflict and dynamics which are not present in CPU scheduling. We address this challenge by bridging between wireless domain and CPU task scheduling. We characterize simultaneous transmission on multiple channels as processors in a multi-processor environment and transmission conflict as task blocking in traditional non-preemptive scheduling. 

We evaluated our schedulability analysis in simulations using TOSSIM \cite{tossim} for the earliest deadline first (EDF) and the deadline-monotonic (DM) scheduling algorithms. Simulations results show that our schedulability analysis is safe and effective in practice. Our analysis hence can be used as an effective schedulability test for admission control of real-time flows.

The rest of the chapter is organized as follows. Section \ref{utilization_sec:bck_related} describes the background on schedulability analysis. Section \ref{utilization_sec:stateoftheart} reviews related work on schedulability analysis for industrial IoT. Section \ref{utilization_sec:sysmodel} describes the system model and Section \ref{utilization_sec:problem} presents the problem formulation. Section \ref{utilization_sec:proposal} presents the schedulability analysis. Section \ref{utilization_sec:disc} extends the proposed schedulability analysis to a general industrial IoT environment. Section \ref{utilization_sec:evaluation} presents the simulation results. Section \ref{utilization_sec:conclude} summarizes the chapter.

%% file: utilization/Background.tex
\section {Background}
\label{utilization_sec:bck_related} 
In general, \emph{end-to-end delay bound analysis} and \emph{utilization bound analysis} are the two broad approaches for schedulability analysis. A utilization bound analysis specifies the maximum possible utilization of all flows in the network and determines the flows as \emph{schedulable} if the total utilization does not exceed the maximum possible utilization in the network. Because of its extremely low runtime overhead, a utilization-based schedulability test is considered one of the most efficient and effective schedulability tests. The end-to-end delay bound based analysis \cite{RTAS11paper, Chengjie, RTSS15paper} requires a separate schedulability test for each flow, which runs in \emph{pseudo-polynomial} time (i.e., \emph{exponential} in the length of the input). However, utilization bound based analysis can provide a single closed-form expression that can run in polynomial time (usually in \emph{linear} time). It thus greatly simplifies various scheduling-control optimization problems, for which pseudo-polynomial time delays bounds is a major hurdle due to its non-linearity, non-convexity, non-differentiability, long execution time, and a large number of constraints (at least $n$ constraints for $n$ flows) \cite{RTAS12_extended}.  In this research, we want to develop a schedulability analysis based on \emph{utilization bound} which is a yet unexplored problem for multi-hop wireless networks.

\section{Related Work}
\label{utilization_sec:stateoftheart}
Real-time scheduling for wireless networks was explored in many early \cite{Stankovic2003_Realtime}  and recent works \cite{li2005scheduling, Wang2009_FlowbasedRealTimeMultichannel, Kanodia2001_Distributed, Lu2002_RAP,  Karenos2006_real,  sensoractuator, Liu2006_JiTS,  GuRTSS, Pereira2007Broadcast, He2007robust}.  However, these works do not focus on schedulability analysis in the network.  The works in  \cite{RTQS, Tarek2004Capacity, Pereira2007Broadcast, capnet, capnetRTSS} discuss schedulability analysis for wireless sensor networks using end-to-end delay bound and they focus on data collection through a routing tree \cite{RTQS}  and/or do not consider multiple channels \cite{RTQS, Tarek2004Capacity}.  In contrast, we consider an industrial IoT based on multiple channels and our analysis is not limited to data collection towards a sink. Furthermore, our analysis is targeted for real-time flows between sensors and actuators for process control purposes and applies to multi-path routing with minimal changes.

Real-time scheduling for industrial IoT based on WirelessHART has received considerable attention in the recent past \cite{RTSS10paper, RTAS11paper, ECRTSpaper, Soldati2009_WirelessHARTevacuation, RTAS12_extended, Zhang2009_WirelessHARTRapid, RoamingHART}.  The works in \cite{Soldati2009_WirelessHARTevacuation, Zhang2009_WirelessHARTRapid} focus on data collection in a tree topology. The works in \cite{modekurthy2018distributed} address graph routing algorithms for WirelessHART networks and that in \cite{RoamingHART} propose a localization system using WirelessHART.  Priority assignment policies for WirelessHART are studied in \cite{ECRTSpaper} and rate selection algorithms are studied in \cite{RTAS12_extended}. 
The work in \cite{RTSS10paper} considers dynamic priority scheduling and does not address any schedulability analysis. To summarize, none of these works focus on schedulability analysis.  The work in \cite{RTAS11paper} presents the first step in establishing a rigorous delay analysis for industrial IoT. Nevertheless, this work does not consider the critical fault-tolerant mechanisms, for achieving reliable communication, like retransmissions and reliable graph routing.  The work in \cite{RTSS15paper} provides a suite of end-to-end delay analysis techniques for schedulability analysis under fixed priority scheduling in WirelessHART networks. In summary, these papers use the worst-case delay that can be obtained by a flow to determine its schedulability. Consequently, these algorithms are less suitable for admission control as many control optimization algorithms execute during run-time and processor cannot be overloaded with other jobs. 

In this chapter, we propose a utilization-based approach for schedulability analysis which provides the simplicity and efficiency in application. A utilization-based analysis was studied in \cite{tarekutilization}  for a single-hop wireless network. In contrast, we focus on multi-hop industrial IoT in which scheduling and analysis are significantly different and challenging as it has to deal with multiple concurrent transmissions on different channels, interferences, and transmission conflicts. Efficient schedulability analysis is particularly useful for online admission control and adaptation (e.g., when network route, topology, or channel condition change) so that the network manager can quickly reassess the schedulability of the flows.

\section{System Model}
\label{utilization_sec:sysmodel}
Because of the worldwide adoption of WirelessHART for process control in challenging industrial environments, we consider an industrial IoT based on the WirelessHART standard \cite{WirelessHART2007_standard}.  WirelessHART networks operate on the 2.4GHz band and are build on the physical layer of IEEE 802.15.4. They form a multi-hop mesh topology of nodes consisting of multiple field devices, multiple access points, and a gateway.  The \emph{network manager} creates routes and transmission schedules. \emph{Access points}   provide redundant paths between the wireless network and the gateway. The \emph{field devices}  are wirelessly networked sensors and actuators. The sensors periodically deliver sample data to the controller through the access points. The controller generates control commands based on the measured process variables and then sends the control commands to the actuators.  Each node has a \emph{half-duplex} omnidirectional radio transceiver,  and hence cannot both transmit and receive at the same time and can receive from at most one sender at a time. 

Transmissions in a WirelessHART network are scheduled based on a multi-channel TDMA (Time Division Multiple Access) protocol. The network employs global time synchronization protocols to synchronize time at all nodes in the network. Each time slot is of $10ms$ and each transmission needs one time slot. A receiver transmits an acknowledgment to notify the sender about a successful reception of a packet. Note that, both the transmission and acknowledgment happen in one 10ms time slot. The network uses the channels defined in IEEE $802.15.4$. It adopts {\slshape channel hopping} in every time slot to achieve high reliability. An excessively noisy channel is assumed to be {\slshape blacklisted} and not used for communication.  We also assume that the network does not allow spatial reuse of channels, i.e.,  a channel can be used by only one node to transmit packet in a time slot. We assume the network adopts a tree routing, where all nodes in the network form a tree rooted at the gateway. Sensor nodes forward data to the controller (located at the gateway) through the upward links. The controller sends control commands to the actuators through the downward links. Note that, we make this assumption only to provide a tight bound on the schedulability analysis and our method can be extended to graph routing, as will be described in Section \ref{utilization_Sec:GenericRouting},  that is typically adopted in industrial IoT and that provides redundant paths for packet delivery for enhanced reliability.

\section{Problem Formulation}\label{utilization_sec:problem} 
Each control loop, also called a \emph{flow},  involves one or more sensors and one or more actuators. Transmissions between access points, sensors, and actuators are scheduled on $m$ ($m\ge 1$) channels. We assume, there are $n$  control loops denoted as $F_1, F_2, \cdots, F_n$. The worst-case execution time of a flow, period (sampling rate of sensors), and the deadline of   $F_i$  are denoted by $C_i$, $T_i$, and  $D_i$, respectively. Note that, real-time wireless standards like WirelessHART and ISA-100 reserve a fixed $\omega$ number of slots for each link to handle both transmissions and re-transmissions. WirelessHART uses $\omega = 2$ for each link to successfully transmit a packet. Thus, if there are a total of $\ell_i$ links  on flow $F_i$'s route, then $C_i$ can be computed as $C_i=  \ell_i \times  \omega $. Note that, $C_i$ can change due to network dynamics and route change. In these situations, we need to run the schedulability test again to make sure the new routes are schedulable on the network.

The set of periodic flows $F$ is called \emph{schedulable}  if there is a transmission schedule such that no deadline is missed. A schedulability test  $\mathbb{S}$     is \emph{sufficient} if any set of flows deemed to be schedulable by $\mathbb{S}$   is indeed schedulable.  If flow $F_i$ involves a maximum of $C_i$ transmissions, then its utilization  $u_i$ is defined as $\frac{C_i}{T_i}$ and the total utilization of all $n$  flows is defined as $\sum_{i=1}^n \frac{C_i}{T_i}$.  In this chapter, our objective is to determine a sufficient schedulability analysis for EDF and DM schedulers based on utilization bound. Note that, utilization-based schedulability test for other schedulers is out of the scope of this chapter.

%% file: utilization/Proposed_Research.tex
\section{Utilization-Based Schedulability Analysis}\label{utilization_sec:proposal}
Here, we first propose an approach for determining a schedulability analysis by bridging the gap between processor and network scheduling. We then discuss transmission conflict delay computation for a flow. Table \ref{utilization_tab:symbols} summarizes the notations used in this section.

\begin{table}[h]
\vspace{0.15in}
\begin{center}
\begin{tabular}{ cc } 
 \hline
 Symbol & Description \\ 
 \hline
 \hline
 $F_i$ & Flow $i$ \\ 
 $T_i$ & Period of $F_i$ \\ 
 $D_i$ & Deadline of $F_i$ \\ 
 $C_i$ & worst-case execution time of $F_i$ \\ 
 $\Delta_i$ & Transmission conflict delay on $F_i$ \\
 $\delta(i,j)$ & Transmission conflict delay caused on $F_i$ \\ & by high priority flow $F_j$ \\ 
 $hp(F_i)i$ & Set of flows that are higher priority than $F_i$ \\ 
$\alpha(i,j)$ & \# of common paths between routes of \\ & $F_i$ and $F_j$ \\ 
$\alpha_1(i,j)$ & \# of common paths with path length \\ & as one between routes of $F_i$ and $F_j$  \\
$\beta(\rho_i,j)$ & \# of common paths between $\rho_i$-th path \\ & of flow $F_i$ and all paths of $F_j$ \\ 
$\beta_1(\rho_i,j)$ & \# of common paths  between $\rho_i$-th path of flow $F_i$ \\ &  and all paths of $F_j$ with path length 1\\
 $\omega$ & \# of transmission slots assigned \\  & for each link in a flow \\ 
 \hline
\end{tabular}
\vspace{0.15in}
\caption{Notations}
 \label{utilization_tab:symbols}
 \vspace{-0.25in}
 \end{center}
\end{table}

\begin{figure}[t]
    \centering
    \begin{subfigure}[b]{0.3\textwidth}
        \includegraphics[width=\textwidth]{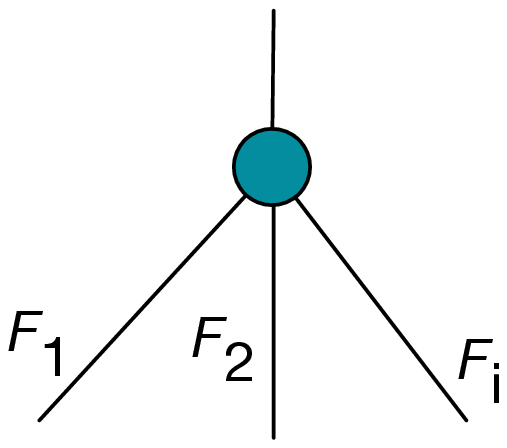}
        \caption{Transmission conflict}
        \label{utilization_fig:conflict}
    \end{subfigure}
    \quad
    \begin{subfigure}[b]{0.3\textwidth}
        \includegraphics[width=\textwidth]{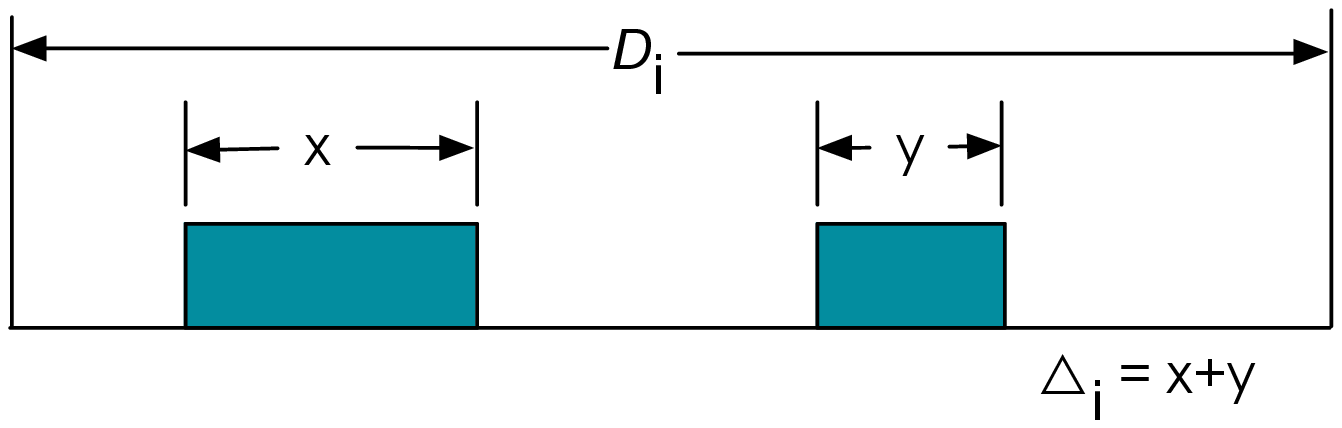}
        \caption{Time loss (cyan color) due to conflict}
        \label{utilization_fig:loss}
    \end{subfigure}
   \caption{Time Loss Due to Transmission Conflict}
   \label{utilization_fig:1ab}
\end{figure}

\subsection{Establishing a Utilization Bound Analysis}\label{utilization_sec:rta}
Channel contention and transmission conflict are the two sources of delays in multi-hop wireless networks. {\bf \emph{Channel contention}} is defined as the delay caused when all available channels in a time slot are assigned to higher priority flows. {\bf \emph{Transmission conflict}} is the delay when two flows share a common node, and the higher priority flow delays the low priority flow at the common node since the common node can transmit/receive a packet for only one flow in one time slot (due to single half-duplex radio). In this section, we establish the utilization-based analysis assuming we know the transmission conflict delay. In the following section, we discuss the transmission conflict delay computation.  

Channel contention delay in industrial IoT can be considered similar to execution delay experienced by a task running on a multi-processor platform when (i) the number of channels in industrial IoT is equal to the number of processors (we use $m$ interchangeably to denote both channels and processors); (ii) length of one time slot (10ms) is equal to the length of one time unit in process scheduling (where a task is non-preemptable); and (iii) period, deadline and worst-case execution time (WCET) of each flow is equal to its analogous task (in the analogous processor task set). This is because a flow (in industrial IoT) and its equivalent task (in processor) are contending for $m$ shared resources. For the same values of period, deadline, and execution time both a flow and its equivalent task generate the same access patterns. Thus, the delay due to channel contention observed by a flow in the network is the same as the delay observed by its equivalent task on a processor. Our technical approach leverages this bridge between multiprocessor scheduling and wireless transmission scheduling.

We first review the results on preemptive and non-preemptive scheduling on multiprocessors. In \emph{preemptive scheduling}, a task upon start can be preempted by any higher priority task any time. In \emph{non-preemptive scheduling}, a task once started can never be preempted by any other task. In non-preemptive scheduling, a higher priority task thus experiences priority inversion where a high priority task is \emph{blocked} by a lower priority task (as it cannot preempt if the lower priority task has already started). EDF is a dynamic priority scheduling policy where, at any time,  the task having the shortest absolute deadline is scheduled first. A set of $n$ real-time tasks with a constrained deadline (i.e. $D_i\le T_i$) is schedulable using preemptive EDF scheduling on $m$ processors \cite{Baruah} if 
\begin{equation}\label{utilization_edf}
\sum_{i=1}^{n} \frac{C_i}{D_i}  \le m - (m-1)  \left ( \max \left\{  \frac{C_i}{D_i} | 1\le i\le n \right\}     \right ).                   
\end{equation}

DM scheduling is a fixed priority scheduling policy where tasks are prioritized based on their relative deadlines. A set of $n$ real-time tasks with a constrained deadline (i.e. $D_i\le T_i$) is schedulable using preemptive DM scheduling on $m$ processors \cite{Baruah} if 
\begin{equation} \label{utilization_dm}
\sum_{i=1}^{n} \frac{C_i}{D_i}  \le  \frac{m}{2} \left (1 -   \max \left\{  \frac{C_i}{D_i} | 1\le i\le n \right\}      \right) + \max \left\{  \frac{C_i}{D_i} | 1\le i\le n \right\} .
\end{equation}
For non-preemptive scheduling the corresponding conditions are derived  by taking into account the maximum blocking time.

To adopt the similar results for industrial IoT, we present our technique as follows. We can use Equation (\ref{utilization_edf}) and Equation (\ref{utilization_dm}) for industrial IoT in the absence of transmission conflict, and when every transmission happens on a separate channel in each time slot (allowing at most $m$ concurrent transmission per time slot). However, transmission conflict poses an additional challenge in the wireless domain. We can model the transmission conflict delays as blocking time in non-preemptive scheduling. Assuming $\Delta_i$ (as computed in Section \ref{utilization_sec:deltaComputation}) denotes the transmission conflict delay caused on flow  $F_i$ by all higher priority flows.  
Then, $D_i -  \Delta_i$ represents the utilization loss due to transmission conflict, i.e., $F_i$ can use at most $D_i -  \Delta_i$ time slots to complete the end-to-end communication. For example, in Fig. \ref{utilization_fig:conflict}, the receiver can receive from at most one transmitter in a time slot, and hence only one flow is assigned a time slot to transmit and the other two flows are blocked during this period. In this example, $F_1$ has the highest priority and hence is assigned time slots $t$ and $t+1$ (let $t$ be the current system time). Similarly, $F_2$ is assigned time slots $t+2$ and $t+3$. During time slots $t$, $t+1$, $t+2$, and $t+3$, flow $F_i$ is blocked and waits for an idle transmission slot regardless of channel availability. Thus, we can consider a loss of four time slots from the relative deadline of $F_i$ (Fig. \ref{utilization_fig:loss}). Namely, from deadline $D_i$, the flow $F_i$ loses at most $\Delta_i$ slots, and hence its effective utilization (as shown in Fig. \ref{utilization_fig:loss}) $\mu_i$ becomes 
$\mu_i =   \frac{C_i}{D_i-  \Delta_i}$ .  Let us define
\begin{equation*}
\begin{split}
&\mu_{\max} = \max\{ \mu_i  | 1\le i\le n   \}; \\ &\qquad  \mu_{\text{sum}}  =  \underset{i=1}{   \overset{n}{\sum}   }   \mu_i.
\end{split}
\end{equation*}

Therefore, from Equations (\ref{utilization_edf}) and (\ref{utilization_dm}), any constrained deadline set of real-time flows  is schedulable on $m$ channels in an industrial IoT that allows at most $m$ concurrent transmissions under EDF scheduling  if 
\begin{equation} \label{utilization_edfwsan}
\begin{split}
\mu_{\text{sum}}  \le m - (m-1)\mu_{\max} \\  0 < \mu_{i} \le 1 \qquad  \forall  1\le i\le n
\end{split}
\end{equation}
and under  DM scheduling  if  
\begin{equation} \label{utilization_dmwsan}
\begin{split}
\mu_{\text{sum}}  \le  \frac{m}{2} \left (1 -\mu_{\max}\right) +  \mu_{\max}\ \\  0 < \mu_{i} \le 1 \qquad \forall 1\le i\le n.
\end{split}
\end{equation}

Note that, Equations (\ref{utilization_edfwsan}) and (\ref{utilization_dmwsan}) consider a preemptive scheduler where preemptions are allowed at the start of a time slot. For non-preemptive schedulers,  $\Delta_i$ also includes the maximum blocking time caused due to priority inversions.

\subsection{Transmission Conflict Delay Computation}
\label{utilization_sec:deltaComputation}
In this section, we first obtain a bound on the number of shared paths between two routes on a tree routing. We then obtain a bound on the number of time slots a packet of flow $F_j$ can delay a packet of flow $F_i$ on a shared path. We then use these two bounds to derive an upper bound on the transmission conflict delay that a flow can experience under the assumption that network routing happens on a bi-directional tree.

\begin{figure}[t]
    \centering
    \begin{subfigure}[b]{0.3\textwidth}
        \includegraphics[width=\textwidth]{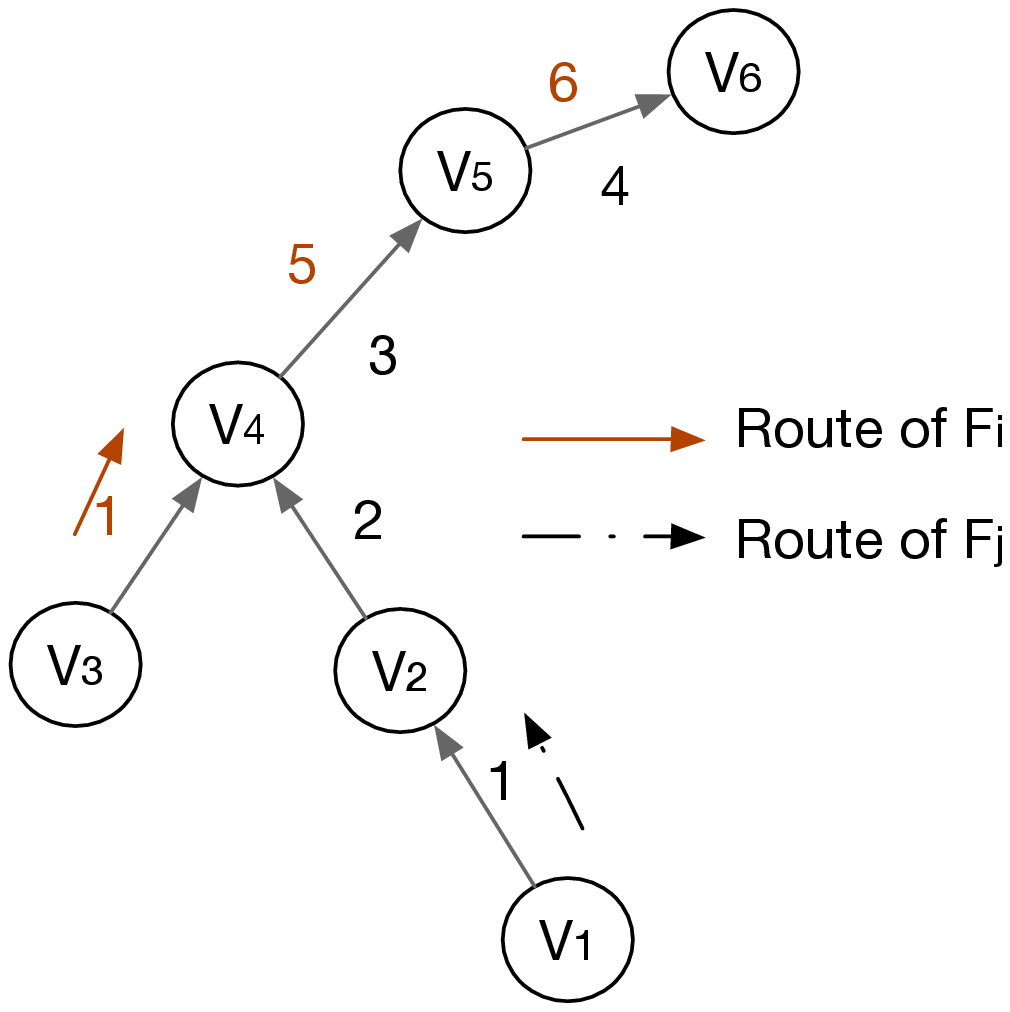}
        \caption{\normalsize $T_i = T_j = 8 ; \delta(i,j) = 3$}
        \label{utilization_fig:delay3}
    \end{subfigure}
    \quad
    \begin{subfigure}[b]{0.3\textwidth}
        \includegraphics[width=\textwidth]{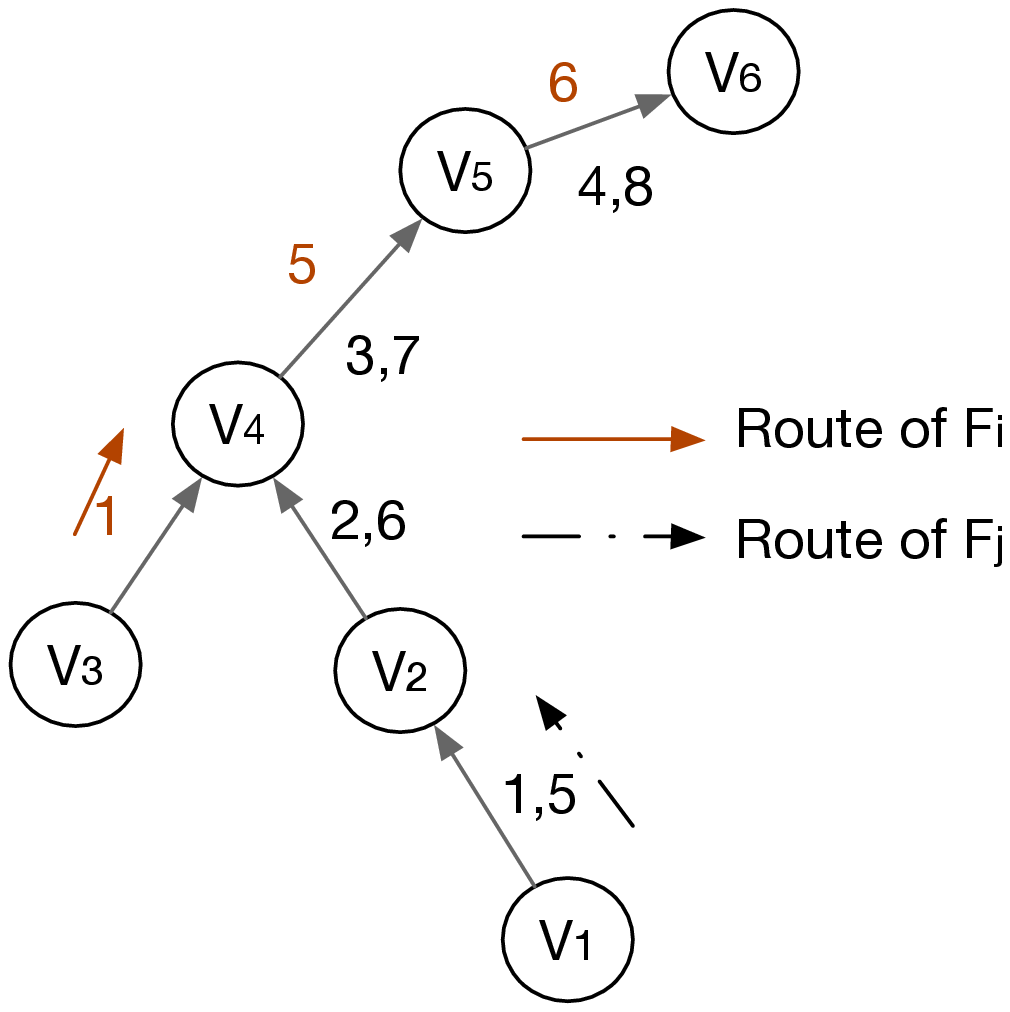}
        \caption{\normalsize $T_i = 8, T_j = 4 ; \delta(i,j) = 3$}
        \label{utilization_fig:delay4}
    \end{subfigure}
   \caption{Example of $F_j$ Delaying $F_i$ on a Tree Route}
   \label{utilization_fig:delay34}
\end{figure}

We define a {\slshape common path} as a set of nodes shared between routes of two flows. A flow $F_i$ experiences a transmission conflict delay by another flow $F_j$ on a common path between the routes of $F_i$ and $F_j$. To estimate the transmission conflict delay, we first estimate the number of common paths between $F_i$ and $F_j$. Specifically, we first bound the number of common paths on the uplink route (which connects a sensor to the controller) and the downlink route (which connects the controller to an actuator). In the uplink route, $F_i$ and $F_j$ transmit messages to the same destination, i.e., the controller, which resides at the root of the tree. In a tree, there can exist only one parent for every node. Therefore, a common path which starts at some node $V_k$ (where $V_k$ can be an intermediate node or the controller), only ends at the controller. Thus, we can conclude that in an uplink route there exists only one common path which starts at some node $V_k$ and ends at the controller. We can use similar reasoning for the downlink route. In a downlink route, there exists only one common path which starts at the controller and ends at an intermediate node $V_l$. Therefore, in a tree route, the number of common paths is limited to $1$ where the common path is said to start at $V_k$ and end at $V_l$. For example, Fig \ref{utilization_fig:delay3} shows that two flows $F_i$ and $F_j$ share a path from $V_4$ to $V_6$.

We now compute the maximum delay caused by one packet of a flow $F_j$ on a packet of flow $F_i$ considering tree routing. At a common node, $V_a$, on the common path, the packet of $F_i$ is delayed at most $3\omega$ times by a packet of $F_j$ \cite{RTSS15paper}. For example, as shown in Fig. \ref{utilization_fig:delay3} flow $F_i$'s transmission from $V_4 \rightarrow V_5$ conflicts with three transmissions of $F_j$, $V_2 \rightarrow V_4, V_4 \rightarrow V_5,$ and $V_5 \rightarrow V_6$.  At time $3\omega + 1$ (time slot after the blocking duration), packets of $F_i$ and $F_j$ have different destination nodes, and they can use different channels to make concurrent transmissions. Note that, if only one channel is available for transmission, then the packet of $F_i$ is blocked due to channel contention and not due to transmission conflict. Thus, we can say that on a common path between two flows, a low priority packet can be delayed due to transmission conflict by a high priority packet on one node. For example, in a common path from $V_a,V_b, V_c, V_d, \cdots V_k$, if at time $\tau$ a packet of $F_j$ delays a packet of $F_i$ by $\delta(i,j)$, then at time $\tau + 3\omega$ the packet of $F_j$ is at node $V_c$ and packet of $F_i$ is at node $V_a$. At time $\tau + 3\omega + 1$, $F_i$ and $F_j$ can make simultaneous transmission thereby keeping a non-decreasing distance between $F_i$ and $F_j$ flows. In summary, we can say that a packet of $F_j$ blocks a packet of $F_i$ by at most $3\omega$ time slots.

We now compute the maximum delay caused by a flow $F_j$ on a flow $F_i$. On a common path, a flow $F_i$ can be delayed by at most $\big\lceil \frac{T_i}{T_j} \big\rceil$ times by a high priority flow $F_j$. Note that, we use a pessimistic scenario in which all packets of flow $F_j$ that spawned in  the interval $[\alpha T_i, (\alpha + 1) T_i]$ interfere $\alpha^{th}$ packet of $F_i$. A more accurate bound can be obtained by using response time analysis which is complicated and takes exponential time to find a solution (which we are trying to avoid in this chapter). Considering a pessimistic value on the number of interfering packets, the total delay caused by flow $F_j$ on a flow $F_i$ is expressed as $$\delta(i,j) = 3 \omega \bigg\lceil \frac{T_i}{T_j} \bigg\rceil.$$ 
For a DM scheduler, a flow $F_i$ can only be delayed by a high priority flow. Therefore, the total delay experienced by flow $F_i$ under a DM scheduler considering all high priority flows (given by $hp(F_i)$) is shown by Equation (\ref{utilization_DMDelay}). \begin{equation}\label{utilization_DMDelay} \Delta_i^{DM} = \sum_{F_j \in hp(F_i)} \delta(i,j) \end{equation} 
In case of an EDF scheduler, the flows have dynamic priority based on their absolute deadlines. Therefore, a flow $F_i$ can be delayed by every other flow $F_j$ where $j \neq i$.  Note that, some flows interfere at most once due to very large periods and some flows interfere multiple times due to very short periods. To incorporate these additional delays, we can extend the total delay computation for DM schedulers to EDF schedulers by considering all flows interfere $F_i$. Under an EDF scheduler, an upper bound of the total delay experienced by a flow $F_i$ is given by Equation (\ref{utilization_EDFDelay}). \begin{equation}\label{utilization_EDFDelay}\Delta_i^{EDF} = \sum_{j \in [1, n] \, \text{and} \,  \neq i} \delta(i,j) \end{equation}

\begin{figure}[t]
    \centering
    \begin{subfigure}[b]{0.3\textwidth}
        \includegraphics[width=\textwidth]{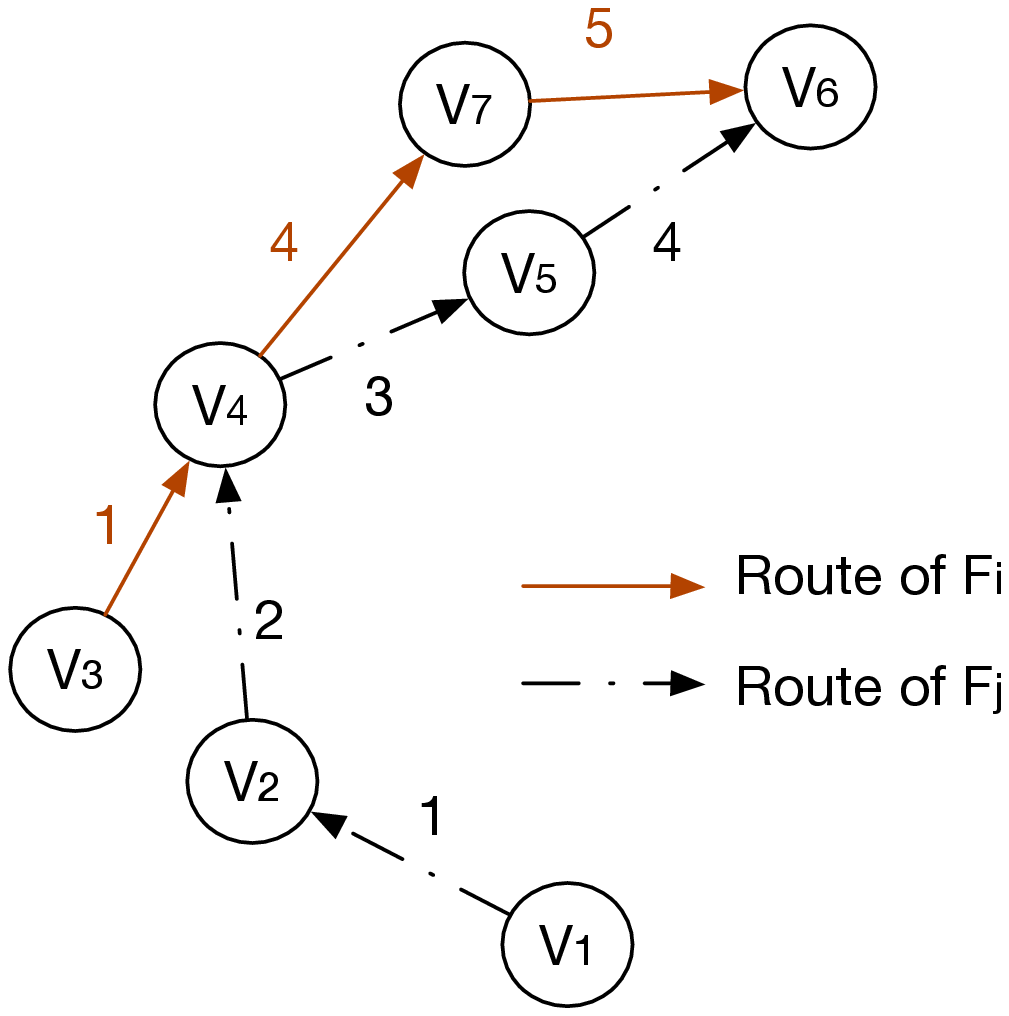}
        \caption{\normalsize $T_i = T_j = 8 ; \delta(i,j) = 2$}
        \label{utilization_fig:delay1}
    \end{subfigure}
    \quad
    \begin{subfigure}[b]{0.3\textwidth}
        \includegraphics[width=\textwidth]{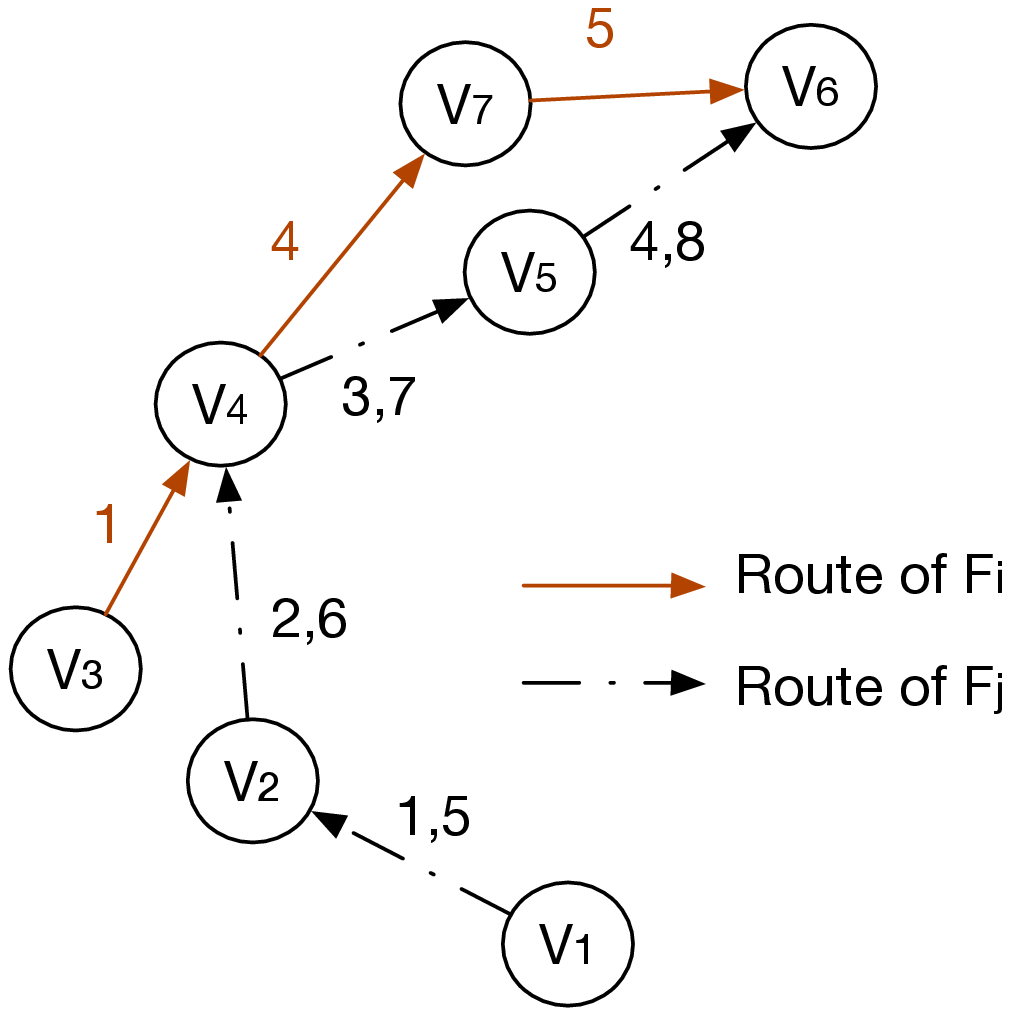}
        \caption{\normalsize $T_i = 8, T_j = 4 ; \delta(i,j) = 2$}
        \label{utilization_fig:delay2}
    \end{subfigure}
   \caption{Example of $F_j$ Delaying $F_i$ on a Graph Route}
   \label{utilization_fig:delay12}
\end{figure}

\section{Extending the Utilization Bound Analysis to Graph Routing and Hierarchical Networking}
\label{utilization_sec:disc}
In this chapter, we consider a network model that supports only tree routing and does not consider the spatial reuse of channels. In this section, we present an approach to handle graph routing algorithms and spatial reuse of channels.

\subsection{Transmission Conflict Delay Computation for Graph Routing Algorithms}
\label{utilization_Sec:GenericRouting}
Routing in industrial IoT is broadly classified into two types: source routing and graph routing. Source routing provides a single route for each flow in the network.  We define, a {\em routing graph} as a directed list of loop-free paths between a source and a destination. Each node in a routing graph must have a minimum of two unique outgoing paths from itself to the destination. Graph routing allows to schedule a packet on multiple links using multiple channels on multiple time slots through multiple paths to deliver a packet to a destination, thereby ensuring high reliability in highly unreliable environments. A routing graph consists of an uplink graph and multiple downlink graphs. An uplink graph connects all sensors to controllers while each downlink graph connects a controller to an actuator. Note that, some recent work on industrial wireless network focuses on route-less approach by relying on a network-wide Glossy flooding [35]. However, such approaches consider very sparse traffic and rely on single channel. Such a model does not have to handle transmission conflict or channel contention. Therefore, our model and approach are significantly different, more challenging, and more general.

In a source or graph routing algorithm, packets can experience different degrees of conflict during communication. For example, a flow can experience a delay at multiple nodes, as shown in Fig. \ref{utilization_fig:delay1}. These different degrees of conflict cause different transmission conflict delay on flows. In the following discussion, we derive an upper bound on the delay that a lower priority flow can experience from the higher priority ones due to conflicts for a network with source routing. Here, we first discuss the transmission conflict delay caused by a packet of $F_j$ on a packet of $F_i$. We then extend the analysis to consider transmission conflict delay caused by all packets of $F_j$ on $F_i$ under a source routing algorithm. We then extend this result to graph routing algorithms.

In a generalized routing, routes of two flows can intersect each other more than once and can generate multiple common paths since there is no limitation on the number of parent nodes. One packet of flow $F_i$ can be delayed at all common path by the same packet of flow $F_j$ since $F_j$ and $F_i$ experience a different delay pattern on the non-common paths. However, on a common path, a packet of $F_j$ can delay a packet of $F_i$ only for $3\omega$ time slots (the same reasoning as discussed in the previous section applies here). Therefore, if routes of two flow $F_i$ and $F_j$ intersect each other $\alpha(i,j)$ times then maximum delay a low priority packet experiences due to a high priority packet is expressed as $3\omega \times \alpha(i,j)$. However, as shown in Fig. \ref{utilization_fig:delay1} routes that intersect on only one node can have a maximum delay of $2\omega$ since the destination node is not the same. Therefore, a packet of flow $F_j$ delays a packet of flow $F_i$ by $3\omega \times \alpha(i,j) - \omega \times \alpha_1(i,j) $  where $\alpha_1(i,j)$ represents the number of intersection in two routes with only one node common to both of them.

We now compute the maximum delay caused by a flow $F_j$ on flow $F_i$. The delay caused by a flow $F_j$ on a flow $F_i$ is given by the Equation (\ref{utilization_delayoneflow}).
 \begin{equation} \label{utilization_delayoneflow} \delta(i,j) = (\alpha(i,j) + \bigg\lceil \frac{T_i}{T_j} \bigg\rceil - 1) 3\omega - \omega \times \alpha_1(i,j)\end{equation}
To prove this, let us assume that there are $\alpha(i,j)$ common paths between the routes of $F_i$ and $F_j$, of these common paths $\alpha_1(i,j)$ common paths consist of only one node, and $\big\lceil \frac{T_i}{T_j} \big\rceil$ packets of $F_j$ interfere one packet of $F_i$. A low priority flow experiences a maximum delay when the first packet of the $F_j$ interferes with the packet of $F_i$ on the first $\alpha(i,j) - 1$ common paths and $\big\lceil \frac{T_i}{T_j} \big\rceil$ packets of $F_j$ interfere $F_i$ on last common path (and the last common path has more than 2 nodes). In this scenario, the delay experienced by $F_i$ at the first $\alpha(i,j) - 1$ nodes is given by $(\alpha(i,j) - 1) 3\omega - \omega \times \alpha_1(i,j)$ and the delay experienced at the last common path is given by $\big\lceil \frac{T_i}{T_j} \big\rceil 3\omega$. Now combining the delays experienced at each node, we can compute the total delay as given in Equation (\ref{utilization_delayoneflow}). Note that, this is one scenario that leads to the maximum transmission conflict delay. However, no other scenario will result in a greater delay. To prove this, let us assume, without loss of generality, the first packet of $F_j$ interferes the packet from $F_i$ on $1$-st to $\eta$-th common path, where $0 < \eta < \alpha(i,j)$, and $\big\lceil \frac{T_i}{T_j} \big\rceil$ packets of $F_j$ interferes at the $\eta$-th common path and last packet interferes from the $\eta$-th common path to the last common path. In this scenario, first packet of $F_j$ cannot interfere after $\eta$-th and second packet of $F_j$ cannot interfere before $\eta$-th. Therefore, applying the same reasoning, a packet of $F_i$ experiences a transmission conflict delay of $\big\lceil \frac{T_i}{T_j} \big\rceil \times 3\omega$ time slot only at the $\eta$-th common path and at all other common paths it experiences a transmission conflict delay of $(\alpha(i,j) - 1) 3\omega - \omega \times \alpha_1(i,j)$ time slots. This delay does not change with the value of $\eta$. Consequently, the delay caused by a flow $F_j$ on a flow $F_i$ is given by the Equation (\ref{utilization_delayoneflow}). For a DM scheduler only high priority flows interfere a low priority flow. Therefore, the total transmission conflict delay for a low priority flow $F_i$ is given as $$\Delta_i^{DM} = \sum_{F_j \in hp(F_i)} \delta(i,j).$$ For an EDF scheduler, all flows interfere all other flows. Therefore, the total transmission conflict delay for a flow $F_i$ is given as $$\Delta_i^{EDF} = \sum_{j \in [1, n] \, \text{and} \, j \neq i} \delta(i,j).$$

We now use the transmission conflict delay result of source routing algorithm and extend it to compute the transmission conflict delay for a graph routing algorithm. In graph routing, we need to consider common paths generated due to multiple paths of the high priority flow. Let us assume the flow $F_i$ has  $\epsilon_i$ number of paths in the graph route, $\rho_i$ denotes one path of the graph route, $\beta(\rho_i,j)$ denotes the number common paths between $\rho_i$-th path of $F_i$ all paths of flow $F_j$ and $\beta_1(\rho_i,j)$ denotes the number of common paths  between $\rho_i$-th path of $F_i$ all paths of flow $F_j$ with only one common node. On a path $\rho_i$ of flow $F_i$, the worst-case transmission conflict delay can be experienced when a packet of $F_j$ interferes on $\beta_1(\rho_i,j)$ common paths between path $\rho_i$ and all paths of $F_j$. Considering that on a path, a packet is at most delayed for $3\omega$ (or $2\omega$ depending on common path length), we can extend transmission conflict delay computation from source routing to transmission conflict delay experienced by a packet on the path $\rho_i$th to be $$\zeta (\rho_i,j) = (\beta(\rho_i,j) + \bigg\lceil \frac{T_i}{T_j} \bigg\rceil - 1) 3\omega  - \beta_1(\rho_i,j) \omega.$$ The total transmission conflict delay caused by $F_j$ on a packet of $F_i$ (considering all paths of $F_i$) is given by  $$\delta(i,j) = \sum_{\rho_i = 1}^{\epsilon_i} \zeta (\rho_i,j).$$ For a DM scheduler, the total transmission conflict delay experienced by a flow $F_i$ is given by $$\Delta_i^{DM} = \sum_{F_j \in hp(F_i)} \delta(i,j).$$ Similarly, for an EDF scheduler, the total transmission conflict delay experienced by a flow $F_i$ is given by $$\Delta_i^{EDF} = \sum_{ j \in [1, n] \, \text{and} \, j \neq i} \delta(i,j).$$

\subsection{Adopting the Utilization Based Analysis through Hierarchical Networking}
Because we derived the above results considering at most $m$ concurrent transmissions in the network, we now propose a hierarchical network-based analysis where this constraint is relaxed for the global network.  Specifically, we consider the network as a collection of subnetworks, where each subnetwork has its own subnetwork manager. Each manager adopts the above result at the subnetwork level. A global network manager coordinates with the subnetwork managers to manage the entire network in a hierarchical fashion. Every subnetwork will involve a unique channel for every transmission in a time slot. Thus if there are $m' (\le m)$ channels used in a subnetwork, then there will be at most $m'$ concurrent transmissions in the subnetwork. Therefore, we can use the results of Eq. \ref{utilization_edfwsan} and Eq. \ref{utilization_dmwsan} in each subnetwork directly.  
 
An important technical challenge in our proposed hierarchical architecture is to deal with the interdependencies among the subnetworks.  For example, if the subnetwork manager of a subnetwork needs to create a local  TDMA schedule (i.e., for the links inside the subnetwork), it may need to wait for its neighboring subnetworks (or some of the neighboring subnetworks) to finish their schedule, this is because of the dependency created by a packet. A packet routed through multiple subnetworks should be scheduled in the earlier subnetworks first. Because feedback flows involve both upwards and downward communication in the WSAN, such dependencies can be cyclic. For example, let us consider $2$  packets $p$  and $q$ such that $p$  needs to be scheduled first in  subnetwork $C_1$ and then in subnetwork $C_2$, and  that   $q$  needs to be scheduled first in  subnetwork $C_2$ and then in subnetwork $C_1$.  In such a scenario, $C_1$  needs to create a schedule after $C_2$  creates, and $C_2$  needs to create a schedule after $C_1$  creates, thereby creating a cyclic dependency.  

Our proposed method to remove these dependencies is to assign sub-deadlines and release offsets for each flow among the subnetworks.  Specifically, for every flow $F_i$ that passes through a subnetwork   $C_j$, we assign a release offset $r_{i,j}$  and a sub-deadline $d_{i,j}$  in the subnetwork. The deadline of flow $F_i$ is equally divided into sub-deadlines $d_{i,j}$ for each subnetwork a flow passes through. The release offset $r_{i,j}$   is equal to the sub-deadline of $F_i$  in the subnetwork where it needs to be scheduled immediately before $C_j$. Thus, subnetwork  $C_j$  needs to schedule $F_i$  within the time window $[r_{i,j},  d_{i,j}]$, thereby requiring no knowledge of the schedule (for $F_i$)  in other subnetworks.

%% file: utilization/evaluation.tex
\section{Evaluation}
\label{utilization_sec:evaluation}

\subsection{Simulation Setup}
We evaluated our utilization-based schedulability analysis results through simulations on TOSSIM \cite{tossim}. We used the topology collected from a wireless sensor network testbed \cite{sha2015implementation} of 74 nodes. Fig. \ref{utilization_fig:nodePosition} shows the collected topology where black lines are transmission links and red dotted lines show interference. Each node is equipped with Chipcon CC2420 radio which is compliant with the IEEE 802.15.4 standard. We implemented a multi-channel TDMA medium access control (MAC) protocol with channel hopping and a network layer to support tree routing. Time in the network was divided into 10ms slots, and clocks were synchronized across the entire network using the Flooding Time Synchronization Protocol (FTSP) \cite{FTSP}. For the sake of simplicity, we used Dijkstra's shortest path algorithm to generate a directional tree with root as the base station. We assumed all links to be bi-directional. We used packet reception ratio (PRR) as a metric for generating the routes. PRR values used in the simulation are obtained from real experiments. Links with a PRR higher than 90\% are used to determine the topology of the network. 

\begin{figure}
	\centering
    \includegraphics[width=0.75\textwidth]{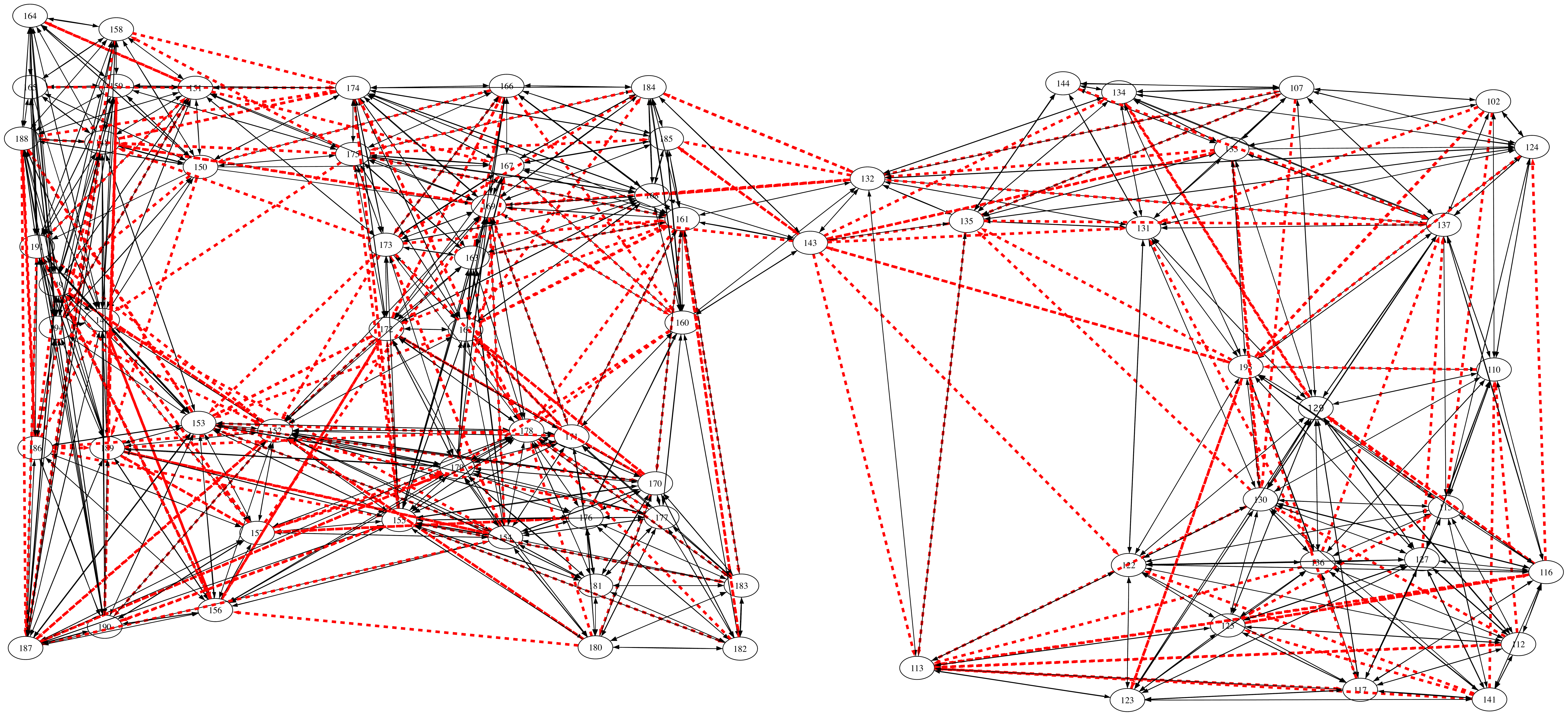}
    \caption{Network topology used in simulation}
    
    \label{utilization_fig:nodePosition}
\end{figure}

We used 5 IEEE 802.15.4 channels 12, 15, 17, 20, 21 for our simulations and the rest are assumed to be blacklisted. For the sake of simplicity, we assumed that spatial diversity of channels is not allowed, and at most one transmission is allowed on a channel.  We used an optimal channel assignment algorithm proposed in \cite{channelallocation} for channel assignment. We used either DM or EDF scheduling algorithm to allocate transmission time slots to each flow. We assigned 2 transmission time slots for each link on a flow. Second transmission slot is provided for redundancy and to account for transmission failures occurring due to channel noise. 

We evaluated our analysis in terms of {\slshape Schedulability ratio} defined as the fraction of the test cases that are deemed schedulable. We used 100 random test cases to obtain the schedulability ratio. We used the number of flows in the network as a parameter for comparison. We generate flows by randomly selecting sources and destinations, and simulate their schedules. One node with the highest degree of connectivity in the topology was selected as an access point. We assigned a random harmonic period in the range $2^{10\sim15}$ms. The deadlines are equal to periods. Priorities of the flows are assigned based on the DM policy. 

\subsection{Results} 
We analyze the effectiveness of our analysis by simulating the complete schedule of transmissions of all flows released within the hyper-period. In Fig. \ref{utilization_fig:resultDM} and Fig. \ref{utilization_fig:resultEDF}, ``Simulation" indicates the fraction of test cases that have no deadline misses in the simulations, and represents conservative upper bounds of schedulability ratios; ``Analytical" indicates the schedulability ratio based on our utilization-based schedulability analysis.

\begin{figure}
	\centering
    \includegraphics[width=0.35\textwidth]{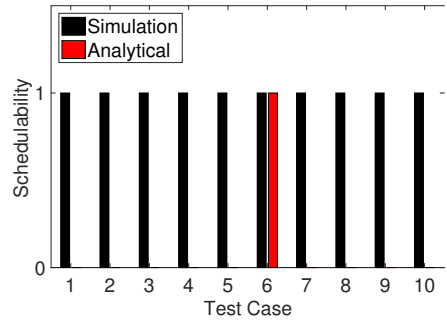}
    \caption{Schedulability for 10 test cases}
    
    \label{utilization_fig:resultDMBar}
\end{figure}

\begin{figure}
	\centering
    \includegraphics[width=0.35\textwidth]{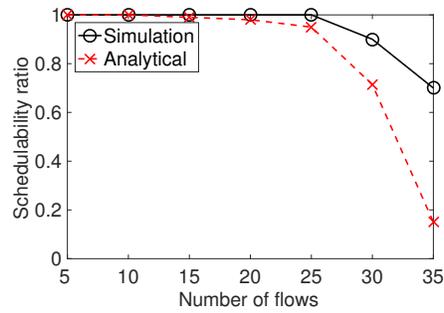}
    \caption{Schedulability ratio under DM scheduling}
    
    \label{utilization_fig:resultDM}
\end{figure}

Fig. \ref{utilization_fig:resultDM} shows the schedulability ratio for 100 test cases under varying number of flows for deadline monotonic scheduling algorithm. In our results, we observed that the analytical result shows a small decrease in schedulability ratio when compared to the simulation result. For 30 control loops, the simulation result shows that 73 test cases were schedulable and analytical results show that only 14 test cases were schedulable. Fig. \ref{utilization_fig:resultDMBar} shows 10 such test cases. We observed that every test case that was said to be schedulable by the analytical result was, in fact, schedulable in simulation. We observed that the early decrease in schedulability ratio for analytical analysis for deadline monotonic scheduling is due to the loose upper bound for schedulability ratio in a multi-processor environment.

\begin{figure}
	\centering
    \includegraphics[width=0.35\textwidth]{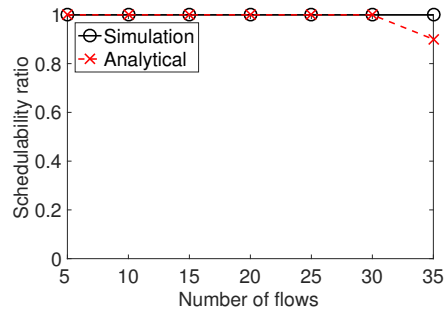}
    \caption{Schedulability ratio under EDF scheduling}
    
    \label{utilization_fig:resultEDF}
\end{figure}

Fig. \ref{utilization_fig:resultEDF} shows the schedulability ratio for 100 test cases with the earliest deadline first scheduling algorithm. For 30 control loops, we observed that the simulation results show a 100\% schedulability and analytical results show 90\% schedulability. Similar to DM scheduling, all test cases that were deemed schedulable under analytical analysis were schedulable under simulation. We observed that delay computation for tree routing gives a tight upper bound and the gap between the analytical schedulability ratio and that based on simulations stems from the pessimism in the multi-processor utilization bounds. We observed a slow decrease in schedulability ratio for analytical analysis for the EDF scheduler. This phenomenon was due to the tighter bounds for schedulability ratio in a multi-processor environment. We also observed that for some cases, the scheduling algorithm also improved the schedulability ratio.

%% file: online_period_selection/ICII19_DynamicPeriod.tex
\chapter{Online Period Selection for Wireless Control Systems}
\label{ch:period}
Recent advancements in Industrial Internet-of-Things and cyber-physical systems through the development of wireless standards like WirelessHART and ISA100 for real-time and reliable communication paved the way for a new industrial trend called {\slshape Industry 4.0}. Industry 4.0 proposes to improve production efficiency by employing smart factories. One approach to improve the production efficiency is to predict the external disturbances and adjust the sampling periods accordingly. In a wireless control system, an online adjustment of the sampling periods can decrease the energy consumption of nodes, especially when the system is within a stable state. Although adjusting the sampling period is beneficial, the stability of the system may be compromised due to external disturbances through an increase in sampling period while decreasing the sampling period can impact the real-time performance and energy of the wireless network. Existing work on online sampling period selection assumes that the controller computes the periods and schedules, and repeatedly broadcasts the new sampling periods and schedules to all the nodes. Such an approach is highly energy consuming and can impact control performance. In contrast, to handle online period selection, we propose an autonomous scheduler based on game theory, where each node in the network generates the schedules locally and without any communication with the others. Our approach can handle any changes in the period and link qualities locally and on the fly. We also propose a heuristic for online sampling period assignment, where each node predicts the state and adjust the period locally. Our evaluation on a case study shows that the proposed approach consumes at least $59\%$ less energy when compared to the state-of-the-art approach.

	\input{online_period_selection/Introduction}
	\input{online_period_selection/related_work}
	\input{online_period_selection/NetworkModel}

	\input{online_period_selection/AutonomousScheduling}
	\input{online_period_selection/PeriodSelection}   
	\input{online_period_selection/evaluation}

    \section{Summary}\label{online_sec:Conclusion}
   In summary, we have proposed an autonomous scheduler which maximizes the number of successful transmission in the network using game theory. In the proposed autonomous scheduler, each node avoids collisions of packets and transmissions under poor link quality.   We have also derived an efficient utilization based schedulability test to determine the schedulability of all flows on the network. We have then proposed a heuristic for sampling period acclimation in wireless control systems. Our evaluation on a case study showed that the proposed approach consumes at least $59\%$ less energy when compared to the state-of-the-art approach. Our autonomous scheduler and sampling period selection algorithm are proposed for a single hop network. In the future, we plan to extend these algorithms for multi-hop networks.

%% file: online_period_selection/Introduction.tex

\section{Introduction}\label{online_sec:introduction}
Industrial Internet-of-Things (IoT) and industrial Cyber-Physical Systems (CPS) evolved from wireless standards like WirelessHART \cite{WirelessHART2007_standard} and ISA100 \cite{isa100_standard} that facilitate low-power, flexible, and cost-efficient real-time communication for a broad range of applications like process control \cite{songRTAS11}, smart manufacturing \cite{RTSS10paper}, smart grid \cite{gungor2011smart}, and data center power management \cite{capnet, capnetRTSS}. 
These wireless standards facilitate closed loop real-time communication between sensors and actuators, where each sensor measures the state of a plant and delivers it to a controller through a wireless network. The controller generates control commands based on the measured state and then sends the control commands to each actuator through a wireless network.  To facilitate reliable and real-time communications in a shared wireless environment, industrial wireless standards employ a Time Division Multiple Access (TDMA) based Media Access Control (MAC) protocol with a high degree of redundancy.

Industrial IoT and CPS paved the way for a new industrial trend under the specification of Industry 4.0 \cite{industry4}. It is a new concept of Internet technologies for industries to improve production efficiency by employing smart factories \cite{iiot}. One approach to enhance the production efficiency of the plants is to predict the external disturbances and to adjust the sampling period accordingly. Such an online sampling period selection enables the controller to counteract the disturbances promptly, efficiently, and without affecting the stability. Furthermore, an online sampling period selection reduces the energy consumption of the sensors and actuators. 

Adjusting the sampling rate of a plant can have undesirable consequences in a wireless control system in which many control loops share the same wireless communication medium. Particularly, selecting a high sampling rate for one control loop may affect latency, real-time performance, and stability of the others. Selecting a low sampling rate can affect the stability of the plant. Therefore, sampling rates should balance real-time communication and control performance, requiring a control and communication co-design approach.

An online sampling rate assignment has to take into account the stability of each plant and the schedulability of each control loop in the network. Typically, existing work formulates this as a non-linear optimization problem with network schedulability and stability constraints \cite{li2018triggering, henriksson2015multiple, peng2016higher}. The non-linear optimization problem is solved at the controller. The controller then disseminates the sampling period of each plant and transmission schedule of all plants. Such an approach can lead to high energy consumption and disrupt the sensing and actuation messages \cite{distributedhart}. To address this issue, the work in \cite{ma2018efficient} placed restrictions on the increase in the sampling period to avoid re-dissemination of a new schedule. Although intuitive, this approach can consume more energy due to the unavailability of a period meeting the restrictions. Another approach to avoid re-dissemination of a schedule is to use any existing distributed scheduler \cite{distributedhart, song2018distributed, zhang2018fd, duquennoy2015orchestra}. Although these approaches can handle changes in the sampling period, they have the following limitations. First, they assume a simple link quality model, where the link quality at the current slot is independent of that observed in the past time slots. Second, frequent dissemination of the sampling period incurs energy overheads at the nodes.

To address the above limitations, we propose an autonomous period selection algorithm for wireless control systems. Our approach is based on selecting a maximum period to ensure a decrease in the Lyapunov function for stability. To handle online period selection, we also propose an autonomous scheduler for wireless control systems using game theory. The objective of the game is to maximize the number of successful transmissions by avoiding collision and transmissions during poor link quality. Our evaluation on a case study shows that the proposed approach consumes at least $59\%$ less energy compared to the holistic controller proposed in \cite{ma2018efficient}, which is the state-of-the-art approach. 

In summary, we make the following contributions in this chapter.
\begin{itemize}
    \item We formulate the autonomous TDMA scheduler as a non-cooperative game. The players choose a strategy that yields maximum payoff for all players, and hence, each node converges to a schedule without communicating with the others. The payoff of the game represents the number of successful transmissions, and the objective of the game is to maximize the payoff. We also provide a linear-time sufficient schedulability analysis for the proposed autonomous scheduler.
    \item We propose an online sampling period assignment algorithm, where each node autonomously selects the sampling period of all control loops by predicting the state of the system. Based on the current state, it predicts the sampling period required to ensure a non-increasing Lyapunov function for stability.
\end{itemize}

The rest of the chapter is organized as follows. Section \ref{online_sec:RelatedWorks} reviews related work. Section \ref{online_sec:sysModel}   describes  the system model. 
Section \ref{online_sec:scheduler} describes the autonomous scheduler. Section \ref{online_sec:periodSelection} describes the dynamic period selection. 
Section \ref{online_sec:evaluation} presents the evaluation and Section \ref{online_sec:Conclusion} summarizes the chapter.

%% file: online_period_selection/related_work.tex
\section{Related Work} \label{online_sec:RelatedWorks}

Real-time routing, real-time scheduling, and schedulability analysis are investigated in \cite{songRTAS11, modekurthy2018distributed, RTSS15paper, modekurthy2018utilization, palattella20146tisch, dujovne20146tisch, zimmerling2017adaptive, RTSS10paper, jin2017reliability, distributedhart, song2018distributed, zhang2018fd}. They assume a centralized or distributed scheduler, which consumes energy at the nodes. The work in \cite{duquennoy2015orchestra} proposes an autonomous scheduler for sparse traffic scenarios to address this issue. However, it assumes a fixed sampling period. In contrast, our work proposes an autonomous period selection and an autonomous scheduler based on game theory. Although there exist several works that use game theory for wireless networks, they do not consider real-time scheduling \cite{samian2015cooperation}.

Scheduling-control co-design for wired networked control systems was studied in  \cite{hscc2015, kishida2018event, liu2018robust, akashi2018self, zhang2017overview}. In contrast, we consider a wireless network. Wireless control co-design was investigated for single-hop networks in \cite{xu2013distributed, sh2} (also see survey \cite{park2018wireless}), and for multi-hop networks in \cite{multihopcodesign}. They assume each sensor samples the plant at fixed periods. Such static sampling is not a good approach for dynamic systems. The work in \cite{xenofon} proposes a self-triggered control with network schedulability constraints based on a virtual link capacity margin. It uses a simplistic network model where it does not avoid collisions in the network. In contrast, our autonomous scheduler avoids collisions.

The work in \cite{li2018triggering, henriksson2015multiple} proposes stability under a self-triggered control with a centralized scheduler. A centralized schedule has to broadcast the entire schedule after every change in period. Thus, a centralized scheduler usually causes 1) high energy overheads and 2) frequent disruptions to the control operation.   To address this limitation, the work in \cite{ma2018efficient} proposes a self-triggered control where the periods only increase in integer multiples of the initial sampling period. In cases where the actual increase in the period needed for stability is small, the approach in \cite{ma2018efficient} does not change the period, thereby making no improvement in energy consumption. It does not decrease the sampling period when the network is underutilized.

All the approaches discussed above use a centralized period selection algorithm, which still requires frequent broadcast of sampling periods. In contrast, we propose an autonomous period selection algorithm where each node selects a period locally and without communicating with others. Furthermore, we propose an autonomous scheduler where nodes can adapt to the sampling period changes without any energy overheads.

%% file: online_period_selection/NetworkModel.tex
\section{System Model}  \label{online_sec:sysModel}
\subsection{Network Model}
In this chapter, we consider an industrial wireless sensor-actuator network (WSAN), which consists of field devices, an access point, and a gateway. The {\slshape field devices} are wirelessly networked sensors and actuators. Each node contains a {\slshape half-duplex} omnidirectional radio transceiver that can receive from at most one sender at a time.  Furthermore, a node cannot both transmit and receive at the same time. {\slshape Access point} provides a path between the wireless network and the gateway. The {\slshape network manager} and the controller remain at the {\slshape gateway}. The network employs feedback {\slshape control loops} between sensors and actuators. Sensors measure process variables and deliver them to the controller through the network. The controller sends control commands to the actuators,   which then operate the control components to adjust physical processes. In this chapter, we assume that the role of the network manager is to create network topology, observe the network operation, and determine offline admission control for new control loops/flows. Note that the network manager does not generate a schedule. 

We consider a WSAN where all sensors/actuators directly communicate with the access point. That is, sensors, actuators, and the access point form a star topology. Many existing WirelessHART and ISA100 deployments use such single-hop topology for technical simplicity. Furthermore, recent low-power wide-area network technologies like LoRa  \cite{petajajarvi2015coverage}, SigFox \cite{centenaro2016long} and SNOW \cite{ton_snow, snow_implementation, saifullah2017snow2} all adopt a single-hop network topology.
Transmissions in the network are scheduled based on a single channel TDMA protocol. Time in the network is globally synchronized and slotted. A transmission and its acknowledgment happen in one slot. Transmissions are scheduled based on the link quality. We use the signal to noise ratio (SNR) to determine the link quality.

\subsection{Control System Model}
\label{online_sec:controlmodel}
We consider there are $n$ real-time control loops in the system denoted by $F = \{F_1, F_2, . . . , F_n\}$. For the rest of the chapter, we use the terms flow and control loop interchangeably.  Each flow corresponds to a linear time-invariant (LTI) control system, which is expressed in its discrete state space expression, as shown in Equation (\ref{online_eq:statespace}).
\begin{equation} \label{online_eq:statespace} \begin{split} x_i(t+1) =& A_i x_i(t) + B_i u_i(t) + v_i(t) \\ y_i(t) =& C_i x_i(t) \end{split} \end{equation}
In Equation (\ref{online_eq:statespace}), $x_i(t)$ represents the state of control system $i$ at $t$-th  time slot , $u_i(t)$ represents the control input, $y_i(t)$ represents the observed output, and $v_i(t)$ represents the discrete time white Gaussian noise. We assume the pair ($A_i$, $B_i$) is controllable, and the pair ($A_i$, $C_i$) is observable. 

The period and deadline of the flow $F_i$ are represented by $T_i$ and $D_i$, respectively. We consider an implicit deadline system, where the deadlines are equal to the periods, i.e., $D_i = T_i$. Since we consider the sampling period changes dynamically with the system, we also consider that the deadline of the packet also changes dynamically. Nevertheless, the relative deadline of the packet is the same as the current sampling period of the plant. For the rest of the chapter, the periods and deadlines are specified in number of time slots.

For each plant, we use a model predictive controller with a cost function given by Equation (\ref{online_eq:cost}), and a final state $x_f \equiv 0$. 
\begin{equation} \begin{aligned} \label{online_eq:cost} J_i(t) = &\frac{1}{2} x_i(t + N)^T P_i^f x_i (t + N) + \\& \sum_{z= t}^{t + N} \frac{1}{2} \big( x_i^T(z) P_i x_i(z) +  u_i^T(z) Q_i u_i(z) \big) \end{aligned} \end{equation}
The cost function is the sum of the weighted average of the state and control input for every time slot in the duration of horizon ($N$), where $P_i$ and $Q_i$ represent the weight for the state and control input, respectively. In Equation (\ref{online_eq:cost}), $P_i^f$ is a positive definite matrix that satisfies the Lyapunov equation $A_i^t P_i^f A_i + Q_i^f = 0$, where $Q_i^f$ is also a positive definite matrix. The cost function used in Equation (\ref{online_eq:cost}) ensures closed loop asymptotic stability of the control loop $i$ \cite{mayne2000constrained}. Note that the value of $u_i(t)$ changes only after $T_i$ time slots and is constant in between.

We assume each control loop $i$ has a \textit{weight} $w_i$. The weight of control determines its importance in ensuring stability. The total cost (inclusive of all control loops) is given by Equation (\ref{online_eq:totalCost}).
\begin{equation}
\label{online_eq:totalCost}
J(t) = \sum_{i \in [i,n]} w_i J_i(t)
\end{equation}

Given the state space equations and control cost for all plants, we propose an online sampling period assignment that changes dynamically with the state of the system. To handle online sampling period assignment, we first propose an autonomous scheduler that does not incur energy overhead.

%% file: online_period_selection/AutonomousScheduling.tex
\section{Link Quality Based Autonomous Scheduler} \label{online_sec:scheduler}
Existing centralized and distributed schedulers for industrial standards like WirelessHART and ISA100 use real-time scheduling policies like EDF (earliest deadline first) and DM (deadline monotonic). However, these schedulers typically assume a simple link quality model where the packet reception rate (PRR) is constant at each time slot. Such a model is not applicable in practice since link quality varies with time, and link quality at a current time slot depends on that of the previous time slots. In this chapter, we consider a general link quality model where link quality varies with time. We use a small neural network to predict the link quality.

For the general link quality model, we develop a new autonomous scheduler where each node determines when to transmit a packet based on the deadline of the packet and link quality. Designing such an autonomous scheduler where all packets meet their deadlines is challenging since all nodes in the network have to compute the best time to transmit all packets, and agree on the schedule without communication. Another important challenge is that the best transmission time slot depends on the dynamically changing link quality.

We model the scheduling problem as a non-cooperative game where each node in the network chooses the Nash equilibrium strategy autonomously and without any coordination. The objective of each node playing the game is to maximize the number of successful transmissions by avoiding deadline misses, collisions, and transmissions on poor links. This formulation is flexible with changing periods while ensuring all nodes converge to the same schedule without communicating with other nodes in the network. Typically, a game can have multiple Nash equilibrium strategies. In the proposed approach, we can break the ties by choosing the strategy which schedules packets with the shortest absolute deadline first. Thus, the game-theoretic autonomous scheduler converges to a solution without any communication. 

The proposed game-theoretic scheduler requires link quality information of all links and period/deadlines of all flows to determine the Nash equilibrium strategy (or transmission schedule). In this section, we assume that all nodes are aware of the sampling rate for all flows. We discuss the autonomous sampling rate selection in Section~\ref{online_sec:periodSelection}. In this section, we first discuss the link quality prediction. We then discuss our game formulation and the existence of Nash equilibrium strategies for a single-slot-game. We then discuss how the nodes compute Nash equilibrium strategies in a multi-slot game. Finally, we present the schedulability analysis for the proposed network scheduler. 

\subsection{Link Quality Prediction}
We use a data-driven prediction using neural networks (ANN), similar to the approach in~\cite{liu2014data}. A node uses ANN to predict the link quality of $\eta$ subsequent time slots between itself and a receiver. In this section, we consider $\eta$ is a design parameter. The choice of an ANN approach arises from the fact that the ANN approach can estimate link quality with greater accuracy~\cite{liu2014data}. The experimental evaluation performed in~\cite{liu2014data} using training information from $20000$ packets show that the average prediction error for a neural network approach is $0.18$. Note that the work in~\cite{liu2014data} only compares the performance of different approaches but does not develop an autonomous scheduler or autonomous period selection. 

The neural network takes the signal strength at the beginning of the time slot as the input. It generates the SNR values of packet transmission and of its acknowledgment for $\eta$ subsequent time slots as output. A node considers the link quality at a time slot to be good if the SNR for packet transmission and its acknowledgment are above a pre-determined threshold, and bad otherwise. Note that, the neural network generates the link qualities for $\eta$ time slots of one link. In a single-hop network, most of the links are positively or negatively correlated to one another~\cite{srinivasan2010k}. The nodes use $\kappa$-factor~\cite{srinivasan2010k}, which is cross-correlation index that depends on channel and power levels, to estimate the link quality of $\eta$ time slots for all links in the network. Thus, the proposed approach is not computation-intensive and can execute on a low power embedded device.

We model the ANN as a three-layer neural network with one input layer, one hidden layer, and one output layer.  The input layer consists of 1 neuron. The output layer consists of $2 \eta$ neurons (one to estimate the SNR of packet reception and another to estimate the SNR of the acknowledgment). The hidden layer consists of $10\eta$ neurons. Note that the proposed neural network setup is similar to that used in~\cite{liu2014data}. We only increase the number of nodes in the hidden layer and the output layer to predict link quality information for $\eta$ time slots. 

To train the ANN, we use 1) signal strength observed on a channel before making a transmission, 2) SNR obtained at the receiver for each packet reception, and 3) SNR of the acknowledgment packet. Signal strength value at the beginning of the time slot is used as the input to the neural network.  SNR values obtained at the receiver and SNR values of packet acknowledgment are used as reference outputs during training. The network manager computes and then broadcasts the weights of the neural network (after the training phase) and the correlation between links to all nodes. 

During the data collection phase, a node collects the signal strength information for a short duration at the beginning of a time slot. It then makes $\eta$ consecutive transmissions to the receiver. Each transmission is followed by an acknowledgment by the receiver. Nodes in the network use a round-robin approach to avoid interference from collisions. Since the autonomous scheduler is designed to avoid collisions, the SNR values from simultaneous transmissions are not required during data collection. The nodes use packet transmissions or piggybacked acknowledgments to send signal strength and the SNR information to the base station. 

In this link quality prediction model, computation-intensive operations take place on the controller, and hence, nodes in the network do not incur significant computation overheads. Furthermore, the controller can train the ANN during design time, thereby not interfering with the control operation.

\subsection{Game Formulation for Autonomous Scheduler}
We formulate the game similar to the prisoner's dilemma where we assume all nodes to be non-cooperative. That is, nodes do not communicate with each other to generate a schedule. They use local information to come up with a strategy to maximize their payoff.  In our game, we consider each node with a packet as a player of the game. 

In a time slot, each node has two strategies: 1) \textit{to transmit} and 2) \textit{not to transmit}. The payoff obtained at each time slot is defined in Equation~(\ref{online_eq:payoff}). The objective of the game is to maximize the payoff, which corresponds to the number of successful transmissions. The payoff function rewards a successful transmission and penalizes the wastage of energy due to transmission failures in each time slot. It also penalizes a deadline miss for each packet. If a node chooses \textit{not to transmit} a packet, then the payoff is $0$. If it chooses \textit{to transmit} a packet and the packet transmission is successful, the payoff is $1$. If it chooses \textit{to transmit} a packet and the packet transmission is not successful, the payoff is $-1$. A packet transmission may fail due to two reasons: 1) collision from simultaneous transmissions or 2) transmission on a poor link (low SNR).
\begin{singlespace}
{
\begin{equation}
	\small
	\label{online_eq:payoff}
    \textit{payoff} = 
\begin{cases}
    1,& \text{ if a packet transmission is successful} \\
    0,& \text{if node decides not to transmit a packet}\\
    -1, & \text{if a packet transmission is not successful} \\
    -2, & \text{if a packet misses its deadline}         
\end{cases}
\end{equation}
}
\end{singlespace}
Given the link quality information and sampling period for all control loops, we present the Nash equilibrium strategy for a single-slot-game. We then extend the analysis to a multi-slot-game and show how a node computes its strategy.

\begin{figure}[t]
    \centering
    \begin{subfigure}[b]{0.3\textwidth}
        \includegraphics[width=\textwidth]{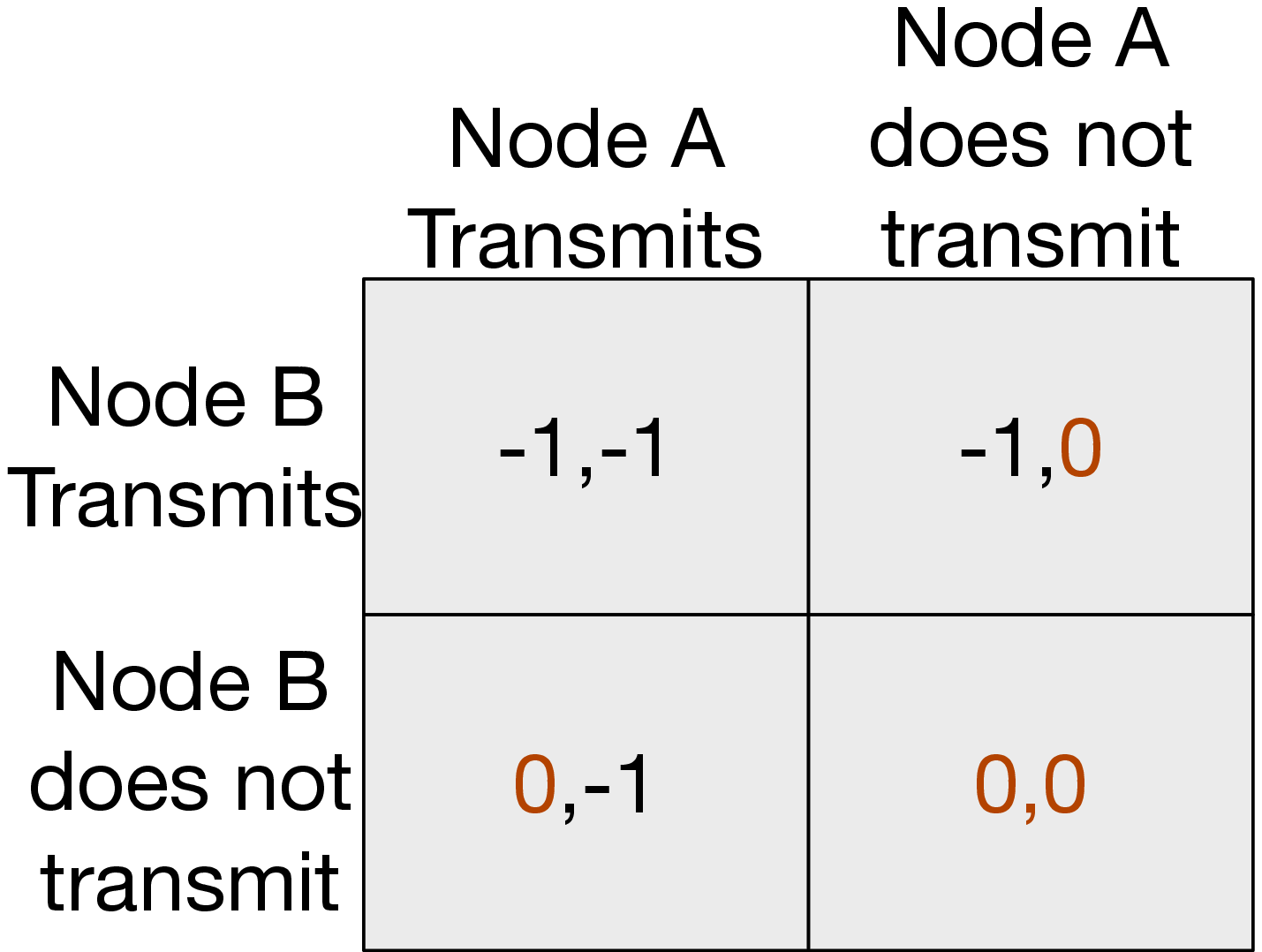}
        \caption{Link qualities for both nodes are bad}
        \label{online_fig:NoNodeHasPacket}
    \end{subfigure}
    \quad
    \begin{subfigure}[b]{0.3\textwidth}
        \includegraphics[width=\textwidth]{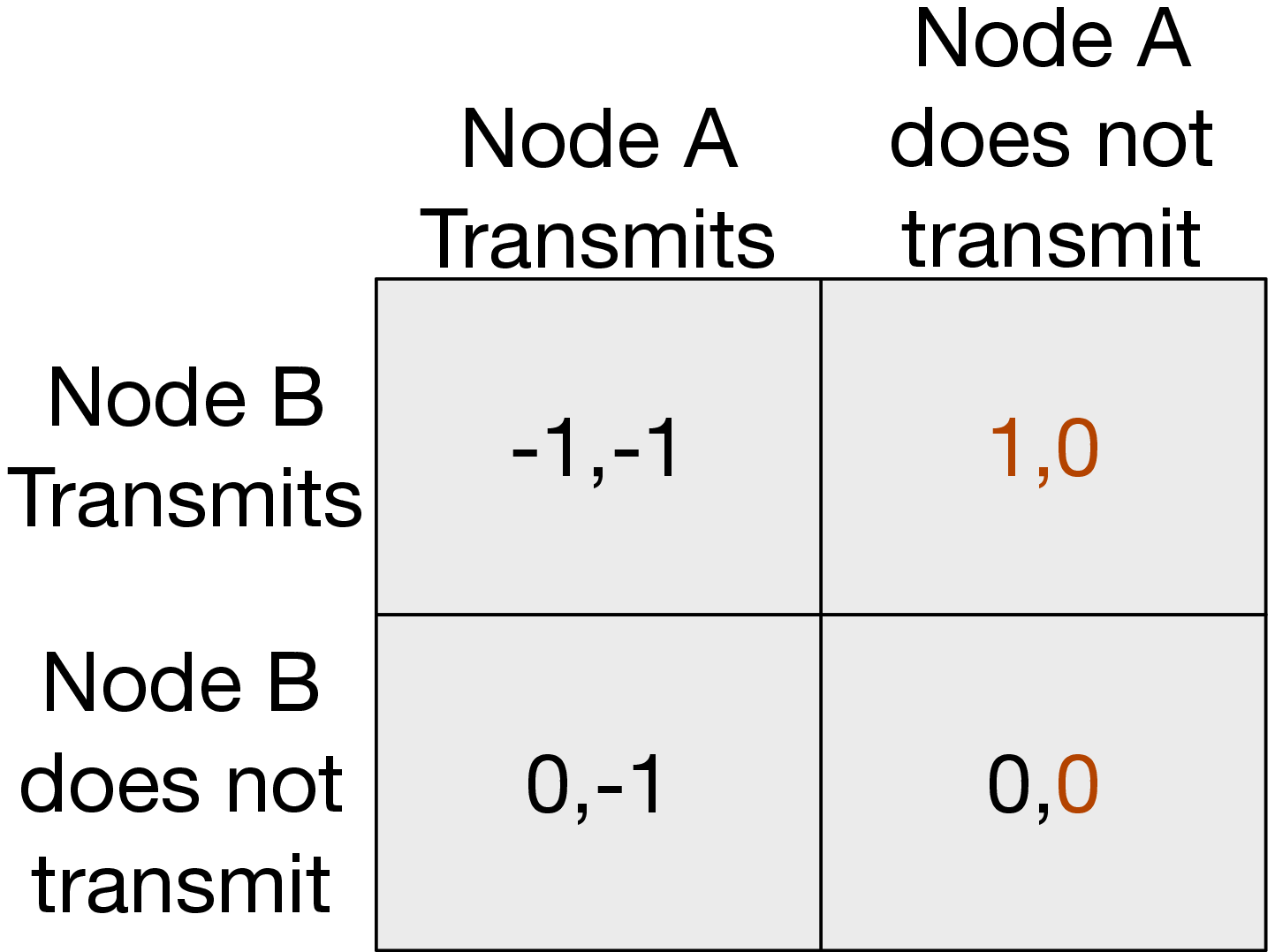}
        \caption{Link quality for node A is bad}
        \label{online_fig:OneNodeHasPacket}
    \end{subfigure}
     \quad
    \begin{subfigure}[b]{0.3\textwidth}
        \includegraphics[width=\textwidth]{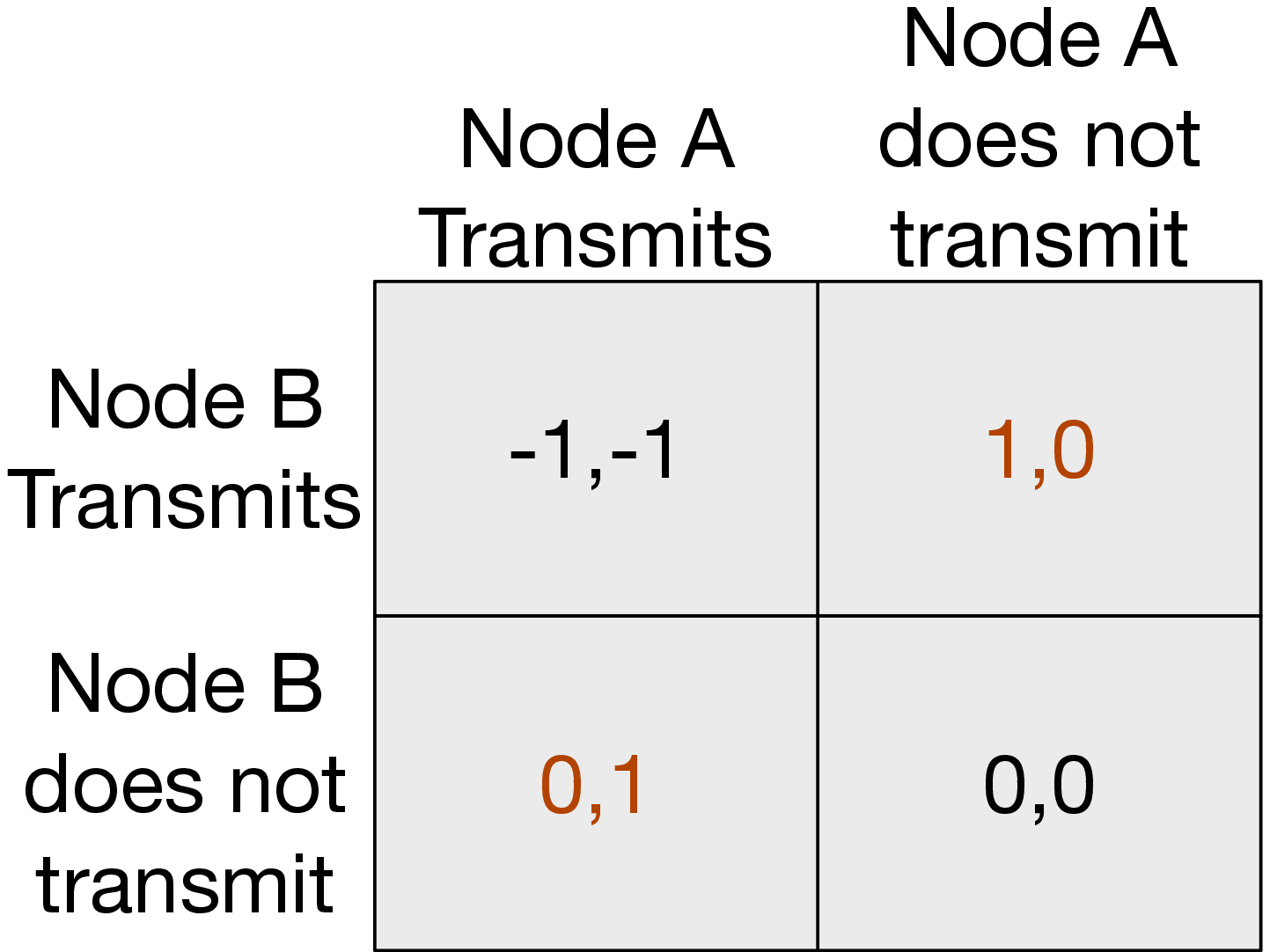}
        \caption{Link qualities for both nodes are good}
        \label{online_fig:BothNodesHavePacket}
    \end{subfigure}
   \caption{Payoff Values for Each Strategy in a Game with Two Players.}
   \label{online_fig:payoff_packet}
\end{figure}

\subsection{Single Time Slot Game}
To better explain the payoff function and Nash Equilibrium for a single time slot game, we provide a simple example with two nodes: node A and node B. A similar approach can be used to compute Nash equilibrium when the number of nodes is greater than 2. Based on the link quality, we can consider three different scenarios. 1) Link quality for node A (the bi-directional link between the access point and node A) and link quality for node B (the bi-directional link between the access point and node B) are bad. 2) Link quality for node A is bad, but link quality for node B is good. 3) Link quality for node A and link quality for node B are good. 

For the first scenario, 
Fig.~\ref{online_fig:NoNodeHasPacket} shows the payoff values obtained for each strategy played by node A and B. If node A chooses \text{to transmit}, the destination does not receive the packet due to poor link quality, and hence, it receives a payoff of $-1$. However, if node A chooses \text{not to transmit}, it receives a payoff of $0$. Thus, the best strategy for node A is \textit{not to transmit} a packet. Similarly, the best strategy for node B is \textit{not to transmit} a packet. Thus, the Nash Equilibrium strategy in this case for both nodes is \textit{not to transmit} a packet. 

For the second scenario, when link quality for A is bad,  Fig.~\ref{online_fig:OneNodeHasPacket} shows the payoff values obtained for each strategy. As shown in Fig.~\ref{online_fig:OneNodeHasPacket}, the best strategy for node A is \textit{not to transmit} a packet, and hence obtain a payoff of $0$. Since node A always chooses \textit{not to transmit} a packet, the best strategy for node B is \text{to transmit} a packet and receive a payoff of $1$. The Nash equilibrium strategy, in this case, is for node B \textit{to transmit} a packet and for node A is \textit{not to transmit} a packet. 

For the third scenario, when link quality for node A and B is good, Fig.~\ref{online_fig:BothNodesHavePacket} shows the payoff values obtained for each strategy.  As shown in Fig.~\ref{online_fig:BothNodesHavePacket}, when both nodes choose \textit{to transmit}, it results in a collision, and hence, node A and B receive a payoff of $-1$. Similarly, when both nodes choose \textit{not to transmit}, it results in a wastage of time slot for both nodes, and hence, they receive a payoff of $0$. When node B chooses \textit{to transmit} and node A chooses \textit{not to transmit}, node B can make a successful transmission. Here, node B gets a payoff of $1$, and node A gets a payoff of $0$ (since its time slot is wasted). Thus, the best strategy of node B is \textit{to transmit} when node A selects \textit{not to transmit}, and the best strategy of node B is \textit{not to transmit} when node A selects \textit{to transmit}.  Thus, there are two Nash equilibria for this scenario. A tie can be broken here based on the shortest absolute deadline. 
The node with the smallest absolute deadline chooses \textit{to transmit} a packet while the other chooses \textit{not to transmit} a packet.

To explain the effect of deadlines and link qualities on the Nash equilibrium strategies, we consider a fourth scenario. In this scenario, we consider node B's link quality is good, and its packet deadline is 2. We consider node A's link quality is bad, and its deadline is 1. Node A has to transmit in the first slot to meet the deadline. In this situation, if node A chooses to transmit a packet, there is a chance the packet transmission is not successful, and it receives a payoff of $-3$ ($-2$ for a deadline miss and $-1$ for a unsucessful transmission). If the packet transmission is successful, then node A gets a payoff of $1$. However, if node A decides not to transmit a packet, it gets a payoff of $-2$. In this case, node A chooses \textit{to transmit} and node B chooses \textit{not to transmit}. Since a node receives a payoff of $-2$ for every deadline miss, it tries to make at least one transmission for every packet. 

In a two player game where both nodes have the same deadlines, there are scenarios where there are two Nash equilibrium strategies. In these scenarios, a mixed strategy Nash equilibrium analysis shows that the probability of switching between the 2 Nash equilibrium strategies is $0.5$. We can extend this result to a $n$ player game where the probability of switching between the $n$ Nash equilibrium strategies is $\frac{1}{n}$. An equal probability for all nodes shows that the choice of the payoff values results in a fair usage of network resources between all the players.

\subsection{Multiple Time Slot Game}
In our autonomous scheduler, each node can predict the link quality of $\eta$ time slots. Therefore, each node plays the game for $\eta$ time slots. The total payoff is computed as the sum of payoffs received in each time slot. For the multi-slot-game, we assume that each node plays a sub-game perfect equilibrium strategy (which is a Nash equilibrium strategy in every time slot). In case of multiple sub-game perfect equilibrium strategy, a tie can be broken here based on the shortest absolute deadline. In a multi-slot-game, we assume that each flow requires a total of $\gamma$ (where $\gamma \ge 2$) time slots for both uplink and downlink communications. Of the available $\gamma$ slots, two are used for transmissions and the rest are used for re-transmissions in case of failures.

In our game, the subgame perfect equilibrium for a $\eta$ time slot game depends on the link quality conditions and deadlines for each packet. Due to a large number of scenarios (considering different link qualities and deadlines), it can be difficult to theoretically analyze all possible scenarios and come up with a single algorithm that can always generate a Nash equilibrium strategy. Therefore, we propose to use the insights from the single-slot-game to reduce the number of strategies explored.

\subsection{Schedulability Analysis}
This section presents the schedulability analysis for the proposed autonomous scheduler. The schedulability analysis determines if the selection of a set of periods will result in all packets meeting their deadlines or not. 

\begin{thm}
Under the proposed autonomous scheduler, the set of flows $F = \{F_i, F_2 \cdots F_n\}$ with periods (equal to deadlines) $T_1$, $T_2$, $\dots$ $T_n$ can meet their deadlines given $p_iT_i$ represents the maximum number of poor link qualities experienced by a flow $F_i$  when Equation~(\ref{online_eq:test}) is satisfied, where $\mu_i = \frac{\gamma}{T_i (1 - p_i)}$. \begin{equation}\label{online_eq:test} \sum_{i = i}^{n} \mu_i \le 1\end{equation}
\begin{proof}
We first show the schedulability analysis when the link quality is good for all links and all slots. We then model the time slots with poor link qualities as blocking time in non-preemptive uni-processor scheduling. We use the non-preemptive schedulability test for the proposed autonomous scheduler.

When the link quality is good for all time slots, each node selects a unique transmission time slot for each packet of a flow based on the shortest absolute deadline (since all link qualities are good). This schedule is the same as the EDF scheduler. Since preemptive uni-processor schedulers consider that a task preempts at the boundaries of the time unit and not in between, the autonomous scheduler (when the link quality is good) maps to a preemptive scheduler. The schedulability test for an implicit deadline task set under preemptive EDF scheduler is $\sum_{i = 1}^{n} \mu_i^G \le 1$, where $\mu_i^G = \frac{c_i}{T_i} = \frac{\gamma}{T_i}$. 

In the proposed autonomous scheduler, a transmission happens only on a good link quality to avoid wastage of energy. When the link quality for a flow with a shorter absolute deadline is bad, then its transmission is blocked until the link quality is good. Meanwhile, a flow with larger absolute deadline with a good link quality can make a transmission.  A similar blocking can be observed in a non-preemptive scheduler where a task with a shorter absolute deadline is blocked by the execution of a larger absolute deadline. Thus, we can map the proposed autonomous scheduler to a non-preemptive EDF scheduler. The schedulability test for non-preemptive uni-processor scheduler is given by $\sum_{i=1}^{n} \mu_i \le 1$, where $\delta$ is the maximum blocking time, and $$\mu_i = \frac{C_i}{T_i - \delta} = \frac{\gamma}{T_i - p_iT_i} = \frac{\gamma}{T_i(1 - p_i)}$$
Thus, the schedulability test for the proposed autonomous scheduler under the assumption that $p_iT_i$ represents the maximum number of poor link qualities experienced by a flow $F_i$ is given by Equation~(\ref{online_eq:test}).
\end{proof}
\end{thm}

%% file: online_period_selection/PeriodSelection.tex
\section{Online Sampling Period Selection} \label{online_sec:periodSelection}
Here, we first formulate the centralized version of an online sampling period selection problem. We then present the autonomous sampling period acclimation algorithm, where each node selects the sampling period without communicating with others.

\subsection{Formulation of a Centralized Online Sampling Period Selection Problem}
In this chapter, we change the sampling period after every $\tau$ time slots instead of changing the sampling period after every interval to ensure low overhead.
Given the state space equations for each flow and schedulability condition, the online sampling period assignment problem is formulated in Problem (\ref{online_eq:dynamicperiodassignment}).  To capture the stability of all control loops, we use the controller cost function given by Equation (\ref{online_eq:totalCost}) as the objective.
We use the schedulability analysis given by Equation (\ref{online_eq:test}) as a constraint. Note that the period of flow $F_i$ changes after every $\tau$ time slots but remains constant within an interval $(k\tau, k\tau + \tau)$, where $k = 1, 2, \cdots \infty$. Thus, any period selected by Problem (\ref{online_eq:dynamicperiodassignment}) results in a schedulable network scenario for the interval of length $\tau$. 
\begin{equation}
\label{online_eq:dynamicperiodassignment}
\begin{aligned}
& {\arg\min_{T_i \forall i}}
& & J(\tilde{t}) \\
& \text{subject to} 
& & \sum_{i = 1}^{n} \mu_i \le 1
\end{aligned}
\end{equation}

Due to the high complexity of the objective function, finding a solution to Problem (\ref{online_eq:dynamicperiodassignment}) can take a long time. Furthermore, the controller must disseminate the sampling periods obtained by solving Problem \ref{online_eq:dynamicperiodassignment} to all the nodes in the network, which causes high energy overhead at the nodes. To address this problem, we propose a sampling period acclimation heuristic, which is presented in the following sections. 

\subsection{Sampling Period Acclimation}
This section presents an approach for sampling period acclimation, where the sampling periods of all plants are varied every $\tau$ time slots. Algorithm \ref{online_alg:overall} presents the sampling period acclimation algorithm. In the proposed approach, the central manager initializes all flows with a sampling period. The initial period assignment has to take into account the schedulability and stability of the plants. Since an optimization function considering both constraints can be computation intensive for a low power embedded device, we choose to execute the initial sampling period assignment at the controller. The initial sampling period assignment can act as a fail-safe assignment when all the flows experience high disturbance. Note that the controller executes the initial sampling period assignment only once (i.e., during the system initialization).

\begin{algorithm}[h]
\footnotesize

	\caption{Sampling Period Acclimation Algorithm} \label{online_alg:overall}
    \SetAlgoLined 
    \SetKwInOut{Input}{Input}
    \SetKwInOut{Output}{Output}
    \DontPrintSemicolon
    \Input{$ T_i[initial]$ = Initial-Sampling-Period-Assignment()}
    \Output{Sampling Period $T_i \, \forall i \text{ in } [0,n]$}
   \vspace{0.025in}
   $x_i(t)$ = State-Estimation() \\
   $x_i(t + \tau)$ = State-Estimation() \\
   $V_i(t) = x_i^T(t)P^L_ix_i(t)$ \\
   $V_i(t+\tau) = x_i^T(t + \tau)P^L_ix_i(t + \tau)$ \\ \vspace{0.025in}
   
   $Inc$, $Dec$ = Segregate-Flows()
   
   \For{i in $Inc$} {
    $T_i$ = Increase-Sampling-Period()
   }
   
   \For{i in $Dec$} {
    $T_i$ = Decrease-Sampling-Period()
   }

\end{algorithm}
\vspace{-0.1in}

After every $\tau$ time slots, the sampling periods for all plants are re-computed by each node autonomously. A node first estimates the current states and future states after $\tau$ time slots (with the current sampling period) of all plants in the network. Based on the current states and the estimated future states, the plants and their respective flows are segregated based on whether their sampling periods increase or decrease. Each sensor first computes the sampling period of all plants with an increase in their sampling period. Each sensor then computes the sampling period of all plants with a decrease in their sampling period. The intuition behind computing increase and decrease in the sampling period separately is to reduce the overhead of computation. An increase in the sampling period may impact the stability of a system. However, it does not impact the schedulability of the flows on the network. Similarly, a decrease in the sampling period impacts the schedulability of the flows, but it may improve the stability of the plant. Thus, splitting the computation of sampling periods can reduce the overhead of calculations on each sensor.   The initial sampling period assignment, state estimations, segregation of plants, and the approaches for increasing/decreasing sampling period are described in detail in the following sections.

\subsection{Initial Sampling Period Assignment}
The initial sampling period assignment must ensure the stability of plants and schedulability of their respective flows on the network. Stability of a plant can be maintained when the sampling period is within the bounded range of maximum ($T_i[\max]$) and minimum ($T_i[\min]$) sampling periods. Schedulability of all flows can be ensured when the total utilization is no greater than $1$, i.e., $ \sum_i^{N} \mu_i \le 1$. We use these constraints in our heuristic to assign initial sampling periods.

Algorithm \ref{online_alg:initial} illustrates our heuristic for sampling period assignment. It takes as input the minimum and maximum sampling periods for each flow that ensures acceptable control performance. It starts by assigning the minimum sampling period to all flows. If all flows in the network are schedulable with the selected sampling periods, it stops.  If a flow is not schedulable, it iterates over the ordered set of flows (decreasing order of the weighted individual control cost). The maximum sampling period $T_i[\max]$ is selected as the new sampling period for a flow. The algorithm stops when a sampling period assignment leads to a schedulable network. In some rare cases, the algorithm can stop when all flows are assigned their maximum sampling periods, and the network is still unschedulable. In this case, the plants are not stabilizable as there is no feasible sampling period assignment with acceptable control performance. Once the algorithm stops, the central manager broadcasts the sampling period to all nodes.

\begin{algorithm}[ht]
\footnotesize

	\caption{Initial Sampling Period Assignment} \label{online_alg:initial}
    \SetAlgoLined 
    \SetKwInOut{Input}{Input}
    \SetKwInOut{Output}{Output}
    \DontPrintSemicolon
    \Input{$ T_i[\min] \, \& \, T_i[\max] \, \forall i \text{ in } [0,n]$}
    \Output{Sampling Period $T_i \, \forall i \text{ in } [0,n]$}
   
   \For{i in [0,n]} {
   	$T_i = T_i[\min]$ \\
	$\mu_i = \frac{\gamma}{T_i}$
   }
   $i = n$ \\
   \While{$\sum_{i = 1}^{n} \mu_i > 1$} {
   	$T_i = T_i[\max]$ \\
	$\mu_i = \frac{\gamma}{T_i}$ \\
	$ i = i -1$
    }    
\end{algorithm}
 \vspace{-0.1in}


\subsection{State Estimation and Correction}
\label{online_sec:stateEstimation}
State estimation has been extensively studied in control theory \cite{simon2006optimal}. Kalman filtering is a popular approach for state estimation, which can predict the state of the plant with high accuracy \cite{simon2006optimal}. Kalman filter uses the past state estimation error to accurately predict the current and future state of the plant. Such an algorithm requires the state space representation, current state of the system, control input applied at each time slot, and the noise experienced at each sample to compute the error in state estimate. Since Kalman filtering is mainly used by the controller to compensate for packet delay or loss, the assumption about the availability of the aforementioned data is valid. However, in the proposed approach, a sensor estimates the state of all plants, which does not have all the required data.  A sensor can obtain the initial state of each plant and the state space representation during the network initialization phase. Nevertheless, a sensor does not have the control input applied, and the noise experienced at each time slot. To address this issue, we propose the following approach to estimate the state of the plant.

Each sensor assumes that all plants use a linear controller given by $u_i(t) = k_ix_i(t)$, where $k_i$ is a constant value. The value of $k_i$ satisfies the Lyapunov stability constraint \begin{equation} \label{online_eq:lyconstraint}(A _i+ B_iK_i)^T P_i^L (A_i + B_iK_i)  - P_i^{L} + Q_i^{L} = 0\end{equation} where, $P_i^{L}$ and $Q_i^{L}$ are positive definite matrices.
A sensor can obtain the values of $K_i$ for all flows, during network initialization. Thus, a sensor can predict the control input applied during the previous step. We use the variance of white Gaussian noise as the past noise in the system. Since the variance is fixed, all nodes use the same value of noise to estimate the state of the plant.  

In the proposed approach, the noise estimation in a plant is assumed to be constant. However, in practice, the noise experienced by each plant can be random. The random nature of noise can lead to errors in state estimation.  A sensor monitoring the plant associated with flow $F_i$ can compute the state estimation error for the plant based on Equation (\ref{online_eq:stateesterror}), where $x_i^{\text{obt}}(t)$ is the measured valued of the state of plant associated with flow $F_i$, and $x_i^{\text{est}}(t)$ is the estimated state of plant associated with flow $F_i$. 
\begin{equation} 
\label{online_eq:stateesterror}
	\xi (t) = x_i^{\text{obt}}(t) - x_i^{\text{est}}(t)
\end{equation}

In some cases, the error in state estimation can be significantly large, and plant associated with flow $F_i$ may require a lower sampling period. Since all nodes are not aware of the state estimation errors for all flows, a sensor that monitors the plant (with large errors) must broadcast the state estimation error to all nodes. To enable this communication, we reserve $\nu$ (where $\nu << \tau$) number of time slots in every $\tau$ time slots (typically the last $\nu$ slots). A sensor with a positive error can use this $\nu$ time slots to transmit the error using CSMA-CA protocol. We choose $\xi (t) > 0$ as the condition to transmit the state estimation error for its simplicity. However, the system designer is free to choose any policy. In practice, the random noise should not be high for all plants at the same time, and hence, a small value of $\nu$ should be sufficient to communicate the errors without collision. Note that CSMA-CA protocol is only used within the $\nu$ time slots to transmit the observed states of the plants with high state estimation errors. All other communications happen using the autonomous TDMA scheduler proposed in Section \ref{online_sec:scheduler} and their real-time performance is not affected by these $\nu$ time slots. For the schedulability test, the reserved slots can be considered as a flow with period $\tau$, execution requirement of $\nu$, and blocking time of $0$.

Here, all sensors receive the state space equations, initial states, and control inputs of all plants from the controller. Since the initial state, control input calculation, and noise calculation for all plants are the same at all sensors, all nodes converge to the same the state estimation for all plants. 

\subsection{Segregating the Flows}
Once the state estimation for current time slot ($t$) is estimated, a similar approach as mentioned in Section \ref{online_sec:stateEstimation} can be used to estimate the state at $t + \tau^{\text{th}}$ time slot using the current sampling period. Given the state of the system $x(t + \tau)$ and $x(t)$, we compute the Lyapunov function for the linear controller at time $t$ and $t + \tau$ (given by $V(t)$ and $V(t + \tau)$, respectively) using Equation (\ref{online_eq:lyafn}). We use the summation over the interval $\tau$ to accurately represent the trend of the Lyapunov function.
\begin{equation} 
\label{online_eq:lyafn}
	V(t) = \sum_{z = t - \tau}^{t} x_i^T(z) P_i^{L} x_i(z)
\end{equation}

We want $V(t)$ to be non-increasing in time for stability. To ensure a decreasing Lyapunov function, we place a tighter restriction on $V(t)$, given by Equation (\ref{online_eq:segineq}).
\begin{equation} \label{online_eq:segineq} V_i(t + \tau)  \le \beta V_i(t) \end{equation}

If the above equation is satisfied with the initial sampling period or the last used sampling period, then there is a steady decrease in the Lyapunov function, and the plant should be reaching the final state. In these situations, the sampling period can be higher than the last used sampling period or the initial sampling period.  However, if the above equation is not satisfied, then it implies that the plant may experience some noise which can affect the stability of the system. In this case, the sampling period of the system should be equal to or lower than the initial sampling period. 

Given the two outcomes of Inequality (\ref{online_eq:segineq}), the plants can be segregated into two categories: plants with an increase in the sampling period and plants with a decrease in the sampling period. The sampling period is first computed for all plants with an increase in the sampling period and followed by all plants with a decrease in the sampling period. The approach to increase and decrease sampling period is discussed in the following sections.

For the simplicity in explanation, we use a uniform value of $\beta$ in this chapter. However, the system designer can choose to use multiple values of $\beta$ depending on the state of the plant. One example is to choose a value of $\beta$ close to $0$ when the system is in an unstable region or approaching a stable region, and close to $1$ when the system is within a stable region. Note that our approach only provides a best-effort approach to ensure the stability and closed loop performance of the plants. We choose to avoid imposing restrictions on the plant model to ensure stability as the restrictions reduce the applicability of the proposed approach. 

\subsection{Increasing Sampling Period}
When the state of a plant is within (or close to) a stable region, then the sampling period of that plant can be increased beyond the initial sampling period. However, a substantial increase in the sampling period may cause the state of the plant to move to an unstable region. Therefore, we formulate the increase in sampling period as an optimization problem. The objective of the problem is to find the maximum sampling rate, which satisfies the constraint specified in Equation (\ref{online_eq:segineq}). Equation (\ref{online_eq:maxTopt}) formulates the sampling period maximization problem. 
\begin{equation}
\label{online_eq:maxTopt}
\begin{aligned}
& {\text{maximize}}
& & T_i \\
& \text{subject to}
& & V_i(t + \tau) \le \beta V_i(t) \\
\end{aligned}
\end{equation}

Our heuristic for solving the optimization problem is to use multiplicative increase and additive decrease of $T_i$ (starting from the initial sampling period) and select the maximum $T_i$ that satisfies the constraint. Note that the sampling period $T_i[\max]$ or $T_i[\min]$ is required when the plant is in an unstable region and moving towards a stable region. When the plant is within a stable region, sampling the state of the plant at $T_i[\max]$ can cause energy overheads. When the plant is in a stable region, the plant can be sampled at a period higher than the $T_i[\max]$ while ensuring the plant does not enter an unstable region. Thus, we use $\tau$ as an upper bound of $T_i$ instead of $T_i[\max]$.

Since all flows in the network are schedulable with the initial period assignment, any increase in the period will only lower the utilization and does not affect the schedulability of the system.  Thus, we do not consider schedulability as a constraint for the optimization problem.

\subsection{Decreasing Sampling Period}
The decrease in the sampling period of a plant affects the schedulability of a system. However, it does not affect the stability of a plant. Therefore, to counteract the plant noise for the next $\tau$ time slots, each node first assigns initial period to the flow. Since the initial period is within the $T_i[\min]$ and $T_i[\max]$, the sampling period should always ensure stability. After the first assignment, a node computes the total utilization for all flows in the network. If the total utilization is no greater than $1$, then periods can be decreased from the current period assignment. The available slack in the utilization (i.e., $1 - \sum_{i =1}^{n} \mu_i$) is divided equally among all the flows with a decrease in the sampling period. For example in a system with two flows with decreasing periods and the total utilization with $T_i[\max]$ assignment is $0.8$, the utilization of each flow ($\frac{\gamma}{T_i}$) is incremented by $0.1$.

%% file: online_period_selection/evaluation.tex
\section{Evaluation} \label{online_sec:evaluation}
We have evaluated the proposed sampling period acclimation (SPA) algorithm and autonomous scheduler against the holistic controller proposed in~\cite{ma2018efficient}. The holistic controller computes the maximum increase in the sampling period that ensures a decreasing Lyapunov function. We have compared the two approaches in terms of average energy consumed per node in the network.  We define \textit{energy consumption} as the average energy consumed per node in milli-Joules required to transmit the state and control information from sensor to actuator. We used the product of the number of transmissions and energy consumed for each communication (0.22mJ with 5.5ms Tx time) to estimate energy consumption per node. We also show the system response, i.e., the state of the system and control cost over time under SPA algorithm.

We executed the proposed SPA algorithm on Matlab and the proposed autonomous scheduler on TOSSIM~\cite{tossim}. Our setup consisted of 5 plants. The state space equation of a plant is expressed as $x(t + 1) = Ax(t) + Bu(t) + w$, where $A = 1.05$, $B = 1$, $x(0) = 50$, and $w = 0.05$. For simplicity, we assume all plants have the same coefficients. For a fair comparison between the proposed approach and holistic controller, we choose a linear controller with $K = -0.95$. The controller satisfies the Lyapunov stability constraint, expressed in Equation~\ref{online_eq:lyconstraint}, where $P^{L} = 1$, and $Q^L = 3$. However, the proposed approach can also work with other controllers like the model predictive controller.  We choose two values of $\beta$. When the state of the system is above $1$, $\beta = 0.5$, and when state of the system is below $1$, $\beta = 1$. The initial sampling period of all the plants was fixed at $100ms$. Since the period was selected to be $100ms$, we selected $\tau = 500ms$ since it allows for a few iterations between the initial assignment and first sampling period change.

\begin{figure}[t]
    \centering
    \begin{subfigure}[b]{0.35\textwidth}
        \includegraphics[width=\textwidth]{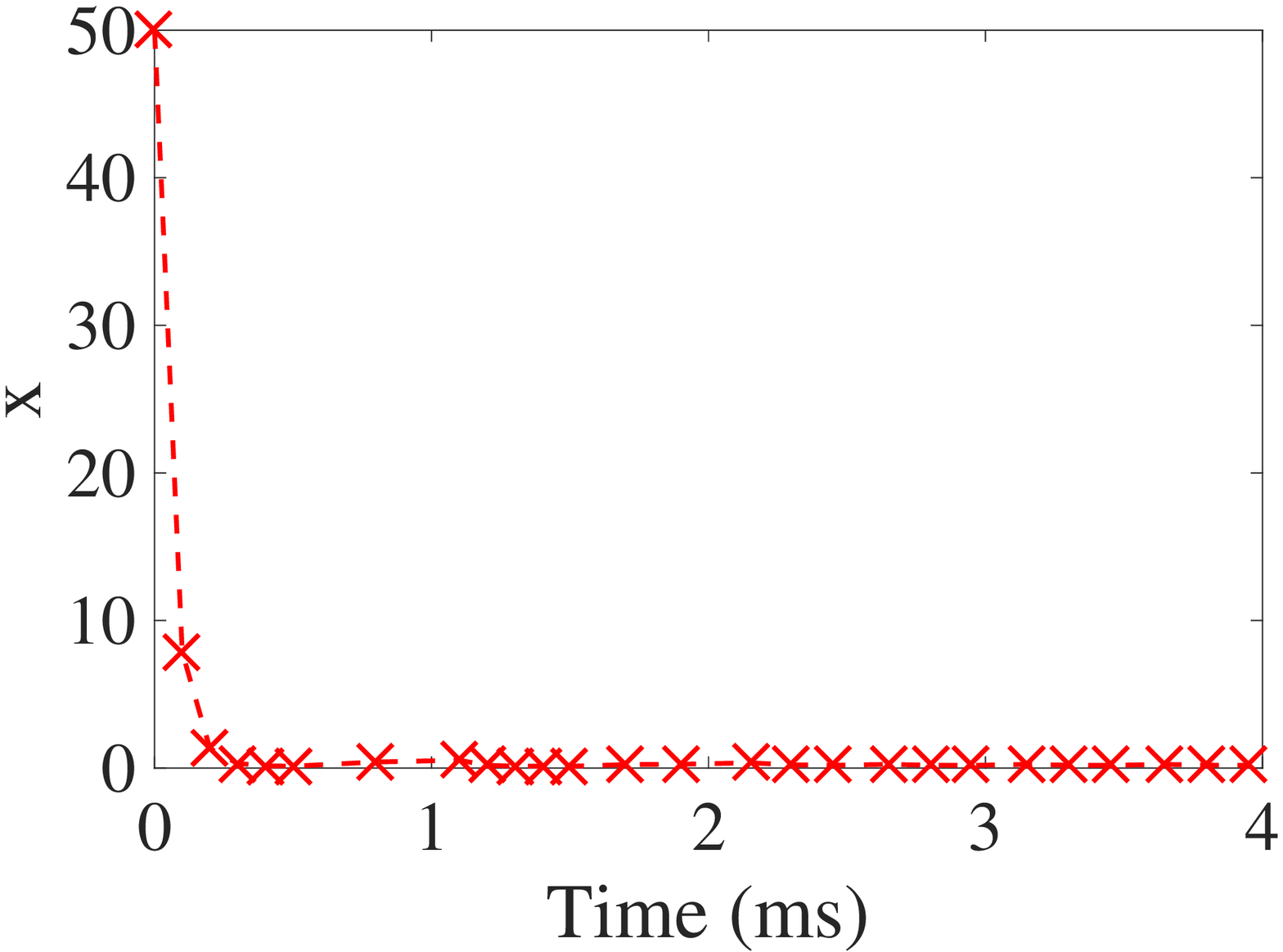}
        \caption{State of a plant}
        \label{online_fig:state}
    \end{subfigure}
    \quad
    \begin{subfigure}[b]{0.35\textwidth}
        \includegraphics[width=\textwidth]{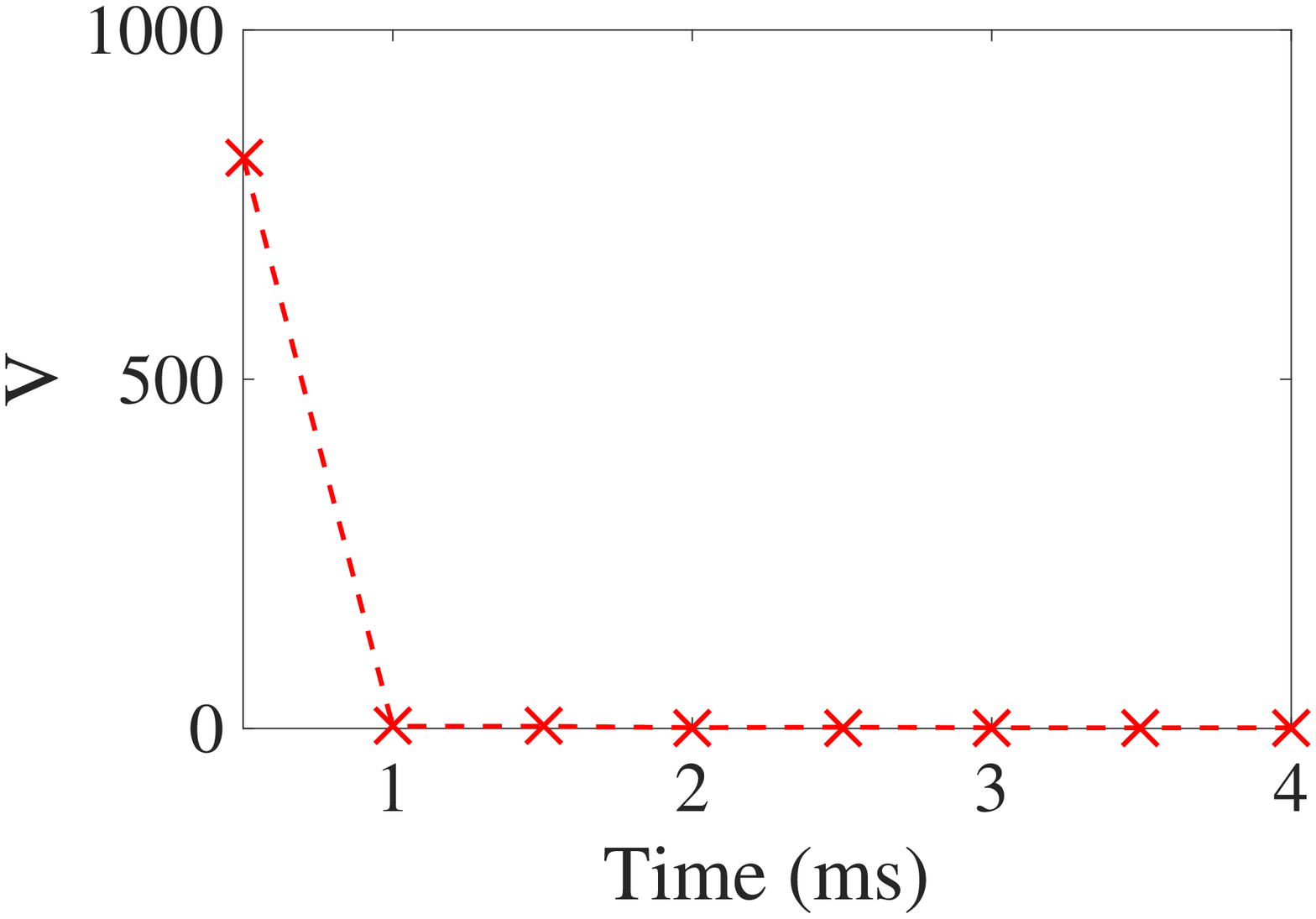}
        \caption{Control cost}
        \label{online_fig:controlcost}
    \end{subfigure}
   \caption{System Response under SPA}
   \label{online_fig:state_control_cost}
\end{figure}

We used a star topology of $11$ nodes; $5$ sensors, $5$ actuators, and $1$ controller, with the controller in the center of the star topology. We used the proposed autonomous TDMA medium access protocol.  Time in the network was divided into $10ms$ time slots. For a fair comparison between the proposed approach and the holistic controller, we assume the link quality to be good for all links, and each node uses $2$ time slots to make a successful transmission. The MAC protocol of both the proposed approach and the holistic controller is similar to that in WirelessHART without graph routing and redundant transmisisons. Since the link quality was assumed to be good, the utilization of each flow was given by $\mu_i = \frac{2}{T_i}$, and the schedulability test was given by $\sum_{i} \mu_i \le 1$.

Fig.~\ref{online_fig:state_control_cost} shows the control cost and state of the system as time increases from $10ms$ to $4s$. We have observed that under SPA, there was a steady decrease in the control cost as well as the state of the system. We have observed that the stabilized around $0.8$ due to the presence of a fixed noise. We have observed that at $0.5s$, the state of the system was below $1$, and the control cost was close to $800$. Since the control cost was so high, SPA algorithm selected a period of $300ms$ for all flows.  The high period of $300ms$ resulted in a significant accumulation of noise. Due to the high noise, we observed that the plant could not be stabilized by any period greater than $100ms$. Thus, there is a sharp decrease in the period at $1s$. A similar effect was observed from $1.5s$ to $2s$, which accounts for the sharp increase and decrease in the period. After $2s$, we observed that the system was in a stable state, and control cost remained constant, and hence, the period was constant.  

\begin{figure}[t]
    \centering
    \begin{subfigure}[b]{0.35\textwidth}
        \includegraphics[width=\textwidth]{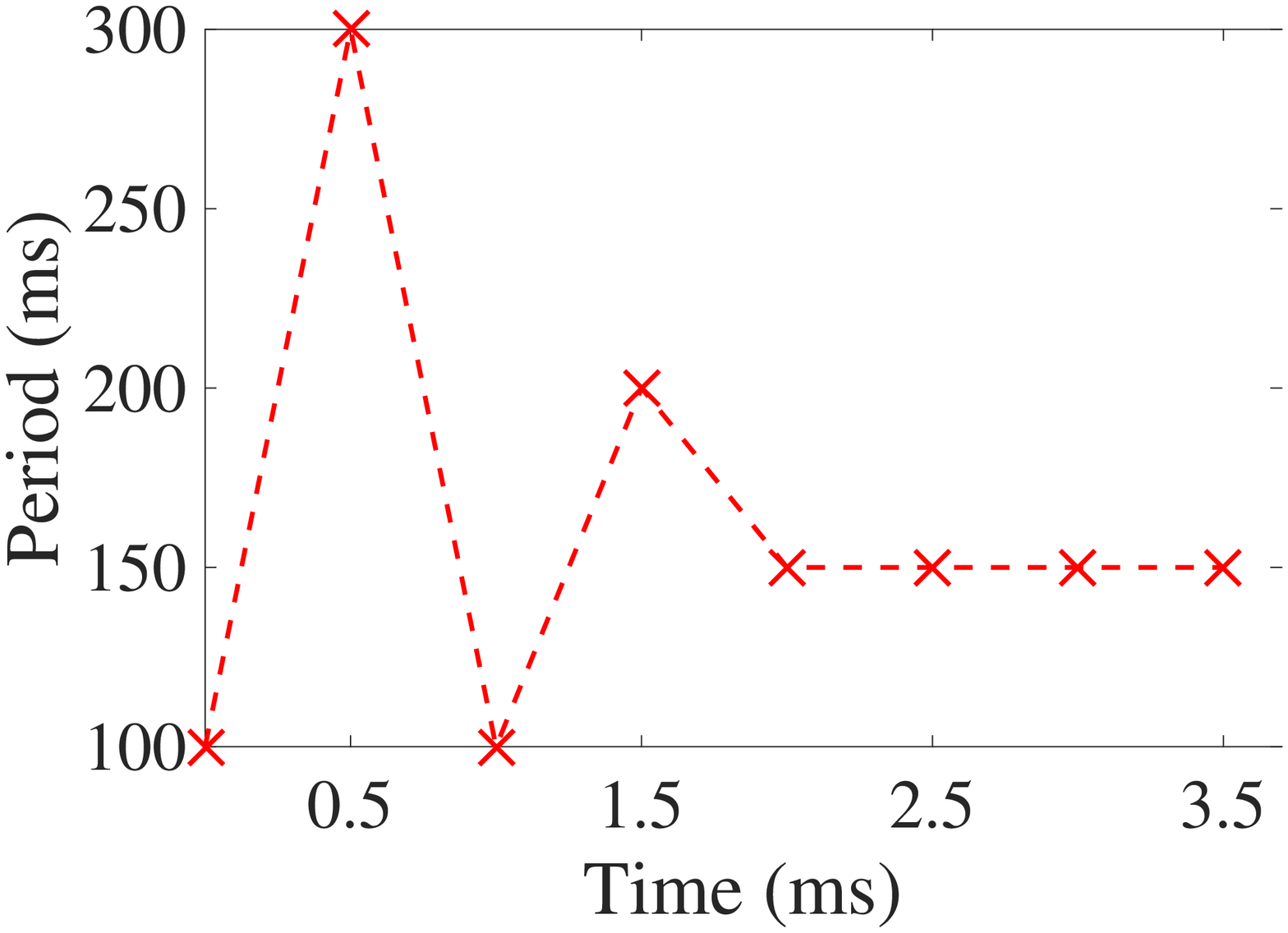}
        \caption{Period selection by SPA}
        \label{online_fig:period}
    \end{subfigure}
    \quad
    \begin{subfigure}[b]{0.35\textwidth}
        \includegraphics[width=\textwidth]{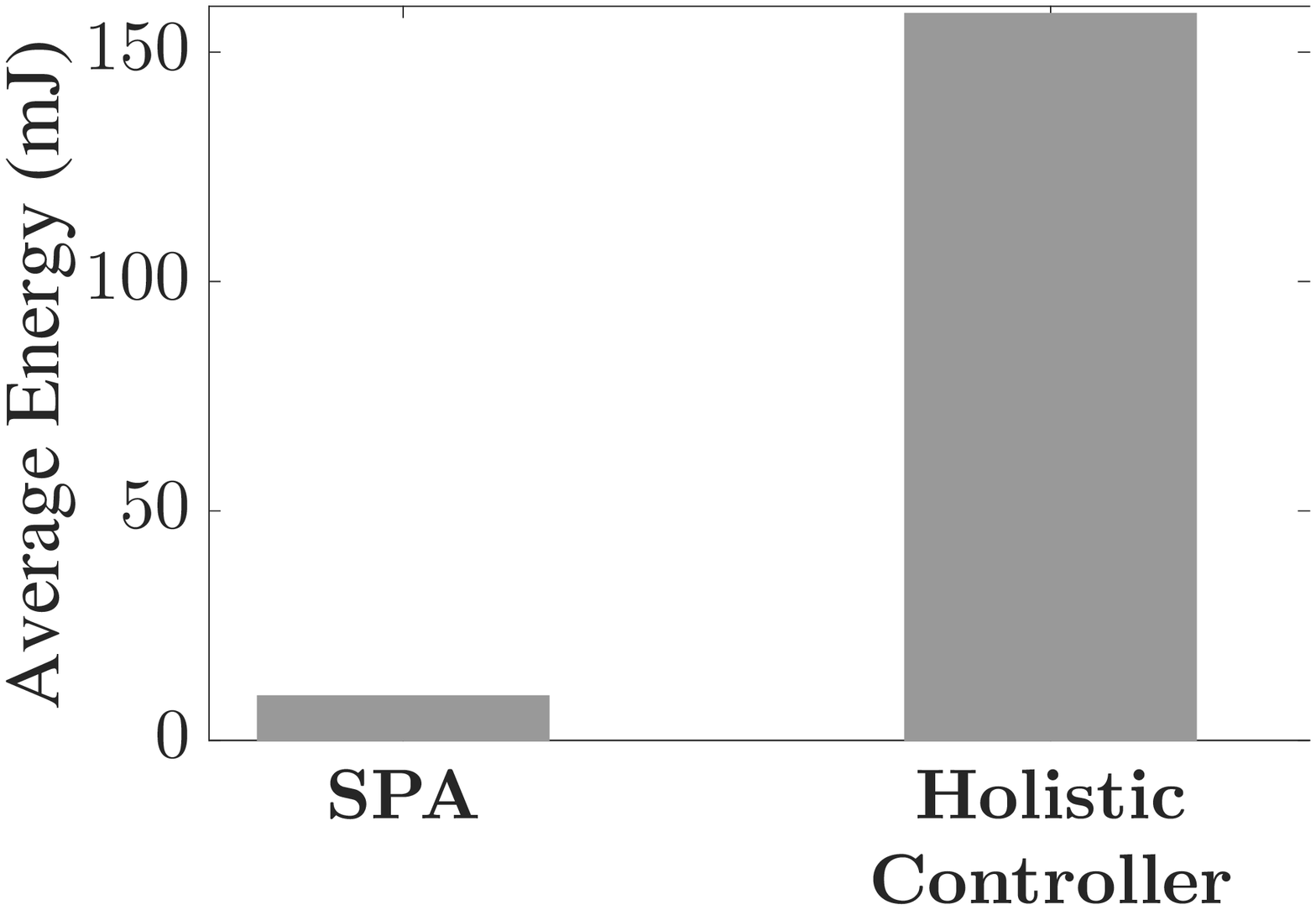}
        \caption{Average energy consumption}
        \label{online_fig:energy}
    \end{subfigure}
   \caption{Performance of SPA}
   \label{online_fig:energy_period}
\end{figure}

For a fair comparison with the holistic controller, we assume the holistic controller also changes the period every $\tau = 500ms$, and a similar cost metric was used. Note that, in both the approaches, the future control cost ($V(t + \tau)$) is intended to decrease by a fraction of the previous control cost. Since the holistic controller can only increase the period in multiples of $2$, the closest value of the period is selected and used for the schedule. The holistic scheduler uses a flooding mechanism to transmit packets in the network, and hence the average energy consumption at the nodes is around $158.4 mJ$ after a $4s$ interval. However, in the proposed approach,  the average energy consumption is $9.68 mJ$, which is a $92\%$ decrease in energy consumption when compared to the holistic controller. If the holistic controller used a single hop communication without any flooding, the average energy consumption is $23.7 mJ$, which shows that SPA consumes $59\%$ less energy. These results show that SPA algorithm minimizes the energy consumption by at least $59\%$ when compared to the holistic controller while ensuring a decrease in control cost.

%% file: conclusion.tex
\chapter{Conclusion}\label{ch:conclusion}
Industrial Internet-of-Things (IoT) facilitate low-power, flexible, and cost-efficient communication for a broad range of industrial applications. Communication in industrial IoT must meet stringent real-time performance and reliability requirements of the applications. Existing works on real-time wireless networks use centralized wireless stack design. Such a design consumes high energy, is less scalable, and is less suitable for handling changes to the network or operating conditions of the application. To address these challenges, the following contributions are made in this dissertation: (1) A scalable and distributed routing algorithm for industrial IoT which generates graph routes with a high degree of redundancy, (2) A local and online scheduling algorithm that is scalable, energy-efficient, and supports network/workload dynamics while ensuring reliability and real-time performance, (3) A fast and efficient test of schedulability that determines if an application meets the real-time performance requirement for given network topology, and (4) A distributed scheduling and control co-design that balances the control performance requirement and real-time performance for industrial IoT. Future work is planned to extend real-time scheduling for a low-power wide-area network. 

The scalable and distributed routing algorithm for industrial IoT is an adaptation of the Bellman-Ford algorithm. We have implemented our protocol in TinyOS and evaluated its effectiveness through both experiments and TOSSIM simulations. Our simulation results using TOSSIM show that it consumes about 86.4\% less energy with 66.1\% reduced convergence time at the cost of 1kB of additional memory compared to the state-of-the-art centralized approaches, thereby, to our knowledge, demonstrating it as the first practical distributed reliable graph routing for WirelessHART networks.

We have proposed DistributedHART - a local and online real-time scheduling system for industrial IoT. DistributedHART enables local scheduling at nodes through time window allocation at the nodes. Within a time window, a node locally schedules a packet for transmission based on any existing real-time scheduling policy. Thus, DistributedHART obviates the need of creating and disseminating a central global schedule thereby reducing resource waste and enhancing scalability. In addition, DistributedHART would lead to higher network utilization in the network at the expense of a slight increase in the end-to-end delay of all control loops. Through experiments on a 130-node testbed as well as large-scale simulations, we observe at least 85\% less energy consumption in DistributedHART compared to existing centralized approach. The performance of our schedulability test suggests that there is still room for improvement. In the future, we will derive an improved schedulability test by deriving tighter delay bounds. Nevertheless, in DistributedHART, we use existing centralized or distributed routing protocols to generate routes between sensors, actuators and controller. Although the routing is distributed, the schedule generated is not optimized for the routing protocol. On the contrary, our schedule generation is based on random vertex coloring. A distributed joint routing and scheduling algo- rithm is highly challenging for WirelessHART networks and we plan to address this problem in the future.

We have proposed the first latency minimizing in-band integration of  multiple SNOWs. We have implemented the proposed integration on SNOW hardware platform and observed, through physical experiments, up to $50\%$ decrease in network latency compared to the existing approach. We have also performed simulations in NS-3 and have observed up to $84\%$ decrease in network latency in an integrated network of  5000 nodes (of 5 SNOWs) based on our approach compared to the existing approach. In the future, we plan to extend this work to support a closed loop communication between sensors and actuators for industrial IoT. we plan on leveraging on edge-based control of plants in wide area deployments to provide a low-latency control for applications while also collecting data to a central location for monitoring the control performance of the entire system.

We have developed a schedulability analysis based on \emph{utilization bound} for multi-hop wireless networks. This approach determines the maximum total utilization of all flows in the network and determines those as \emph{schedulable} if the total utilization does not exceed the maximum possible utilization in the network. Because of its extremely low runtime overhead, a utilization-based schedulability test is considered one of the most efficient and effective schedulability tests. This work is the inception of a new horizon on utilization-based analysis for industrial IoT, which can direct the wireless community in the same way the real-time systems research today evolved from Liu and Layland's utilization bound.  Our result can trigger many research directions in the line of real-time scheduling, scheduling-control co-design, control performance optimization, routing, priority assignment, and mixed-criticality real-time wireless sensor and actuator networks. Our future work involves analyzing the effects of assigning sub-deadlines for large networks, packet loss, and the trade-offs among various control performance metrics.

We have proposed a distributed scheduling and control co-design that balances the control performance requirement and real-time performance for industrial IoT. We have proposed an autonomous scheduler which maximizes the number of successful transmission in the network using game theory. In the proposed autonomous scheduler, each node avoids collisions of packets and transmissions under poor link quality.   We have also derived an efficient utilization based schedulability test to determine the schedulability of all flows on the network. We have then proposed a heuristic for sampling period acclimation in wireless control systems. Our evaluation on a case study showed that the proposed approach consumes at least $59\%$ less energy when compared to the state-of-the-art approach. Our autonomous scheduler and sampling period selection algorithm are proposed for a single hop network. In the future, we plan to extend these algorithms for multi-hop networks.

%% file: references/references.tex
\addcontentsline{toc}{chapter}{References}
\bibliographystyle{abbrv}
\bibliography{references/references_final}

%% file: abstract.tex
\addcontentsline{toc}{chapter}{Abstract}
\begin{center}
 {\bf ABSTRACT}
 \\ {\bf REAL-TIME CONTROL OVER WIRELESS NETWORKS} \\
 by \\ {\bf VENKATA PRASHANT MODEKURTHY} \\ {\bf August 2020}

\end{center}
\vspace{0.1in}
\noindent {\bf Advisor:} Dr. Abusayeed Saifullah \\
\noindent {\bf Co-Advisor:} Dr. Sanjay Madria \\
\noindent {\bf Major:} Computer Science \\
\noindent {\bf Degree:} Doctor of Philosophy\\

Industrial internet of Things (IIoT) are gaining popularity for use in large-scale applications such as oil-field management (e.g., $74\times 8$km$^2$ East Texas Oil-field), smart farming, smart manufacturing, smart grid, and data center power management. These applications require the wireless stack to provide a scalable, reliable, low-power and low-latency communication. To realize a predictable and reliable communication in a highly unreliable wireless environment, industrial wireless standards use a centralized wireless stack design. In a centralized wireless stack design, a central manager generates routes and a communication schedule for a multi-channel time division multiple access communication (TDMA) based medium access control (MAC). However, a centralized wireless stack design is highly energy consuming, not scalable, and does not support frequent changes to networks or workloads.
To address these challenges, the following contributions are made in this dissertation: (1) A scalable and distributed routing algorithm for industrial IoT which generates graph routes, which offer a high degree of redundancy, (2) A local and online scheduling algorithm that is scalable, energy-efficient, and supports network/workload dynamics while ensuring reliability and real-time performance, (3) An approach to minimize latency for in-band integration of multiple low-power networks, (4) A fast and efficient test of schedulability that determines if an application meets the real-time performance requirement for given network topology, and (5) A distributed scheduling and control co-design that balances the control performance requirement and real-time performance for industrial IoT.